\documentclass[12pt,a4paper]{article}

\usepackage{lmodern}[lmr]

\usepackage[hidelinks]{hyperref}

\usepackage{physics, amsfonts, amsmath, amssymb, amsthm, color, enumitem, float, graphicx, epstopdf, lipsum, mathrsfs, mathtools, multirow, nccmath, nonfloat, pdflscape, rotating, tocloft}

\usepackage[margin=0.70in]{geometry}

\usepackage[round]{natbib}
\usepackage[doublespacing]{setspace}
\usepackage[table]{xcolor}

\def\be{\begin{equation}}
\def\ee{\end{equation}}
\def\bea{\begin{eqnarray}}
\def\eea{\end{eqnarray}}

\author{}
\title{}

\DeclareMathOperator*{\argmin}{\arg\!\min}

\DeclareMathOperator*{\plim}{p\!\lim}

\DeclareMathOperator*{\sgn}{\normalfont\textrm{sgn}}

\renewcommand{\arraystretch}{1.2}

\begin{document}
\newcommand\blfootnote[1]{
\begingroup
\renewcommand\thefootnote{}\footnote{#1}
\addtocounter{footnote}{-1}
\endgroup
}

\newtheorem{corollary}{Corollary}
\newtheorem{definition}{Definition}
\newtheorem{lemma}{Lemma}
\newtheorem{proposition}{Proposition}
\newtheorem{remark}{Remark}
\newtheorem{theorem}{Theorem}
\newtheorem{assumption}{Assumption}
\newtheorem{example}{Example}

\numberwithin{corollary}{section}
\numberwithin{definition}{section}
\numberwithin{equation}{section}
\numberwithin{proposition}{section}
\numberwithin{remark}{section}
\numberwithin{theorem}{section}

\allowdisplaybreaks[4]

\begin{titlepage}

\begin{center}
{\Large \textbf{Robust Estimation and Inference for \\High-Dimensional Panel Data Models}\blfootnote{Gao and Peng would like to acknowledge the Australian Research Council (ARC) for its financial support under Grant Number: DP250100063, and Peng would also like to thank the ARC for its financial support under DP210100476.  Liu's research was financially supported by National Natural Science Foundation of China under Grant Number 72203114. Yan acknowledges the financial support by the NSFC under the grant number 72303142 and the Fundamental Research Funds for the Central Universities under grant numbers 2022110877 and 2023110099. The authors contributed equally to this paper and are credited in alphabetical order.

$^*$Department of Econometrics and Business Statistics,  Monash University. 

$^\dag$School of Finance, Nankai University. 

$^\ddag$School of Statistics and Management, Shanghai University of Finance and Economics. }
\bigskip}

\bigskip

{\sc Jiti Gao$^{\ast}$} and {\sc Fei Liu$^{\dag}$}  and {\sc Bin Peng$^{\ast}$} and {\sc Yayi Yan$^{\ddag}$}

\bigskip

\today

\end{center}

\begin{abstract}
This paper provides the relevant literature with a complete toolkit for conducting robust estimation and inference about the parameters of interest involved in a high-dimensional panel data framework.  Specifically, (1) we allow for non-Gaussian, serially and cross-sectionally correlated and heteroskedastic error processes, (2) we develop an estimation method for high-dimensional long-run covariance matrix using a thresholded estimator, (3) we also allow for the number of regressors to grow faster than the sample size. 

Methodologically and technically, we develop two Nagaev--types of concentration inequalities: one for a partial sum and the other for a quadratic form, subject to a set of easily verifiable conditions. Leveraging these two inequalities, we derive a non-asymptotic bound for the LASSO estimator, achieve asymptotic normality via the node-wise LASSO regression, and establish a sharp convergence rate for the thresholded heteroskedasticity and autocorrelation consistent (HAC) estimator. 

We demonstrate the practical relevance of these theoretical results by investigating a high-dimensional panel data model with interactive effects. Moreover, we conduct extensive numerical studies using simulated and real data examples.

\medskip
\medskip

	
\medskip
	

\end{abstract}

\end{titlepage}

\section{Introduction}\label{Sec1}

With the emergence of the data-rich world, numerous applications in many fields, including business and economics, focus on panel data regressions within a high-dimensional framework, where the number of variables can be very large and potentially exceed the sample size. For example, in the realm of asset pricing, academic researchers and financial analysts have put forth over five hundred risk factors and individual firm characteristics (which continue to grow) to explain the cross-sectional relationship of stock return (\citealp{feng2020taming}). Given the necessity for innovative statistical methodologies to dissect such data,  there is an increasing interest in discussing regularized estimation of high-dimensional (HD) panel data models (e.g., \citealp{vogt2022cce,babii2022machine,belloni2023high}). Nonetheless, achieving robust estimation and inference using HD panel data remains a challenge, particularly when idiosyncratic errors present non-Gaussianity, time series autocorrelation (TSA), and cross-sectional dependence (CSD). Our main goal is to enrich the relevant literature (e.g., \citealp{FLY2015, KP2019,gupta2023robust}) with a complete toolkit for conducting robust inference about the parameters of interest involved in HD panel data models associated with complex dependence structure.  

That said, in this paper, (i) we consider HD panel data models in which the idiosyncratic errors  exhibit TSA, CSD, heteroscedasticity, as well as heavy--tailed behaviors; (ii) we debias the least absolute shrinkage and selection operator (LASSO) estimator to achieve asymptotic normality, and explore the estimation of HD long--run covariance matrix using a thresholded estimator; and (iii) we accommodate such scenarios that the number of regressors may increase more rapidly than the sample size. 

To achieve these  goals with the aforementioned data features, we establish two Nagaev--types of concentration inequalities, one for partial sum and the other for quadratic form, subject to a set of easily verifiable conditions. Leveraging these two inequalities, we derive a non-asymptotic bound for the  LASSO estimator, achieve asymptotic normality via the nodewise LASSO regression, and obtain a sharp convergence rate for the thresholded heteroskedasticity and autocorrelation consistent (HAC) estimator. Furthermore, we considerably generalize the main ideas to investigate a class of widely used HD panel data models (e.g., \citealp{vogt2022cce, belloni2023high}) associated with interactive effects.

Up to this point, it is worth commenting on some key references, and outlining our contributions to the relevant literature accordingly. 

\begin{enumerate}[leftmargin=*, itemsep=0.5pt, parsep=0.5pt, topsep=0.6pt]
\item We establish a set of toolkit to analyze HD panel data based on some easily verifiable conditions, which also enriches the literature of HD time series analysis. For instance, the newly developed toolkit allows for the generalization or derivation of restrictions such as Assumption 3.1 in \cite{babii2022machine} and Assumptions 4 and 5 in \cite{gupta2023robust} from a set of fundamental conditions.

\item \cite{vogt2022cce} examine a HD panel data model with interactive effects, and present error bounds for penalized estimators under the cross-sectional independence condition. In a similar vein, \cite{belloni2023high} delve into quantile regression within a comparable framework. In this context, cross-sectional independence remains crucial for establishing the corresponding asymptotic properties, and a valid inference procedure is absent. Building on this line of research, we respectively investigate HD panel data regression with and without a factor structure, and establish inference while accommodating non-Gaussian, serially and cross-sectionally correlated, and heteroskedastic error processes.

\item Under a Gaussian assumption, \cite{baek2023local} establish non-asymptotic error bounds for the estimation of a HD long-run covariance matrix. In their research, the Gaussian condition is employed to derive a concentration inequality for a quadratic form, controlling the maximum deviation of each element in the HD long-run covariance estimator from its true value. Similarly, \cite{BAI2020} explore the use of a thresholded HAC estimator to achieve valid inference for a fixed-dimensional panel data model under a sub-Gaussian assumption, and \cite{gao2023higher} propose a dependent wild bootstrap method to deal with cross--sectional dependence for the fixed--dimensional setting.  Our contribution in this paper involves developing two general Nagaev--types of concentration inequalities for quadratic forms under the context of HD panel data, encompassing heavy-tailed behaviours with temporal and cross-sectional correlation. 

\item  \cite{adamek2023lasso} and \cite{babii2022machine} develop inferential procedures for HD regression via HAC estimators, and they show spectral norm consistent only if the dimensionality of regressors is much smaller than the number of observations. Furthermore, the pooled HAC estimator proposed by \cite{babii2022machine} is not robust in the presence of CSD. Our contribution to this area of study involves proposing a thresholded HAC estimator with a sharp rate, allowing for valid inference in the HD setting where the dimensionality of regressors can be much larger than the number of observations.
\end{enumerate}

The rest of this paper is organized as follows.  Section \ref{Sec2} presents the main results of this paper, and Section \ref{Sec3} showcases their practical relevance by considering a HD panel data model with interactive effects. Section \ref{Sec4} conducts Monte Carlo simulations to examine the theoretical findings. Section \ref{Sec5} consists of the empirical results by studying a HD asset pricing model. Section \ref{Sec6} concludes. Appendix \ref{AP.A1} verifies several assumptions in the paper. A numerical implementation procedure is presented in Appendix \ref{AP.A2}. Many other technical details are given in Appendix \ref{AP.A3}--\ref{AP.A5} and online Appendices B and C.

Before proceeding further, we introduce a set of notations which are repeatedly used throughout. For any $\mathbf{x} \in \mathbb{R}^{n}$, $|\mathbf{x}|_p\coloneqq  (\sum_{i=1}^{n}|x_i|^p)^{1/p}$, in which $x_i$ stands for the $i^{th}$ element of $\mathbf{x}$; for a random vector $\mathbf{v}$, $\|\mathbf{v}\|_p\coloneqq (E|\mathbf{v}|_{p}^p )^{1/p}$; $|\cdot|$ denotes the absolute value of a scalar or the cardinality of a set; for an $m\times n$ matrix $\mathbf{A} = (A_{ij})_{ i\leq m, j\leq n}$, $|\mathbf{A}|_{p}$ denotes the induced $\ell_p$ matrix norm such that  $|\mathbf{A}|_{p} \coloneqq \max_{\mathbf{x} \neq \mathbf{0}}|\mathbf{A}\mathbf{x}|_p/|\mathbf{x}|_p$; $|\mathbf{A}|_{\mathrm{max}} \coloneqq\max_{ i\leq m,j\leq n}|A_{ij}|$ denotes the matrix element-wise max norm; $|\mathbf{A}|_F$ denotes the Frobenius norm; $|\mathbf{A}|_{*} \coloneqq \sum_{k=1}^{\min(m,n)}\psi_{k}(\mathbf{A})$ denotes the nuclear norm, where $\psi_{k}(\cdot)$ is the $k^{th}$ largest singular value of a matrix;  $\psi_{\mathrm{max}}(\mathbf{A})$ and $\psi_{\mathrm{min}}(\mathbf{A})$ denote the largest and smallest singular values of $\mathbf{A}$, respectively; for two conformable matrix $\mathbf{A}$ and $\mathbf{B}$, let $\mathbf{A}\circ \mathbf{B}$ be the Hadamard product; $\mathbf{I}_a$ stands for an $a\times a$ identity matrix; let $\mathbf{1}_a$ stand for a $a\times 1$ vector of ones; let $\to_P$, $\to_D$ and $=_D$ denote convergence in probability, convergence in distribution and equal in distribution, respectively;  for two constants $a$ and $b$, $a \asymp b$ stands for $a=O(b)$ and $b=O(a)$; for a given positive integer $L$, let $[L]=\{1,2,\ldots, L\}$.

\section{Methodology \& Asymptotic Theory}\label{Sec2}

We start this section by introducing a set of HD panel data. In what follows, $i\in [N]$ and $t\in [T]$ always index individuals and time periods respectively, and both $(N,T)$ may diverge. Additionally, let $j\in [d]$ index the variables, and $d$ may diverge along with $(N,T)$. Accordingly, for a given time period $t\in [T]$ and a given variable $j\in [d]$, we denote the following $N\times 1$ vector:
\begin{eqnarray*}
\mathbf{u}_j(\mathcal{F}_t)= [u_{j1}(\mathcal{F}_t),\ldots, u_{jN}(\mathcal{F}_t)]^\top\quad\text{with}\quad \mathcal{F}_{t} = (\pmb{\varepsilon}_t,\pmb{\varepsilon}_{t-1},\ldots),
\end{eqnarray*}
where   $\{\pmb{\varepsilon}_t\}_{t \in \mathbb{Z}}$ are a sequence of independent and identically distributed (i.i.d.) random vectors, and each $u_{ji}(\cdot)$ is a measurable function. Define the coupled version of $\mathcal{F}_t$ in order to measure dependence:
$\mathcal{F}_t^*= (\pmb{\varepsilon}_t, \ldots ,\pmb{\varepsilon}_{1},\pmb{\varepsilon}_{0}^*,\pmb{\varepsilon}_{-1},\ldots)$, where $\pmb{\varepsilon}_{0}^*$ is an i.i.d. copy of $\{\pmb{\varepsilon}_t\}$.  Let us introduce Assumption 1 below.

\begin{assumption}\label{ASSUMPTION1}
Let $E(\mathbf{u}_j(\mathcal{F}_0)) = 0$ for $\forall j\in [d]$. There exist constants $\alpha > 2$ and $q > 2$ such that 
$$\max_{j\in [d]}\sup_{|\mathbf{w}|_2<\infty}\|\mathbf{w}^\top[\mathbf{u}_j (\mathcal{F}_t) -\mathbf{u}_j (\mathcal{F}_t^*)] \|_q = O(t^{-\alpha}).$$
\end{assumption}

Assumption \ref{ASSUMPTION1} generalizes the functional dependence measure of \cite{wu2005nonlinear} to a panel data setting, allowing for a wide range of data generating processes (DGPs). For example, one often chooses $\mathbf{w} = (1,0,\ldots, 0)^\top$ and $\mathbf{w} = (\frac{1}{\sqrt{N}},\ldots, \frac{1}{\sqrt{N}})^\top$ in some settings (e.g., \citealp[Eq. (66)]{Pesaran2006}; \citealp[Assumption C.iii]{bai2009panel}). This condition regulates the decay rate of temporal dependence after taking cross--sectionally weighted averages. By doing so, $u_{ji}(\mathcal{F}_t)$'s are allowed to exhibit cross--sectional dependence of various unknown forms. In Example \ref{Exam1} of Appendix \ref{AP.A1}, we show the flexibility of Assumption \ref{ASSUMPTION1}.  

Assumption \ref{ASSUMPTION1} is sufficient for us to establish two Nagaev--types of sharp concentration inequalities, allowing for correlation over both dimensions and non-Gaussianity. These results will be used to study LASSO and HD long-run covariance estimation respectively, and they are also generally applicable to deal with some other scenarios.

\begin{lemma}\label{LEMMA1}
Let Assumption \ref{ASSUMPTION1} hold. 
\begin{itemize}
\item [1.] There exist constants $C_1>0$, $C_2>0$ and $C_3 > 0$ such that for $\forall j\in [d]$ and any $x \geq c_q\sqrt{NT}$ with some $c_q>0$,
{\small
\begin{eqnarray*}
\Pr\left(\max_{1\leq t\leq T}\left| \sum_{s=1}^{t}\sum_{i=1}^{N}u_{ji}(\mathcal{F}_s)\right|\geq x \right) \leq C_1\frac{TN^{q/2}}{x^q} + C_2 \exp\left( - C_3 \frac{x^2}{TN} \right).
\end{eqnarray*}}

\item [2.] Let $x = c_q N\sqrt{TM\ell \log(d)}$ for some $c_q>0$ and any $M>1$, and $\{a_k\}_{k=-\infty}^{\infty}$ be a sequence of non-negative constants satisfying $a_{|k|} = 0$ for $|k|>\ell$ with $\ell  = O(T^{\gamma})$ and $0<\gamma<1$.  Then for $\gamma < \theta < 1$ and $q > 4$, there exist constants $C_4>0$ and $C_5>0$ such that for $\forall j, j' \in [d]$
{\small
\begin{eqnarray*}
&&\Pr\left(\left|\sum_{t,s=1}^{T}\sum_{k,l=1}^{N} a_{t-s} [u_{jk}(\mathcal{F}_t)u_{j'l}(\mathcal{F}_s)-E(u_{jk}(\mathcal{F}_t)u_{j'l}(\mathcal{F}_s))]\right| \geq x \right)\\
&\leq&   \frac{C_4 \log T}{x^{q/2}} \left(\frac{(T\ell)^{q/4}}{T^{(\alpha -1)\theta q/2}}+T\ell^{q/2-1-(\alpha-1)\theta q/2}+T\right) + \frac{C_5 }{(d\wedge T)^{M}} .
\end{eqnarray*}}
\end{itemize}

\end{lemma}

Notably, $\mu_u$ is guaranteed to be fixed under Assumption \ref{ASSUMPTION1}, and $M$ can be a sufficiently large constant. With this lemma in hand, we are ready to present the benchmark model and investigate its estimation and inference.

\subsection{The Benchmark Model \& Its Estimation}

In this section, we consider the following benchmark model:
\begin{eqnarray}
\mathbf{y}_{t} = \mathbf{X}_{t}\pmb{\beta}_0  + \mathbf{e}_{t},\label{EQUATION1} 
\end{eqnarray}
where $\mathbf{X}_t$ is a $N\times d$ matrix such that
\begin{eqnarray*}
\mathbf{X}_{t}=[\mathbf{g}_1(\mathcal{F}_t),\ldots, \mathbf{g}_d(\mathcal{F}_t)]=[\mathbf{x}_{1t},\ldots, \mathbf{x}_{Nt}]^\top \quad\text{with}\quad\mathbf{x}_{it} = [x_{1,it},\ldots,x_{d,it}]^\top.
\end{eqnarray*} 
Accordingly, $\mathbf{y}_{t}=[y_{1t},\ldots, y_{Nt}]^\top$, $\mathbf{e}_{t}=[e_{1t},\ldots, e_{Nt}]^\top$, and $\pmb{\beta}_0=[\beta_{0,1},\ldots,\beta_{0,d}]^\top$. Moreover, we let $J\coloneqq\{ j\in [d]\mid \beta_{0,j} \neq 0\}$ and $J^c \coloneqq [d] \setminus J$, and accordingly denote $s\coloneqq |J| $ and $\beta_{\min}\coloneqq \min_{j\in J} |\beta_{0,j}|$. 

Model \eqref{EQUATION1} has a simple form for us to understand how dependence over both dimensions influences typical HD analysis, such as \cite{babii2022machine} and  \cite{baek2023local}. It is straightforward to allow addictive fixed effects by using the usual demean procedure. In Section \ref{Sec3}, we extend the main ideas to discuss a class of HD panel data models associated with interactive fixed effects, which as mentioned in the introduction has drawn attentions (e.g., \citealp{vogt2022cce,belloni2023high}, and references therein) in recent years.

Before proceeding further, we introduce some suitable regularity conditions on the regressors and errors permitting non-Gaussianity, TSA, CSD, and heteroscedasticity.  

\begin{assumption}\label{ASSUMPTION2}
\item
\begin{enumerate}[leftmargin=*, itemsep=0.5pt, parsep=0.5pt, topsep=0.6pt]

\item Let $\mathbf{g}_{j}(\mathcal{F}_t)\circ\mathbf{e}_{t} \eqqcolon \mathbf{u}_j(\mathcal{F}_t)$ satisfy
$$\max_{j,j' \in [d]} \sup_{|\mathbf{w}|_2<\infty}\|\mathbf{w}^\top[\mathbf{g}_j(\mathcal{F}_t) \circ \mathbf{g}_{j'}(\mathcal{F}_t) - \mathbf{g}_j(\mathcal{F}_t^*) \circ \mathbf{g}_{j'}(\mathcal{F}_t^*)]\|_{q} = O(t^{-\alpha}).$$

\item Let $\mathbb{V}\coloneqq\{\mathbf{v}\in \mathbb{R}^{d} \mid \mathbf{v}\ne \mathbf{0}, |\mathbf{v}_{J^c}|_1 \leq 3|\mathbf{v}_J|_1\}$, where $\mathbf{v}_J$ and $\mathbf{v}_{J^c}$ include elements of $\mathbf{v}$ indexed by $J$ and $J^c$ respectively. Suppose that $\psi_{\pmb{\Sigma}_x}(J) = \min_{ \mathbf{v}\in \mathbb{V}} \frac{\mathbf{v}^\top\pmb{\Sigma}_x \mathbf{v}}{\mathbf{v}^\top\mathbf{v}} > 0$, where $\pmb{\Sigma}_{x} = \frac{1}{N}\sum_{i=1}^{N}E(\mathbf{x}_{it}\mathbf{x}_{it}^\top)$.
\end{enumerate}
\end{assumption}

Note that Assumption \ref{ASSUMPTION2}.1 is similar to Assumption \ref{ASSUMPTION1}. Assumption \ref{ASSUMPTION2}.2 formulates the compatibility condition, and enriches the literature by regulating the singular values of the second moment of $\mathbf{x}_{it}$. Typically, the literature imposes such a condition on a sample covariance matrix (e.g., Assumption A2 in \citealp{chernozhukov2021lasso}). 

Note also that Assumption \ref{ASSUMPTION2} allows that the elements of $\mathbf{x}_{it}$ can be both serially dependent on the time--series dimension and cross-sectionally dependent on the cross--sectional dimension. While we do not consider the case where there are dynamic features involved in our panel data, it is possible to modify our method and theory to allow for strong cross--sectional dependence and lagged dependent variables in $\mathbf{x}_{it}$ under strict exogeneity conditions on the errors. We plan to leave such topics for future research. 

We now come back to investigate the model \eqref{EQUATION1}. The first two steps are about penalized estimation and variable selection. 

\smallskip

\hrule \smallskip
\begin{enumerate}[leftmargin=*, itemsep=0.5pt, parsep=0.5pt, topsep=0.6pt]
\item[] \textbf{Step 1}: Estimate $\pmb{\beta}_0$ via LASSO:
\begin{equation}\label{EQUATION2}
\widehat{\pmb{\beta}} = \argmin_{\pmb{\beta}\in\mathbb{R}^{d}} \frac{1}{2NT}|\mathbf{y} - \mathbf{X}\pmb{\beta}|_2^2 + \omega_1 |\pmb{\beta}|_{1},
\end{equation}
where $\mathbf{y}= [\mathbf{y}_1^\top,\ldots,\mathbf{y}_T^\top]^\top$, $\mathbf{X} = [\mathbf{X}_1^\top,\ldots,\mathbf{X}_T^\top]^\top$, and $\omega_1$ is a tuning parameter at the order $\omega_1\asymp \sqrt{\log (d)/(NT)}$.

\item[] \textbf{Step 2}: Update the estimate via weighted LASSO:
\begin{equation}\label{EQUATION3}
\widehat{\pmb{\beta}}_{\omega} = \argmin_{\pmb{\beta}\in\mathbb{R}^{d}} \frac{1}{2NT}|\mathbf{y} - \mathbf{X}\pmb{\beta}|_2^2 + \omega_1 \sum_{j=1}^ {d}g_{j}|\beta_j|,
\end{equation}
where $\{g_{j}\}$ are a sequence of predetermined weights and may depend on $\widehat{\pmb{\beta}}$.
\end{enumerate}
\smallskip
\hrule 

\smallskip

Given $g_j = \frac{1}{|\widehat{\beta}_j|}$ or $g_j = \frac{\omega^*}{\max(|\widehat{\beta}_j|,\ \omega^*)}$ with a threshold $\omega^*$, model \eqref{EQUATION3} becomes either an adaptive LASSO (\citealp{zou2006adaptive}) or a conservative LASSO (\citealp{caner2018asymptotically}). In connection with Lemma \ref{LEMMA1}, we are now able to establish the following results.

\begin{lemma}\label{LEMMA2}
Suppose that Assumptions \ref{ASSUMPTION1} and \ref{ASSUMPTION2} hold. Let $C_0,C_1,C_2$ be positive constants such that $s\le \frac{\psi_{\pmb{\Sigma}_x}(J) }{2C_0}\sqrt{\frac{NT}{\log d}}$ and $1 - C_1\Big(\frac{d T^{1-q/2}}{(\log d)^{q/2}} +  d^{-C_2}\Big)\eqqcolon C_{NT}>0$. Then the following two results hold:
\begin{enumerate}[leftmargin=*, itemsep=0.5pt, parsep=0.5pt, topsep=0.6pt]
\item $|\widehat{\pmb{\beta}}-\pmb{\beta}_0|_2 \leq \frac{8\sqrt{s}\ \omega_1 }{\psi_{\pmb{\Sigma}_x}(J)}$ and $|\widehat{\pmb{\beta}}-\pmb{\beta}_0|_1 \leq \frac{32s\ \omega_1 }{\psi_{\pmb{\Sigma}_x}(J) }$
with probability larger than $C_{NT}$.

\item Suppose further that $s\leq \frac{\psi_{\mathrm{min}}(\pmb{\Sigma}_{x}) }{2C_0}\sqrt{\frac{NT}{\log d}}$, $\beta_{\min} > \frac{\sqrt{s}\omega_1}{\psi_{\pmb{\Sigma}_x}(J)} (1 + 2\max_{j\in J}|g_j|)$ and $\min_{j\in J^c}|g_j| \geq  (\frac{2s|\pmb{\Sigma}_x|_{\mathrm{max}}}{\psi_{\pmb{\Sigma}_x}(J)} + \frac{\psi_{\pmb{\Sigma}_x}(J)}{16} ) (\frac{1}{2}+\max_{j\in J}|g_j| )$, where $s=|J|$ is the same as defined in equation (\ref{EQUATION1}). Then $\sgn(\widehat{\pmb{\beta}}_{w}) = \sgn(\pmb{\beta}_0)$
with probability larger than $C_{NT}$.
\end{enumerate}
\end{lemma}

Under the condition:
\begin{eqnarray}\label{EQUATION4}
\frac{dT^{1-q/2}}{(\log d)^{q/2}}\to 0,
\end{eqnarray}
\noindent Lemma \ref{LEMMA2} presents the non-asymptotic error bounds of Step 1, and implies consistent variable selection of Step 2. Notably, Lemma \ref{LEMMA2} does not require specific tail behavior, and model \eqref{EQUATION4} infers that $d$ may grow at a polynomial order of the sample size. 

If we impose further, for example, an exponential tail assumption, $d$ can even diverge at an exponential rate as in \cite{van2014asymptotically}, wherein an i.i.d. assumption and an exponential tail are adopted to facilitate the development. There is a trade-off between the tail behavior of the error components and the  divergence rate of $d$.

\subsection{Estimation and Inference}\label{Sec2.2}

To proceed, we start to approximate the following two matrices:

\begin{enumerate}[leftmargin=*, itemsep=0.5pt, parsep=0.5pt, topsep=0.6pt]
\item the inverse of $\mathbf{X}^\top\mathbf{X}/(NT)$,

\item $\frac{1}{NT}\sum_{i, j=1}^{N}\sum_{t, s=1}^{T}E [\mathbf{x}_{it}\mathbf{x}_{js}^\top e_{it} e_{js} ]\eqqcolon \pmb{\Theta}$.
\end{enumerate}

The second one is a HD long-run covariance matrix permitting heavy-tailed errors. In what follows, we approximate these two matrices using node-wise LASSO and thresholded HAC covariance estimation.

\smallskip

We further introduce some notation. For $\forall j \in [d]$, let $\mathbf{X}_j$ and $\mathbf{X}_{-j}$ respectively be the $j^{th}$ column of $\mathbf{X}$ and the sub-matrix of $\mathbf{X}$ with the $j^{th}$ column removed. Accordingly, we let
\begin{eqnarray*}
\pmb{\gamma}_{j} = [ E(\mathbf{X}_{-j}^\top\mathbf{X}_{-j})]^{-1}
E(\mathbf{X}_{-j}^\top\mathbf{X}_{j}) \eqqcolon \pmb{\Sigma}_{x,-j,-j}^{-1}\pmb{\Sigma}_{x,-j,j}\quad\text{and}\quad
\pmb{\eta}_j = \mathbf{X}_{j} - \mathbf{X}_{-j} \pmb{\gamma}_j,
\end{eqnarray*}
and it is easy to check that $|\pmb{\gamma}_{j}|_2 <\infty$ and $E(\mathbf{X}_{-j}^\top\pmb{\eta}_j) = \mathbf{0}$. We then run the following node-wise LASSO regression to debias the LASSO estimate and perform inference:
\begin{equation}\label{EQUATION5}
\widehat{\pmb{\gamma}}_{j} = \argmin_{\mathbf{b}\in\mathbb{R}^{d-1}}\frac{1}{NT}|\mathbf{X}_{j} - \mathbf{X}_{-j}\mathbf{b}|_2^2 + 2\widetilde{\omega}_j|\mathbf{b}|_1,
\end{equation}
where $\{\widetilde{\omega}_j\}_{j=1}^{d}$ are a sequence of penalty terms at the order $\widetilde{\omega}_j \asymp \sqrt{\log (d)/(NT)}$. Finally, we propose the debiased LASSO estimator in Step 3 below.

\smallskip

\hrule \smallskip
\begin{enumerate}[leftmargin=*, itemsep=0.5pt, parsep=0.5pt, topsep=0.6pt]
\item[] \textbf{Step 3}: The debiased LASSO estimator is 
\begin{eqnarray}\label{EQUATION6}
\widehat{\pmb{\beta}}_{bc} = \widehat{\pmb{\beta}} + \widehat{\pmb{\Omega}}_x\frac{\mathbf{X}^\top(\mathbf{y}-\mathbf{X}\widehat{\pmb{\beta}})}{NT},
\end{eqnarray}
where {\small $\widehat{\pmb{\Omega}}_x = \widehat{\mathbf{T}}^{-1}\widehat{\mathbf{C}}$, $\widehat{\mathbf{T}} = \mathrm{diag}(\widehat{\tau}_1^2,\ldots,\widehat{\tau}_{d}^2)$ with  $\widehat{\tau}_j^2 = \frac{1}{NT}|\mathbf{X}_{j} - \mathbf{X}_{-j}\widehat{\pmb{\gamma}}_{j}|_2^2 + \widetilde{\omega}_j|\widehat{\pmb{\gamma}}_{j}|_1$, and $\widehat{\mathbf{C}} = (\widehat{c}_{j,k})_{d\times d}$ with $\widehat{c}_{j,k} =-\widehat{\gamma}_{j,k} \mathbb{I}(j\ne k)+\mathbb{I}(j=k)$, and $\widehat{\gamma}_{j,k}$ being the $k^{th}$ element of $\widehat{\pmb{\gamma}}_{j}$}.
\end{enumerate}
\smallskip
\hrule 

\smallskip

Here, $\widehat{\pmb{\Omega}}_x$ serves as an asymptotic approximation to the inverse of $\mathbf{X}^\top\mathbf{X}/(NT)$. To study Step 3, we impose the following conditions.

\begin{assumption}\label{ASSUMPTION3}
\item Let $\pmb{\Sigma}_x^{-1} \eqqcolon \pmb{\Omega}_x =(\Omega_{x,j,k} )_{d\times d} $. Suppose that 
\begin{enumerate}[leftmargin=*, itemsep=0.5pt, parsep=0.5pt, topsep=0.6pt]
\item $(s+ \max_{ j \in [d]}s_j )\log (d)/\sqrt{NT} \to 0$, where $s_j = |\{k\neq j\mid \Omega_{x,j,k} \neq 0 \}|$ denotes the sparsity with respect to the rows of $\pmb{\Omega}_x$, and $s=|J|$ is the same as defined in equation (\ref{EQUATION1}).

\item $\max_j|\pmb{\gamma}_j|_1<\infty$, $\psi_{\mathrm{min}}(\pmb{\Sigma}_x)>0$, and $\psi_{\mathrm{min}}(\pmb{\Theta})>0$.
\end{enumerate}
\end{assumption}

Assumption 3.1 basically requires $\frac{s\, \log(d)}{\sqrt{NT}}\rightarrow 0$, while Assumption 3.2 imposes some conditions on two population matrices. Using them, we establish Theorem \ref{THEOREM1}, which applies to all elements of $\pmb{\beta}_0$, including those 0's, after choosing certain $\rho$.

\begin{theorem}\label{THEOREM1}
Let \eqref{EQUATION4} and Assumptions \ref{ASSUMPTION1}-\ref{ASSUMPTION3} hold.  Let $\pmb{\rho}=[\rho_1,\ldots, \rho_d]^\top$ be a $d \times 1$ vector such that $|H| < \infty$, where $H = \{j \in [d]\mid \rho_j \neq 0\}$. As $(N,T) \to (\infty,\infty)$, 
\be
\sqrt{NT}(\pmb{\rho}^\top\pmb{\Omega}_x\pmb{\Theta}\pmb{\Omega}_x\pmb{\rho})^{-1/2}\pmb{\rho}^\top (\widehat{\pmb{\beta}}_{bc} - \pmb{\beta}_0 ) \to_D N\left(0,1\right).
\label{jitigao1}
\ee
\end{theorem}

Before we make the asymptotic normality in (\ref{jitigao1}) feasible for inferential purposes, we propose to estimate $\pmb{\Theta}$, and then develop a thresholded HAC covariance estimator. 

\smallskip

\hrule \smallskip
\begin{enumerate}[leftmargin=*, itemsep=0.5pt, parsep=0.5pt, topsep=0.6pt]
\item[] \textbf{Step 4}: Define the  estimator of $\pmb{\Theta}$ by
\begin{equation}\label{EQUATION7}
T_u(\widehat{\pmb{\Theta}}_{\ell}) = \left(\widehat{\Theta}_{\ell,kl} \mathbb{I}(|\widehat{\Theta}_{\ell,kl}|\geq u) \right)_{ k,l\leq d},
\end{equation}
where $\widehat{\pmb{\Theta}}_{\ell} \coloneqq  \frac{1}{NT}\sum_{i,j=1}^{T}\sum_{s,t=1}^{T}a((t-s)/\ell)\mathbf{x}_{it}\mathbf{x}_{js}^\top \widehat{e}_{it} \widehat{e}_{js}$, $\widehat{e}_{it} = y_{it} - \mathbf{x}_{it}^\top \widehat{\pmb{\beta}}$, and $u$ is the threshold parameter at the order $u \asymp  \sqrt{\ell \log (d) / T}$, in which $a(\cdot)$ and $\widehat{\Theta}_{\ell,kl}$ are assumed to satisfy Assumption 4 below.
\end{enumerate}
\smallskip
\hrule

\begin{assumption}\label{ASSUMPTION4}
\item
\begin{enumerate}[leftmargin=*, itemsep=0.5pt, parsep=0.5pt, topsep=0.6pt]
\item $a(\cdot)$ is a symmetric and Lipschitz continuous function defined on $[-1,1]$, $a(0)=1$ and $\lim_{|x|\to 0 } \frac{1-a(x)}{|x|^{q_a}} = C_{q_a}$ for $q_a \in \{1,2\}$ and  $0 < C_{q_a} < \infty$. Additionally, $\ell \to \infty$ and $\ell \log (d)/T \to 0$.

\item Let $|\pmb{\Gamma}_t|_2 = O(t^{-(q_a+\epsilon)})$ for some $\epsilon >1$, where $\pmb{\Gamma}_t \coloneqq E(\frac{1}{N}\sum_{i,j=1}^{N}\mathbf{x}_{i0}\mathbf{x}_{jt}^\top e_{i0}e_{jt})$. 

\item Let $\max_{k\in [d]}\sum_{l=1}^{d}|\Theta_{\ell,kl}^{p}| \leq C(d)$ for some $0\leq p < 1$, where $$\pmb{\Theta}_{\ell} = \left(\Theta_{\ell,kl}\right)_{k,l\leq d}\coloneqq \frac{1}{NT}\sum_{i, j=1}^{N}\sum_{t, s=1}^{T}a((t-s)/\ell)E [\mathbf{x}_{it}\mathbf{x}_{js}^\top e_{it} e_{js}].$$
\end{enumerate}
\end{assumption}

The conditions on $a(\cdot )$ of Assumption \ref{ASSUMPTION4}.1 are satisfied by a number of commonly used kernels, such as the Bartlett and Parzen kernels. The last two conditions of Assumption \ref{ASSUMPTION4}.1 ensure the spectral norm consistency of our thresholded HAC covariance matrix estimator. On top of Assumption \ref{ASSUMPTION1}, Assumption \ref{ASSUMPTION4}.2 further requires an algebraic decay rate of the auto-covariance matrix, which is used to derive the order of a bias term involved in truncating the long-run covariance matrix. Assumption \ref{ASSUMPTION4}.3 controls the order of elements in $\pmb{\Theta}_{\ell}$, and allows for the presence of many ``small'' but nonzero elements. 

\begin{theorem}\label{THEOREM2}
Let Assumption \ref{ASSUMPTION4} and the conditions of Theorem \ref{THEOREM1} hold with $q>4$. Additionally, let $E(e_{it} \mid \mathbf{X}) = 0$,  $\lim\sup_{T\rightarrow \infty} s^2\sqrt{\ell \log (d)/T}<\infty$ and $\frac{d^3T\log T}{(T\ell \log d)^{q/4} }\to 0$.  Then we have
$$ |T_u(\widehat{\pmb{\Theta}}_{\ell}) - \pmb{\Theta} |_2 = O(\ell^{-q_a}) + O_P\left((\ell \log d / T)^{(1-p)/2}C(d) \right).$$
\end{theorem}

Note that the technical conditions (i.e., $E(e_{it} \mid \mathbf{X}) = 0$,  $\lim\sup_{T\rightarrow \infty} s^2\sqrt{\ell \log (d)/T}<\infty$ and $\frac{d^3T\log T}{(T\ell \log d)^{q/4} }\to 0$) assumed in Theorem \ref{THEOREM2} are only used to ensure that the estimation errors in $\{\widehat{e}_{it}\}$ are negligible. Theorem \ref{THEOREM2} immediately infers that
\be
\sqrt{NT}(\pmb{\rho}^\top\widehat{\pmb{\Omega}}_xT_u(\widehat{\pmb{\Theta}}_{\ell})\widehat{\pmb{\Omega}}_x\pmb{\rho})^{-1/2}\pmb{\rho}^\top (\widehat{\pmb{\beta}}_{bc} - \pmb{\beta}_0 ) \to_D N\left(0,1\right).
\nonumber
\ee

Up to this point, our investigation about model \eqref{EQUATION1} is completed. Again, the above investigation considerably extends and enriches the relevant literature by allowing for the idiosyncratic errors to exhibit TSA, CSD, heteroskedasticity, as well as heavy-tailed behavior.

\section{Extension to HD Panel Data Models with Interactive  Effects}\label{Sec3}

In this section, we generalize the above investigation for model \eqref{EQUATION1} to a class of widely used interactive fixed--effect models of the form:
\begin{eqnarray}\label{EQUATION8}
y_{it} &=& \mathbf{x}_{it}^\top \pmb{\beta}_0 + \pmb{\lambda}_{0i}^\top \mathbf{f}_{0t}  + e_{it},
\end{eqnarray}
which admits the following matrix representation:
\begin{equation}\label{EQUATION9}
\mathbf{y} = \mathbf{X} \pmb{\beta}_0 + \mathrm{vec}\left(\pmb{\Xi}_0\right)  + \mathbf{e},
\end{equation}
where $\pmb{\Xi}_0 = \mathbf{F}_0\pmb{\Lambda}_0^\top$, $\mathbf{F}_0 = [\mathbf{f}_{01},\ldots,\mathbf{f}_{0T}]^\top$, $\pmb{\Lambda}_0 = [\pmb{\lambda}_{01},\ldots,\pmb{\lambda}_{0N}]^\top$, and $\mathbf{e}$ is defined accordingly. Note that $\pmb{\Xi}_0$ admits the singular value decomposition (SVD): $$\pmb{\Xi}_0 = \mathbf{U}_0\mathbf{D}_0\mathbf{V}_0^\top,$$ where $\mathbf{U}_0$ and $\mathbf{V}_0$ are $T\times T$ and $N\times N$ respectively. Let $\mathbf{U}_{0,[r]} \in \mathbb{R}^{T\times r} $ and $\mathbf{V}_{0,[r]} \in \mathbb{R}^{N\times r}$ be the sub-matrices of singular vectors associated with the largest $r$ singular values of $\pmb{\Xi}_0$, where $r$ is assumed to be fixed and finite throughout the rest of this paper.

The latent factor structure induces a nonconvex regression problem, and we propose to deal with it via a nuclear norm constraint (e.g., \citealp{moon2018nuclear,belloni2023high}). We again propose a multiple--step procedure to investigate \eqref{EQUATION9}: (i) In Step 1, we use an $\ell_1$-nuclear norm penalized estimation method to obtain consistent estimator, which however suffers from substantial shrinkage bias (see Proposition \ref{Thm5} for details); and (ii) In Step 2, we propose an iterative estimation method that removes the shrinkage bias in order to obtain robust estimation and inference.

\smallskip

\hrule\smallskip
\begin{enumerate}[leftmargin=*, itemsep=0.5pt, parsep=0.5pt, topsep=0.6pt]
\item[] \textbf{Step 1}: We estimate $\pmb{\beta}_0$ and $\pmb{\Xi}_0$ by
\begin{equation}\label{EQUATION10}
(\widetilde{\pmb{\beta}}, \widetilde{\pmb{\Xi}}) = \argmin_{\pmb{\beta},\ \pmb{\Xi}} \frac{1}{2NT}|\mathbf{y} - \mathbf{X}\pmb{\beta} - \mathrm{vec}\left(\pmb{\Xi}\right)|_2^2 + \omega_1 |\pmb{\beta}|_1 + \frac{\omega_2}{\sqrt{NT}}|\pmb{\Xi}|_{*},
\end{equation}
where $\widetilde{\pmb{\beta}}=[\widetilde{\beta}_1,\ldots,\widetilde{\beta}_d]^\top$, and two tuning parameters are at order $\omega_1\asymp \sqrt{\log(d)/NT}$ and $\omega_2\asymp\max(1/\sqrt{N},1/\sqrt{T})$.  

We then estimate $r$ by $\widehat{r} = \sum_{k=1}^{\min\{T,N\}}\mathbb{I} ( \psi_k(\widetilde{\pmb{\Xi}}) \geq (\omega_2\sqrt{NT}|\widetilde{\pmb{\Xi}}|_2)^{1/2} ),$ and obtain a preliminary estimator of $\pmb{\Lambda}_0$ by $
\widetilde{\pmb{\Lambda}} = \sqrt{N}\widetilde{\mathbf{V}}_{[\widehat{r}]}$, where $\psi_k(A)$ denotes the $k$--th largest eigenvalue of matrix $A$, $\widetilde{\mathbf{V}}_{[\widehat{r}]}$ includes the singular vectors associated with the largest $\widehat{r}$ singular values of $\widetilde{\mathbf{V}}$, and $\widetilde{\mathbf{V}}$ is obtained from the SVD decomposition $\widetilde{\pmb{\Xi}} = \widetilde{\mathbf{U}}\widetilde{\mathbf{D}}\widetilde{\mathbf{V}}^\top$.

\item[] \textbf{Step 2}: Set $\widehat{\pmb{\Lambda}}^{(0)} = \widetilde{\pmb{\Lambda}}$. For $l\geq 1$, we update the estimators by weighted LASSO:
\begin{equation*}
(\widehat{\pmb{\beta}}^{(l)}, \widehat{\mathbf{F}}^{(l)}) = \argmin_{\pmb{\beta},\ \mathbf{F}} \frac{1}{2NT} |\mathbf{y} - \mathbf{X} \pmb{\beta} - (\widehat{\pmb{\Lambda}}^{(l-1)} \otimes \mathbf{I}_T)\mathrm{vec}\left(\mathbf{F}\right) |_2^2 + \omega_3\sum_{j=1}^{d}g_j|\beta_j|,
\end{equation*}
where $g_j = \mathbb{I}(|\widetilde{\beta}_j| < \omega_3)$ for some $\omega_3>0$, and 
$\widehat{\pmb{\Lambda}}^{(l)}$ corresponds to the first $\widehat{r}$ eigenvalues of $\frac{1}{NT}\sum_{t=1}^{T}(\mathbf{y}_t-\mathbf{X}_t\widehat{\pmb{\beta}}^{(l)})(\mathbf{y}_t-\mathbf{X}_t\widehat{\pmb{\beta}}^{(l)})^\top$.  Iterate the above procedure until numerical convergence, and denote the final estimators by $\widehat{\pmb{\beta}}$, $\widehat{\mathbf{F}}$ and $\widehat{\pmb{\Lambda}}$.
\end{enumerate}
\smallskip
\hrule

\smallskip

To proceed, we need to introduce the following assumptions. 

\begin{assumption}\label{ASSUMPTION5}
\item
\begin{enumerate}[leftmargin=*, itemsep=0.5pt, parsep=0.5pt, topsep=0.6pt]

\item Recall $(\omega_1, \omega_2)$ of (\ref{EQUATION10}). \, Let $|\mathbf{E}|_2/\sqrt{NT}\leq \omega_2/2$, where  $\mathbf{E} = [\mathbf{e}_1,\ldots,\mathbf{e}_T]^\top$.

\item Define
\begin{eqnarray*}
\mathbb{C} \coloneqq \Big\{\pmb{\beta}\in\mathbb{R}^{d},\pmb{\Xi}\in\mathbb{R}^{T\times N} \mid && \omega_1 |\pmb{\beta}_{J^c}|_1 +  \frac{\omega_2}{\sqrt{NT}} |\mathbf{M}_{\mathbf{U}_{0,[r]}}\pmb{\Xi}\mathbf{M}_{\mathbf{V}_{0,[r]}} |_{*}  \\
&& \leq 3\omega_1 |\pmb{\beta}_J|_1+3\frac{\omega_2}{\sqrt{NT}} |\pmb{\Xi} - \mathbf{M}_{\mathbf{U}_{0,[r]}}\pmb{\Xi}\mathbf{M}_{\mathbf{V}_{0,[r]}} |_* \Big\}.
\end{eqnarray*}
For $\forall (\pmb{\beta},\pmb{\Xi}) \in \mathbb{C}$, there exists a constant $\kappa_c>0$ such that
$$
\frac{1}{NT}|\mathbf{X}\pmb{\beta}+\mathrm{vec}(\pmb{\Xi})|_2^2\geq \kappa_c|\pmb{\beta}|_2^2 + \kappa_c\frac{1}{NT}|\mathrm{vec}(\pmb{\Xi})|_2^2.
$$

\item Let $|\mathbf{F}_0^\top\mathbf{F}_0/T - \pmb{\Sigma}_{f}|_F = O_P(1/\sqrt{T})$ and $|\pmb{\Lambda}_0^\top\pmb{\Lambda}_0/N - \pmb{\Sigma}_{\lambda}|_F = O_P(1/\sqrt{N})$ for some positive definite matrices $\pmb{\Sigma}_{f}$ and $\pmb{\Sigma}_{\lambda}$. Suppose that there exist constants $n_1>\cdots>n_r>0$ such that $n_j$ equals the $j^{th}$ largest eigenvalue of $\pmb{\Sigma}_{\lambda}^{1/2}\pmb{\Sigma}_{f}\pmb{\Sigma}_{\lambda}^{1/2}$.
\end{enumerate}
\end{assumption}

Assumption \ref{ASSUMPTION5}.1 requires the idiosyncratic error matrix to have an operator norm of order $\max(\sqrt{N},\sqrt{T})$, and nests a class of high--dimensional MA($\infty$) processes as special cases (e.g., see Example \ref{Exam1} of Appendix \ref{AP.A1} for details). Assumption \ref{ASSUMPTION5}.2 is often referred to as the restricted strong convexity condition in the literature (e.g., \citealp{negahban2011estimation,miao2023high}). The first part of Assumption \ref{ASSUMPTION5}.3 is commonly used in the relevant literature, and the second part of Assumption \ref{ASSUMPTION5}.3 requires that the eigenvalues of $\pmb{\Sigma}_{\lambda}^{1/2}\pmb{\Sigma}_{f}\pmb{\Sigma}_{\lambda}^{1/2}$ are distinct in order to identify the corresponding eigenvectors.

\begin{assumption}\label{ASSUMPTION6}
\item
\begin{enumerate}[leftmargin=*, itemsep=0.5pt, parsep=0.5pt, topsep=0.6pt]
\item Let $(\omega_1, \omega_2)$ of (\ref{EQUATION10}) satisfy $\frac{\sqrt{s}\omega_3}{\beta_{\min}} \to 0$, $\frac{\max(\sqrt{s}\omega_1,\omega_2)}{\omega_3}\to0$, and $s \cdot \max(\sqrt{s}\omega_1, \omega_2) \to 0$ as $(N,T)\rightarrow (\infty, \infty)$, where $s=|J|$ is the same as defined in equation (\ref{EQUATION1}).

\item Let  $ \psi_{\mathrm{min}}\left(\pmb{\Sigma}_{J}\right) > 0$, where $\pmb{\Sigma}_{J} \coloneqq \plim\mathbf{D}(\pmb{\Lambda}_0)$, $\mathbf{D}(\pmb{\Lambda}_0) = \frac{\sum_{t=1}^{T}\widetilde{\mathbf{X}}_{J,t}^\top\mathbf{M}_{\pmb{\Lambda}_0}\widetilde{\mathbf{X}}_{J,t}}{NT}$, 
$\widetilde{\mathbf{X}}_{J,t} = \mathbf{X}_{J,t} -  \frac{\sum_{s=1}^{T}a_{st}\mathbf{X}_{J,s}}{T}$, and $a_{st} = \mathbf{f}_{0t}^\top(\frac{\mathbf{F}_0^\top\mathbf{F}_0}{T})\mathbf{f}_{0s}$.
\end{enumerate}
\end{assumption}

The first condition of Assumption \ref{ASSUMPTION6}.1 requires that the non-zero elements of $\pmb{\beta}_{0}$ are not too small. The second condition of Assumption \ref{ASSUMPTION6}.1, together with the first condition, ensures the consistency of variable selection. The third condition of Assumption \ref{ASSUMPTION6}.1 imposes an extra restriction on the diverging rate of the number of nonzero elements in $\pmb{\beta}_0$, which is used to validate the compatibility condition. Assumption \ref{ASSUMPTION6}.2 requires the matrix to be positive definite, and is standard.

The following proposition establishes the consistency of model selection and an asymptotic distribution.

{\small

\begin{proposition}\label{PROPOSITION1}
Let \eqref{EQUATION4} and Assumptions \ref{ASSUMPTION1}, \ref{ASSUMPTION2}, \ref{ASSUMPTION5} and \ref{ASSUMPTION6} hold.  As $(N,T) \to (\infty,\infty)$, we have the following results:
\begin{enumerate}[leftmargin=*, itemsep=0.5pt, parsep=0.5pt, topsep=0.6pt]
\item $\Pr(\widehat{r} = r) \to 1$ and $\Pr (\sgn(\widehat{\pmb{\beta}}^{(l)}) = \sgn(\pmb{\beta}_0) ) \to 1$ for any $l\geq 1$.

\item Suppose that $\mathbf{E}$ is independent of $\mathbf{X}$, $\pmb{\Lambda}_0$ and $\mathbf{F}_0$, and  satisfies the conditions of Example \ref{Exam1}.3 of Appendix \ref{AP.A1}. In addition, let $\max_{j \in J} E(x_{j,it}^4)<\infty$, $E|\mathbf{f}_{0t}|_2^4 < \infty$, $E|\pmb{\lambda}_{0i}|_2^4<\infty$, $N/T \to \alpha$ with $\alpha$ being a positive constant, and $s^{3/2}\max(1/\sqrt{N},1/\sqrt{T})\to 0$. Conditioning on the event $\{\sgn(\widehat{\pmb{\beta}}) = \sgn(\pmb{\beta}_0)\}$, we have 
$$
\sqrt{NT}\pmb{\rho}^\top(\widehat{\pmb{\beta}}_{J}-\pmb{\beta}_{0,J})\to_D N\left(\alpha^{1/2}\pmb{\rho}^\top\pmb{\mu}_{\xi} + \alpha^{-1/2}\pmb{\rho}^\top\pmb{\mu}_{\zeta} ,\pmb{\rho}^\top\pmb{\Sigma}_{J}^{-1}\pmb{\Theta}_{J}\pmb{\Sigma}_{J}^{-1}\pmb{\rho}\right),
$$
where {\small $\pmb{\Omega}_e = E(\mathbf{e}_t\mathbf{e}_t^\top)$,  $\pmb{\Theta}_{J} \coloneqq \lim \mathrm{Var} (\frac{1}{\sqrt{NT}}\sum_{t=1}^{T}\widetilde{\mathbf{X}}_{J,t}^\top\mathbf{M}_{\pmb{\Lambda}_0}\mathbf{e}_{t} )$ $\pmb{\mu}_{\xi} \coloneqq \plim \pmb{\xi}$, $\pmb{\mu}_{\zeta} \coloneqq \plim \pmb{\zeta}$, and 
\begin{eqnarray*}
\pmb{\xi} &=& - \mathbf{D}^{-1}(\pmb{\Lambda}_0) \frac{1}{NT}\sum_{t,s=1}^{T}\frac{\widetilde{\mathbf{X}}_{J,t}^\top\pmb{\Lambda}_0}{N} (\pmb{\Lambda}_0^\top\pmb{\Lambda}_0/N )^{-1} (\mathbf{F}_0^\top\mathbf{F}_0/T)^{-1}\mathbf{f}_{0s} \sum_{i=1}^{N}E(e_{it}e_{is}) , \\
\pmb{\zeta} &=& - \mathbf{D}^{-1}(\pmb{\Lambda}_0) \frac{1}{NT}\sum_{t=1}^{T}\mathbf{X}_{J,t}^\top\mathbf{M}_{\pmb{\Lambda}_0}\pmb{\Omega}_e\pmb{\Lambda}_0(\pmb{\Lambda}_0^\top\pmb{\Lambda}_0/N)^{-1}(\mathbf{F}_0^\top\mathbf{F}_0/T)^{-1}\mathbf{f}_{0t}.
\end{eqnarray*}}
\end{enumerate}
\end{proposition}
}

Finally, we deal with the two biases and establish valid inference in Step 3.

\smallskip

\hrule\smallskip
\begin{enumerate}[leftmargin=*, itemsep=0.5pt, parsep=0.5pt, topsep=0.6pt]
\item[] \textbf{Step 3}: We defined the bias corrected estimator as follows:
\begin{eqnarray*}
&&\widehat{\pmb{\beta}}_{J, \text{bc}} = \widetilde{\pmb{\beta}}_{J, \text{bc}} - \frac{1}{N}\widehat{\pmb{\mu}}_{\zeta},\quad \widetilde{\pmb{\beta}}_{J,\text{bc}} = 2\widehat{\pmb{\beta}}_{J}  -( \widehat{\pmb{\beta}}_{J,S_1} + \widehat{\pmb{\beta}}_{J,S_2})/2,\nonumber \\
&&\widehat{\pmb{\mu}}_{\zeta} = - \mathbf{D}(\widehat{\pmb{\Lambda}})^{-1} \frac{1}{NT}\sum_{t=1}^T \mathbf{X}_{J,t}^\top \mathbf{M}_{\widehat{\pmb{\Lambda}}} T_u (\widehat{\pmb{\Omega}}_e )  \widehat{\pmb{\Lambda}} \left(\frac{\widehat{\mathbf{F}}^\top \widehat{\mathbf{F}}}{T}\right)^{-1} \widehat{\mathbf{f}}_t,
\end{eqnarray*}
where $\widehat{\pmb{\beta}}_{J,S_1}$ and $\widehat{\pmb{\beta}}_{J,S_2}$ are obtained using sample from $\{(i,t)\mid i\in [N],t\in S_1\}$ and $\{(i,t)\mid i\in [N],t\in S_2\}$ respectively,  $S_1 = \{1,\ldots, \lfloor T/2\rfloor \} $, $S_2=  \{\lfloor T/2\rfloor+1,\ldots, T \}$,  $T_u(\cdot)$ is the same as in \eqref{EQUATION7} by replacing the order of $u$ with $u \asymp  \sqrt{\log (N) /T}$, $\widehat{\pmb{\Omega}}_{e} = \frac{1}{T}\sum_{t=1}^T \widehat{\mathbf{e}}_{t}\widehat{\mathbf{e}}_{t}^\top$, and $\widehat{\mathbf{e}}_{t} = \mathbf{y}_t- \mathbf{X}_{J,t}\widehat{\pmb{\beta}}_J - \widehat{\pmb{\Lambda}}\widehat{\mathbf{f}}_t$.
\end{enumerate}
\smallskip
\hrule

\smallskip

\begin{proposition}\label{PROPOSITION2}
Let the conditions of Proposition \ref{PROPOSITION1}.2 hold, and $\delta_{NT} \coloneqq \min (\sqrt{T},\sqrt{N} )$. Recall that $s=|J|$ is the same as defined in equation (\ref{EQUATION1}). Then the following results hold:

{\small
\begin{enumerate}[leftmargin=*]
\item If $\max_{i\in [N]}|\pmb{\lambda}_{0i}|_2=O_P(\sqrt{\log N})$ and $\max_{i\in [N]}\sum_{j=1}^{N}|\omega_{e,ij}^{p_e}| \leq C_e(N)$ for some $0\leq p_e < 1$ with $\pmb{\Omega}_e = (\omega_{e,ij})_{i,j\in [N]}$, $$|T_u (\widehat{\pmb{\Omega}}_e ) - \pmb{\Omega}_e|_2 = O_P\left((\log (N) /\delta_{NT}^2)^{(1-p_e)/2}C_e(N)\right).$$

\item If $\sqrt{s}(\log (N) / T)^{(1-p_e)/2}C_e(N) \to 0$, we have
$$\sqrt{NT}(\pmb{\rho}^\top\widehat{\pmb{\Sigma}}_{J}^{-1}\widehat{\pmb{\Theta}}_{J}\widehat{\pmb{\Sigma}}_{J}^{-1}\pmb{\rho})^{-1/2}\pmb{\rho}^\top(\widehat{\pmb{\beta}}_{J,\mathrm{bc}}-\pmb{\beta}_{0,J})\to_D N\left(0,1\right),$$ where {\small $\widehat{\pmb{\Sigma}}_{J} = \frac{\sum_{t=1}^{T}\widehat{\mathbf{X}}_{J,t}^\top\mathbf{M}_{\widehat{\pmb{\Lambda}}}\widehat{\mathbf{X}}_{J,t}}{NT}$, $\widehat{\pmb{\Theta}}_{J} = \frac{\sum_{t,s=1}^{T}a((t-s)/\ell) (\widehat{\mathbf{X}}_{J,t}^\top\mathbf{M}_{\widehat{\pmb{\Lambda}}}\widehat{\mathbf{e}}_{t} ) (\widehat{\mathbf{X}}_{J,s}^\top\mathbf{M}_{\widehat{\pmb{\Lambda}}}\widehat{\mathbf{e}}_{s} )^\top}{NT}$, and $\widehat{\mathbf{X}}_{J,t}$} is defined in the same way as $\widetilde{\mathbf{X}}_{J,t}$ with $\mathbf{F}_0$ and $\mathbf{f}_{0t}$ replaced by $\widehat{\mathbf{F}}$ and $\widehat{\mathbf{f}}_{t}$, respectively.
\end{enumerate}
}

\end{proposition}

Before we prove the theoretical results in Appendix B, we evaluate the finite--sample performance of the proposed estimation and inferential procedure by simulated and real datasets.

\section{Simulations}\label{Sec4}

A detailed numerical implementational procedure is given in Appendix  \ref{AP.A2}. Using it, we evaluate the results of Section \ref{Sec2.2} by considering the following data generating process:
$$
\textbf{DGP1}: \quad y_{it} = \alpha_i + \mathbf{x}_{it}^\top\pmb{\beta}_{0} + e_{it},\quad \mathbf{e}_t = \rho_e\mathbf{e}_{t-1} + \pmb{\Sigma}_{e}^{1/2}\pmb{\varepsilon}_{e,t},
$$
where $\alpha_i = \frac{1}{T}\sum_{t=1}^{T}(x_{1,it} + x_{2,it})$ is an individual fixed--effect, $\pmb{\Sigma}_{\varepsilon} = \{\delta_{\varepsilon_e}^{|i-j|}\}_{i,j\in [N]}$, $\pmb{\varepsilon}_{e,t}$ follows from an $N$-dimensional $t$-distribution with a degree freedom of 5, and $\rho_e, \delta_{\varepsilon_e}\in\{0.2,0.5\}$, corresponding to low and moderate dependence in the dynamics of error innovations.  

DGP1 also allows for both the serial and cross--sectional correlations in $\mathbf{x}_{it}$ as follows: $\mathbf{x}_{l,t} = \rho_x + 0.2\mathbf{x}_{l,t-1} + \pmb{\Sigma}_{x}^{1/2}\pmb{\varepsilon}_{l,t} \ \ \text{for} \ \ l\in [d]$, 
where $\pmb{\Sigma}_{x} = \{0.2^{|i-j|}\}_{i,j\in [N]}$, $\rho_x \sim N(0,1)$, $\pmb{\varepsilon}_{l,t}$ follows from an $N$-dimensional $t$-distribution with a degree freedom of 5, and $\{\pmb{\varepsilon}_{l,t}\}_{l}$ is mutually independent for $l \in [d]$. Here, we let $d \in\{50,500\}$, $\beta_{0,j} = 0.2+0.1j$ for $j\in [5]$ and $\beta_{0,j} = 0$ for $j\ge 6$. When running regression, we deal with the fixed effects $\alpha_i$ by removing the time mean on both sides of DGP1. Therefore, we  use the demeaned data to conduct the estimation procedure of Section \ref{Sec2.2}.

For DGP1, we consider two sets of sample sizes, which are $N, T \in (20,30,40)$ and $N,T \in (50,100,200,400)$. For each pair of $(N,T)$, we conduct $1000$ replications. In addition, we measure the accuracy of LASSO estimates by the root mean squared error (RMSE) $\text{RMSE}(\widehat{\pmb{\beta}}) \coloneqq \sqrt{\frac{1}{1000}\sum_{j=1}^{1000}|\widehat{\pmb{\beta}}^{(j)} - \pmb{\beta}_0|_F^2},$ where $\widehat{\pmb{\beta}}^{(j)}$ is the estimate of $\pmb{\beta}_0$ at the $j^{th}$ replication. In order to evaluate the finite sample performance of our estimation and inferential procedure, we calculate the empirical coverage rates (ECR) for the nonzero elements in $\pmb{\beta}_0$, i.e., $\beta_{0,1}$-$\beta_{0,5}$ based on Steps 3 \& 4 of Section \ref{Sec2.2}. We then take the average across these elements for ease of presentation. For comparison, we also report the empirical coverage rates (denoted by ECR2) using the HAC estimator considered in \cite{babii2022machine}, which is not robust in the presence of CSD. We also report the ratio of sign consistency (RSC) of the adaptive LASSO procedure (\citealp{zou2006adaptive}), i.e., the ratio of $\{\widehat{\pmb{\beta}}_{w}^{(j)} =_s \pmb{\beta}_0\}_{j=1}^{1000}$. Tables \ref{Table1.1}-\ref{Table1.3} report these results for the cases with $d = 50$ and $d = 500$ respectively.

\begin{center}

\fbox{Tables \ref{Table1.1}-\ref{Table1.3} near here}

\end{center}

In view of Tables \ref{Table1.1} and \ref{Table1.2}, a few facts emerge. First, as expected, RMSE of the LASSO estimator decreases as both $N$ and $T$ increase. Second, as expected, RMSE increases if either $\rho_e$ or $\delta_{\varepsilon_e}$ increases. Third, when $\rho_e = 0.2$, ECR is very close to its nominal level even when $T=20$, although the distortion in ECR increases with the increase of $\rho_e$ and alters with $\delta_{\varepsilon_e}$ slightly. Fourth, when $\rho_e = 0.5$, ECR converges to its nominal level only with an increase in $T$. This is consistent with our theoretical prediction that the estimation error of long-rung covariance matrix is independent of $N$. Fifth, ECR2 is always below its nominal level and the distortion in ECR2 increases with the increase of $\delta_{\varepsilon_e}$. This is not surprising since the HAC estimator in \cite{babii2022machine} is not robust to the presence of CSD and just includes a proportion of the asymptotic variance. Finally, the adaptive LASSO procedure can correctly identify the sparsity pattern as long as the sample size is not so small. When the sample size is relatively small, RSC converges to one rapidly as the sample size increases. In Tables \ref{Table1.1} and \ref{Table1.3}, we find similar patterns. Interestingly, for $d = 500$ and $\rho_e = 0.2$, ECR tends to be larger than its nominal level when $T$ is relatively small. In addition, RMSE for $d = 500$ is slightly larger than that for $d = 50$, which is consistent with that RMSE should be proportion to $\sqrt{\log d}$ according to Lemma \ref{LEMMA2}.1.

We then evaluate the findings of Section \ref{Sec3}. Consider the following DGP:
$$
\textbf{DGP2}: \quad y_{it} = \mathbf{x}_{it}^\top\pmb{\beta}_{0} + \pmb{\lambda}_{0i}^\top \mathbf{f}_{0t} + e_{it},\quad \mathbf{e}_t = \rho_e\mathbf{e}_{t-1} + \pmb{\Sigma}_{e}^{1/2}\pmb{\varepsilon}_{e,t},
$$
where $\pmb{\lambda}_{0i} = [x_{1,i1},x_{2,i1}]^\top$ and $\mathbf{f}_{0t} = [x_{3,1t},x_{4,1t}]^\top$. Here, $\mathbf{x}_{it}$, $\pmb{\beta}_0$ and $\{e_{it}\}$ are generated in exactly the same way as in DGP1. We compute the RMSE of the first and second stage estimators for DGP2, denoted them by $\text{RMSE1}$ and $\text{RMSE2}$ respectively. To evaluate the finite sample performance of our inference procedure, we compute the empirical coverage rates (ECR) of the non-zero elements in $\pmb{\beta}_0$, i.e., $\beta_{0,1}$-$\beta_{0,5}$. We then take the average across these elements. We also report the average shares of the relevant variables included (TPR, true positive rate) and the average shares of the irrelevant variables included (FPR, false positive rate) of the conservative LASSO procedure, as well as the exact estimation rate (EER) of the number of factors by using the singular value thresholding procedure. Tables \ref{Table2.1}-\ref{Table2.3} show the results of DGP2 for $d = 50$ and $d = 500$ respectively.

\begin{center}

\fbox{Tables \ref{Table2.1}-\ref{Table2.3} near here}

\end{center}

Tables \ref{Table2.1} and \ref{Table2.2} reveal some notable points. First, the first-stage $\ell_1$-nuclear norm penalized estimator has a larger RMSE due to the regularization, which slowly vanishes as the sample size increases. On the other hand, the second-stage has much smaller RMSE, which decreases as both $N$ and $T$ increase. Second, for $\rho_e = 0.5$, the finite sample coverage probabilities are smaller than their nominal level (95\%) when $T$ is small, but are quite close to 95\% as $T$ increases. Third,  CSD and TSA in $\{e_{it}\}$ significantly affect the accuracy of our model selection procedure. For each pair of $(N,T)$,  larger $\rho_e$ and $\delta_{\varepsilon_e}$ infer larger FPR. {Nevertheless, our procedure can pick up all the relevant variables when either $N$ or $T$ is larger than 40, and FPR converges to zero quickly when either $N$ or $T$ is large.} Finally, the proposed singular value thresholding procedure can correctly determine the number of factors {when the sample size is not too small}.

Tables \ref{Table2.1} and \ref{Table2.3} report the simulation results of DGP2 for $d = 500$. We find similar patterns. In an analogous way to DGP1, RMSE1 for $d = 500$ is slightly larger than that for $d = 50$, which is consistent with our theoretical prediction. Notably, increasing $d$ does not significantly affect the estimation accuracy of the number of factors.

\section{Empirical Application}\label{Sec5}

Firm level characteristics which potentially can predict future stock returns have drawn considerable attention in the literature of asset pricing (see, \cite{KELLY2019501}, \cite{ChenZimmermann2021}, \cite{belloni2023high}, and many references therein, for example). In this section, we aim to select the firm characteristics that provide incremental information about U.S. monthly stock returns and to estimate how selected characteristics affect expected returns. Importantly, we highlight the necessity of accounting for cross-sectional dependence and demonstrate its implications on inferring significant firm characteristics.

We collect the return data of different firms of S\&P 500 from Center for Research in Security Prices (CRSP) (available at \url{www.crsp.org}), and match them with firm characteristics assembled by \cite{ChenZimmermann2021}, which are available at \url{www.openassetpricing.com}. These data are recorded monthly, and we pay attention to the period from  Jan of 1990 to Dec of 2023  which in total gives 408 time periods.  After matching ``permno" code in both datasets and removing firm characteristics with missing values, we end up with 161 firms (i.e., $N=161$) and 60 characteristics (i.e., $d=60$). We  present the variable names with their means and standard deviations in Table \ref{tb1}, and refer interested readers to \cite{ChenZimmermann2021} for detailed definitions of these variables. As these variables are measured in different units and have significant differences in terms of their standard deviations, we normalize each variable (including stock return) to ensure mean 0 and standard deviation 1 before regression.

We begin by presenting the results using model \eqref{EQUATION1} with fixed effects as in the simulation studies. To run estimation, we always regress the return of firm $i$ at time $t$ ($y_{it}$) on the firm level characteristics of firm $i$ at time period $t-1$ ($\mathbf{x}_{i,t-1}$). The estimation procedure is identical to those presented in Section \ref{Sec3}, so we do not repeat the details here. To show the presence of the weak CSD of the estimation residuals, we conduct the CD test\footnote{The asymptotic distribution of the CD test follows the standard normal distribution, so at the 5\% significance level, the critical values are $\pm 1.96$. We refer interested readers to \cite{Pesaran2004} for more details.} of \cite{Pesaran2004} on the residuals, and obtain the test statistic 61.32 which is a sign of the existence of weak CSD. Additionally, we run Jarque-Bera test based on the estimation residuals for each $i$, and find that we reject the null for all $i$'s, which infers that no individual has normally distributed residuals. Therefore, it shows the necessity of accounting for non-Gaussian error process in theory. 

\begin{table}[hbt]
\setlength{\tabcolsep}{4pt} 
\renewcommand{\arraystretch}{0.6} 
\scriptsize\centering
\caption{Summary Statistics of Firm Characteristics}\label{tb1}
\begin{tabular}{lrrlrrlrr}
\hline
& Mean & Std &  & Mean & Std &  & Mean & Std \\
accruals & 0.040 & \multicolumn{1}{r|}{0.092} & delfinl & -0.018 & \multicolumn{1}{r|}{0.105} & pricedelaytstat & 1.784 & 1.418 \\
announcementreturn & 0.006 & \multicolumn{1}{r|}{0.085} & dellti & -0.007 & \multicolumn{1}{r|}{0.052} & rdipo & 0.000 & 0.022 \\
assetgrowth & -0.126 & \multicolumn{1}{r|}{0.991} & delnetfin & -0.008 & \multicolumn{1}{r|}{0.136} & rds & -822.066 & 22660.487 \\
bmdec & 0.640 & \multicolumn{1}{r|}{1.177} & dolvol & -3.819 & \multicolumn{1}{r|}{3.375} & realizedvol & -0.029 & 0.030 \\
beta & 0.883 & \multicolumn{1}{r|}{0.670} & earningsconsistency & 0.059 & \multicolumn{1}{r|}{1.511} & residualmomentum & -0.036 & 0.327 \\
betafp & 0.908 & \multicolumn{1}{r|}{0.523} & equityduration & -16.823 & \multicolumn{1}{r|}{19.766} & returnskew & -0.137 & 0.871 \\
betaliquidityps & 0.016 & \multicolumn{1}{r|}{0.397} & grltnoa & 0.035 & \multicolumn{1}{r|}{0.103} & returnskew3f & -0.117 & 0.810 \\
bidaskspread & 0.015 & \multicolumn{1}{r|}{0.024} & high52 & 0.804 & \multicolumn{1}{r|}{0.224} & roe & -0.065 & 7.261 \\
bookleverage & -3.531 & \multicolumn{1}{r|}{55.527} & idiovol3f & -0.024 & \multicolumn{1}{r|}{0.027} & sharerepurchase & 0.521 & 0.500 \\
cheq & -1.365 & \multicolumn{1}{r|}{10.219} & idiovolaht & -0.028 & \multicolumn{1}{r|}{0.024} & totalaccruals & -0.006 & 0.203 \\
chinv & -0.006 & \multicolumn{1}{r|}{0.041} & illiquidity & 0.000 & \multicolumn{1}{r|}{0.000} & volmkt & -0.128 & 0.257 \\
chnncoa & 0.003 & \multicolumn{1}{r|}{0.117} & maxret & -0.062 & \multicolumn{1}{r|}{0.081} & volsd & -11.827 & 49.032 \\
chnwc & 0.007 & \multicolumn{1}{r|}{0.087} & noa & -0.513 & \multicolumn{1}{r|}{0.742} & volumetrend & -0.006 & 0.018 \\
chtax & 0.001 & \multicolumn{1}{r|}{0.023} & netdebtfinance & -0.012 & \multicolumn{1}{r|}{0.086} & xfin & -0.009 & 0.159 \\
convdebt & -0.065 & \multicolumn{1}{r|}{0.247} & netequityfinance & 0.002 & \multicolumn{1}{r|}{0.128} & betavix & 0.000 & 0.019 \\
coskewacx & 0.142 & \multicolumn{1}{r|}{0.238} & opleverage & 0.828 & \multicolumn{1}{r|}{0.780} & dnoa & -0.070 & 0.858 \\
coskewness & 0.253 & \multicolumn{1}{r|}{0.348} & pctacc & 2.120 & \multicolumn{1}{r|}{23.208} & hire & -0.019 & 0.225 \\
delcoa & -0.013 & \multicolumn{1}{r|}{0.081} & pcttotacc & -0.407 & \multicolumn{1}{r|}{20.494} & zerotrade & 0.788 & 2.330 \\
delcol & -0.011 & \multicolumn{1}{r|}{0.061} & pricedelayrsq & 0.288 & \multicolumn{1}{r|}{0.317} & zerotradealt1 & 0.707 & 2.350 \\
delequ & -0.022 & \multicolumn{1}{r|}{0.188} & pricedelayslope & -1.224 & \multicolumn{1}{r|}{111.948} & zerotradealt12 & 0.794 & 2.287 \\
\hline
\end{tabular}
\end{table}

The weighted Lasso procedure identifies five significant determinants in future stock returns, which are Earnings announcement return, Pastor-Stambaugh liquidity beta, Idiosyncratic risk (3 factor), Maximum return over month, Realized volatility, respectively. In contrast, by examining the estimated confidence intervals for each element in the debiased Lasso estimates, we find that the predictor `Maximum return over month' is insignificant with a slope coefficient of -0.132 and a 95\% confidence interval of (-0.278, 0.014). It is not surprising that the predictor `Maximum return over month' survives in the weighted Lasso selection procedure as the value `-0.1321' is quite large in this study. In addition, the estimated confidence intervals indicate five other (marginally) significant predictors, which may be due to the well-known multiple-hypothesis-testing problem in this case. Note that even when there is no truly significant firm characteristics, we expect to identify about three significant predictors ($60\times 5\%$) due to pure sampling variation. Then, four firm level characteristics are selected, and are reported in Table \ref{tb2} above. 

The above selection result shares certain similarity with \cite{belloni2023high}.  For example, both studies acknowledge the importance of liquidity beta. While they focus on the selection at different quantile, we emphasize the importance of accounting for dependence of error terms. For the sake of comparison, we also examine the estimated confidence intervals using the HAC estimator considered in \cite{babii2022machine}, which is not robust to the presence of CSD. These estimated confidence intervals indicate that there are 21 significant firm characteristics, which is unusually large. We do expect that these estimated confidence intervals are quite narrow as the HAC estimator considered in \cite{babii2022machine} only just include a proportion of the asymptotic variance with the presence of CSD. This point can be also seen in the constructed confidence intervals reported in Table \ref{tb2}. Overall, the above estimation and selection results highlight the necessity of accounting for cross-sectional dependence when selecting significant firm characteristics.

{\small

\begin{table}
\setlength{\tabcolsep}{4pt} 
\renewcommand{\arraystretch}{0.6} 
\footnotesize
\caption{Selected Key Variables for FE Model}\label{tb2}
\begin{tabular}{llcccc}
\hline
Key Variables & Definitions & Est & \multicolumn{1}{c}{CI} & \multicolumn{1}{c}{CI (no CSD)} \\
announcementreturn & Earnings announcement return & 0.047 & \multicolumn{1}{r|}{(0.015, 0.079)}  & (0.031, 0.063) \\
betaliquidityps & Pastor-Stambaugh liquidity beta & -0.061 & \multicolumn{1}{r|}{(-0.088, -0.035)}  & (-0.073, -0.050) \\
idiovol3f &  Idiosyncratic risk (3 factor) & -0.378 & \multicolumn{1}{r|}{(-0.579, -0.178)}  & (-0.429, -0.328)\\
realizedvol &  Realized volatility & 0.492 & \multicolumn{1}{r|}{(0.257, 0.727)}  & (0.430, 0.555)\\
\hline
\end{tabular}
\end{table}

}

We next consider the model \eqref{EQUATION8} with unobservable common factors, of which the latter one is also considered in  \cite{KELLY2019501} but under the finite dimensional framework. We also conduct the CD test to detect the presence of the weak CSD of the estimation residuals, and obtain the test statistic of $-11.78$. For both models, we end up with $T=407$. This result indicates the presence of weak CSD even after accounting for unobservable common factors (strong CSD).

{\small

\begin{table}
\setlength{\tabcolsep}{4pt} 
\renewcommand{\arraystretch}{0.6} 
\footnotesize
\caption{Selected Key Variables for IFE Model}\label{tb3}
\begin{tabular}{ll ccc}
\hline
Key Variables & Definitions & Est & \multicolumn{1}{c}{CI}  \\
announcementreturn & Earnings announcement return & 0.042 & \multicolumn{1}{r|}{(0.019, 0.065)}  \\
betaliquidityps & Pastor-Stambaugh liquidity beta & -0.032 & \multicolumn{1}{r|}{(-0.058, -0.007)}  \\
high52 &  52 Week High & -0.056 & \multicolumn{1}{r|}{(-0.077, -0.035)}  \\
returnskew3f &  Idiosyncratic skewness & 0.029 & \multicolumn{1}{r|}{(-0.044, -0.015)}  \\
\hline
\end{tabular}
\end{table}

}

Again, since our weighted LASSO procedure may select irrelevant variables in finite sample, we further eliminate irrelevant firm level characteristics based on the constructed confidence intervals. Then, four firm level characteristics are selected, and are reported in Table \ref{tb3} above. Compared to the estimates associated with the selected variables using the fixed effects model, interestingly, we find that after accounting for the common factors, the estimates tend to have narrower confidence intervals overall. This point is important for detecting the significant firm characteristics because the literature (e.g., \citealp{rapach2013}) shows that certain predictors do provide useful signals for forecasting stock returns but these signals are of small magnitudes and are hidden by the large uncertainty of error innovations. Instead, the large uncertainty of error innovations can be captured by using the method of interactive fixed effects. It then shows the necessity of including the factor structure practically, and so demonstrates the practical relevance of Section \ref{Sec3}.

\section{Conclusion}\label{Sec6}

In this paper, we propose a robust inferential procedure for the proposed HD panel data models. Specifically, (i) we pay attention to non-Gaussian, serially and cross-sectionally correlated and heteroskedastic error processes; (ii) we develop an estimation method for high-dimensional long-run covariance matrix using a thresholded estimator; and (iii) we allow for the number of regressors to grow faster than the sample size. In order to establish the corresponding theory, we derive two Nagaev-types of sharp concentration inequalities, one for a partial sum and the other for a quadratic form, subject to a set of easily verifiable conditions. Leveraging these two inequalities, we develop a non-asymptotic bound for the LASSO estimator, establish an asymptotic normality via the nodewise LASSO regression, and derive a sharp convergence rate for the thresholded HAC estimator. 

We believe that our study provides the relevant literature with a complete toolkit for conducting estimation and inference for the parameters of interest within a class of HD panel data settings. We also demonstrate the practical relevance of these estimation and inferential methods by investigating a class of HD panel data models associated with interactive effects. Moreover, we conduct extensive numerical studies using simulated and real data examples.

{\footnotesize
\setlength{\bibsep}{2.pt plus 0ex}
\bibliography{HD-Panel}
}

\bigskip

\renewcommand{\thesection}{A}

\section*{Appendix A}

In this Appendix, Appendix \ref{AP.A1} verifies several assumptions. A numerical implementation procedure is presented in Appendix \ref{AP.A2}. Tables for Section 4 are listed in Appendix \ref{AP.A3}. Appendix \ref{AP.A4} lists some technical lemmas, while their proofs are collected in Appendix \ref{AP.B1}. Appendix \ref{AP.A5} provide the full proofs of the main results listed in Section \ref{Sec2}.

\renewcommand{\theequation}{A.\arabic{equation}}
\renewcommand{\thesection}{A.\arabic{section}}
\renewcommand{\thefigure}{A.\arabic{figure}}
\renewcommand{\thetable}{A.\arabic{table}}
\renewcommand{\thelemma}{A.\arabic{lemma}}
\renewcommand{\theassumption}{A.\arabic{assumption}}
\renewcommand{\thetheorem}{A.\arabic{theorem}}
\renewcommand{\theproposition}{A.\arabic{proposition}}
\setcounter{equation}{0}
\setcounter{lemma}{0}
\setcounter{section}{0}
\setcounter{table}{0}
\setcounter{figure}{0}
\setcounter{assumption}{0}
\setcounter{proposition}{0}

\section{Justification of Assumptions}\label{AP.A1}

\begin{example}\label{Exam1}
Consider a high--dimensional MA($\infty$) process of the form: $\mathbf{e}_t = \sum_{j=0}^{\infty} \mathbf{B}_j\pmb{\varepsilon}_{t-j},$ where $\mathbf{B}_j$'s are $N\times N$ matrices. Suppose that (a) $|\mathbf{B}_j|_{2} = O(j^{-\alpha})$ for some $\alpha > 2$, and (b) $\{\varepsilon_{it}\}$ is independent over $(i,t)$ and $E|\varepsilon_{it}|^q <\infty$ for some $q>2$.  Then
\begin{itemize}[leftmargin=*]
\item [(i)] $\sup_{|\mathbf{w}|_2<\infty}\|\mathbf{w}^\top[\mathbf{e}_j - \mathbf{e}_j^*] \|_q = O(j^{-\alpha})$, where $\mathbf{e}_j^*$ the coupled version of $\mathbf{e}_j$ with $\pmb{\varepsilon}_0^*$ replacing $\pmb{\varepsilon}_0$;
\item [(ii)] $\left|\mathbf{E}\right|_2 = O_P (\max(\sqrt{N},\sqrt{T}) )$ if $E(\varepsilon_{it}^4) <\infty$, where $\mathbf{E} = [\mathbf{e}_1,\ldots,\mathbf{e}_{T}]^\top$;
\item [(iii)] Suppose further for some $|\rho|<1$, $|\mathbf{B}_j|_{2} = O(\rho^j)$, $|\mathbf{B}_j|_{1} = O(\rho^j)$, $|\mathbf{B}_j|_{\infty} = O(\rho^j)$ and $E(\varepsilon_{it}^8) <\infty$. Then $\{\mathbf{e}_t\}$ also satisfy Assumption C of \cite{bai2009panel} which regulates the correlation along both dimensions of $(i, t)$.
\end{itemize}
\end{example}

\begin{proof}[Verification of Example 1]

\noindent (i).	Without loss of generality, let $q=4$ in what follows. For notational simplicity, let $\mathbf{B}_j = \{B_{j,kl}\}_{k,l\in [N]}$ and $\mathbf{w}^\top \mathbf{B}_j = (B_{j,\, \centerdot 1},\ldots,B_{j,\, \centerdot N})$. As $\{\varepsilon_{it}\}$ are independent over $i$, we can write
\begin{eqnarray*}
&&\|\mathbf{w}^\top(\mathbf{e}_j - \mathbf{e}_j^*) \|_4^4 =E\left|\mathbf{w}^\top \mathbf{B}_j(\varepsilon_{0} - \varepsilon_{0}^*)\right|^4 = E\left|\sum_{l=1}^{N}B_{j,\, \centerdot l}^2(\varepsilon_{l,0} - \varepsilon_{l,0}^*)^2\right|^2 \nonumber \\
&&+ 4E\left|\sum_{l=1}^{N-1}\sum_{k=l+1}^{N}B_{j,\, \centerdot l}B_{j,\, \centerdot k}(\varepsilon_{l,0} - \varepsilon_{l,0}^*)(\varepsilon_{k,0} - \varepsilon_{k,0}^*) \right|^2\nonumber \\
&\le &O(1)\left(\sum_{l=1}^{N}B_{j,\, \centerdot l}^2\right)^2 +O(1) \sum_{l=1}^{N-1}\sum_{k=l+1}^{N}B_{j,\, \centerdot l}^2B_{j,\, \centerdot k}^2
\le O(1) \left(\mathbf{w}^\top \mathbf{B}_j\mathbf{B}_j^\top \mathbf{w} \right)^2 = O(|\mathbf{B}_j|_2^4),
\end{eqnarray*} 
where the first inequality follows from some direct calculation, and the second inequality follows from $\sum_{l=1}^{N-1}\sum_{k=l+1}^{N}B_{j,\, \centerdot l}^2B_{j,\, \centerdot k}^2\leq  (\sum_{l=1}^{N}B_{j,\, \centerdot l}^2)^2$.

Based on the above development, we have
$
\|\mathbf{w}^\top(\mathbf{e}_j - \mathbf{e}_j^*) \|_4 = O(j^{-\alpha})$.

\noindent (ii). Note that $\mathbf{E}^\top = \sum_{j=0}^{\infty} \mathbf{B}_j\pmb{\varepsilon}_{-j}$ with $\pmb{\varepsilon}_{-j} = [\pmb{\varepsilon}_{1-j},\ldots,\pmb{\varepsilon}_{T-j}]$. Since $\pmb{\varepsilon}_{-j}$ is an $N\times T$ random matrix consisting of mutually independent random elements such that $E(\varepsilon_{it}^4)<\infty$, we have $E|\pmb{\varepsilon}_{-j}|_2 = O(\max(\sqrt{N}, \sqrt{T}))$ (cf., \citealp{latala2005some}), which implies
$$
E|\mathbf{E}^\top|_2 \leq \sum_{j=0}^{\infty}|\mathbf{B}_j|_2 E|\pmb{\varepsilon}_{-j}|_2 = O(\max(\sqrt{N}, \sqrt{T})).
$$

\noindent (iii). Refer to Example 1 in \cite{gao2023higher}. 
\end{proof}

\section{Numerical Implementation}\label{AP.A2}

When estimating model \eqref{EQUATION1}, we use the following modified BIC (cf., \citealp{wang2009shrinkage}) to select the tuning parameter $\omega_1$ by $\widehat{\omega}_{1} = \argmin_{w_{1}}\frac{1}{NT}|\mathbf{y}-\mathbf{X}\widehat{\pmb{\beta}}_{\omega_1}|_2^2 + |J_{\omega_1}|\frac{\log(NT)}{NT}\log(\log (d))$,
where $\widehat{\pmb{\beta}}_{\omega_1}$ denotes the LASSO estimator obtained using the tuning parameter $\omega_{1}$ and $J_{\omega_1} = \{j\mid
\widehat{\beta}_{\omega_1,j} \neq 0\}$. When constructing $\widehat{\pmb{\Omega}}_x$, similarly, we choose $\widetilde{\omega}_j$ for $1\leq j \leq d$ using the above modified BIC. When estimating the long-run covariance matrix, we use the Bartlett kernel function and let $\ell = \lceil0.75 T^{1/3}\rceil$ following the suggestion of \cite{stock2020introduction}. To select the threshold level $u$ for thresholded HAC estimator, we use the 2-fold cross validation method as suggested by \cite{bickel2008covariance}.

To obtain the first-stage $\ell_1$-nuclear norm penalized estimator for model \eqref{EQUATION8}, following \cite{belloni2023high} we propose a modified BIC to select the best pair of tuning parameters ($\omega_1,\omega_2$). Given a pair ($\omega_{1},\omega_{2}$), our estimation method produces the estimates ($\widetilde{\pmb{\beta}}(\omega_1,\omega_2)$, $\widetilde{\pmb{\Xi}}(\omega_1,\omega_2)$). Then the estimates of $\omega_1$ and $\omega_2$ is given by
\begin{eqnarray*} 
(\widehat{\omega}_{1},\widehat{\omega}_{2}) &=& \argmin_{\omega_{1},\omega_{2}}\frac{1}{NT}|\mathbf{y}-\mathbf{X}\widetilde{\pmb{\beta}}(\omega_1,\omega_2)-\mathrm{vec}(\widetilde{\pmb{\Xi}}(\omega_1,\omega_2))|_2^2 \nonumber\\
&& + |J(\omega_1,\omega_2)|\frac{\log(NT)}{NT}\log(\log (d)) + |\widehat{r}(\omega_1,\omega_2)|\frac{N+T}{NT}  ,
\end{eqnarray*}
where $J(\omega_1,\omega_2) = \{j\mid
\widehat{\beta}_{j}(\omega_1,\omega_2) \neq 0\}$ and $\widehat{r}(\omega_1,\omega_2)$ denotes the rank of $\widetilde{\pmb{\Xi}}(\omega_1,\omega_2)$. For the second-stage iterated weighted LASSO estimator, we use the following modified BIC to select the tuning parameter $\omega_3$ by $\widehat{\omega}_{3} = \argmin_{\omega_{3}}\frac{1}{NT}|\mathbf{y}-\mathbf{X}\widehat{\pmb{\beta}}_{\omega_3}-\mathrm{vec}(\widehat{\mathbf{F}}_{\omega_3}\widehat{\pmb{\Lambda}}_{\omega_3}^\top)|_F^2 + |J_{\omega_3}|\frac{\log(NT)}{NT}\log(\log (d))$,
where $\widehat{\pmb{\beta}}_{\omega_3}$, $\widehat{\mathbf{F}}_{\omega_3}$ and $\widehat{\pmb{\Lambda}}_{\omega_3}$ denote the estimates of $\pmb{\beta}_0$, $\mathbf{F}_0$ and $\pmb{\Lambda}_0$ obtained using the tuning parameter $\omega_{3}$ and $J_{\omega_3} = \{j\mid \widehat{\beta}_{\omega_3,j} \neq 0\}$. In addition, we use the 2-fold cross validation method to select the penalty term $u$ when estimating $\pmb{\Omega}_{e}$.

{\small

\section{Tables for the Simulation Results in Section 4}\label{AP.A3}

\begin{table}[htb]
\renewcommand{\arraystretch}{0.6} 
\caption{Simulation Results of DGP 1}\label{Table1.1}
\begin{tabular}{cccccc c cccc }
\hline
$d = 50$&     & \multicolumn{4}{c}{$\rho_e=0.2,\delta_{\varepsilon_e}=0.2$} & & \multicolumn{4}{c}{$\rho_e=0.2,\delta_{\varepsilon_e}=0.5$}\\
\cline{3-6} \cline{8-11}
\hline
$N$ & $T$ & \text{RMSE} & ECR & ECR2 & RSC & & \text{RMSE} & ECR & ECR2 & RSC\\
\hline
20   & 20  & 0.264  & 0.931  & 0.893  & 0.478      &   & 0.286  & 0.920  & 0.840  & 0.478  \\
& 30  & 0.211  & 0.939  & 0.902  & 0.723      &   & 0.231  & 0.921  & 0.849  & 0.663  \\
& 40  & 0.183  & 0.938  & 0.906  & 0.873      &   & 0.195  & 0.931  & 0.854  & 0.848  \\
30   & 20  & 0.214  & 0.942  & 0.894  & 0.667      &   & 0.231  & 0.926  & 0.840  & 0.669  \\
& 30  & 0.173  & 0.939  & 0.903  & 0.904      &   & 0.188  & 0.929  & 0.850  & 0.886  \\
& 40  & 0.149  & 0.939  & 0.899  & 0.975      &   & 0.162  & 0.935  & 0.857  & 0.970  \\
40   & 20  & 0.187  & 0.945  & 0.897  & 0.820      &   & 0.202  & 0.927  & 0.845  & 0.789  \\
& 30  & 0.152  & 0.936  & 0.894  & 0.968      &   & 0.165  & 0.934  & 0.853  & 0.944  \\
& 40  & 0.128  & 0.944  & 0.910  & 0.992      &   & 0.140  & 0.938  & 0.850  & 0.991  \\ \hline
&     & \multicolumn{4}{c}{$\rho_e=0.5,\delta_{\varepsilon_e}=0.2$}& & \multicolumn{4}{c}{$\rho_e=0.5,\delta_{\varepsilon_e}=0.5$}\\
\cline{3-6} \cline{8-11}      
\hline
$N$ & $T$ & \text{RMSE} & ECR & ECR2 & RSC & & \text{RMSE} & ECR & ECR2 & RSC\\
\hline
20 & 20  & 0.297  & 0.892  & 0.865  & 0.496  &   &  0.332  & 0.867  & 0.805  & 0.467  \\
& 30  & 0.245  & 0.898  & 0.872  & 0.712  &   &  0.276  & 0.879  & 0.813  & 0.698  \\
& 40  & 0.214  & 0.895  & 0.873  & 0.829  &   &  0.214  & 0.895  & 0.873  & 0.829  \\
30 & 20  & 0.246  & 0.904  & 0.862  & 0.649  &   &  0.277  & 0.883  & 0.807  & 0.656  \\
& 30  & 0.200  & 0.905  & 0.880  & 0.865  &   &  0.230  & 0.892  & 0.821  & 0.867  \\
& 40  & 0.174  & 0.906  & 0.880  & 0.948  &   &  0.200  & 0.892  & 0.824  & 0.920  \\
40 & 20  & 0.212  & 0.897  & 0.868  & 0.814  &   &  0.239  & 0.882  & 0.801  & 0.815  \\
& 30  & 0.175  & 0.906  & 0.874  & 0.936  &   &  0.199  & 0.895  & 0.825  & 0.939  \\
& 40  & 0.153  & 0.916  & 0.887  & 0.985  &   &  0.173  & 0.900  & 0.827  & 0.985  \\ 
\hline   \hline 
$d=500$&     & \multicolumn{4}{c}{$\rho_e=0.2,\delta_{\varepsilon_e}=0.2$} & & \multicolumn{4}{c}{$\rho_e=0.2,\delta_{\varepsilon_e}=0.5$}\\
\cline{3-6} \cline{8-11}
\hline
$N$ & $T$ & \text{RMSE} & ECR & ECR2 & RSC & & \text{RMSE} & ECR & ECR2 & RSC\\
\hline
20   & 20  &     0.385  & 0.978  & 0.899  & 0.032       &   &  0.414  & 0.972  & 0.862  & 0.021  \\
& 30  &     0.307  & 0.972  & 0.908  & 0.117       &   &  0.332  & 0.964  & 0.864  & 0.096  \\
& 40  &     0.259  & 0.976  & 0.914  & 0.271       &   &  0.281  & 0.969  & 0.861  & 0.199  \\
30   & 20  &     0.314  & 0.983  & 0.908  & 0.098       &   &  0.339  & 0.977  & 0.853  & 0.078  \\
& 30  &     0.249  & 0.977  & 0.912  & 0.304       &   &  0.270  & 0.971  & 0.866  & 0.260  \\
& 40  &     0.215  & 0.973  & 0.911  & 0.565       &   &  0.233  & 0.966  & 0.860  & 0.463  \\
40   & 20  &     0.267  & 0.985  & 0.907  & 0.211       &   &  0.291  & 0.980  & 0.853  & 0.163  \\
& 30  &     0.217  & 0.980  & 0.910  & 0.532       &   &  0.236  & 0.974  & 0.863  & 0.428  \\
& 40  &     0.186  & 0.972  & 0.906  & 0.794       &   &  0.202  & 0.968  & 0.851  & 0.700  \\
\hline
&     & \multicolumn{4}{c}{$\rho_e=0.5,\delta_{\varepsilon_e}=0.2$}& & \multicolumn{4}{c}{$\rho_e=0.5,\delta_{\varepsilon_e}=0.5$}\\
\cline{3-6} \cline{8-11}     
\hline 
$N$ & $T$ & \text{RMSE} & ECR & ECR2 & RSC & & \text{RMSE} & ECR & ECR2 & RSC\\
\hline
20 & 20  &     0.432  & 0.958  & 0.875  & 0.027    &   &  0.458  & 0.946  & 0.823  & 0.033  \\
& 30  &     0.350  & 0.949  & 0.883  & 0.095    &   &  0.376  & 0.937  & 0.821  & 0.110  \\
& 40  &     0.296  & 0.951  & 0.884  & 0.209    &   &  0.320  & 0.939  & 0.831  & 0.208  \\
30 & 20  &     0.353  & 0.968  & 0.880  & 0.076    &   &  0.377  & 0.954  & 0.815  & 0.090  \\
& 30  &     0.283  & 0.957  & 0.890  & 0.236    &   &  0.307  & 0.947  & 0.839  & 0.232  \\
& 40  &     0.247  & 0.953  & 0.885  & 0.427    &   &  0.268  & 0.935  & 0.824  & 0.425  \\
40 & 20  &     0.301  & 0.975  & 0.881  & 0.149    &   &  0.324  & 0.961  & 0.823  & 0.150  \\
& 30  &     0.248  & 0.963  & 0.888  & 0.387    &   &  0.268  & 0.947  & 0.829  & 0.373  \\
& 40  &     0.212  & 0.951  & 0.881  & 0.642    &   &  0.232  & 0.937  & 0.817  & 0.601  \\
\hline 
\end{tabular}
\end{table}

\begin{table}[htb]
\renewcommand{\arraystretch}{0.6} 
\caption{Simulation Results of DGP 1 for $d = 50$}\label{Table1.2}
\begin{tabular}{cccccc c cccc }
\hline
&     & \multicolumn{4}{c}{$\rho_e=0.2,\delta_{\varepsilon_e}=0.2$} & & \multicolumn{4}{c}{$\rho_e=0.2,\delta_{\varepsilon_e}=0.5$}\\
\cline{3-6} \cline{8-11}
\hline
$N$ & $T$ & \text{RMSE} & ECR & ECR2 & RSC & & \text{RMSE} & ECR & ECR2 & RSC\\
50  & 50   &    0.106  & 0.943  & 0.905  & 1.000    &   &     0.115  & 0.931  & 0.844  & 0.999  \\
& 100  &    0.073  & 0.948  & 0.913  & 1.000    &   &     0.080  & 0.942  & 0.862  & 1.000  \\
& 200  &    0.053  & 0.946  & 0.921  & 1.000    &   &     0.058  & 0.944  & 0.867  & 1.000  \\
& 400  &    0.037  & 0.944  & 0.916  & 1.000    &   &     0.041  & 0.942  & 0.863  & 1.000  \\
100 & 50   &    0.075  & 0.946  & 0.908  & 1.000    &   &     0.082  & 0.944  & 0.856  & 1.000  \\
& 100  &    0.052  & 0.946  & 0.911  & 1.000    &   &     0.058  & 0.938  & 0.850  & 1.000  \\
& 200  &    0.037  & 0.946  & 0.918  & 1.000    &   &     0.041  & 0.946  & 0.859  & 1.000  \\
& 400  &    0.026  & 0.951  & 0.923  & 1.000    &   &     0.029  & 0.950  & 0.867  & 1.000  \\
200 & 50   &    0.053  & 0.948  & 0.906  & 1.000    &   &     0.059  & 0.939  & 0.851  & 1.000  \\
& 100  &    0.037  & 0.948  & 0.917  & 1.000    &   &     0.041  & 0.944  & 0.859  & 1.000  \\
& 200  &    0.027  & 0.949  & 0.916  & 1.000    &   &     0.029  & 0.949  & 0.864  & 1.000  \\
& 400  &    0.019  & 0.948  & 0.919  & 1.000    &   &     0.019  & 0.948  & 0.919  & 1.000  \\
400& 50   &    0.038  & 0.944  & 0.904  & 1.000    &   &     0.042  & 0.942  & 0.847  & 1.000  \\
& 100  &    0.027  & 0.951  & 0.913  & 1.000    &   &     0.029  & 0.945  & 0.854  & 1.000  \\
& 200  &    0.019  & 0.952  & 0.916  & 1.000    &   &     0.021  & 0.948  & 0.867  & 1.000  \\
& 400  &    0.013  & 0.950  & 0.917  & 1.000    &   &     0.015  & 0.948  & 0.862  & 1.000  \\
\hline
&     & \multicolumn{4}{c}{$\rho_e=0.5,\delta_{\varepsilon_e}=0.2$}& & \multicolumn{4}{c}{$\rho_e=0.5,\delta_{\varepsilon_e}=0.5$}\\
\cline{3-6} \cline{8-11}     
\hline 
$N$ & $T$ & \text{RMSE} & ECR & ECR2 & RSC & & \text{RMSE} & ECR & ECR2 & RSC\\
\hline
50  & 50   &    0.123  & 0.907  & 0.878  & 0.999   &   &     0.141  & 0.892  & 0.813  & 0.997  \\
& 100  &    0.087  & 0.924  & 0.892  & 1.000   &   &     0.100  & 0.921  & 0.833  & 1.000  \\
& 200  &    0.062  & 0.931  & 0.904  & 1.000   &   &     0.071  & 0.929  & 0.849  & 1.000  \\
& 400  &    0.044  & 0.936  & 0.911  & 1.000   &   &     0.050  & 0.932  & 0.854  & 1.000  \\
100 & 50   &    0.088  & 0.920  & 0.888  & 1.000   &   &     0.100  & 0.910  & 0.825  & 1.000  \\
& 100  &    0.062  & 0.922  & 0.895  & 1.000   &   &     0.071  & 0.917  & 0.835  & 1.000  \\
& 200  &    0.044  & 0.939  & 0.910  & 1.000   &   &     0.050  & 0.932  & 0.854  & 1.000  \\
& 400  &    0.031  & 0.939  & 0.911  & 1.000   &   &     0.035  & 0.937  & 0.855  & 1.000  \\
200 & 50   &    0.062  & 0.918  & 0.885  & 1.000   &   &     0.071  & 0.912  & 0.821  & 1.000  \\
& 100  &    0.044  & 0.932  & 0.903  & 1.000   &   &     0.050  & 0.927  & 0.840  & 1.000  \\
& 200  &    0.032  & 0.934  & 0.903  & 1.000   &   &     0.035  & 0.936  & 0.847  & 1.000  \\
& 400  &    0.023  & 0.939  & 0.908  & 1.000   &   &     0.025  & 0.936  & 0.855  & 1.000  \\
400 & 50   &    0.045  & 0.918  & 0.880  & 1.000   &   &     0.050  & 0.907  & 0.817  & 1.000  \\
& 100  &    0.032  & 0.930  & 0.894  & 1.000   &   &     0.036  & 0.925  & 0.833  & 1.000  \\
& 200  &    0.023  & 0.936  & 0.903  & 1.000   &   &     0.025  & 0.932  & 0.851  & 1.000  \\
& 400  &    0.016  & 0.940  & 0.910  & 1.000   &   &     0.018  & 0.939  & 0.855  & 1.000  \\
\hline 
\end{tabular}
\end{table}

\begin{table}[htb]
\renewcommand{\arraystretch}{0.6} 
\caption{Simulation Results of DGP 1 for $d = 500$}\label{Table1.3}
\begin{tabular}{cccccc c cccc }
\hline
&     & \multicolumn{4}{c}{$\rho_e=0.2,\delta_{\varepsilon_e}=0.2$} & & \multicolumn{4}{c}{$\rho_e=0.2,\delta_{\varepsilon_e}=0.5$}\\
\cline{3-6} \cline{8-11}
\hline
$N$ & $T$ & \text{RMSE} & ECR & ECR2 & RSC & & \text{RMSE} & ECR & ECR2 & RSC\\
\hline
50  & 50   &     0.148  & 0.968  & 0.913  & 0.993    &   &      0.163  & 0.966  & 0.852  & 0.953  \\
& 100  &     0.104  & 0.960  & 0.910  & 1.000    &   &      0.111  & 0.955  & 0.857  & 1.000  \\
& 200  &     0.073  & 0.959  & 0.916  & 1.000    &   &      0.079  & 0.955  & 0.866  & 1.000  \\
& 400  &     0.053  & 0.948  & 0.908  & 1.000    &   &      0.057  & 0.949  & 0.847  & 1.000  \\
100 & 50   &     0.105  & 0.970  & 0.902  & 1.000    &   &      0.113  & 0.959  & 0.850  & 1.000  \\
& 100  &     0.075  & 0.966  & 0.909  & 1.000    &   &      0.080  & 0.959  & 0.856  & 1.000  \\
& 200  &     0.052  & 0.957  & 0.917  & 1.000    &   &      0.056  & 0.953  & 0.861  & 1.000  \\
& 400  &     0.037  & 0.964  & 0.931  & 1.000    &   &      0.039  & 0.958  & 0.877  & 1.000  \\
200 & 50   &     0.075  & 0.970  & 0.902  & 1.000    &   &      0.081  & 0.966  & 0.850  & 1.000  \\
& 100  &     0.052  & 0.969  & 0.915  & 1.000    &   &      0.056  & 0.962  & 0.858  & 1.000  \\
& 200  &     0.037  & 0.959  & 0.915  & 1.000    &   &      0.040  & 0.954  & 0.862  & 1.000  \\
& 400  &     0.026  & 0.945  & 0.913  & 1.000    &   &      0.028  & 0.943  & 0.855  & 1.000  \\
400& 50   &     0.053  & 0.970  & 0.905  & 1.000    &   &      0.057  & 0.967  & 0.847  & 1.000  \\
& 100  &     0.038  & 0.973  & 0.916  & 1.000    &   &      0.040  & 0.968  & 0.864  & 1.000  \\
& 200  &     0.027  & 0.955  & 0.911  & 1.000    &   &      0.028  & 0.956  & 0.857  & 1.000  \\
& 400  &     0.019  & 0.953  & 0.918  & 1.000    &   &      0.020  & 0.953  & 0.869  & 1.000  \\
\hline
&     & \multicolumn{4}{c}{$\rho_e=0.5,\delta_{\varepsilon_e}=0.2$}& & \multicolumn{4}{c}{$\rho_e=0.5,\delta_{\varepsilon_e}=0.5$}\\
\cline{3-6} \cline{8-11}      
\hline
$N$ & $T$ & \text{RMSE} & ECR & ECR2 & RSC & & \text{RMSE} & ECR & ECR2 & RSC\\
\hline
50  & 50   &    0.174  & 0.951  & 0.894  & 0.914   &   &     0.191  & 0.945  & 0.828  & 0.864  \\
& 100  &    0.118  & 0.946  & 0.900  & 1.000   &   &     0.130  & 0.936  & 0.835  & 1.000  \\
& 200  &    0.083  & 0.948  & 0.905  & 1.000   &   &     0.092  & 0.942  & 0.844  & 1.000  \\
& 400  &    0.060  & 0.939  & 0.905  & 1.000   &   &     0.066  & 0.937  & 0.837  & 1.000  \\
100 & 50   &    0.119  & 0.948  & 0.885  & 1.000   &   &     0.130  & 0.937  & 0.825  & 1.000  \\
& 100  &    0.084  & 0.949  & 0.896  & 1.000   &   &     0.092  & 0.935  & 0.835  & 1.000  \\
& 200  &    0.059  & 0.948  & 0.912  & 1.000   &   &     0.065  & 0.937  & 0.850  & 1.000  \\
& 400  &    0.041  & 0.952  & 0.925  & 1.000   &   &     0.046  & 0.948  & 0.869  & 1.000  \\
200 & 50   &    0.084  & 0.955  & 0.885  & 1.000   &   &     0.092  & 0.943  & 0.825  & 1.000  \\
& 100  &    0.059  & 0.949  & 0.896  & 1.000   &   &     0.065  & 0.940  & 0.840  & 1.000  \\
& 200  &    0.042  & 0.943  & 0.906  & 1.000   &   &     0.046  & 0.939  & 0.854  & 1.000  \\
& 400  &    0.029  & 0.941  & 0.913  & 1.000   &   &     0.032  & 0.940  & 0.845  & 1.000  \\
400 & 50   &    0.060  & 0.950  & 0.879  & 1.000   &   &     0.065  & 0.938  & 0.817  & 1.000  \\
& 100  &    0.042  & 0.954  & 0.903  & 1.000   &   &     0.046  & 0.944  & 0.841  & 1.000  \\
& 200  &    0.030  & 0.943  & 0.904  & 1.000   &   &     0.033  & 0.941  & 0.846  & 1.000  \\
& 400  &    0.021  & 0.948  & 0.916  & 1.000   &   &     0.023  & 0.945  & 0.860  & 1.000  \\
\hline 
\end{tabular}
\end{table}

\begin{table}[htb]
\setlength{\tabcolsep}{4pt} 
\renewcommand{\arraystretch}{0.65} 
\caption{Simulation Results of DGP 2}\label{Table2.1}
\centering
\scalebox{0.85}{\begin{tabular}{cc cccccc c cccccc }
\hline
$d = 50$&     & \multicolumn{6}{c}{$\rho_e=0.2,\delta_{\varepsilon_e}=0.2$} & & \multicolumn{6}{c}{$\rho_e=0.2,\delta_{\varepsilon_e}=0.5$}\\
\cline{3-8} \cline{10-15}
\hline
$N$ & $T$ & RMSE1& RMSE2 & ECR &TPR & FPR&EER & &  RMSE1& RMSE2&TPR & FPR & RSC&EER\\
\hline
20   & 20  &      0.499  & 0.248  & 0.660   & 0.959 & 0.555  & 0.841  &  &   0.385  & 0.241  & 0.710   & 0.953 & 0.526  & 0.856  \\ 
& 30  &      0.378  & 0.165  & 0.768   & 0.994 & 0.364  & 0.898  &  &   0.284  & 0.163  & 0.796   & 0.993 & 0.333  & 0.908  \\ 
& 40  &      0.358  & 0.139  & 0.806   & 1.000 & 0.278  & 0.941  &  &   0.241  & 0.129  & 0.832   & 0.999 & 0.222  & 0.956  \\ 
30   & 20  &      0.388  & 0.175  & 0.734   & 0.995 & 0.370  & 0.901  &  &   0.308  & 0.169  & 0.758   & 0.989 & 0.369  & 0.914  \\ 
& 30  &      0.262  & 0.108  & 0.835   & 1.000 & 0.194  & 0.964  &  &   0.220  & 0.108  & 0.840   & 1.000 & 0.168  & 0.963  \\ 
& 40  &      0.219  & 0.081  & 0.873   & 1.000 & 0.101  & 0.983  &  &   0.188  & 0.084  & 0.873   & 1.000 & 0.096  & 0.990  \\ 
40   & 20  &      0.355  & 0.143  & 0.756   & 1.000 & 0.255  & 0.930  &  &   0.268  & 0.132  & 0.783   & 1.000 & 0.265  & 0.952  \\ 
& 30  &      0.224  & 0.082  & 0.860   & 1.000 & 0.105  & 0.982  &  &   0.194  & 0.083  & 0.860   & 1.000 & 0.122  & 0.983  \\ 
& 40  &      0.187  & 0.063  & 0.893   & 1.000 & 0.058  & 0.982  &  &   0.163  & 0.066  & 0.898   & 1.000 & 0.051  & 0.989  \\ 			\hline
&     & \multicolumn{5}{c}{$\rho_e=0.5,\delta_{\varepsilon_e}=0.2$} & & \multicolumn{5}{c}{$\rho_e=0.5,\delta_{\varepsilon_e}=0.5$}\\
\cline{3-8} \cline{10-15}
\hline
$N$ & $T$ & RMSE1& RMSE2 & ECR &TPR& FPR&EER &           &  RMSE1& RMSE2 & ECR &TPR& FPR&EER\\
\hline
20   & 20  &         0.481  & 0.266  & 0.640   & 0.949& 0.698   & 0.746    &  &  0.394  & 0.260  & 0.673   & 0.936 & 0.714  & 0.789  \\ 
& 30  &         0.377  & 0.185  & 0.739   & 0.988& 0.586   & 0.810    &  &  0.304  & 0.180  & 0.764   & 0.983 & 0.576  & 0.811  \\ 
& 40  &         0.376  & 0.157  & 0.766   & 1.000& 0.526   & 0.805    &  &  0.264  & 0.143  & 0.811   & 0.998 & 0.486  & 0.852  \\ 
30   & 20  &         0.345  & 0.183  & 0.713   & 0.987& 0.538   & 0.819    &  &  0.308  & 0.185  & 0.726   & 0.982 & 0.557  & 0.808  \\ 
& 30  &         0.257  & 0.121  & 0.804   & 0.999& 0.381   & 0.874    &  &  0.230  & 0.123  & 0.809   & 0.999 & 0.390  & 0.871  \\ 
& 40  &         0.227  & 0.096  & 0.846   & 1.000& 0.241   & 0.910    &  &  0.204  & 0.098  & 0.853   & 1.000 & 0.271  & 0.907  \\ 
40   & 20  &         0.295  & 0.148  & 0.737   & 0.999& 0.390   & 0.862    &  &  0.264  & 0.146  & 0.743   & 0.998 & 0.461  & 0.876  \\ 
& 30  &         0.218  & 0.094  & 0.832   & 1.000& 0.244   & 0.920    &  &  0.203  & 0.097  & 0.833   & 1.000 & 0.287  & 0.901  \\ 
& 40  &         0.193  & 0.075  & 0.869   & 1.000& 0.151   & 0.934    &  &  0.175  & 0.080  & 0.871   & 1.000 & 0.171  & 0.920  \\ 
\hline \hline
$d = 500$&     & \multicolumn{6}{c}{$\rho_e=0.2,\delta_{\varepsilon_e}=0.2$} & & \multicolumn{6}{c}{$\rho_e=0.2,\delta_{\varepsilon_e}=0.5$}\\
\cline{3-8} \cline{10-15}
\hline
$N$ & $T$ & RMSE1& RMSE2 & ECR & TPR&FPR&EER & &  RMSE1& RMSE2  & TPR&FPR& RSC&EER\\
\hline
20   & 20  &      0.794  & 0.438  & 0.463   & 0.837  & 0.800 & 0.650  &  &   0.614  & 0.457  & 0.552   & 0.726  & 0.735 & 0.773  \\ 
& 30  &      0.588  & 0.312  & 0.546   & 0.967  & 0.575 & 0.838  &  &   0.402  & 0.310  & 0.632   & 0.951  & 0.497 & 0.898  \\ 
& 40  &      0.563  & 0.245  & 0.577   & 0.995  & 0.451 & 0.856  &  &   0.327  & 0.226  & 0.690   & 0.987  & 0.373 & 0.934  \\ 
30   & 20  &      0.601  & 0.313  & 0.584   & 0.965  & 0.579 & 0.819  &  &   0.448  & 0.320  & 0.657   & 0.933  & 0.529 & 0.891  \\ 
& 30  &      0.352  & 0.192  & 0.711   & 0.999  & 0.345 & 0.945  &  &   0.302  & 0.200  & 0.722   & 0.997  & 0.366 & 0.949  \\ 
& 40  &      0.300  & 0.141  & 0.764   & 1.000  & 0.252 & 0.976  &  &   0.254  & 0.148  & 0.773   & 1.000  & 0.235 & 0.983  \\ 
40   & 20  &      0.540  & 0.243  & 0.619   & 0.995  & 0.435 & 0.867  &  &   0.387  & 0.233  & 0.687   & 0.993  & 0.427 & 0.928  \\ 
& 30  &      0.301  & 0.142  & 0.761   & 0.999  & 0.262 & 0.974  &  &   0.260  & 0.147  & 0.767   & 0.998  & 0.268 & 0.975  \\ 
& 40  &      0.246  & 0.097  & 0.833   & 1.000  & 0.166 & 0.986  &  &   0.218  & 0.104  & 0.836   & 1.000  & 0.186 & 0.987  \\ 
\hline
&     & \multicolumn{5}{c}{$\rho_e=0.5,\delta_{\varepsilon_e}=0.2$} & & \multicolumn{5}{c}{$\rho_e=0.5,\delta_{\varepsilon_e}=0.5$}\\
\cline{3-8} \cline{10-15}
\hline
$N$ & $T$ & RMSE1& RMSE2 & ECR  & TPR&FPR&EER &           &  RMSE1& RMSE2 & ECR  & TPR&FPR&EER\\
\hline
20   & 20  &      0.763  & 0.452  & 0.466   & 0.793 & 0.878  & 0.578   &  &    0.618  & 0.467  & 0.542  & 0.699 & 0.854   & 0.679  \\ 
& 30  &      0.591  & 0.327  & 0.549   & 0.960 & 0.791  & 0.695   &  &    0.429  & 0.325  & 0.626  & 0.925 & 0.739   & 0.781  \\ 
& 40  &      0.572  & 0.259  & 0.564   & 0.991 & 0.703  & 0.731   &  &    0.356  & 0.236  & 0.689  & 0.985 & 0.646   & 0.824  \\ 
30   & 20  &      0.517  & 0.319  & 0.610   & 0.943 & 0.710  & 0.726   &  &    0.430  & 0.323  & 0.649  & 0.912 & 0.709   & 0.784  \\ 
& 30  &      0.348  & 0.204  & 0.703   & 0.998 & 0.578  & 0.846   &  &    0.318  & 0.211  & 0.706  & 0.994 & 0.592   & 0.855  \\ 
& 40  &      0.307  & 0.157  & 0.748   & 1.000 & 0.463  & 0.908   &  &    0.274  & 0.161  & 0.755  & 0.999 & 0.503   & 0.894  \\ 
40   & 20  &      0.425  & 0.241  & 0.658   & 0.990 & 0.605  & 0.803   &  &    0.358  & 0.243  & 0.676  & 0.980 & 0.637   & 0.836  \\ 
& 30  &      0.291  & 0.155  & 0.747   & 0.999 & 0.490  & 0.883   &  &    0.273  & 0.162  & 0.743  & 0.997 & 0.534   & 0.882  \\ 
& 40  &      0.250  & 0.109  & 0.807   & 1.000 & 0.377  & 0.923   &  &    0.232  & 0.114  & 0.810  & 1.000 & 0.428   & 0.906  \\ 
\hline 
\end{tabular}}
\end{table}

\begin{table}[htb]
\setlength{\tabcolsep}{4pt} 
\renewcommand{\arraystretch}{0.65} 
\caption{Simulation Results of DGP 2 for $d=50$}\label{Table2.2}
\centering
\scalebox{0.95}{\begin{tabular}{cc cccccc c cccccc }
\hline
&     & \multicolumn{6}{c}{$\rho_e=0.2,\delta_{\varepsilon_e}=0.2$} & & \multicolumn{6}{c}{$\rho_e=0.2,\delta_{\varepsilon_e}=0.5$}\\
\cline{3-8} \cline{10-15}
\hline
$N$ & $T$ & RMSE1& RMSE2 & ECR &TPR & FPR&EER & &  RMSE1& RMSE2&TPR & FPR & RSC&EER\\
\hline
50  & 50   &             0.154  & 0.050  & 0.903 &       1.000    & 0.018  & 0.995          &            &        0.133  & 0.053  & 0.900  &1.000& 0.020  & 0.997  \\
& 100  &             0.119  & 0.035  & 0.920       &  1.000   & 0.007  & 1.000          &            &              0.089  & 0.037  & 0.925 &1.000 & 0.008  & 1.000  \\
& 200  &             0.126  & 0.024  & 0.938       & 1.000    & 0.003  & 0.997          &            &              0.063  & 0.026  & 0.935 &1.000 & 0.000  & 1.000  \\
& 400  &             0.090  & 0.017  & 0.945       &1.000     & 0.000  & 1.000          &            &              0.078  & 0.018  & 0.944&1.000  & 0.000  & 1.000  \\
100 & 50   &             0.115  & 0.034  & 0.906    &1.000        & 0.015  & 1.000          &            &           0.098  & 0.037  & 0.906 &1.000 & 0.010  & 1.000  \\
& 100  &             0.077  & 0.024  & 0.928       & 1.000    & 0.000  & 1.000          &            &              0.067  & 0.026  & 0.926 &1.000 & 0.000  & 1.000  \\
& 200  &             0.063  & 0.017  & 0.929       & 1.000    & 0.003  & 1.000          &            &              0.045  & 0.018  & 0.924&1.000  & 0.000  & 1.000  \\
& 400  &             0.059  & 0.012  & 0.945       & 1.000    & 0.003  & 1.000          &            &              0.031  & 0.013  & 0.941 &1.000 & 0.000  & 1.000  \\
200 & 50   &             0.132  & 0.024  & 0.916    &1.000         & 0.003  & 0.997          &            &           0.123  & 0.026  & 0.914  &1.000& 0.015  & 0.997  \\
& 100  &             0.062  & 0.017  & 0.914       &  1.000   & 0.000  & 1.000          &            &              0.050  & 0.019  & 0.913 &1.000 & 0.000  & 1.000  \\
& 200  &             0.038  & 0.012  & 0.943       & 1.000    & 0.000  & 1.000          &            &              0.033  & 0.013  & 0.944 &1.000 & 0.000  & 1.000  \\
& 400  &             0.034  & 0.009  & 0.931       & 1.000    & 0.000  & 1.000          &            &              0.023  & 0.009  & 0.931 &1.000 & 0.000  & 1.000  \\
400 & 50   &             0.088  & 0.017  & 0.902    & 1.000       & 0.000  & 1.000          &            &           0.089  & 0.019  & 0.902  &1.000& 0.000  & 1.000  \\
& 100  &             0.058  & 0.012  & 0.904       &  1.000  & 0.000  & 1.000          &            &              0.060  & 0.013  & 0.906 &1.000 & 0.000  & 1.000  \\
& 200  &             0.034  & 0.008  & 0.933       & 1.000   & 0.000  & 1.000          &            &              0.026  & 0.009  & 0.936 &1.000 & 0.000  & 1.000  \\
& 400  &             0.020  & 0.006  & 0.931       & 1.000   & 0.000  & 1.000          &            &              0.017  & 0.006  & 0.942 &1.000 & 0.000  & 1.000  \\
\hline
&     & \multicolumn{5}{c}{$\rho_e=0.5,\delta_{\varepsilon_e}=0.2$} & & \multicolumn{5}{c}{$\rho_e=0.5,\delta_{\varepsilon_e}=0.5$}\\
\cline{3-8} \cline{10-15}
\hline
$N$ & $T$ & RMSE1& RMSE2 & ECR &TPR& FPR&EER &           &  RMSE1& RMSE2 & ECR &TPR& FPR&EER\\
\hline
50  & 50   &       0.154  & 0.061  & 0.877  &1.000& 0.080  & 0.972              &                  &          0.142  & 0.064  & 0.864&1.000  & 0.120  & 0.952  \\
& 100  &             0.121  & 0.042  & 0.913&1.000  & 0.068  & 0.970        &                  &                0.099  & 0.045  & 0.903 &1.000 & 0.048  & 0.970  \\
& 200  &             0.145  & 0.029  & 0.926&1.000  & 0.010  & 0.995        &                  &                0.079  & 0.031  & 0.923 &1.000 & 0.030  & 0.987  \\
& 400  &             0.100  & 0.021  & 0.938 &1.000 & 0.005  & 1.000        &                  &                0.090  & 0.022  & 0.939 &1.000 & 0.008  & 1.000  \\
100 & 50   &          0.107  & 0.041  & 0.881 &1.000 & 0.033  & 0.985          &                  &              0.100  & 0.045  & 0.882 &1.000 & 0.050  & 0.980  \\
& 100  &             0.080  & 0.029  & 0.910 &1.000 & 0.013  & 0.997      &                  &                  0.072  & 0.031  & 0.909 &1.000 & 0.015  & 0.997  \\
& 200  &             0.063  & 0.020  & 0.922 &1.000 & 0.010  & 1.000      &                  &                  0.050  & 0.022  & 0.921 &1.000 & 0.010  & 1.000  \\
& 400  &             0.066  & 0.014  & 0.942 &1.000 & 0.008  & 1.000      &                  &                  0.041  & 0.016  & 0.934 &1.000 & 0.018  & 0.997  \\
200 & 50   &          0.126  & 0.029  & 0.892 &1.000 & 0.025  & 0.992          &                  &              0.080  & 0.031  & 0.894 &1.000 & 0.050  & 0.992  \\
& 100  &             0.054  & 0.021  & 0.906 &1.000 & 0.000  & 1.000       &                  &                 0.051  & 0.023  & 0.901 &1.000 & 0.003  & 1.000  \\
& 200  &             0.040  & 0.014  & 0.938 &1.000 & 0.003  & 1.000       &                  &                 0.037  & 0.015  & 0.942 &1.000 & 0.010  & 1.000  \\
& 400  &             0.034  & 0.010  & 0.930 &1.000 & 0.000  & 1.000       &                  &                 0.027  & 0.011  & 0.927&1.000  & 0.000  & 1.000  \\
400 & 50   &          0.090  & 0.021  & 0.876 &1.000 & 0.000  & 1.000          &                  &              0.091  & 0.023  & 0.876 &1.000 & 0.003  & 1.000  \\
& 100  &             0.061  & 0.015  & 0.895 &1.000 & 0.005  & 0.997       &                  &                 0.042  & 0.016  & 0.909 &1.000 & 0.010  & 1.000  \\
& 200  &             0.028  & 0.010  & 0.923 &1.000 & 0.000  & 1.000       &                  &                 0.027  & 0.011  & 0.926 &1.000 & 0.000  & 1.000  \\
& 400  &             0.021  & 0.007  & 0.929 &1.000 & 0.000  & 1.000       &                  &                 0.019  & 0.008  & 0.936 &1.000 & 0.000  & 1.000  \\
\hline 
\end{tabular}}
\end{table}

\begin{table}[htb]
\setlength{\tabcolsep}{4pt} 
\renewcommand{\arraystretch}{0.65} 
\centering
\caption{Simulation Results of DGP 2 for $d = 500$}\label{Table2.3}
\scalebox{0.95}{\begin{tabular}{cc cccccc c cccccc }
\hline
&     & \multicolumn{6}{c}{$\rho_e=0.2,\delta_{\varepsilon_e}=0.2$} & & \multicolumn{6}{c}{$\rho_e=0.2,\delta_{\varepsilon_e}=0.5$}\\
\cline{3-8} \cline{10-15}
\hline
$N$ & $T$ & RMSE1& RMSE2 & ECR & TPR&FPR&EER & &  RMSE1& RMSE2  & ECR  & TPR&FPR&EER\\
\hline
50  & 50   &            0.192  & 0.054  & 0.899 &1.000 & 0.028  & 0.995         &            &             0.172  & 0.056  & 0.900  &1.000& 0.035  & 0.995  \\
& 100  &            0.144  & 0.035  & 0.915 &1.000 & 0.008  & 0.997     &            &                 0.115  & 0.037  & 0.926  &1.000& 0.005  & 0.997  \\
& 200  &            0.163  & 0.024  & 0.923 &1.000 & 0.003  & 0.997    &            &                  0.080  & 0.026  & 0.930  &1.000& 0.000  & 1.000  \\
& 400  &            0.113  & 0.017  & 0.937 &1.000 & 0.003  & 1.000    &            &                  0.101  & 0.018  & 0.937  &1.000 & 0.005  & 1.000  \\
100 & 50   &            0.142  & 0.035  & 0.906 &1.000 & 0.025  & 0.995       &            &               0.124  & 0.037  & 0.912  &1.000 & 0.013  & 0.995  \\
& 100  &            0.095  & 0.024  & 0.908 &1.000 & 0.000  & 1.000     &            &                 0.085  & 0.026  & 0.919  &1.000 & 0.000  & 1.000  \\
& 200  &            0.074  & 0.017  & 0.934 &1.000 & 0.008  & 1.000    &            &                  0.058  & 0.018  & 0.941  &1.000 & 0.000  & 1.000  \\
& 400  &            0.075  & 0.012  & 0.942 &1.000 & 0.000  & 1.000    &            &                  0.040  & 0.013  & 0.939  &1.000 & 0.000  & 1.000  \\
200 & 50   &            0.164  & 0.024  & 0.900 &1.000 & 0.000  & 1.000       &            &               0.158  & 0.026  & 0.905  &1.000  & 0.000  & 1.000  \\
& 100  &            0.076  & 0.017  & 0.914 &1.000 & 0.003  & 1.000   &            &                   0.064  & 0.019  & 0.912  &1.000 & 0.000  & 1.000  \\
& 200  &            0.047  & 0.012  & 0.939 &1.000 & 0.000  & 1.000  &            &                    0.042  & 0.012  & 0.938  &1.000 & 0.000  & 1.000  \\
& 400  &            0.039  & 0.008  & 0.943 &1.000 & 0.000  & 1.000  &            &                    0.028  & 0.009  & 0.940  &1.000& 0.000  & 1.000  \\
400 & 50   &            0.112  & 0.017  & 0.899 &1.000 & 0.003  & 1.000    &            &                  0.109  & 0.019  & 0.900  &1.000 & 0.000  & 1.000  \\
& 100  &            0.075  & 0.012  & 0.909 &1.000 & 0.000  & 1.000    &            &                  0.074  & 0.013  & 0.911  &1.000& 0.000  & 1.000  \\
& 200  &            0.039  & 0.008  & 0.938 &1.000 & 0.000  & 1.000   &            &                   0.031  & 0.009  & 0.935  &1.000 & 0.000  & 1.000  \\
& 400  &            0.025  & 0.006  & 0.943 &1.000 & 0.000  & 1.000   &            &                   0.022  & 0.006  & 0.939  &1.000& 0.000  & 1.000  \\
\hline
&     & \multicolumn{5}{c}{$\rho_e=0.5,\delta_{\varepsilon_e}=0.2$} & & \multicolumn{5}{c}{$\rho_e=0.5,\delta_{\varepsilon_e}=0.5$}\\
\cline{3-8} \cline{10-15}
\hline
$N$ & $T$ & RMSE1& RMSE2 & ECR  & TPR&FPR&EER &           &  RMSE1& RMSE2 & ECR  & TPR&FPR&EER\\
\hline
50  & 50   &           0.197  & 0.062  & 0.887 &1.000 & 0.118  & 0.965         &                  &                 0.184  & 0.066  & 0.879&1.000  & 0.145  & 0.950  \\
& 100  &                 0.151  & 0.042  & 0.904  &1.000& 0.028  & 0.990     &                  &               0.129  & 0.044  & 0.910 &1.000 & 0.035  & 0.980  \\
& 200  &                 0.174  & 0.030  & 0.923 &1.000 & 0.008  & 0.997    &                  &                0.097  & 0.032  & 0.919 &1.000 & 0.035  & 0.970  \\
& 400  &                 0.125  & 0.021  & 0.928  &1.000& 0.005  & 1.000    &                  &                0.111  & 0.023  & 0.930 &1.000 & 0.020  & 0.992  \\
100 & 50   &              0.132  & 0.041  & 0.890 &1.000 & 0.035  & 0.990        &                  &               0.128  & 0.044  & 0.893 &1.000 & 0.048  & 0.980  \\
& 100  &                 0.098  & 0.029  & 0.909 &1.000 & 0.003  & 0.997     &                  &               0.091  & 0.031  & 0.914 &1.000 & 0.013  & 0.997  \\
& 200  &                 0.077  & 0.021  & 0.926 &1.000 & 0.015  & 1.000    &                  &                0.065  & 0.022  & 0.926 &1.000 & 0.015  & 1.000  \\
& 400  &                 0.083  & 0.014  & 0.936 &1.000 & 0.005  & 0.997    &                  &                0.050  & 0.016  & 0.932 &1.000 & 0.005  & 1.000  \\
200 & 50   &              0.161  & 0.029  & 0.884 &1.000 & 0.005  & 0.970       &                  &                0.098  & 0.031  & 0.882 &1.000 & 0.023  & 0.980  \\
& 100  &                 0.068  & 0.021  & 0.899 &1.000 & 0.005  & 1.000     &                  &               0.066  & 0.023  & 0.890 &1.000 & 0.010  & 1.000  \\
& 200  &                 0.049  & 0.014  & 0.930 &1.000 & 0.003  & 1.000    &                  &                0.046  & 0.015  & 0.932 &1.000 & 0.005  & 1.000  \\
& 400  &                 0.042  & 0.010  & 0.939 &1.000 & 0.008  & 1.000    &                  &                0.033  & 0.011  & 0.936 &1.000 & 0.005  & 1.000  \\
400 & 50   &              0.114  & 0.021  & 0.877 &1.000 & 0.003  & 1.000        &                  &               0.112  & 0.023  & 0.878 &1.000 & 0.020  & 0.997  \\
& 100  &                 0.076  & 0.015  & 0.898 &1.000 & 0.005  & 0.997    &                  &                0.050  & 0.016  & 0.908 &1.000 & 0.005  & 1.000  \\
& 200  &                 0.034  & 0.010  & 0.935 &1.000 & 0.000  & 1.000   &                  &                 0.033  & 0.011  & 0.934 &1.000 & 0.005  & 1.000  \\
& 400  &                 0.025  & 0.007  & 0.934 &1.000 & 0.000  & 1.000   &                  &                 0.024  & 0.008  & 0.939 &1.000 & 0.000  & 1.000  \\
\hline 
\end{tabular}}
\end{table}

\section{Technical Lemmas} \label{AP.A4}

In this section, we list several technical lemmas, and their proofs are given in Appendix C of the online document.

\begin{lemma}\label{LemmaB1}
Let $X_1,\ldots,X_T$ be mean zero independent random variables with $\max_{t\in [T]}\|X_t\|_q<\infty$ for some $q>2$. Then,
$$\textstyle
\Pr\left(\max_{t\in [T]}\sum_{s=1}^{t}X_s \geq x \right) \leq \left(1 + \frac{2}{q} \right)^{q} \frac{\sum_{t=1}^{T}\|X_t\|_q^q}{x^q} + 2\exp\left(  \frac{- 2e^{-q}(q+2)^{-2} x^2}{\sum_{t=1}^{T}\|X_t\|_2^2}\right).
$$
\end{lemma}

Corollary 1.8 in \cite{nagaev1979large} shows that the above inequality holds for $\sum_{s=1}^{T}X_s$, while \cite{borovkov1973notes} proves that the above inequality also holds for $\max_{t\in [T]}\sum_{s=1}^{t}X_s$.

\begin{lemma}\label{LemmaB2}
Let $X_1,\ldots,X_T$ be mean zero independent random variables, and let $x > y$ and $y \geq (4\sum_{t=1}^{T}\|X_t\|_q^q)^{1/q}$.
\begin{itemize}[leftmargin=*, itemsep=0.5pt, parsep=0.5pt, topsep=0.6pt]
\item [1.] 	 Suppose $\max_{t\in [T]}\|X_t\|_q<\infty$ for some $1\leq q \leq2$. Then,
$$\textstyle
\Pr\left(\sum_{t=1}^{T}X_t \geq x \right) \leq \sum_{t=1}^{T}\Pr\left(X_t \geq y \right) + \left(\frac{e^2 \sum_{t=1}^{T}\|X_t\|_q^q}{xy^{q-1}} \right)^{x/(2y)};
$$

\item [2.] 	 Suppose $\max_{t\in [T]}\|X_t\|_q<\infty$ for some $q \geq 2$, $\beta = q/(q+2)$ and $\alpha = 1-\beta$. Then,
$$\textstyle
\Pr\left(\sum_{t=1}^{T}X_t \geq x \right) \leq \sum_{t=1}^{T}\Pr\left(X_t \geq y \right) + \exp\left(\frac{-\alpha^2x^2}{2e^{t}\sum_{t=1}^{T}\|X_t\|_2^2 }\right)+ \left(\frac{ \sum_{t=1}^{T}\|X_t\|_q^q}{\beta xy^{q-1}} \right)^{\beta x/y}.
$$
\end{itemize}
\end{lemma}

Lemma \ref{LemmaB2} is Corollaries 1.6--1.7 in \cite{nagaev1979large}.

\smallskip

Before proceeding further, we define 	$\delta_{q}(u_j,t) \coloneqq  \sup_{|\mathbf{w}|_2<\infty}\|\mathbf{w}^\top[\mathbf{u}_j (\mathcal{F}_t) -\mathbf{u}_j (\mathcal{F}_t^*)] \|_q$,
\begin{eqnarray*}
\delta_{q}(x_jx_{j'},t) &\coloneqq & \sup_{|\mathbf{w}|_2<\infty}\|\mathbf{w}^\top[\mathbf{g}_j(\mathcal{F}_t) \circ \mathbf{g}_{j'}(\mathcal{F}_t) - \mathbf{g}_j(\mathcal{F}_t^*) \circ \mathbf{g}_{j'}(\mathcal{F}_t^*)]\|_{q},\nonumber \\
\Delta_{q}(u_j,m) &\coloneqq & \sum_{k=m}^\infty \delta_{q}(u_j,k),\quad\Delta_{\nu}(x_jx_{j'},m) \coloneqq  \sum_{k=m}^\infty \delta_{\nu}(x_jx_{j'},k).
\end{eqnarray*}

\begin{lemma}\label{LemmaB3}
\item
\begin{enumerate}[leftmargin=*, itemsep=0.5pt, parsep=0.5pt, topsep=0.6pt]
\item Let Assumption \ref{ASSUMPTION1} hold. Then for some constant $c_q>0$ only depending on $q$
$$
\|\max_{t\in [T]}|S_{j,Nt}| \|_q 
\leq c_q \sqrt{NT} \Delta_{2}(u_j,0) + c_q \sqrt{N}T^{1/q}\sum_{m=1}^{\infty}(m\wedge T)^{1/2-1/q}\delta_q(u_j,m),
$$
where $S_{j,Nt} = \sum_{s=1}^{t}\sum_{i=1}^{N}x_{j,is}e_{is}$;

\item Let Assumption \ref{ASSUMPTION2}.1 hold. Then for some constant $c_q>0$ only depending on $q$
\begin{eqnarray*}
&&\left\|\max_{t\in [T]}\left|\sum_{s=1}^{t}\sum_{i=1}^{N}[x_{j,is}x_{j',is}-E(x_{j,is}x_{j',is})]\right|\ \right\|_{q}\\
&\leq&c_{q} \sqrt{NT} \Delta_{2}(x_jx_{j'},0) + c_{q} \sqrt{N}T^{1/q}\sum_{m=1}^{\infty}(m\wedge T)^{1/2-1/q}\delta_q(x_jx_{j'},m).
\end{eqnarray*}

\end{enumerate}	
\end{lemma}

\begin{lemma}\label{LemmaB4}
Let Assumption \ref{ASSUMPTION2}.1 hold. Then for any $x \geq \sqrt{NT}\mu_x^{1+1/q}$ and some positive constants $C_1,C_2,C_3 > 0$ not depending on $x,N,T$  
$$
\Pr\left(\max_{t\in [T]}\left|\sum_{k=1}^{t}\sum_{i=1}^{N}[x_{j,ik}x_{j',ik} - E(x_{j,ik}x_{j',ik})]\right|\geq x\right)
\leq  C_1\frac{TN^{q/2}}{x^q} + C_2 \exp\left( - C_3 \frac{x^2}{TN} \right),
$$
where $\mu_{x} \coloneqq \sum_{m=1}^{\infty}\mu_{x,m} < \infty$ and $\mu_{x,m} \coloneqq \left(m^{q/2-1}\delta_{q}^{q}(x_jx_{j'},m)\right)^{1/(q+1)}$.
\end{lemma}

\begin{lemma}\label{LemmaB5}
Let Assumption \ref{ASSUMPTION2} hold.
\begin{enumerate}[leftmargin=*, itemsep=0.5pt, parsep=0.5pt, topsep=0.6pt]
\item Let $\widehat{\pmb{\Sigma}}_{x} = \mathbf{X}^\top\mathbf{X}/(NT)$. For any $\mathbf{v}\in\mathbb{V}$, we have
$$
\mathbf{v}^\top \widehat{\pmb{\Sigma}}_{x} \mathbf{v} \geq \mathbf{v}^\top \pmb{\Sigma}_{x} \mathbf{v} -  C_0s\sqrt{\log d/(NT)}|\mathbf{v}|_2^2
$$
with probability larger than $1 - C_1 (\frac{dT^{1-q/2}}{(\log d)^{q/2}} +  d^{-C_2} )$
for some constant $C_0,C_1,C_2>0$ not depending on $N$ and $T$.

\item $|\mathbf{X}_{-j}^\top\pmb{\eta}_j/(NT)|_{\infty} = O_P(\sqrt{\log d/(NT)})$,  

$|\widehat{\pmb{\gamma}}_j - \pmb{\gamma}_j|_2 = O_P(\sqrt{s_j \log d/(NT)})$,

and $|\widehat{\pmb{\gamma}}_j - \pmb{\gamma}_j|_1 = O_P(s_j \sqrt{ \log d/(NT)})$ for $j=1,2,\ldots,d$,  

where $s_j = |\{k\neq j\mid \Omega_{j,k} \neq 0 \}|$ denotes the row sparsity with respect to rows of $\pmb{\Omega}_x = \pmb{\Sigma}_x^{-1}$.

\item Let $\pmb{\Omega}_{x,j}$ be the $j^{th}$ row of $\pmb{\Omega}_x$,  $\max_{j\in [d]}s_j = o(\sqrt{NT/\log d})$ and $\widetilde{\omega}_j= O(\sqrt{\log d/(NT)})$ for $1\leq j\leq d$, then $1/\widehat{\tau}_j^2 = O_P(1)$,  $|\widehat{\pmb{\Omega}}_{x,j} - \pmb{\Omega}_{x,j}|_1 = O_P(s_j\sqrt{\log d/(NT)})$, $|\widehat{\pmb{\Omega}}_{x,j} - \pmb{\Omega}_{x,j}|_2 = O_P(\sqrt{s_j\log d/(NT)})$, $|\pmb{\Omega}_{x,j}|_1 = O(s_j^{1/2})$, and  $|\widehat{\pmb{\Omega}}_{x,j}|_1 = O_P(s_j^{1/2})$.
\end{enumerate}
\end{lemma}

We further let $\widetilde{\mathbf{u}}_{it} \coloneqq E[\mathbf{u}_{it} \mid \mathcal{F}_{t-m,t}]$, $\mathbf{u}_{it} \coloneqq \mathbf{x}_{it}e_{it}$, $\overline{\mathbf{u}}_t \coloneqq \frac{1}{\sqrt{N}}\sum_{i=1}^{N}\mathbf{u}_{it}$, $\overline{u}_{j,t}$ be the $j^{th}$ element of $\overline{\mathbf{u}}_t$,  $\widetilde{\overline{\mathbf{u}}}_t \coloneqq \frac{1}{\sqrt{N}}\sum_{i=1}^{N}\widetilde{\mathbf{u}}_{it}$,  $\widetilde{\overline{u}}_{j,t}$ be the $j^{th}$ element of $\widetilde{\overline{\mathbf{u}}}_t$, $\mathcal{F}_{t-m,t} \coloneqq (\pmb{\varepsilon}_t,\ldots,\pmb{\varepsilon}_{t-m})$, and $\Psi_{q}(u_j,m) \coloneqq  (\sum_{k = m}^{\infty}\delta_{q}^{q'}(u_j,k) )^{1/q'}$ with $q' = 2\wedge q$. 

\begin{lemma}\label{LemmaB6}
Let Assumption \ref{ASSUMPTION1} hold. Then
\begin{itemize}[leftmargin=*, itemsep=0.5pt, parsep=0.5pt, topsep=0.6pt]
\item [1.] $\|\sum_{t=1}^{T}b_t\overline{u}_{j,t}\|_q = O ((\sum_{t=1}^{T}b_t^2)^{1/2} )$ and  $\|\sum_{t=1}^{T}b_t(\overline{u}_{j,t}-\widetilde{\overline{u}}_{j,t})\|_q = O ((\sum_{t=1}^{T}b_t^2)^{1/2}\Delta_q(u_j,m+1) )$ for any fixed sequences $\{b_t\}_{t\in [T]}$.

\item [2.] Let $a_1,a_2, \ldots\in \mathbb{R}$, $A_T = (\sum_{t=1}^{T}a_t^2)^{1/2}$, and $d_{j,m,q} = \sum_{t=0}^{\infty}\min(\delta_{q}(u_j,t),\Psi_q(u_j,m+1))$. Further, denote $L_{NT,j,j'} = \sum_{1\leq s < t \leq T}a_{t-s} \overline{u}_{j,t}\overline{u}_{j',s}$ and $\widetilde{L}_{NT,j,j'} = \sum_{1\leq s < t \leq T}a_{t-s} \widetilde{\overline{u}}_{j,t}\widetilde{\overline{u}}_{j',s}$, for some $q\geq 4$ we have
$$
\|L_{NT,j,j'} - E(L_{NT,j,j'}) - \widetilde{L}_{NT,j,j'} - E(\widetilde{L}_{NT,j,j'})\|_{q/2}
= O(\sqrt{T}A_T (d_{j,m,q}+d_{j',m,q} ));
$$

\item [3.] $\|\sum_{s,t=1}^{T}a_{s,t}\left(\overline{u}_{j,t}\overline{u}_{j',s}-E(\overline{u}_{j,t}\overline{u}_{j',s})\right)\|_{q/2} = O (\sqrt{T}B_T )$ for some $\forall q\geq 4$ and any given fixed sequences $\{a_{s,t}\}_{1\leq s,t\leq T}$, where $B_T^2 = \max \{\max_{t\in [T]}\sum_{s=1}^Ta_{s,t}^2, \max_{s\in [T]}\sum_{t=1}^Ta_{s,t}^2 \}$.
\end{itemize}
\end{lemma}

\begin{lemma}\label{LemmaB7}
Let Assumption \ref{ASSUMPTION1} hold with $q > 4$.
\begin{itemize}[leftmargin=*, itemsep=0.5pt, parsep=0.5pt, topsep=0.6pt]
\item [1.] Let $x_T > 0$ with $T^{1+\delta}/x_T\to 0$ for some $\delta > 0$. For $0 < \theta < 1$, there exists a constant $c_{q,\delta,\theta}$ only depending on $q$, $\delta$ and $\theta$ such that
$$
\Pr\left(|L_{NT,j,j'} - E(L_{NT,j,j'})| \geq x_T \right) \leq c_{q,\delta,\theta}x_T^{-q/2} (\log T) (T^{q/2-(\alpha-1)\theta q/2}+T).
$$

\item [2.] Let $a_{k} = 0$ for $k>\ell$ and $x_T > 0$ such that $T^{\delta}\sqrt{T\ell}/x_T\to 0$ for some $\delta > 0$. For $0 < \theta < 1$, there exists a constant $c_{q,\delta,\theta}$ only depending on $q$, $\delta$ and $\theta$ such that
$$
\Pr\left(|L_{NT,j,j'} - E(L_{NT,j,j'})| \geq x_T \right)
\leq c_{q,\delta,\theta}x_T^{-q/2} (\log T) ((T\ell)^{q/4}T^{-(\alpha-1)\theta q/2}+T\ell^{q/2-1-(\alpha-1)\theta q/2}+T).
$$
\end{itemize}
\end{lemma}

\begin{lemma}\label{LemmaB8}
Let Assumptions \ref{ASSUMPTION1}--\ref{ASSUMPTION4} hold with $q>4$.
\begin{itemize}[leftmargin=*, itemsep=0.5pt, parsep=0.5pt, topsep=0.6pt]
\item [1.] If $\frac{d^2T\log T}{(T\ell \log d)^{q/4} }\to 0$,  
$\displaystyle\max_{  k,l\in [d]} |\frac{1}{T}\sum_{t,s=1}^{T}a((t-s)/\ell) (\overline{u}_{k,t}\overline{u}_{l,s} - E(\overline{u}_{k,t}\overline{u}_{l,s}) ) | = O_P(\sqrt{\ell \log d /T}).$

\item [2.] If $\frac{s^2\sqrt{\ell \log d}}{\sqrt{T}} < \infty$, $\displaystyle
\max_{ k,l\in [d]} |\frac{1}{T}\sum_{t,s=1}^{T}a((t-s)/\ell) (\widehat{\overline{u}}_{k,t}-\overline{u}_{k,t} ) (\widehat{\overline{u}}_{l,s} - \overline{u}_{l,s} ) | = O_P(\sqrt{\ell \log d /T}).$

\item [3.] If $\frac{d^3T\log T}{(T\ell \log d)^{q/4} }\to 0$, $\displaystyle
\max_{ k,l\in [d]} |\frac{1}{T}\sum_{t,s=1}^{T}a((t-s)/\ell) (\widehat{\overline{u}}_{k,t}-\overline{u}_{k,t} ) \overline{u}_{l,s} | = o_P(\sqrt{\ell \log d /T}).
$
\end{itemize}
\end{lemma}

\section{Results of Section \ref{Sec2}}\label{AP.A5}
\begin{proof}[Proof of Lemma \ref{LEMMA1}]
\item

\noindent (1).
For notational simplicity, let $\mu_u = \sum_{m=1}^{\infty}\mu_{u,m}$ and $ \mu_{u,m} =  (m^{q/2-1}\delta_q^q(u_j,m) )^{1/(q+1)}$.

We apply the martingale approximation technique and approximate $u_{j,it} = x_{j,it}e_{it}$ by $m$-dependent processes. Then we prove this lemma by using delicate block techniques and the results on independent random variables.

Let $\{\lambda_m\}_{m=1}^T$ be a positive sequence such that $\sum_{m=1}^{T}\lambda_m \leq 1$. For any $t\geq 1$ and $m\geq 0$, define $S_{j,Nt,m} = \sum_{i=1}^{N}\sum_{k=1}^{t}u_{j,ik,m}$ and $u_{j,ik,m} = E\left[x_{j,ik}e_{ik} \mid \mathcal{F}_{k,k-m}\right]$, where $\mathcal{F}_{k,k-m} = (\pmb{\varepsilon}_k,\pmb{\varepsilon}_{k-1},\ldots,\pmb{\varepsilon}_{k-m})$. Therefore, $u_{j,ik,m}$ and $u_{j,ik',m}$ are independent if $|k-k'|>m$. Write $
u_{j,it} = u_{j,it} - u_{j,it,T} + \sum_{m=1}^{T}(u_{j,it,m}-u_{j,it,m-1}) + u_{j,it,0}$. Let $M_{t,m} = \sum_{k=1}^{t}\sum_{i=1}^{N}(u_{j,ik,m}-u_{j,ik,m-1})$ and $M_{T,m}^* = \max_{1\leq t\leq T}|M_{t,m}|$, and thus $S_{j,NT,T}-S_{j,NT,0} = \sum_{t=1}^{T}\sum_{i=1}^{N}(u_{j,it,T}-u_{j,it,0}) = \sum_{m=1}^{T}M_{T,m}$. Then, write
\begin{eqnarray*}
\Pr\left(\max_{1\leq t\leq T}|S_{j,Nt}| \geq 5x \right) &\leq& \sum_{m=1}^{T}\Pr\left(M_{T,m}^*\geq 3\lambda_m x\right)  + \Pr\left(\max_{1\leq t \leq T} | S_{j,Nt,0}|\geq x\right)\\
&&+  \Pr\left(\max_{1\leq t \leq T} | S_{j,Nt} - S_{j,Nt,T}| \geq x\right) 
\eqqcolon  I_1 + I_2 + I_3.
\end{eqnarray*}

Consider $I_1$. Let $Y_{l,m} = \sum_{i=1}^{N}\sum_{k=1+(l-1)m}^{\min(lm,T)}(u_{j,ik,m}-u_{j,ik,m-1})$, $b = \lfloor T/m \rfloor + 1$, $W_{s,m}^{e} = \sum_{l=1}^{s}((1+(-1)^l)/2) \cdot Y_{l,m}$,  $W_{s,m}^{o} = \sum_{l=1}^{s}((1-(-1)^l)/2) \cdot Y_{l,m}$ and $\lfloor t \rfloor_m = \lfloor t/m \rfloor \times m$. Hence, for any $x > 0$ we have
\begin{eqnarray*}
\Pr\left(M_{T,m}^*\geq 3\lambda_m x\right)
&\leq& \Pr\left(\max_{1\leq s\leq b}|W_{s,m}^e|\geq \lambda_m x\right) + \Pr\left(\max_{1\leq s\leq b}|W_{s,m}^o|\geq \lambda_m x\right)\\
&& + \frac{T}{m} \Pr\left(\max_{1\leq t \leq m}|M_{t,m}| \geq \lambda_m x\right) \eqqcolon  I_{1,1} + I_{1,2} + I_{1,3}.
\end{eqnarray*}

For $I_{1,1}$, since $Y_{2,m},Y_{4,m},\ldots$ are independent, then by using Nagaev inequality for i.i.d. random variables (see Lemma \ref{LemmaB1}), we have
\begin{eqnarray*}
\textstyle	\Pr\left(\max_{1\leq s\leq b}|W_{s,m}^e|\geq \lambda_m x\right) \leq \left(1 + 2/q\right)^q\frac{\sum_{l=1}^{b}\|Y_{l,m}\|_q^q}{(\lambda_mx)^q} + 2\exp\left( - \frac{2(\lambda_mx)^2}{e^q(q+2)^2\sum_{l=1}^{b}\|Y_{l,m}\|_2^2} \right).
\end{eqnarray*}
In addition, by the proof of Lemma \ref{LemmaB3}, we have $$
\|Y_{l,m}\|_q \leq \sqrt{(q-1)(\min(lm,T)-(l-1)m)}\sqrt{N} \delta_q(u_j,m),$$
which follows that
$$
I_{1,1} \leq \left(1 + 2/q\right)^q(q-1)^{q/2} \frac{T}{x^q} \frac{N^{q/2}m^{q/2-1}\delta_q^q(u_j,m)}{\lambda_m^q}+ 2\exp\left( - \frac{2(\lambda_mx)^2}{e^q(q+2)^2 NT\delta_2^2(u_j,m)}\right).
$$
Similarly,
$$
I_{1,2} \leq \left(1 + 2/q\right)^q(q-1)^{q/2} \frac{T}{x^q} \frac{N^{q/2}m^{q/2-1}\delta_q^q(u_j,m)}{\lambda_m^q}+ 2\exp\left( - \frac{2(\lambda_m x)^2}{e^q(q+2)^2 NT\delta_2^2(u_j,m)}\right).
$$
Consider $I_{1,3}$. By Doob $L_p$ maximal inequality and Burkholder inequality, we have
\begin{eqnarray*}
E\left(\max_{1\leq t\leq m}|M_{t,m}|^q\right) &=& E\left(\max_{1\leq t\leq m}|\sum_{k=t}^{m}\sum_{i=1}^{N}(u_{j,ik,m}-u_{j,ik,m-1})|^q\right)\\
&\leq& (q/(q-1))^q E\left(|\sum_{k=1}^{m}\sum_{i=1}^{N}(u_{j,ik,m}-u_{j,ik,m-1})|^q\right)\\
&\leq& (q/(q-1))^q (q-1)^{q/2} N^{q/2} m^{q/2}\delta_{q}^q(u_j,m),
\end{eqnarray*}
which follows that
\begin{eqnarray*}
\Pr\left(\max_{1\leq t \leq m}|M_{t,m}| \geq \lambda_m x \right) &\leq&  (q/(q-1))^q (q-1)^{q/2} \frac{N^{q/2} m^{q/2}\delta_{q}^q(u_j,m)}{x^q\lambda_m^q}.
\end{eqnarray*}
Combining the above analysis and choose $\lambda_m = \mu_{u,m} / \mu_u$, we have
\begin{eqnarray*}
\sum_{m=1}^{T}\Pr\left(M_{T,m}^*\geq 3\lambda_m x\right) &\leq& c_q \frac{TN^{q/2}}{x^q}\sum_{m=1}^{T} \frac{m^{q/2-1}\delta_q^q(u_j,m)}{\lambda_m^q}+ 4\sum_{m=1}^{T}\exp\left( -c_q \frac{2(\lambda_mx)^2}{NT\delta_2^2(u_j,m)}\right)\\
&\leq& c_q \frac{TN^{q/2}}{x^q}\mu^{q+1}+ 4\sum_{m=1}^{T}\exp\left( -c_q \frac{2(\mu_{u,m}x)^2}{NT\mu^2\delta_2^2(u_j,m)}\right)
\end{eqnarray*}
for some constant $c_q > 0$ only depending on $q$. 

Consider $I_2$. Note that $\{u_{j,it,0}\}$ are independent random variables, by using Nagaev inequality and $\|\sum_{i=1}^{N}u_{j,it,0} \|_q \leq \sqrt{N}\delta_q(u_j,0)$, we have
\begin{eqnarray*}
\Pr\left(\max_{1\leq t \leq T} | S_{j,Nt,0}|\geq x\right)
&\leq& \textstyle \left(1 + 2/q\right)^q\frac{T\|\sum_{i=1}^{N}u_{j,it,0} \|_q^q}{x^q} + 2\exp\left( - \frac{2x^2}{e^q(q+2)^2T\| \sum_{i=1}^{N}u_{j,it,0} \|_2^2} \right)\\
&\leq& \textstyle\left(1 + 2/q\right)^q\frac{TN^{q/2}\delta_q^q(u_j,0)}{x^q} + 2\exp\left( - \frac{2x^2}{e^q(q+2)^2TN\delta_2^2(u_j,0)} \right).
\end{eqnarray*}

Finally, consider $I_3$. Since $\|\max_{1\leq t \leq T} | S_{j,Nt} - S_{j,Nt,T}| \|_q \leq 3\sqrt{q-1}\sqrt{NT}\sum_{m=T+1}^{\infty}\delta_q(u_j,m)$ by the proof of Lemma \ref{LemmaB3}, we have
\begin{eqnarray*}
\Pr\left(\max_{1\leq t \leq T} | S_{j,Nt} - S_{j,Nt,T}|\geq x\right) &\leq& \frac{1}{x^q} (3\sqrt{q-1}\sqrt{NT}\sum_{m=T+1}^{\infty}\delta_q(u_j,m))^q\\
&\leq& c_q\frac{1}{x^q} N^{q/2}T(\sum_{m=T+1}^{\infty}\mu_{u,m})^{q+1}
\end{eqnarray*}
in view of 
\begin{eqnarray*}
&&	(\sum_{m=T+1}^{\infty}\delta_q(u_j,m))^q\leq(\sum_{m=T+1}^{\infty}\delta_q^{q/(q+1)}(u_j,m))^{q+1}\\
&\leq& (\sum_{m=T+1}^{\infty}(m/T)^{\frac{q/2-1}{q+1}}\delta_q^{q/(q+1)}(u_j,m))^{q+1}
=T^{1-q/2}(\sum_{m=T+1}^{\infty}\mu_{u,m})^{q+1}.
\end{eqnarray*}
Combining the above analysis, we then have for any $x  > 0$,
\begin{eqnarray*}
\Pr\left(\max_{1\leq t\leq T}|S_{j,Nt}|\geq x\right) &\leq& c_q\frac{TN^{q/2}}{x^q}\left(\mu^{q+1}+\delta_q^q(u_j,0)\right) + 4\sum_{m=1}^{\infty}\exp\left( -c_q \frac{\mu_{u,m}^2x^2}{NT\mu_u^2\delta_2^2(u_j,m)}\right)\\
&&+2\exp\left( - \frac{c_qx^2}{TN\delta_2^2(u_j,0)} \right).
\end{eqnarray*}
Note that if $x = \sqrt{NT}\mu_{u}^{1+1/q} y$ with $y>0$, then $\mu_{u,m}^2x^2 /(NT\mu_u^2\delta_2^2(u_j,m))\geq m^{1-2/q}y^2$ and thus
$$
\sum_{m=1}^{\infty}\exp\left( -c_q \frac{\mu_{u,m}^2x^2}{NT\mu_u^2\delta_2^2(u_j,m)}\right) \leq \sum_{m=1}^{\infty}\exp\left( -c_q m^{1-2/q}y^2\right).
$$
In addition, since for any $s>0$$$\sup_{y\geq 1}\sum_{m=1}^{\infty}\exp\left( -m ^sy^2\right) \exp(y^2) = \sum_{m=1}^{\infty}\exp\left( -m ^s\right) \exp(1)$$ and $$\sum_{m=1}^{\infty}\exp\left( -m ^sy^2\right) \leq \sum_{m=1}^{\infty}\exp\left( -m^s\right)\exp(1-y^2),$$ thus for some $c_q'>0$ $$\sum_{m=1}^{\infty}\exp\left( -c_q \frac{\mu_{u,m}^2x^2}{NT\mu_u^2\delta_2^2(u_j,m)}\right) \leq c_q^\prime\exp(-c_qx^2(NT\mu_u^{2+2/q})).$$

\smallskip

\noindent (2). Recall that $\widetilde{\mathbf{u}}_{it} = E[\mathbf{u}_{it} \mid \mathcal{F}_{t-m,t}]$, $\mathbf{u}_{it} = \mathbf{x}_{it}e_{it}$, $\overline{\mathbf{u}}_t = \frac{1}{\sqrt{N}}\sum_{i=1}^{N}\mathbf{u}_{it}$, $\overline{u}_{j,t}$ is the $j^{th}$ element of $\overline{\mathbf{u}}_t$,  $\widetilde{\overline{\mathbf{u}}}_t = \frac{1}{\sqrt{N}}\sum_{i=1}^{N}\widetilde{\mathbf{u}}_{it}$ and  $\widetilde{\overline{u}}_{j,t}$ is the $j^{th}$ element of $\widetilde{\overline{\mathbf{u}}}_t$. For notation simplicity, let $\ell = T^{\gamma}$ and $L_{NT,jj'} = \sum_{1\leq s < t \leq T}a_{t-s} \overline{u}_{j,t}\overline{u}_{j',s}$. For $\gamma < \theta <1$, let $m_T = \lfloor T^{\sqrt{\theta}} \rfloor$, $\widetilde{\overline{u}}_{j,t,\theta} = \frac{1}{\sqrt{N}}\sum_{i=1}^{N}\widetilde{u}_{j,it,\theta}$ and $Q_{NT,\theta} = \sum_{1\leq s < t \leq T}a_{t-s} \widetilde{\overline{u}}_{j,t,\theta}\widetilde{\overline{u}}_{j',s,\theta}$. By using Lemmas \ref{LemmaB6} (2) and (3), we have for any $x_T>0$
\begin{eqnarray*}
\Pr(|L_{NT,j,j'} - E(L_{NT,j,j'}) - Q_{NT,\theta} + E(Q_{NT,\theta})|\geq x_T/2) \leq C_{q,M} x_{T}^{-q/2}(T\ell)^{q/4}T^{- (\alpha-1)\sqrt{\theta} q/2}.
\end{eqnarray*}
Split $[T]$ into blocks $B_1,\ldots,B_{b_T}$ with block size $2m_T$ and define $$Q_{NT,\theta,k} =\sum_{t\in B_k}\sum_{1\leq s \leq t}a_{t-s} \widetilde{\overline{u}}_{j,t,\theta}\widetilde{\overline{u}}_{j',s,\theta}.$$
By using Lemma \ref{LemmaB2} (2) and Lemma \ref{LemmaB6} (3), we have for any $M>1$
\begin{eqnarray*}
&&\Pr(|Q_{NT,\theta} - E(Q_{NT,\theta})|\geq x_T/2)\\
&\leq& \sum_{k=1}^{b_T} \Pr(|Q_{NT,\theta,k} - E(Q_{NT,\theta,k})|\geq x_T/C_{q,M,\theta}) + \left\{ \frac{C_{q,M,\theta}Tm_{T}^{-1}(m_T\ell)^{q/4}}{(T\ell)^{q/4}}\right\}^{C_{q,M,\theta}} \\
&& + C_{\theta} \exp\left\{ \frac{c_{q}^2 \log d}{C_q}\right\} \leq \sum_{k=1}^{b_T} \Pr(|Q_{NT,\theta,k} - E(Q_{NT,\theta,k})|\geq x_T/C_{q,M,\theta}) + C_{q,M,\theta}(d^{-M}+T^{-M}).
\end{eqnarray*}
By using Lemma \ref{LemmaB7} (2), we have
\begin{eqnarray*}
&&\Pr(|Q_{NT,\theta,k} - E(Q_{NT,\theta,k})|\geq x_T/C_{q,M,\theta})\\ &\leq& O(1)x_T^{-q/2}\log T \left((T^{\sqrt{\theta}}\ell)^{q/4}T^{-(\alpha-1)\theta q/2} T^{\sqrt{\theta}} \ell^{q/2-1-(\alpha-1)\theta q/2} + T^{\sqrt{\theta}}\right).
\end{eqnarray*}
Combining the above analysis, the proof is complete.
\end{proof}

\begin{proof}[Proof of Lemma \ref{LEMMA2}]
\item
\noindent (1). This proof is conditional on the following two events
$\mathcal{A}_{NT} = \left\{\left|\frac{1}{NT}\mathbf{e}^\top \mathbf{X}\right|_\infty \leq \omega_1/2\right\}$ and  $\mathcal{B}_{NT} = \left\{|\widehat{\pmb{\Sigma}}_{x} - \pmb{\Sigma}_{x}|_{\mathrm{max}} \leq \frac{C_0}{16}\sqrt{\log d/(NT)} \right\}$ for some constant $C_0 > 0$. Note that by using Lemma \ref{LEMMA1} (1) and $\omega_1  \asymp\sqrt{\log d/(NT)}$, the event $\mathcal{A}_{NT} = \left\{|\frac{1}{NT}\mathbf{e}^\top \mathbf{X}|_\infty \leq \omega_1/2\right\}$ holds with probability larger than $ C_{NT}$ for some $C_1,C_2>0$. Similarly, by using Lemma \ref{LemmaB4} (1) and $x = \frac{C_0}{16}\sqrt{\log d/(NT)}$ for some constant $C_0 > 0$, $\mathcal{B}_{NT} = \left\{|\widehat{\pmb{\Sigma}}_{x} - \pmb{\Sigma}_x|_{\mathrm{max}} \leq x \right\}$ holds with probability larger than $C_{NT}$ for some $C_1,C_2>0$. Hence, we have
$$
\Pr(\mathcal{A}_{NT}\cap\mathcal{B}_{NT}) = \Pr(\mathcal{A}_{NT}) + \Pr(\mathcal{B}_{NT}) - \Pr(\mathcal{A}_{NT} \cup \mathcal{B}_{NT})\geq C_{NT}.
$$

By the nature of minimizing procedure in \eqref{EQUATION2}, we have
$$
\frac{1}{2NT}|\mathbf{y} - \mathbf{X}\widehat{\pmb{\beta}}|_2^2 +  \omega_1|\widehat{\pmb{\beta}}|_1 \leq  \frac{1}{2NT}|\mathbf{y} - \mathbf{X}\pmb{\beta}_0|_2^2 +  \omega_1|\pmb{\beta}_0|_1.
$$

Since $\frac{1}{NT}\mathbf{e}^\top \mathbf{X}(\widehat{\pmb{\beta}}-\pmb{\beta}_0) \leq |\frac{1}{NT}\mathbf{e}^\top \mathbf{X}|_\infty |\widehat{\pmb{\beta}}-\pmb{\beta}_0|_1 \leq \frac{\omega_1}{2} |\widehat{\pmb{\beta}}-\pmb{\beta}_0|_1$, we have
$$
\frac{1}{NT}|\mathbf{X}(\widehat{\pmb{\beta}}-\pmb{\beta}_0)|_2^2 +  2\omega_1|\widehat{\pmb{\beta}}|_1 \leq  \omega_1|\widehat{\pmb{\beta}}-\pmb{\beta}_0|_1 +  2\omega_1|\pmb{\beta}_0|_1.
$$

Note that $|\widehat{\pmb{\beta}}-\pmb{\beta}_0|_1 + |\pmb{\beta}_0|_1 - |\widehat{\pmb{\beta}}|_1 = |\widehat{\pmb{\beta}}_{J}-\pmb{\beta}_{0,J}|_1 + |\pmb{\beta}_{0,J}|_1 - |\widehat{\pmb{\beta}}_J|_1 \leq 2 |\widehat{\pmb{\beta}}_{J}-\pmb{\beta}_{0,J}|_1$, we have
$$
\frac{1}{NT}|\mathbf{X}(\widehat{\pmb{\beta}}-\pmb{\beta}_0)|_2^2 + \omega_1|\widehat{\pmb{\beta}}-\pmb{\beta}_0|_1 \leq 4\omega_1 |\widehat{\pmb{\beta}}_{J}-\pmb{\beta}_{0,J}|_1,
$$
and thus
$
\frac{1}{NT}|\mathbf{X}(\widehat{\pmb{\beta}}-\pmb{\beta}_0)|_2^2 + \omega_1|\widehat{\pmb{\beta}}_{J^c}-\pmb{\beta}_{0,J^c}|_1 \leq 3 \omega_1 |\widehat{\pmb{\beta}}_{J}-\pmb{\beta}_{0,J}|_1.
$

By using Lemma \ref{LemmaB5} (1) and the condition $s\leq \frac{\psi_{\pmb{\Sigma}_x}(J) }{2C_0}\sqrt{\frac{NT}{\log d}}$, we have
$$
\frac{1}{NT}|\mathbf{X}(\widehat{\pmb{\beta}}-\pmb{\beta}_0)|_2^2\geq (\psi_{\pmb{\Sigma}_x}(J)- C_0s\sqrt{\log d/(NT)})|\widehat{\pmb{\beta}}-\pmb{\beta}_0|_2^2\geq \frac{\psi_{\pmb{\Sigma}_x}(J)}{2}|\widehat{\pmb{\beta}}-\pmb{\beta}_0|_2^2
$$
conditional on the event $\mathcal{A}_{NT}\cap\mathcal{B}_{NT}$.

Hence,
$$
\frac{\psi_{\pmb{\Sigma}_x}(J)}{2} |\widehat{\pmb{\beta}}-\pmb{\beta}_0|_2^2\leq 4\omega_1 |\widehat{\pmb{\beta}}_{J}-\pmb{\beta}_{0,J}|_1 \leq 4\omega_1 \sqrt{s}|\widehat{\pmb{\beta}}_{J}-\pmb{\beta}_{0,J}|_2
\leq 4\omega_1 \sqrt{s}|\widehat{\pmb{\beta}}-\pmb{\beta}_0|_2,
$$
which follows that $
|\widehat{\pmb{\beta}}-\pmb{\beta}_0|_2 \leq \frac{8\sqrt{s}\omega_1 }{\psi_{\pmb{\Sigma}_x}(J)} $
and
$
|\widehat{\pmb{\beta}}-\pmb{\beta}_0|_1 \leq 4 |\widehat{\pmb{\beta}}_{J}-\pmb{\beta}_{0,J}|_1\leq 4\sqrt{s}|\widehat{\pmb{\beta}}_{J}-\pmb{\beta}_{0,J}|_2\leq \frac{32s\omega_1}{\psi_{\pmb{\Sigma}_x}(J)}.
$

\smallskip
\noindent (2). Similar to part (1), this proof is conditional on the event $\mathcal{A}_{NT}\cap\mathcal{B}_{NT}$.

Define $\vec{\mathbf{b}} = (\sgn(\beta_j)g_j)_{j\in J}$. Due to the property of convex optimization, $\widehat{\pmb{\beta}}_{w} \in \mathbb{R}^{d}$ is a solution if and only if there exists a subgradient
$$
\vec{\mathbf{g}} \in \partial\sum_{j=1}^{d}g_j|\widehat{\beta}_j| = \left\{\mathbf{z} \in \mathbb{R}^{d} \mid z_j = \sgn(\widehat{\beta}_j)g_j \ \text{for}\ \widehat{\beta}_j \neq 0 \ \text{and} \ |z_j|\leq g_j \ \text{elsewhere} \right\}
$$
such that $\frac{1}{NT}\mathbf{X}^\top\mathbf{X}\widehat{\pmb{\beta}}_{w} - \frac{1}{NT}\mathbf{X}^\top\mathbf{Y} + \omega_1\vec{\mathbf{g}}  = \mathbf{0}$.

Hence, $\sgn(\widehat{\pmb{\beta}}_{w}) = \sgn(\pmb{\beta}_0)$ if and only if
$$
\frac{1}{NT}\mathbf{X}_{J^c}^\top\mathbf{X}_J\left(\widehat{\pmb{\beta}}_{w,J} -\pmb{\beta}_{0,J} \right) - \frac{1}{NT}\mathbf{X}_{J^c}^\top\mathbf{e}  = -  \omega_1\vec{\mathbf{g}}_{J^c},
$$
$$
\frac{1}{NT}\mathbf{X}_{J}^\top\mathbf{X}_J\left(\widehat{\pmb{\beta}}_{w,J} -\pmb{\beta}_{0,J} \right) - \frac{1}{NT}\mathbf{X}_{J}^\top\mathbf{e}  = -  \omega_1\vec{\mathbf{g}}_{J} = -  \omega_1\vec{\mathbf{b}},
$$
$$
\sgn(\widehat{\pmb{\beta}}_{w,J}) = \sgn(\pmb{\beta}_{0,J})\quad \text{and}\quad \widehat{\pmb{\beta}}_{w,J^c} = \pmb{\beta}_{0,J^c} = \mathbf{0},
$$
where $\mathbf{X}_{J}$ denotes the sub-matrix of $\mathbf{X}$ only containing the columns indexed by $J$, and $\pmb{\beta}_{0,J}$ denotes the sub-vector of $\pmb{\beta}$ containing the elements indexed by $J$.

By standard results in matrix perturbation theory, Lemma \ref{LemmaB5} (1) and $s\leq \frac{\psi_{\pmb{\Sigma}_x}(J) }{2C_0}\sqrt{\frac{NT}{\log d}}$, we have
\begin{eqnarray*}
&&\left|\psi_{\mathrm{min}}((NT)^{-1}\mathbf{X}_{J}^\top\mathbf{X}_J)  - \psi_{\mathrm{min}}(\pmb{\Sigma}_{x,J,J}) \right|\\
&\leq& \left|(NT)^{-1}\mathbf{X}_{J}^\top\mathbf{X}_J - \pmb{\Sigma}_{x,J,J} \right|_2 \leq s \left|(NT)^{-1}\mathbf{X}_{J}^\top\mathbf{X}_J - \pmb{\Sigma}_{x,J,J} \right|_{\mathrm{max}}\leq \frac{\psi_{\pmb{\Sigma}_x}(J)}{2}.
\end{eqnarray*}
where $\pmb{\Sigma}_{x,J,J}$ is the sub-matrix of $\pmb{\Sigma}_{x}$.

Therefore, conditional on the set $\mathcal{B}_{NT} = \left\{|\widehat{\pmb{\Sigma}}_{x} - {\pmb{\Sigma}}_x|_{\mathrm{max}} \leq \frac{C_0}{16}\sqrt{\log d/(NT)} \right\}$, we have $$\psi_{\mathrm{min}}((NT)^{-1}\mathbf{X}_{J}^\top\mathbf{X}_J) \geq \frac{\psi_{\pmb{\Sigma}_x}(J)}{2}$$ and thus  $(NT)^{-1}\mathbf{X}_{J}^\top\mathbf{X}_J$ is nonsingular. 

Given the invertibility of $(NT)^{-1}\mathbf{X}_{J}^\top\mathbf{X}_J$, to prove $\sgn(\widehat{\pmb{\beta}}_{w}) = \sgn(\pmb{\beta}_0)$, it suffices to show 
\begin{equation}\label{Eq.A1}
\sgn\left(\widehat{\pmb{\beta}}_{w,J}\right) = \sgn\left[\pmb{\beta}_{0,J} + \left((NT)^{-1}\mathbf{X}_{J}^\top\mathbf{X}_J\right)^{-1}\left((NT)^{-1}\mathbf{X}_{J}^\top\mathbf{e}-  \omega_1\vec{\mathbf{b}}\right)\right]= \sgn(\pmb{\beta}_{0,J} )
\end{equation}
and for any $j \in J^c$
\begin{eqnarray}\label{Eq.A2}
&&\left|(NT)^{-1}\mathbf{X}_{j}^\top\mathbf{X}_J\left((NT)^{-1}\mathbf{X}_{J}^\top\mathbf{X}_J\right)^{-1}\left((NT)^{-1}\mathbf{X}_{J}^\top\mathbf{e}-  \omega_1\vec{\mathbf{b}}\right) -  \frac{1}{NT}\mathbf{X}_{j}^\top\mathbf{e}\right|\nonumber \\
&=& |-\omega_1\vec{g}_{j}|\leq \omega_1g_j.
\end{eqnarray}

Consider \eqref{Eq.A1} first. It is sufficient to show $|(\frac{1}{NT}\mathbf{X}_{J}^\top\mathbf{X}_J)^{-1}(\frac{1}{NT}\mathbf{X}_{J}^\top\mathbf{e}-  \omega_1\vec{\mathbf{b}})|_{\infty} < \beta_{\mathrm{min}}$.
Since $|(\frac{1}{NT}\mathbf{X}_{J}^\top\mathbf{X}_J)^{-1}|_{2}\leq\frac{2}{\psi_{\pmb{\Sigma}_x}(J)}$ and $\beta_{\min} > \frac{\sqrt{s}\omega_1}{\psi_{\pmb{\Sigma}_x}(J)} (1 + 2\max_{j\in J}|g_j|)$, we have
\begin{eqnarray*}
&&\left|\left(\frac{1}{NT}\mathbf{X}_{J}^\top\mathbf{X}_J\right)^{-1}\left(\frac{1}{NT}\mathbf{X}_{J}^\top\mathbf{e}-  \omega_1\vec{\mathbf{b}}\right)\right|_{\infty}
\leq \left|\left(\frac{1}{NT}\mathbf{X}_{J}^\top\mathbf{X}_J\right)^{-1}\right|_{\infty}\left|\frac{1}{NT}\mathbf{X}_{J}^\top\mathbf{e}-  \omega_1\vec{\mathbf{b}}\right|_{\infty}\\
&\leq&\frac{2\sqrt{s}}{\psi_{\pmb{\Sigma}_x}(J) }\left(\left|\frac{1}{NT}\mathbf{X}_{J}^\top\mathbf{e}\right|_{\infty} + |\omega_1\vec{\mathbf{b}}|_{\infty} \right)
\leq \frac{2\sqrt{s}}{\psi_{\pmb{\Sigma}_x}(J)}\left(\omega_1/2+ \omega_1\max_{j\in J}|g_j| \right) < \beta_{\min}.
\end{eqnarray*}

Next consider \eqref{Eq.A2}. By the triangle inequality, it suffices to show for any $j \in J^c$
$$
\left|\frac{1}{NT}\mathbf{X}_{j}^\top\mathbf{X}_J(\frac{1}{NT}\mathbf{X}_{J}^\top\mathbf{X}_J)^{-1}(\frac{1}{NT}\mathbf{X}_{J}^\top\mathbf{e}-  \omega_1\vec{\mathbf{b}}) \right| + \left|  \frac{1}{NT}\mathbf{X}_{j}^\top\mathbf{e}\right| \leq \omega_1 g_j.
$$
The first term is bounded by
\begin{eqnarray*}
&&\left|\frac{1}{NT}\mathbf{X}_{j}^\top\mathbf{X}_J(\frac{1}{NT}\mathbf{X}_{J}^\top\mathbf{X}_J)^{-1}(\frac{1}{NT}\mathbf{X}_{J}^\top\mathbf{e}-  \omega_1\vec{\mathbf{b}}) \right|\\
&\leq& \left|\frac{1}{NT}\mathbf{X}_{j}^\top\mathbf{X}_J(\frac{1}{NT}\mathbf{X}_{J}^\top\mathbf{X}_J)^{-1}\right|_1 \times\left( \left|\frac{1}{NT}\mathbf{X}_{J}^\top\mathbf{e}\right|_{\infty} + |\omega_1\vec{\mathbf{b}}|_{\infty}\right)\\
&\leq&\frac{2s}{\psi_{\pmb{\Sigma}_x}(J)}\left|\frac{1}{NT}\mathbf{X}_{j}^\top\mathbf{X}_J\right|_\infty \times\left( \left|\frac{1}{NT}\mathbf{X}_{J}^\top\mathbf{e}\right|_{\infty} + |\omega_1\vec{\mathbf{b}}|_{\infty}\right).
\end{eqnarray*}
Since $\max_{j\in J^c}|\frac{1}{NT}\mathbf{X}_{j}^\top\mathbf{X}_J|_\infty \leq |\pmb{\Sigma}_x|_{\mathrm{max}} + \frac{C_0}{16}\sqrt{\log d/(NT)}$ by Lemma \ref{LemmaB5} (1), the left hand of \eqref{Eq.A2} is bounded by 
$$ 
\left(\frac{2s|\pmb{\Sigma}_x|_{\mathrm{max}}}{\psi_{\pmb{\Sigma}_x}(J)} + \frac{\psi_{\pmb{\Sigma}_x}(J)}{16}\right)\left(\frac{\omega_1}{2}+\omega_1 \max_{j\in J}|g_j|\right).
$$

The proof is completed given
$\min_{j\in J^c}|g_j| \geq  \left(\frac{2s|\pmb{\Sigma}_x|_{\mathrm{max}}}{\psi_{\pmb{\Sigma}_x}(J)} + \frac{\psi_{\pmb{\Sigma}_x}(J)}{16}\right)\left(\frac{1}{2}+\max_{j\in J}|g_j|\right).$
\end{proof}

\begin{proof}[Proof of Theorem \ref{THEOREM1}]
\item
\noindent (1). 
Due to the property of convex optimization, $\widehat{\pmb{\beta}} \in \mathbb{R}^{d}$ is a solution if and only if there exists a subgradient
$$
\vec{\mathbf{g}} \in \partial\sum_{j=1}^{d}|\widehat{\beta}_j| = \left\{\mathbf{z} \in \mathbb{R}^{d} \mid z_j = \sgn(\widehat{\beta}_j) \ \text{for}\ \widehat{\beta}_j \neq 0 \ \text{and} \ |z_j|\leq 1 \ \text{elsewhere} \right\}
$$
such that $|\vec{\mathbf{g}}|_{\infty} \leq 1$ and $\frac{1}{NT}\mathbf{X}^\top\mathbf{X}\widehat{\pmb{\beta}} - \frac{1}{NT}\mathbf{X}^\top\mathbf{Y} + \omega_1\vec{\mathbf{g}}  = \mathbf{0}$.

Hence, we have
\begin{eqnarray*}
\sqrt{NT}\pmb{\rho}^\top\left(\widehat{\pmb{\beta}}_{bc} - \pmb{\beta}_0\right) &=& \pmb{\rho}^\top\widehat{\pmb{\Omega}}_x\frac{1}{\sqrt{NT}}\mathbf{X}^\top\mathbf{e} + \sqrt{NT}\pmb{\rho}^\top\left(\mathbf{I}_{d} - \widehat{\pmb{\Omega}}_x\widehat{\pmb{\Sigma}}_{x} \right)(\widehat{\pmb{\beta}} -\pmb{\beta}_0)\\
&\eqqcolon & J_1 + J_2.
\end{eqnarray*}
Here, $J_2$ is the error resulting from using an approximate inverse of $\widehat{\pmb{\Sigma}}_{x}$, which will be proved asymptotically negligible. In addition, the bias term $\omega_1\widehat{\pmb{\Omega}}\vec{\mathbf{g}}$ resulting from the penalization of the parameters is known to be equal to $\frac{1}{NT}\mathbf{X}^\top\left(\mathbf{Y}-\mathbf{X}\widehat{\pmb{\beta}}\right)$.

Hence, we have
$$
\sqrt{NT}\pmb{\rho}^\top\left(\widehat{\pmb{\beta}}_{bc} - \pmb{\beta}_0\right) = \pmb{\rho}^\top\widehat{\pmb{\Omega}}\frac{1}{\sqrt{NT}}\mathbf{X}^\top\mathbf{e} + \sqrt{NT}\pmb{\rho}^\top\left(\mathbf{I}_{d} - \widehat{\pmb{\Omega}}\widehat{\pmb{\Sigma}}_{x} \right)(\widehat{\pmb{\beta}} -\pmb{\beta}_0).
$$

We will then provide: (1) a central limit theorem for $J_1$; and (2) a verification of asymptotic negligibility of $J_2$.

Consider $J_2$ first. By \cite{yuan2010high}, we have
$$
\Omega_{x,j,j} = [{\Sigma}_{x,j,j}-\pmb{\Sigma}_{x,j,-j}\pmb{\Sigma}_{x,-j,-j}^{-1}\pmb{\Sigma}_{x,-j,j}]^{-1} \quad \text{and} \quad \pmb{\Omega}_{x,j,-j} =-\Omega_{x,j,j}\pmb{\Sigma}_{x,j,-j}\pmb{\Sigma}_{x,-j,-j}^{-1},
$$
where $\Omega_{x,j,j}$ denotes the $j^{th}$ diagonal entry of $\pmb{\Omega}_x$, $\pmb{\Omega}_{x,j,-j}$ is the $1 \times (d-1)$ vector obtained by removing the $j^{th}$ entry of the $j^{th}$ row of $\pmb{\Omega}_x$, $\Sigma_{x,j,j}$ is the $j^{th}$ diagonal entry of $\pmb{\Sigma}_x$, $\pmb{\Sigma}_{x,-j,-j}$ is the submatrix of $\pmb{\Sigma}_x$ with the $j^{th}$ row and column removed, $\pmb{\Sigma}_{x,j,-j}$ is the $j^{th}$ row of $\pmb{\Sigma}_x$ with its $j^{th}$ entry removed, $\pmb{\Sigma}_{x,-j,j}$ is the $j^{th}$ column of $\pmb{\Sigma}_x$ with its $j^{th}$ entry removed. Define the $(d-1)\times 1$ vector
\begin{equation*}
\pmb{\gamma}_{j} = \argmin_{\mathbf{b}\in\mathbb{R}^{d-1}}\frac{1}{NT}E|\mathbf{X}_{j} - \mathbf{X}_{-j}\mathbf{b}|_2^2
\end{equation*}
such that
$$
\pmb{\gamma}_{j} = \left(\frac{1}{NT}E(\mathbf{X}_{-j}^\top\mathbf{X}_{-j})\right)^{-1}
\frac{1}{NT}E(\mathbf{X}_{-j}^\top\mathbf{X}_{j}) = \pmb{\Sigma}_{x,-j,-j}^{-1}\pmb{\Sigma}_{x,-j,j}
$$
and $|\pmb{\gamma}_{j}|_2\leq |\pmb{\Sigma}_{x,-j,-j}^{-1}|_{2}|\pmb{\Sigma}_{x,-j,j}|_2<\infty$.
Define $\pmb{\eta}_j = \mathbf{X}_{j} - \mathbf{X}_{-j} \pmb{\gamma}_j$ and thus $E(\mathbf{X}_{-j}^\top\pmb{\eta}_j) = \mathbf{0}$. Define $\tau_j^2 = \frac{1}{NT}E|\mathbf{X}_{j} - \mathbf{X}_{-j}\pmb{\gamma}_{j}|_2^2 = {\Sigma}_{x,j,j}-\pmb{\Sigma}_{x,j,-j}\pmb{\Sigma}_{x,-j,-j}^{-1}\pmb{\Sigma}_{x,-j,j} = 1/\Omega_{x,j,j}$. Observe that $\pmb{\Omega}_{x,j,-j} = - \pmb{\gamma}_j^\top/\tau_j^2$, and thus we can write $\pmb{\Omega}_x = \mathbf{T}^{-2}\mathbf{C}$, where $\mathbf{T}$ and $\mathbf{C}$ is defined similarly to $\widehat{\mathbf{T}}^{-1}$ and $\widehat{\mathbf{C}}$ but with $\tau_j^2$ and $\pmb{\gamma}_j$ replacing $\widehat{\tau}_j^2$ and $\widehat{\pmb{\gamma}}_j$.

Let $\widehat{\pmb{\Omega}}_{x,j}$ and $\widehat{\mathbf{C}}_{x,j}$ denote the $j^{th}$ row of $\widehat{\pmb{\Omega}}_x$ and $\widehat{\mathbf{C}}_x$, respectively, and thus $\widehat{\pmb{\Omega}}_{x,j} = \widehat{\tau}_j^{-2}\widehat{\mathbf{C}}_j$. By the properties of Karush–Kuhn–Tucker (KKT) condition, we have
$$
|\widehat{\pmb{\Sigma}}_{x}\widehat{\pmb{\Omega}}_{x,j}^\top - \mathbf{e}_j|_{\infty} \leq\lambda_{j}/\widehat{\tau}_j^{2}=O_P(\sqrt{\log d/(NT)}),
$$
where $\mathbf{e}_j$ is the $j^{th}$ unit column vector. Hence, by using Lemma \ref{LemmaB5} (3), we have
\begin{eqnarray*}
|J_2| &\leq&  |\pmb{\rho}|_1\times |\sqrt{NT}(\widehat{\pmb{\beta}} -\pmb{\beta}_0)|_1\times\max_{j\in H}|\widehat{\pmb{\Sigma}}_{x}\widehat{\pmb{\Omega}}_{x,j}^\top - \mathbf{e}_j|_{\infty} \\
&=&O_P(s\sqrt{\log d}) O_P(\sqrt{\log d/(NT)}) = o_P(1). 
\end{eqnarray*}

Next, consider $J_1$. Write
$$
J_1 = \pmb{\rho}^\top (\widehat{\pmb{\Omega}}_x-\pmb{\Omega}_x )\frac{1}{\sqrt{NT}}\mathbf{X}^\top\mathbf{e} + \pmb{\rho}^\top\pmb{\Omega}_x\frac{1}{\sqrt{NT}}\mathbf{X}^\top\mathbf{e} \eqqcolon J_{1,1} + J_{1,2}.
$$
For $J_{1,1}$, by using Lemma \ref{LemmaB5} (3), we have
\begin{eqnarray*}
|J_{1,1}| &\leq& |\sum_{j\in H}(\widehat{\pmb{\Omega}}_{x,j} - \pmb{\Omega}_{x,j})\rho_j|_{1} \times |\frac{1}{\sqrt{NT}}\mathbf{X}^\top\mathbf{e}|_{\infty}\\
&=& O_P(\sum_{j\in H}s_j\sqrt{\log d/(NT)}) O_P(\sqrt{\log d})=o_P(1). 
\end{eqnarray*}

To complete the proof, it is sufficient to show that \ $J_{1,1} \to_D N(0,1)$.

We next prove the asymptotic normality by using martingale apprxoimation technique and the martingale central limit theorem. Define $S_{j,NT} = \frac{1}{\sqrt{T}}\sum_{t=1}^{T}\overline{u}_{j,t}$ and $\overline{u}_{j,t} = \frac{1}{\sqrt{N}}\sum_{i=1}^{N}X_{j,it}e_{it}$. Hence, $J_{1,1} = \sum_{k\in H}\sum_{j=1}^{d} \rho_k \Omega_{x,k,j} S_{j,NT}$.

Let $S_{j,NT,L} = \frac{1}{\sqrt{T}}\sum_{t=1}^{T}\sum_{l=0}^{L-1}\mathcal{P}_{t-l} (\overline{u}_{j,t})$ and $\widehat{S}_{j,NT,L} = \frac{1}{\sqrt{T}}\sum_{t=1}^{T}\sum_{l=0}^{L-1}\mathcal{P}_{t} (\overline{u}_{j,t+l})$, in which  $\mathcal{P}_{t}(\cdot) = E[\cdot \mid \mathcal{F}_t] - E[\cdot \mid \mathcal{F}_{t-1}]$ , $L\to \infty$ and $L/T\to0$. Note that $\overline{u}_{j,t} = \sum_{l=0}^{\infty}\mathcal{P}_{t-l} (\overline{u}_{j,t})$ and $\{\mathcal{P}_{t-l} (\overline{u}_{j,t})\}_{t=1}^T$ is a sequence of martingale differences, and thus
\begin{eqnarray*}
\|S_{j,NT,L} - S_{j,NT}\|_2 &=&\left\|\frac{1}{\sqrt{T}}\sum_{t=1}^{T}\sum_{l=L}^{\infty}\mathcal{P}_{t-l} (\overline{u}_{j,t}) \right\|_2 \leq \frac{1}{\sqrt{T}}\sum_{l=L}^{\infty}\left\|\sum_{t=1}^{T} \mathcal{P}_{t-l} (\overline{u}_{j,t})\right\|_2 \\
&\leq&  O(1)\frac{1}{\sqrt{T}}\sum_{l=L}^{\infty}\left\{\sum_{t=1}^{T} E\left[ \left(\mathcal{P}_{t-l} (\overline{u}_{j.t})\right)^2\right]\right\}^{1/2} \le  O(1)\sum_{l=L}^{\infty}\delta_{2}(u_j,l) \to 0,
\end{eqnarray*}
where the second inequality follows from Burkholder inequality, and the last step follows from $L\to \infty$. Similarly, by using Burkholder inequality, we have as $L\to \infty$ and $L/T\to0$
\begin{eqnarray*}
\|\widehat{S}_{j,NT,L} - S_{j,NT,L}\|_2 &\leq& \frac{1}{\sqrt{T}}\sum_{l=0}^{L-1} \left\|\sum_{t=1}^{l}\mathcal{P}_{t-l} (\overline{u}_{j,t})\right\|_2  + \frac{1}{\sqrt{T}}\sum_{l=0}^{L-1}\left\|\sum_{t=T-l+1}^{T}\mathcal{P}_{t} (\overline{u}_{j,t+l})\right\|_2  \to 0.
\end{eqnarray*}
Hence, we have $\|\widehat{S}_{j,NT,L} - S_{j,NT}\|_2 \to 0$. Note that $\{\sum_{l=0}^{L-1}\mathcal{P}_{t} (\overline{u}_{j,t+l})\}_{t=1}^{T}$ is a sequence of martingale differences subject to $\mathscr{F}_t$, so the asymptotic normality can be easily obtained by using a standard martingale central limit theory. 

The proof is now completed.
\end{proof}

\begin{proof}[Proof of Theorem \ref{THEOREM2}]
\item 
Write
\begin{eqnarray*}
T_u(\widehat{\pmb{\Theta}}_{\ell}) - \pmb{\Theta} &=& \left(T_u(\widehat{\pmb{\Theta}}_{\ell}) - E(\widetilde{\pmb{\Theta}}_{\ell})\right)  + \left( E(\widetilde{\pmb{\Theta}}_{\ell}) - \pmb{\Theta}\right) \eqqcolon J_3 + J_4,
\end{eqnarray*}
where $\widetilde{\pmb{\Theta}}_{\ell} = \left(\Theta_{\ell,kl}\right)_{k,l\leq d} = \frac{1}{NT}\sum_{i, j=1}^{N}\sum_{t, s=1}^{T}a((t-s)/\ell) \mathbf{x}_{it}\mathbf{x}_{js}^\top e_{it} e_{js}$.

Consider the bias term $J_4$ first. Note that $\pmb{\Gamma}_k = E(\overline{\mathbf{u}}_0\overline{\mathbf{u}}_{k}^\top)$ and thus
\begin{eqnarray*}
J_4 &=& \sum_{k=-\ell}^{\ell} [a(k/\ell) - 1] \pmb{\Gamma}_k - \sum_{k=-\ell}^{\ell}\frac{k}{T} [a(k/\ell) - 1] \pmb{\Gamma}_k -  \sum_{k=\ell+1}^{T-1}\frac{T-k}{k} \left(\pmb{\Gamma}_k + \pmb{\Gamma}_k^\top\right)\\
&\eqqcolon &J_{4,1} + J_{4,2} + J_{4,3}.
\end{eqnarray*}
For $J_{4,1}$, by Assumption \ref{ASSUMPTION4} (1), for $\forall\epsilon > 0$, we choose $\nu_\epsilon > 0$ such that  
\begin{eqnarray*}
|k/\ell| < \nu_\epsilon\quad \text{and}\quad \left|\frac{1-a(k/\ell)}{|k/\ell|^{q_a}}-C_q \right| < \epsilon.
\end{eqnarray*}
Letting $\ell_T^* = \lfloor\nu_\epsilon\ell \rfloor$, write
\begin{eqnarray*}
\ell^{q_a} \times J_{4,1} &=& \sum_{k=-\ell_T^*}^{\ell_T^*}\frac{ a\left(k/\ell\right)-1}{\left|k/\ell\right|^{q_a}}|k|^{q_a} \pmb{\Gamma}_k + \sum_{k=\ell_T^*+1}^{\ell}\frac{ a\left(k/\ell\right)-1}{\left|k/\ell\right|^q}|k|^q (\pmb{\Gamma}_k + \pmb{\Gamma}_k^\top).
\end{eqnarray*}
Since $|\pmb{\Gamma}_k|_2 = O(k^{-(q_a+\epsilon)})$ for some $\epsilon > 1$, the first term on the RHS converges to $-C_q\sum_{k=-\infty}^{\infty}|k|^{q_a}\pmb{\Gamma}_k$ in spectral norm. For the second term, since $|a(\cdot)|\leq M$, $\left|\frac{1-a(k/\ell)}{|k/\ell|^{q_a}} \right| \leq (M+1)/\nu_\epsilon^{q_a}$ due to the fact that $k/\ell\geq \nu_\epsilon$. Then this term in spectral norm is bounded by
\begin{eqnarray*}
(M+1)/\nu_\epsilon^{q_a} \sum_{k=\ell_T^*+1}^{\infty}k^2|\pmb{\Gamma}_k+\pmb{\Gamma}_k^\top|_2 \to 0.
\end{eqnarray*}
Therefore, $|J_{4,1}|_2 = O(\ell^{-q_a})$. Similarly, it can be shown that $|J_{4,2}|_2 = o(\ell^{-q_a})$ and $|J_{4,3}|_2 = o(\ell^{-q_a})$, which implies that $|J_4|_2 = O(\ell^{-q_a})$.

Now consider $J_3$. We first prove $\max_{1\leq k,l\leq d}|\widehat{\Theta}_{\ell,kl} - \Theta_{\ell,kl}| = O_P(\sqrt{\ell \log d / T})$. Write
\begin{eqnarray*}
\widehat{\Theta}_{\ell,kl} - \Theta_{\ell,kl} &=& \frac{1}{T}\sum_{t,s=1}^{T}a((t-s)/\ell)\left(\widehat{\overline{u}}_{k,t}\widehat{\overline{u}}_{l,s} - \overline{u}_{k,t}\overline{u}_{l,s}\right) + \frac{1}{T}\sum_{t,s=1}^{T}a((t-s)/\ell)\left(\overline{u}_{k,t}\overline{u}_{l,s} - E(\overline{u}_{k,t}\overline{u}_{l,s})\right) \\
&\eqqcolon & J_5 + J_6.
\end{eqnarray*}
By using Lemma \ref{LemmaB8} (1), we have $\max_{1\leq k,l\leq d}|J_6| = O_P(\sqrt{\ell \log d /T})$. For $J_5$, write
\begin{eqnarray*}
J_5&=&\frac{1}{T}\sum_{t,s=1}^{T}a((t-s)/\ell)\left(\widehat{\overline{u}}_{k,t}-\overline{u}_{k,t}\right)\left(\widehat{\overline{u}}_{l,s} - \overline{u}_{l,s}\right) + \frac{1}{T}\sum_{t,s=1}^{T}a((t-s)/\ell)\left(\widehat{\overline{u}}_{k,t}-\overline{u}_{k,t}\right)\overline{u}_{l,s}\\
&& + \frac{1}{T}\sum_{t,s=1}^{T}a((t-s)/\ell)\overline{u}_{k,s}\left(\widehat{\overline{u}}_{l,t}-\overline{u}_{l,t}\right).
\end{eqnarray*}
Applying Lemmas \ref{LemmaB8} (2)-(3) to the above three terms, we have $J_5 = o_P(\sqrt{\ell\log d/T})$. Hence, we have proved $\max_{1\leq k,l\leq d}|\widehat{\Theta}_{\ell,kl} - \Theta_{\ell,kl}| = O_P(\sqrt{\ell \log d / T})$. Next, write
$$
J_3 = \left|T_u(\widehat{\pmb{\Theta}}_{\ell}) - E(\widetilde{\pmb{\Theta}}_{\ell})\right|_2 \leq \left|T_u(\widehat{\pmb{\Theta}}_{\ell}) - T_u(E(\widetilde{\pmb{\Theta}}_{\ell}))\right|_2+\left|T_u(E(\widetilde{\pmb{\Theta}}_{\ell})) - E(\widetilde{\pmb{\Theta}}_{\ell})\right|_2 \eqqcolon J_{3,1} + J_{3,2}.
$$
The term $J_{3,2}$ is bounded by
$$
\max_{1\leq k\leq d}\sum_{l=1}^{d}|\Theta_{\ell,kl}|\mathbb{I}(|\Theta_{\ell,kl}|\leq u) \leq u^{1-p} C(d) = O\left( (\ell\log d/T)^{(1-p)/2}C(d)\right).
$$
For $J_{3,1}$,
\begin{eqnarray*}
J_{3,1} &\leq& \max_{1\leq k\leq d}\sum_{l=1}^{d}|\widehat{\Theta}_{\ell,kl}|\mathbb{I}(|\widehat{\Theta}_{\ell,kl}| \geq u,|\Theta_{\ell,kl}| < u) + \max_{1\leq k\leq d}\sum_{l=1}^{d}|\Theta_{\ell,kl}|\mathbb{I}(|\widehat{\Theta}_{\ell,kl}| < u,|\Theta_{\ell,kl}| \geq u)\\
&& + \max_{1\leq k\leq d}\sum_{l=1}^{d}|\widehat{\Theta}_{\ell,kl}- {\Theta}_{\ell,kl}|\mathbb{I}(|\widehat{\Theta}_{\ell,kl}| \geq u,|\Theta_{\ell,kl}| \geq u) \eqqcolon J_{3,11} + J_{3,12} + J_{3,13}.
\end{eqnarray*}
For $J_{3,13}$, by using $\max_{1\leq k,l\leq d}|\widehat{\Theta}_{\ell,kl} - \Theta_{\ell,kl}| = O_P(\sqrt{\ell \log d / T})$, we have
$$
J_{3,13}\leq \max_{1\leq k,l\leq d}|\widehat{\Theta}_{\ell,kl} - \Theta_{\ell,kl}| \max_{1\leq k\leq d} \sum_{l=1}^{d} |\Theta_{\ell,kl}|^{p}u^{-p} =  O_P\left( (\ell\log d/T)^{(1-p)/2}C(d)\right).
$$
For $J_{3,11}$, write
\begin{eqnarray*}
J_{3,11}&\leq& \max_{1\leq k\leq d}\sum_{l=1}^{d}|\widehat{\Theta}_{\ell,kl}-{\Theta}_{\ell,kl}|\mathbb{I}(|\widehat{\Theta}_{\ell,kl}| \geq u,|\Theta_{\ell,kl}| < u) + \max_{1\leq k\leq d}\sum_{l=1}^{d}|{\Theta}_{\ell,kl}|\mathbb{I}(|\Theta_{\ell,kl}| < u)\\
&\eqqcolon & J_{3,111} + J_{3,112}.
\end{eqnarray*}
For $J_{3,112}$, we have $J_{3,112} \leq u^{1-p} C(d).$ For $J_{3,111}$, take some $a \in (0,1)$, then
\begin{eqnarray*}
J_{3,112}&\leq& \max_{1\leq k\leq d}\sum_{l=1}^{d}|\widehat{\Theta}_{\ell,kl}-{\Theta}_{\ell,kl}|\mathbb{I}(|\widehat{\Theta}_{\ell,kl}| \geq u,|\Theta_{\ell,kl}| < au)\\
&& + \max_{1\leq k\leq d}\sum_{l=1}^{d}|\widehat{\Theta}_{\ell,kl}-{\Theta}_{\ell,kl}| \mathbb{I}(|\widehat{\Theta}_{\ell,kl}| \geq u,au\leq|\Theta_{\ell,kl}| < u) \\
&\leq& \max_{1\leq k,l\leq d}|\widehat{\Theta}_{\ell,kl}-{\Theta}_{\ell,kl}|\times\max_{1\leq k\leq d}\sum_{l=1}^{d}\mathbb{I}(|\widehat{\Theta}_{\ell,kl}-\Theta_{\ell,kl}| > (1-a)u) \\
&& + \max_{1\leq k,l\leq d}|\widehat{\Theta}_{\ell,kl}-{\Theta}_{\ell,kl}| \times C(d)(au)^{1-p}.
\end{eqnarray*} 
Note that since $\max_{1\leq k,l\leq d}|\widehat{\Theta}_{\ell,kl} - \Theta_{\ell,kl}| = O_P(\sqrt{\ell \log d / T})$, then for some large $C>0$ and $u= C\sqrt{\ell \log d / T}$, we have
$$
\Pr\left(\max_{1\leq k\leq d}\sum_{l=1}^{d}\mathbb{I}(|\widehat{\Theta}_{\ell,kl}-\Theta_{\ell,kl}| > au) > 0 \right) = \Pr\left(\max_{1\leq k,l\leq d}|\widehat{\Theta}_{\ell,kl}-{\Theta}_{\ell,kl}| > (1-a)u \right)\to0.
$$
Combining the above analysis, we have $J_{3,11} = O_P\left( (\ell\log d/T)^{(1-p)/2}C(d)\right)$.

For $J_{3,12}$, we have
\begin{eqnarray*}
J_{3,12} &\leq& \max_{1\leq k\leq d}\sum_{l=1}^{d}\left(|\widehat{\Theta}_{\ell,kl}-\Theta_{\ell,kl}|+|\widehat{\Theta}_{\ell,kl}|\right)\mathbb{I}(|\widehat{\Theta}_{\ell,kl}| < u,|\Theta_{\ell,kl}| \geq u)\\
&\leq&\max_{1\leq k,l\leq d}|\widehat{\Theta}_{\ell,kl}-{\Theta}_{\ell,kl}|\times\max_{1\leq k\leq d}\sum_{l=1}^{d}\mathbb{I}(|\Theta_{\ell,kl}| \geq u) + u \max_{1\leq k\leq d}\sum_{l=1}^{d}\mathbb{I}(|\Theta_{\ell,kl}| \geq u) \nonumber \\
&=& O_P\left( (\ell\log d/T)^{(1-p)/2}C(d)\right).
\end{eqnarray*}

Combining the above results, the proof is now completed.
\end{proof}

\newpage

\setcounter{page}{1}

\begin{center}
{\large \textbf{Online Supplementary Appendices B and C to  \\``Robust Estimation and Inference for \\High--Dimensional Panel Data Models"}\blfootnote{

$^*$Department of Econometrics and Business Statistics,  Monash University. 

$^\dag$School of Finance, Nankai University. 

$^\ddag$School of Statistics and Management, Shanghai University of Finance and Economics. }
\bigskip}

\bigskip

{\sc Jiti Gao$^{\ast}$} and {\sc Fei Liu$^{\dag}$}  and {\sc Bin Peng$^{\ast}$} and {\sc Yayi Yan$^{\ddag}$}

\bigskip

\end{center}

In this file, we provide two appendices. Specifically, in Appendix B, Appendix \ref{AP.B1} presents the proofs of technical lemmas collected in Appendix \ref{AP.A4}. Appendix \ref{AP.B2} provide the full proofs of the main results listed in Section \ref{Sec3}. Appendix C provides some technical lemmas, which are used to prove the results of Section \ref{Sec3} and the asymptotic properties of the oracle least squares estimator for the HD panel models with interactive fixed effects.

\renewcommand{\thesection}{B}

\section*{Appendix B} 

\renewcommand{\theequation}{B.\arabic{equation}}
\renewcommand{\thesection}{B.\arabic{section}}
\renewcommand{\thefigure}{B.\arabic{figure}}
\renewcommand{\thetable}{B.\arabic{table}}
\renewcommand{\thelemma}{B.\arabic{lemma}}
\renewcommand{\theassumption}{B.\arabic{assumption}}
\renewcommand{\thetheorem}{B.\arabic{theorem}}
\renewcommand{\theproposition}{B.\arabic{proposition}}

\setcounter{equation}{0}
\setcounter{lemma}{0}
\setcounter{section}{0}
\setcounter{table}{0}
\setcounter{figure}{0}
\setcounter{assumption}{0}
\setcounter{proposition}{0}

\section{Proofs of the Technical Lemmas}\label{AP.B1}

\begin{proof}[Proof of Lemma \ref{LemmaB3}]
\item
\noindent (1). We apply the martingale approximation technique and approximate $x_{j,it}e_{it}$ by $m$-dependent processes. Then we prove this lemma by using delicate block techniques and the results on independent random variables. 

For $t\geq 1$ and $m\geq 0$, define $S_{j,Nt,m} \coloneqq \sum_{i=1}^{N}\sum_{k=1}^{t}u_{j,ik,m}$ and $u_{j,ik,m} \coloneqq E\left[x_{j,ik}e_{ik} \mid \mathcal{F}_{k,k-m}\right]$, where $\mathcal{F}_{k,k-m} = (\pmb{\varepsilon}_k,\pmb{\varepsilon}_{k-1},\ldots,\pmb{\varepsilon}_{k-m})$. Therefore, $u_{j,ik,m}$ and $u_{j,ik',m}$ are independent if $|k-k'|>m$. Write
$$
u_{j,it} = u_{j,it} - u_{j,it,T} + \sum_{m=1}^{T}(u_{j,it,m}-u_{j,it,m-1}) + u_{j,it,0},
$$
which follows that
\begin{eqnarray*}
\|\max_{1\leq t\leq T}|S_{j,Nt}| \|_q &\leq& \|\max_{1\leq t \leq T} | S_{j,Nt} - S_{j,Nt,T}| \|_q + \sum_{m=1}^{T} \|\max_{1\leq t \leq T} | S_{j,Nt,m} - S_{j,Nt,m-1}| \|_q \\
&& + \|\max_{1\leq t \leq T} | S_{j,Nt,0} | \|_q \nonumber \\
&\coloneqq & I_1 + I_2 + I_3.
\end{eqnarray*}

Consider $I_2$ first. By the stationarity of $\{\sum_{i=1}^{N}u_{j,ik,m}\}_{k=1}^{T}$, we have
$$
\|\max_{1\leq t \leq T} | S_{j,Nt,m} - S_{j,Nt,m-1}| \|_q 
= \|\max_{0 \leq t \leq T-1} | \sum_{k = T-t}^{T} \sum_{i=1}^{N}(u_{j,ik,m} - u_{j,ik,m-1})| \|_q.
$$
Note that $\{u_{j,i(T-k),m} - u_{j,i(T-k),m-1}\}_{k=0}^{T-1}$ are martingale differences, and thus 
$$
\left\{\left|\sum_{k = T-t}^{T} \sum_{i=1}^{N}(u_{j,ik,m} - u_{j,ik,m-1})\right|\right\}_{t=0}^{T-1}
$$
is a positive submartingale with respect to $(\pmb{\varepsilon}_{T-t-m},\pmb{\varepsilon}_{T-t-m+1},\ldots)$. Therefore, by Doob $L_p$ maximal inequality, we have
$$
\|\max_{1\leq t \leq T} | S_{j,Nt,m} - S_{j,Nt,m-1}| \|_q \leq q/(q-1) \|S_{j,NT,m} - S_{j,NT,m-1}\|_q.
$$
Define $Y_{l,m} = \sum_{i=1}^{N}\sum_{k=1+(l-1)m}^{\min(lm,T)}(u_{j,ik,m}-u_{j,ik,m-1})$ and $b = \lfloor T/m \rfloor + 1$, we have
$$
|S_{j,NT,m} - S_{j,NT,m-1}| = \left| \sum_{l = 1}^{b}Y_{l,m} \right|.
$$
Note that $Y_{1,m},Y_{3,m},\ldots$ are independent and $Y_{2,m},Y_{4,m},\ldots$ are also independent, then by Rosenthal inequality for independent variables (e.g., \citealp{johnson1985best}),
$$
\| \sum_{l \text{ is odd}}Y_{l,m}\|_q \leq \frac{14.5q}{\log q}\left( \| \sum_{l \text{ is odd}}Y_{l,m}\|_2 + \left(\sum_{l \text{ is odd}}\| Y_{l,m}\|_q^q\right)^{1/q}\right),
$$
for $q \geq 2$ we have
\begin{eqnarray*}
\|S_{j,NT,m} - S_{j,NT,m-1}\|_q &\leq& \frac{14.5q}{\log q}\left[\| \sum_{l \text{ is odd}}Y_{l,m}\|_2 + \|\sum_{l \text{ is even}}Y_{l,m}\|_2 \right. \\
&&\left. + \left(\sum_{l \text{ is odd}}\| Y_{l,m}\|_q^q\right)^{1/q} + \left(\sum_{l \text{ is even}}\| Y_{l,m}\|_q^q\right)^{1/q} \right].
\end{eqnarray*}
In addition, by Burkholder inequality for martingale differences, we have
$$
\|Y_{l,m}\|_q^2 \leq (q-1)\sum_{k = 1 + (l-1)m}^{\min(lm,T)} \|\sum_{i=1}^{N}(u_{j,ik,m}-u_{j,ik,m-1})\|_q^2.
$$
For $\|\sum_{i=1}^{N}(u_{j,ik,m}-u_{j,ik,m-1})\|_q$, by Assumption \ref{ASSUMPTION1} we have
\begin{eqnarray*}
\|\sum_{i=1}^{N}(u_{j,ik,m}-u_{j,ik,m-1})\|_q &=& \|\sum_{i=1}^{N}E\left(u_{j,i}(\mathcal{F}_k)\mid \mathcal{F}_{k,k-m}\right)-E\left(u_{j,i}(\mathcal{F}_k)\mid\mathcal{F}_{k,k-m+1}\right) \|_q\\
&=&\|\sum_{i=1}^{N}E\left(u_{j,i}(\mathcal{F}_m)\mid \mathcal{F}_{m,0}\right)-E\left(u_{j,i}(\mathcal{F}_m)\mid\mathcal{F}_{m,1}\right) \|_q\\
&=&\|\sum_{i=1}^{N}E\left(u_{j,i}(\mathcal{F}_m)\mid \mathcal{F}_{m,0}\right)-E\left(u_{j,i}(\mathcal{F}_m^*)\mid\mathcal{F}_{m,0}\right) \|_q\\
&\leq&\sqrt{N} \delta_q(u_j,m).
\end{eqnarray*}
Hence, for $1\leq l \leq b$ we have
$$
\|Y_{l,m}\|_q \leq \sqrt{(q-1)(\min(lm,T)-(l-1)m)}\sqrt{N} \delta_q(u_j,m).
$$
Then, for $1\leq m\leq T$, we have
$$
\|S_{j,NT,m} - S_{j,NT,m-1}\|_q\leq \frac{29 q}{\log q}\left(\sqrt{NT} \delta_2(u_j,m) + \sqrt{N}\sqrt{q-1}T^{1/q}m^{1/2-1/q}\delta_q(u_j,m)\right).
$$
By the above developments and $q/(q-1)\leq 2$ when $q\geq 2$, we have
\begin{eqnarray*}
&&\sum_{m=1}^{T}\|\max_{1\leq t \leq T} | S_{j,Nt,m} - S_{j,Nt,m-1}| \|_q \\
&\leq& \frac{87q}{\log q}\left(\sqrt{NT} \sum_{m=1}^{T}\delta_2(u_j,m)+\sqrt{N}\sqrt{q-1}T^{1/q}\sum_{m=1}^{T}m^{1/2-1/q}\delta_q(u_j,m)\right).
\end{eqnarray*}

Similar to the above developments, for $I_1$, we have
$$
\|\max_{1\leq t \leq T} | S_{j,Nt} - S_{j,Nt,T}| \|_q \leq 3\sqrt{q-1}\sqrt{NT}\sum_{m=T+1}^{\infty}\delta_q(u_j,m).
$$

For $I_3$, note that $\{u_{j,it,0}\}$ are independent random variables, by Doob $L_p$ maximal inequality and Rosenthal inequality for independent variables, we have
\begin{eqnarray*}
\|\max_{1\leq t \leq T} | S_{j,Nt,0} | \|_q &\leq& q/(q-1) \|S_{j,NT,0} \|_q\\
&\leq & q/(q-1)\frac{14.5q}{\log q}\left( \sqrt{T}\| \sum_{i=1}^{N}u_{j,it,0} \|_2 + T^{1/q}\left(\| \sum_{i=1}^{N}u_{j,it,0} \|_q^q\right)^{1/q}\right).
\end{eqnarray*}
Note that $u_{j,it,0} = E(u_{j,i}(\mathcal{F}_t) \mid \mathcal{F}_{t,t}) =_D E(u_{j,i}(\mathcal{F}_0) \mid \pmb{\varepsilon}_0)$ and $E(u_{j,i}(\mathcal{F}_0^*) \mid \pmb{\varepsilon}_0) = 0$, so we have
$$
\|\sum_{i=1}^{N}u_{j,it,0} \|_q \leq \sqrt{N}\delta_q(u_j,0).
$$
Hence, $
\|\max_{1\leq t \leq T} | S_{j,Nt,0} | \|_q \leq\frac{29}{\log q}\left( \sqrt{NT}\delta_q(u_j,0) + T^{1/q}\sqrt{N}\delta_q(u_j,0)\right)$.

Combining the above analysis, we have
\begin{eqnarray*}
\|\max_{1\leq t\leq T}|S_{j,Nt}| \|_q &\leq& \sqrt{NT} \left[ \frac{87q}{\log q}\sum_{m=1}^{T}\delta_2(u_j,m) + 3\sqrt{q-1}\sum_{m=T+1}^{\infty}\delta_q(u_j,m)+\frac{29q}{\log q}\delta_2(u_j,0)\right]\\
&& + \sqrt{N}T^{1/q}\left[\frac{87q\sqrt{q-1}}{\log q}\sum_{m=1}^{T}m^{1/2-1/q}\delta_q(u_j,m) +  \frac{29q}{\log q}\delta_q(u_j,0)\right]\\
&\leq&c_q \sqrt{NT} \Delta_{2}(u_j,0) + c_q \sqrt{N}T^{1/q}\sum_{m=1}^{\infty}(m\wedge T)^{1/2-1/q}\delta_q(u_j,m).
\end{eqnarray*}
The proof of part (1) is now complete.

\smallskip

\noindent (2). The proof is similar with that of part (1), so it is omitted here.
\end{proof}

\begin{proof}[Proof of Lemma \ref{LemmaB4}]
\item
The proof of Lemma \ref{LemmaB4} is similar with that of Lemma \ref{LEMMA1} (1), so it is omitted here.
\end{proof}

\begin{proof}[Proof of Lemma \ref{LemmaB5}]
\item 
\noindent (1). Note that
$$
|\mathbf{v}|_1 = |\mathbf{v}_{\mathbf{J}}|_1 + |\mathbf{v}_{\mathbf{J}^c}|_1 \leq 4|\mathbf{v}_{\mathbf{J}}|_1\leq 4\sqrt{s}|\mathbf{v}_{\mathbf{J}}|_2,
$$
which follows that
\begin{eqnarray*}
\mathbf{v}^\top \widehat{\pmb{\Sigma}}_{x} \mathbf{v} &\geq& \mathbf{v}^\top {\pmb{\Sigma}}_x \mathbf{v} - |\widehat{\pmb{\Sigma}}_{x} - {\pmb{\Sigma}}_x|_{\mathrm{max}}|\mathbf{v}|_1^2\geq \mathbf{v}^\top {\pmb{\Sigma}}_{x} \mathbf{v} - 16s|\widehat{\pmb{\Sigma}}_{x} - {\pmb{\Sigma}}_x|_{\mathrm{max}}|\mathbf{v}|_2^2.
\end{eqnarray*}

Using Lemma \ref{LemmaB4} with $x = \frac{C_0}{16}\sqrt{\log d/(NT)}$ for some constant $C_0 > 0$, $\mathcal{B}_{NT} = \{|\widehat{\pmb{\Sigma}}_{x} - {\pmb{\Sigma}}_{x}|_{\mathrm{max}} \leq x  \}$ holds with probability larger than $1 - C_1 (\frac{dT^{1-q/2}}{(\log d)^{q/2}} +  d^{-C_2} )$ for some $C_1,C_2>0$.

\smallskip

\noindent (2). Given the condition $|\pmb{\gamma}_j|_1<\infty$, we first show the weak dependence properties of $\pmb{\eta}_j$. By the construction of $\pmb{\eta}_j$, we have $\eta_{j,it} = x_{j,it} - \mathbf{x}_{-j,it}^\top \pmb{\gamma}_j$, where $\mathbf{x}_{-j,it}$ is the sub-vector of $\mathbf{x}_{it}$ with the $j^{th}$ column removed. Then, by using Assumption \ref{ASSUMPTION2}.1 and $|\pmb{\gamma}_j|_1<\infty$, for any $j'\neq j$, we have
\begin{eqnarray*}
&&\left\|\frac{1}{\sqrt{N}}\sum_{i=1}^{N}(x_{j',it}\eta_{j,it} - x_{j',it}^*\eta_{j,it}^*)\right\|_{q}\\
&\leq& \left\|\frac{1}{\sqrt{N}}\sum_{i=1}^{N}(x_{j',it}x_{j,it} - x_{j',it}^*x_{j,it}^*)\right\|_{q} + \sum_{l=1,\neq j}^{d}\left\|\frac{1}{\sqrt{N}}\sum_{i=1}^{N}(x_{j',it}x_{l,it}- x_{j',it}^*x_{l,it}^*)\right\|_{q} |\lambda_{j,l}|\\
&\leq&\max_{1\leq j\leq d}\left\|\frac{1}{\sqrt{N}}\sum_{i=1}^{N}(x_{j',it}x_{j,it} - x_{j',it}^*x_{j,it}^*)\right\|_{q} (1 + |\pmb{\gamma}_j|_1) =O(t^{-\alpha}).
\end{eqnarray*}

Given the weak dependence properties of $x_{j',it}\eta_{j,it}$ and $\widetilde{\omega}_j\asymp \sqrt{\log d/(NT)}$, the proof of part (2) is then similar with that of Lemma \ref{LEMMA2}, so it is omitted here.

\smallskip

\noindent (3). The Karush-Kuhn-Tucker (KKT) condition of the optimization problem \eqref{EQUATION5} implies that there exists a subgradient $\vec{\mathbf{k}}_j$ such that
\begin{eqnarray*}
-\frac{1}{NT}\mathbf{X}_{-j}^\top\left(\mathbf{X}_j - \mathbf{X}_{-j}\widehat{\pmb{\gamma}}_{j}\right) + \widetilde{\omega}_j\vec{\mathbf{k}}_j  &=& \mathbf{0}\\
-\widehat{\pmb{\gamma}}_j^\top\frac{1}{NT}\mathbf{X}_{-j}^\top\left(\mathbf{X}_j - \mathbf{X}_{-j}\widehat{\pmb{\gamma}}_{j}\right) + \widetilde{\omega}_j|\widehat{\pmb{\gamma}}_j|_1  &=& \mathbf{0}
\end{eqnarray*}
which implies that
\begin{eqnarray*}
\widehat{\tau}_j^2 &=& (\mathbf{X}_j - \mathbf{X}_{-j}\widehat{\pmb{\gamma}}_j)^\top\mathbf{X}_{j}/(NT) \\
&=& \frac{\pmb{\eta}_j^\top\pmb{\eta}_j}{NT} + \frac{\pmb{\eta}_j^\top\mathbf{X}_{-j}\pmb{\gamma}_j}{NT} - \frac{(\widehat{\pmb{\gamma}}_j - \pmb{\gamma}_j)^\top\mathbf{X}_{-j}^\top\mathbf{X}_{-j}\pmb{\gamma}_j}{NT} - \frac{(\widehat{\pmb{\gamma}}_j - \pmb{\gamma}_j)^\top\mathbf{X}_{-j}^\top\pmb{\eta}_j}{NT}.
\end{eqnarray*}

Hence, for all $j \in [d]$, we have
\begin{eqnarray*}
|\widehat{\tau}_j^2 - \tau_j^2| &\leq& \left|\frac{\pmb{\eta}_j^\top\pmb{\eta}_j}{NT} - \tau_j^2\right| + \left|\frac{\pmb{\eta}_j^\top\mathbf{X}_{-j}\pmb{\gamma}_j}{NT}\right| + \left|\frac{(\widehat{\pmb{\gamma}}_j - \pmb{\gamma}_j)^\top\mathbf{X}_{-j}^\top\mathbf{X}_{-j}\pmb{\gamma}_j}{NT}\right| + \left|\frac{(\widehat{\pmb{\gamma}}_j - \pmb{\gamma}_j)^\top\mathbf{X}_{-j}^\top\pmb{\eta}_j}{NT}\right| \\
&\coloneqq & I_4 + I_5 + I_6 + I_{7}.
\end{eqnarray*}

For $I_4$, it can be shown that $ |\frac{\pmb{\eta}_j^\top\pmb{\eta}_j}{NT} - \tau_j^2 | = O_P((NT)^{-1/2})$ by using similar arguments as the proof of Lemma \ref{LemmaB6} (1). For $I_5$, by part (2) of this lemma, we have
$$
I_5\leq |\pmb{\eta}_j^\top\mathbf{X}_{-j}/(NT)|_\infty|\pmb{\gamma}_j|_1\leq \sqrt{s_j}|\pmb{\eta}_j^\top\mathbf{X}_{-j}/(NT)|_\infty|\pmb{\gamma}_j|_2=O_P(\sqrt{s_j\log d/(NT)}).
$$

Consider $I_6$. Note that by the KKT condition, we have
\begin{eqnarray*}
-\frac{1}{NT}\mathbf{X}_{-j}^\top\left(\mathbf{X}_j - \mathbf{X}_{-j}\widehat{\pmb{\gamma}}_{j}\right) + \lambda_{j}\vec{\mathbf{k}}_j  &=& \mathbf{0}\\
-\frac{1}{NT}\mathbf{X}_{-j}^\top\left(\mathbf{X}_{-j}\pmb{\gamma}_j+\pmb{\eta}_j - \mathbf{X}_{-j}\widehat{\pmb{\gamma}}_{j}\right) + \lambda_{j}\vec{\mathbf{k}}_j &=& \mathbf{0}.
\end{eqnarray*}
Hence, by using part (2) of this lemma, we have
$$
\left|\frac{(\widehat{\pmb{\gamma}}_j - \pmb{\gamma}_j)^\top\mathbf{X}_{-j}^\top\mathbf{X}_{-j}}{NT}\right|_\infty = |\mathbf{X}_{-j}^\top\pmb{\eta}_j/(NT) + \lambda_{j}\vec{\mathbf{k}}_j|_\infty = O_P(\sqrt{\log d/(NT)})
$$
and
$$
I_6\leq \left|\frac{(\widehat{\pmb{\gamma}}_j - \pmb{\gamma}_j)^\top\mathbf{X}_{-j}^\top\mathbf{X}_{-j}}{NT}\right|_\infty |\pmb{\gamma}_j|_1 = O_P(\sqrt{s_j\log d/(NT)}).
$$
Similarly, $I_{7} = O_P(s_j\log d/(NT))$. Therefore, $|\widehat{\tau}_j^2 - \tau_j^2| = O_P(\sqrt{s_j\log d/(NT)})$.

Note that $\tau_j^2 = 1/ \Omega_{x,j,j} \geq 1/\psi_{\mathrm{max}}(\pmb{\Omega}_x)  = \psi_{\mathrm{min}}(\pmb{\Sigma}_x) > 0$, and thus $\widehat{\tau}_j^2 \geq \tau_j^2 - |\widehat{\tau}_j^2 - \tau_j^2| >0$ with probability approaching to 1 as $|\widehat{\tau}_j^2 - \tau_j^2| = O_P(\sqrt{s_j\log d/(NT)})$, which also implies that
$$
\left|\frac{1}{\widehat{\tau}_j^2} - \frac{1}{\tau_j^2}\right| = \frac{|\widehat{\tau}_j^2 - \tau_j^2|}{\widehat{\tau}_j^2\tau_j^2} = O_P(\sqrt{s_j\log d/(NT)})  .
$$

Next, consider $|\widehat{\pmb{\Omega}}_{x,j} - \pmb{\Omega}_{x,j}|_1$. As $\widehat{\pmb{\Omega}}_{x,j} = \widehat{\mathbf{C}}_j/\widehat{\tau}_j^2$, by using part (2) of this lemma and $|\pmb{\gamma}_j|_1 = O(\sqrt{s_j})$, we have
\begin{eqnarray*}
|\widehat{\pmb{\Omega}}_{x,j} - \pmb{\Omega}_{x,j}|_1 &=& |\widehat{\mathbf{C}}_j/\widehat{\tau}_j^2 - \mathbf{C}_j/\tau_j^2|_1 \leq |1/\widehat{\tau}_j^2 - 1/\tau_j^2| + |\widehat{\pmb{\gamma}}_j/\widehat{\tau}_j^2 - \pmb{\gamma}_j/\tau_j^2|_1\\
&\leq& |1/\widehat{\tau}_j^2 - 1/\tau_j^2| + |\widehat{\pmb{\gamma}}_j/\widehat{\tau}_j^2 - \pmb{\gamma}_j/\widehat{\tau}_j^2|_1+|\pmb{\gamma}_j/\widehat{\tau}_j^2 - \pmb{\gamma}_j/\tau_j^2|_1\\
&=&  O_P(\sqrt{s_j\log d/(NT)}) + O_P(s_j\sqrt{\log d/(NT)}) +  O_P(s_j\sqrt{\log d/(NT)})\\
&=&  O_P(s_j\sqrt{\log d/(NT)}).
\end{eqnarray*}
Similarly, by using $|\pmb{\gamma}_j|_2 = O(1)$, we have $|\widehat{\pmb{\Omega}}_{x,j} - \pmb{\Omega}_{x,j}|_2 =  O_P(\sqrt{s_j\log d/(NT)})$.
\end{proof}

\begin{proof}[Proof of Lemma \ref{LemmaB6}]
\item 
\noindent (1). Consider $\|\sum_{t=1}^{T}b_t(\overline{u}_{j,t}-\widetilde{\overline{u}}_{j,t})\|_q$ first. Define $D_{t,k} = E[\overline{u}_{j,t} \mid \mathcal{F}_{t-k,t}] - E[\overline{u}_{j,t} \mid \mathcal{F}_{t-k+1,t}]$, then $\{D_{t,k},\ t = T,\ldots,1\}$ form martingale differences with respect to $\mathcal{F}_{t-k,\infty}$. In addition, by Jensen inequality, we have $\|D_{t,k}\|_q \leq \delta_q(u_j,k)$.

By using Burkholder inequality and Minkowski inequality, for $q \geq 2$, we have
\begin{eqnarray*}
\left\|\sum_{t=1}^{T}b_tD_{t,k}\right\|_q^2 \leq O(1)\sum_{t=1}^{T}\|b_tD_{t,k} \|_q^2 = O(1)\sum_{t=1}^{T} b_t^2 \delta_q^2(u_j,k).
\end{eqnarray*}
Since $\overline{u}_{j,t} - \widetilde{\overline{u}}_{j,t} = \sum_{k=m+1}^{\infty}D_{t,k}$, the result follows.

The proof of $\|\sum_{t=1}^{T}b_t\overline{u}_{j,t}\|_q$ is similar to above, so omitted here.

\smallskip

\noindent (2). In what follows, let
\begin{eqnarray*}
Z_{j,t} &=& \sum_{s=1}^{t}a_{t+1-s}\overline{u}_{j,s},\quad \widetilde{Z}_{j',t} = \sum_{s=1}^{t}a_{t+1-s}\widetilde{\overline{u}}_{j',s},\\
L_{T,j,j'}^{\diamond} &=& \sum_{1\leq s < t \leq T}a_{t-s} \overline{u}_{j,t}\widetilde{\overline{u}}_{j',s}= \sum_{t=2}^{T} \overline{u}_{j,t}\widetilde{Z}_{j',t-1}.
\end{eqnarray*}
Also, let $Z_{j,t,\{k\}}, \overline{u}_{j,t,\{k\}}$ be the coupled version of $Z_t$ and $x_t$ replacing $\pmb{\varepsilon}_k$ with $\pmb{\varepsilon}_k^\prime$. 

Define the projection operator $\mathcal{P}_{t}(\cdot) = E[\cdot \mid \mathcal{F}_t] - E[\cdot \mid \mathcal{F}_{t-1}]$. By Jensen inequality,
\begin{eqnarray*}
&&\|\mathscr{P}_k(L_{T,j,j'} - L_{T,j,j'}^{\diamond})\|_{q/2} \\
&=& \left\|E\left[\sum_{t=2}^{T}[\overline{u}_{j,t}(Z_{j',t-1}-\widetilde{Z}_{j',t-1}) \mid \mathcal{F}_{k}\right] - E\left[\sum_{t=2}^{T}[\overline{u}_{j,t,\{k\}}(Z_{j',t-1,\{k\}}-\widetilde{Z}_{j',t-1,\{k\}}) \mid \mathcal{F}_{k}\right] \right\|_{q/2}\\
&\leq&\left\|\sum_{t=2}^{T}\overline{u}_{j,t,\{k\}}(Z_{j',t-1}-\widetilde{Z}_{j',t-1}-Z_{j',t-1,\{k\}}+\widetilde{Z}_{j',t-1,\{k\}})\right\|_{q/2}\\
&& + \left\|\sum_{t=2}^{T}[\overline{u}_{j,t}- \overline{u}_{j,t,\{k\}}](Z_{j',t-1}-\widetilde{Z}_{j',t-1})\right\|_{q/2}\\
&\eqqcolon &I_{8} + I_{9}.
\end{eqnarray*}

We consider $I_{8}$ first. Since $\|\overline{u}_{j',t} -  \overline{u}_{j',t,\{k\}}\|_q\leq\delta_q(u_{j'},t-k)$ and $\|\overline{u}_{j',t}- \widetilde{\overline{u}}_{j',t}\|_q\leq\Psi_{q}(u_j',m+1)$, we thus obtain that
\begin{eqnarray*}
\|\overline{u}_{j',t} - \widetilde{\overline{u}}_{j',t} - \overline{u}_{j',t,\{k\}} +  \widetilde{\overline{u}}_{j',t,\{k\}} \|_q\leq 2\min\left(\delta_q(u_{j'},t-k),\Psi_{q}(u_{j'},m+1)\right).
\end{eqnarray*}
In addition, by Minkowski inequality, Cauchy-Schwarz inequality and part (1) of this lemma,  we have
\begin{eqnarray*}
I_{8} &=& \left \|\sum_{s=1}^{T-1}(\overline{u}_{j',s} - \widetilde{\overline{u}}_{j',s} - \overline{u}_{j',s,\{k\}} +  \widetilde{\overline{u}}_{j',s,\{k\}}) \sum_{t=s+1}^{T}a_{t-s}\overline{u}_{j,t,\{k\}}\right\|_{q/2}\\
&\leq& \sum_{s=1}^{T-1}\left\|\overline{u}_{j',s} - \widetilde{\overline{u}}_{j',s} - \overline{u}_{j',s,\{k\}} +  \widetilde{\overline{u}}_{j',s,\{k\}} \right\|_q\cdot \left\|\sum_{t=s+1}^{T}a_{t-s}\overline{u}_{j,t,\{k\}} \right\|_q\\
&=& O(A_T)\sum_{s=1}^{T-1}\min \left(\delta_q(u_{j'},s-k),\Psi_q(u_{j'},m+1)\right).
\end{eqnarray*}

We now consider $I_{9}$. By part (1) of this lemma, we have  
\begin{eqnarray*}
\max_{1\leq t\leq T}\|Z_{j',t-1}-\widetilde{Z}_{j',t-1}\|_4\leq O(1)A_T\Delta_q(u_{j'},m+1).
\end{eqnarray*}
Also, note that $I_{9}$ is actually $I_{9, k}$, and the sub-index $k$ is suppressed previously for notational simplicity. In addition, by , we have
\begin{eqnarray*}
\sum_{k=-\infty}^{T}I_{9,k}^2 &\leq &O(1)\sum_{k=-\infty}^{T}A_T^2\Delta_q^2(u_{j'},m+1)\left(\sum_{t=1}^{T}\delta_q(u_{j},t-k)\right)^2 \\
&\leq& O(1)A_T^2\Delta_q^2(u_{j'},m+1) \sum_{k=-\infty}^{T}\sum_{t=1}^{T}\delta_q(u_j,t-k)
\\ 
&=&O(1)TA_T^2\Delta_q^2(u_{j'},m+1).
\end{eqnarray*}
Similarly, $\sum_{k=-\infty}^{T}I_{8,k}^2 = O(1)TA_T^2d_{j',m,q}^2$. Since $\{\mathscr{P}_k(L_T - L_T^{\diamond})\}_{k}$ is a sequence of martingale differences, by Burkholder inequality and $\Delta_q(u_{j'},m+1) \leq d_{j',m,q}$, we have
\begin{eqnarray*}
&&\|L_{T,j,j'} -E(L_{T,j,j'}) -L_{T,j,j'}^{\diamond} -E(L_{T,j,j'}^{\diamond})\|_{q/2}^2 \nonumber \\
& = & \left\|\sum_{k=-\infty}^{T}\mathscr{P}_k(L_{T,j,j'} - L_{T,j,j'}^{\diamond})\right\|_{q/2}^2\\
&\leq& O(1)\sum_{k=-\infty}^{T}\|\mathscr{P}_k(L_{T,j,j'} - L_{T,j,j'}^{\diamond})\|_{q/2}^2\nonumber \\
&=& O(1)TA_T^2d_{j',m,q}^2.
\end{eqnarray*}
Similarly, we have $
\|\widetilde{L}_{T,j,j'} -E(\widetilde{L}_{T,j,j'}) -L_{T,j,j'}^{\diamond} -E(L_{T,j,j'}^{\diamond})\|_{q/2}^2 = O(1)TA_T^2d_{j,m,q}^2$.

The proof is now complete.

\smallskip

\noindent (3). The proof is similar to that of part (2), so omitted here.
\end{proof}

\begin{proof}[Proof of Lemma \ref{LemmaB7}]
\item 
\noindent (1). Let $m = \lfloor T^{\theta} \rfloor$ for some $0< \theta < 1$. By Lemmas \ref{LemmaB6} (2) and (3) and using Markov inequality, we have
$$
\Pr(|L_{NT,j,j'} - E(L_{NT,j,j'}) - \widetilde{L}_{NT,j,j'} - E(\widetilde{L}_{NT,j,j'})|\geq x_T/2) \leq C_q x_{T}^{-q/2}T^{q/2 - (\alpha-1)\theta q/2}.
$$
Then to complete the proof, it is sufficient to show
$$
\Pr(|\widetilde{L}_{NT,j,j'} - E(\widetilde{L}_{NT,j,j'})|\geq x_T/2) \leq C_{q,\delta,\theta} x_{T}^{-q/2}(T \log T)(m^{q/2 -1-(\alpha-1)\theta q/2} + 1).
$$
since
\begin{eqnarray*}
&&\Pr(|L_{NT,j,j'} - E(L_{NT,j,j'})|\geq x_T) \nonumber \\ 
&\leq& \Pr(|L_{NT,j,j'} - E(L_{NT,j,j'}) - \widetilde{L}_{NT,j,j'} - E(\widetilde{L}_{NT,j,j'})|\geq x_T/2)\\
&&+\Pr(| \widetilde{L}_{NT,j,j'} - E(\widetilde{L}_{NT,j,j'})|\geq x_T/2).
\end{eqnarray*}

\smallskip

Let $z_T$ satisfy $T^{1+\delta} / z_T \to 0$ and $l_T = \lfloor -\log(\log T)/ (\log \theta) \rfloor$ such that $T^{\theta^{l_T}} \leq e$. Let $y_T = z_T/(2l_T)$ and $1<m<T/4$.

Define $\widetilde{\mathbf{u}}_{it,\theta} = E[\mathbf{u}_{it} \mid \mathcal{F}_{t-\lfloor m^{\theta}\rfloor,t}]$, $\widetilde{\overline{\mathbf{u}}}_{t,\theta} = \frac{1}{\sqrt{N}}\sum_{i=1}^{N}\widetilde{\mathbf{u}}_{it,\theta}$, $Y_{j',t,1} = \sum_{s=1}^{t-3m-1}a_{s,t}\widetilde{\overline{u}}_{j',s}$ and $Z_{j',t,1} = \sum_{s=1 \vee (t-3m)}^{t}a_{s,t}\widetilde{\overline{u}}_{j',s}$. Define $Y_{j',t,2}$ and $Z_{j',t,2}$ similarly by replacing $\widetilde{\overline{u}}_{j',s}$ with $\widetilde{\overline{u}}_{j',s,\theta}$. Observe that $\widetilde{\overline{u}}_{j,t}$ and $\widetilde{\overline{u}}_{j,t,\theta}$ are independent of $Y_{j',t,l}$ for $l=1,2$. 

We first consider $\sum_{t=1}^{T}(\widetilde{\overline{u}}_{j,t}Z_{j',t,1} - \widetilde{\overline{u}}_{j,t,\theta}Z_{j',t,2})$. Split $[T]$ into blocks $B_1,\ldots,B_{b_T}$ with block size $4m$ and define $W_{T,b} = \sum_{t\in B_b}(\widetilde{\overline{u}}_{j,t}Z_{j',t,1} - \widetilde{\overline{u}}_{j,t,\theta}Z_{j',t,2})$. Let $y_T$ satisfy $y_T < z_T /2$ and $T^{1+\delta/2} / y_T \to 0$. Since $W_{T,b}$ and $W_{T,b'}$ are independent if $|b-b'| > 1$, by Lemma \ref{LemmaB2} (1), Lemma \ref{LemmaB6} (3) and Lemma \ref{LemmaB6} (2), for any $M > 1$, there exists a constant $C_{q,M,\delta,\theta}$ such that
\begin{eqnarray*}
&& \Pr\left(\left|\sum_{t=1}^{T}(\widetilde{\overline{u}}_{j,t}Z_{j',t,1} - \widetilde{\overline{u}}_{j,t,\theta}Z_{j',t,2}) - E\left(\sum_{t=1}^{T}(\widetilde{\overline{u}}_{j,t}Z_{j',t,1} - \widetilde{\overline{u}}_{j,t,\theta}Z_{j',t,2})\right) \right| \geq y_T  \right)\\
&\leq& C_{q,M,\delta,\theta} y_T^{-M} + \sum_{b=1}^{b_T} \Pr\left(|W_{T,b}-E(W_{T,b})| \geq y_T /C_{M,\delta} \right)\\
&\leq& C_{q,M,\delta,\theta} y_T^{-M} + C_{q,M,\delta,\theta} y_T^{-q/2}  T m^{q/2-1-(\alpha-1)\theta q/2}.
\end{eqnarray*}

We next deal with the term $\sum_{t=1}^{T}(\widetilde{\overline{u}}_{j,t}Y_{j',t,1} - \widetilde{\overline{u}}_{j,t,\theta}Y_{j',t,2})$. Split $[T]$ into blocks $B_{1}^*,\ldots,B_{b_T}^*$ with size $m$. Define $R_{T,b} = \sum_{t\in B_{b}^*  } (\widetilde{\overline{u}}_{j,t}Y_{j',t,1} - \widetilde{\overline{u}}_{j,t,\theta}Y_{j',t,2})$. Let $\xi_b$ be the $\sigma$-fields generated by $(\pmb{\varepsilon}_{l_b},\pmb{\varepsilon}_{l_{b-1}},\ldots)$, where $l_b = \max\{B_{b}^*\}$. Note that $\{R_{T,b}\}_{b \text{ is odd}}$ is a martingale sequence with respect to $\{\xi_b\}_{b \text{ is odd}}$, and so are $\{R_{T,b}\}_{b \text{ is even}}$ and $\{\xi_b\}_{b \text{ is even}}$. By Lemma 1 of \cite{haeusler1984exact}, for any $M>1$ there exists a constant $C_{M,\delta}$ such that 
\begin{eqnarray*}
&&\Pr\left(\left|\sum_{t=1}^{T}(\widetilde{\overline{u}}_{j,t}Y_{j',t,1} - \widetilde{\overline{u}}_{j,t,\theta}Y_{j',t,2})\right| \geq y_T\right)\\
&\leq& C_{M,\delta} y_T^{-M} + 4\Pr\left(\left|\sum_{b=1}^{B_T^*}E( R_{T,b}^2 \mid \xi_{b-2})\right| > y_T^2/(\log T)^{3/2}\right)\\
&& + \sum_{b=1}^{B_T^*}\Pr\left(\left|R_{T,b}\right| \geq y_T / \log T\right)\nonumber \\
&\coloneqq & I_{10} + I_{11} + I_{12}.
\end{eqnarray*}

Consider $I_{12}$ first. Since $({\overline{u}}_{j,t}, {\overline{u}}_{j,t,\theta})$ and  $(Y_{j',t,1},Y_{j',t,2})$ are independent, by using similar arguments as Lemma \ref{LemmaB6} (2), we have
$$
\|R_{T,b}\|_{q} \leq C_q (mT)^{q/2} m^{-(\alpha-1) \theta q},
$$
which implies that
$$
I_{12} \leq C_q y_T^{-q} (\log y_T)^{q} T^{q/2+1}m^{q/2-1-(\alpha-1)\theta q}.
$$

Consider $I_{11}$. Let $r_{j,s-t} = E(\widetilde{\overline{u}}_{j,t}\widetilde{\overline{u}}_{j,s})$ and $r_{j,s-t,\theta} = E(\widetilde{\overline{u}}_{j,t,\theta}\widetilde{\overline{u}}_{j,s,\theta})$, then we have 
\begin{eqnarray*}
\sum_{b=1}^{b_T^*}E(R_{T,b}^2 \mid \xi_{b,2}) &\leq& 2\sum_{b=1}^{b_T^*} \left[\sum_{s,t \in B_b^*} (r_{j,s-t} Y_{j',s,1}Y_{j',t,1}+ r_{j,s-t,\theta}Y_{j',s,2}Y_{j',t,2} )\right]\\
&=& \sum_{1\leq s\leq t\leq T}a_{j,s,t} \widetilde{\overline{u}}_{j',t}\widetilde{\overline{u}}_{j',s} +  \sum_{1\leq s\leq t\leq T}a_{j,s,t,\theta} \widetilde{\overline{u}}_{j',t,\theta}\widetilde{\overline{u}}_{j',s,\theta},
\end{eqnarray*}
in which $|a_{j,s,t}|\leq CT$ and $|a_{j,s,t,\theta}|\leq CT$ since $\sum_{l\in \mathbb{Z}} |r_{j,l}| < \infty$ and $\sum_{l\in \mathbb{Z}} |r_{j,l,\theta}| < \infty$. Define $U(T,m,x_T) = \sup_{\{a_{j,s,t}\}} \Pr\left(\left|\sum_{1\leq s\leq t\leq T}a_{j,s,t} \widetilde{\overline{u}}_{j',t}\widetilde{\overline{u}}_{j',s} - E\left(\sum_{1\leq s\leq t\leq T}a_{s,t} \widetilde{\overline{u}}_{j',t}\widetilde{\overline{u}}_{j',s}\right) \right| \geq x_T\right)$, where the supremum is taken over all arrays $\{a_{s,t}\}$ such that $|a_{s,t}|\leq 1$. Therefore, we have
$$
I_{11} \leq C_{\theta} U(T,m,y_T^2/(T(\log T)^2)) + C_{\theta} U(T,\lfloor m^{\theta} \rfloor,y_T^2/(T(\log T)^2)).
$$

Combining the above analysis, we have shown that $U(T,m,z_T)$ is bounded from above by (up to a constant)
\begin{eqnarray}\label{EqB.1}
&&U(T,\lfloor m^{\theta} \rfloor,z_T - 2y_T) + U(T,m,y_T^2/(T(\log T)^2))  \nonumber \\
&&+U(T,\lfloor m^{\theta} \rfloor,y_T^2/(T(\log T)^2)) +  y_T^{-M} + y_T^{-q/2}  T m^{q/2-1-(\alpha-1)\theta q/2} \nonumber \\
&& + y_T^{-q} (\log y_T)^{q} T^{q/2+1}m^{q/2-1-(\alpha-1)\theta q}.
\end{eqnarray}

Since $\sup_{a_{j,s,t}}\|L_{T,j,j'} - E(L_{T,j,j'}) - \widetilde{L}_{T,j,j'}\|_{q/2} = O(T)$ by Lemma \ref{LemmaB6} (3), by applying \eqref{EqB.1} recursively to deal with the term $U(T,m,y_T^2/(T(\log T)^2))$ for $p$ times such that $(y_t/T)^{-2^p q} = O(y_T^{-(M+1)})$, we have (up to a constant)
\begin{eqnarray}\label{EqB.2}
U(T,m,z_T) &\leq&U(T,\lfloor m^{\theta} \rfloor,z_T - 2y_T) + y_T^{-M} + y_T^{-q/2}  T m^{q/2-1-(\alpha-1)\theta q/2} \nonumber \\
&& + y_T^{-q} (\log y_T)^{q} T^{q/2+1}m^{q/2-1-(\alpha-1)\theta q}.
\end{eqnarray}
Using similar arguments, we can show for $1\leq m\leq 3$
$$
U(T,m,z_T/(2l_T)) \leq O \left(z_T^{-M} + z_T^{-q/2}  T \log T
+ z_T^{-q} (\log z_T)^{q+1} T^{q/2+1}\right).
$$

Using \eqref{EqB.2} for $l_T-1$ times, we have
\begin{eqnarray}\label{EqB.3}
U(T, m , z_T) &\leq& O(1) (\log z_T)^{q+1} \left(z_T^{-q/2}  T
+ z_T^{-q} T^{q/2+1}\right) (m^{q/2-1-(\alpha-1) \theta q/2}+1).
\end{eqnarray}
Plugging \eqref{EqB.3} back into \eqref{EqB.1} for the terms $U(T,m,y_T^2/(T(\log T)^2))$ and
$U(T,\lfloor m^{\theta} \rfloor,y_T^2/(T(\log T)^2))$ and using $T^{1+\delta/2}/y_T\to 0$, we have
\begin{eqnarray}\label{EqB.4}
U(T, m , z_T) &\leq& U(T,\lfloor m^{\theta} \rfloor,z_T - 2y_T) +O(1) z_T^{-q} T (m^{q/2-1-(\alpha-1) \theta q/2}+1).
\end{eqnarray}
Finally, by using \eqref{EqB.4} for $l_T-1$ times, we obtain
$$
U(T, m , z_T) \leq O(1)z_T^{-q} T \log T (m^{q/2-1-(\alpha-1) \theta q/2}+1).
$$
The proof is now complete.

\smallskip

\noindent (2). For $l\geq1$, define $m_{T,l} = \lfloor T^{\theta^l}\rfloor$, $\widetilde{\mathbf{u}}_{it,l} = E[\mathbf{u}_{it} \mid \mathcal{F}_{t-m_{T,l},t}]$, $\widetilde{\overline{\mathbf{u}}}_{t,l} = \frac{1}{\sqrt{N}}\sum_{i=1}^{N}\widetilde{\mathbf{u}}_{it,l}$ and $Q_{NT,l} = \sum_{1\leq s < t \leq T}a_{t-s} \widetilde{\overline{u}}_{j,t,l}\widetilde{\overline{u}}_{j',s,l}$. Let $l_T = \lceil - \log(\log T)/\log \theta \rceil$ and thus $m_{T,l_T} \leq e$. By using Lemmas \ref{LemmaB6} (2) and (3), we have
$$
\Pr(|L_{NT,j,j'} - E(L_{NT,j,j'}) - Q_{NT,1} + E(Q_{NT,1})|\geq x_T/l_T) \leq C_{q,\theta} (\log T)^{1/2} x_{T}^{-q/2}(T\ell)^{q/4}T^{- (\alpha-1)\theta q/2}.
$$
Let $l_T'$ be the smallest $l$ such that $m_{T,l}\leq \ell/4$. For $1\leq l \leq l_T'$, split $[1,T]$ into blocks $B_{1}^l,\ldots,B_{b_{T,l}}^l$ with size $\ell + m_{T,l}$. Define $W_{T,l,b} = \sum_{t\in B_b^l}\sum_{1\leq s \leq t}a_{t-s} \widetilde{\overline{u}}_{j,t,l}\widetilde{\overline{u}}_{j',s,l}$ and $$W_{T,l,b}' = \sum_{t\in B_b^l}\sum_{1\leq s \leq t}a_{t-s} \widetilde{\overline{u}}_{j,t,l+1}\widetilde{\overline{u}}_{j',s,l+1}.$$ Then, by using Lemma \ref{LemmaB2} (1) and using  Lemma \ref{LemmaB6} (3), we have for any $C>2$
\begin{eqnarray}\label{EqB.5}
&&\Pr(|Q_{NT,l} - E(Q_{NT,l}) - (Q_{NT,l+1} - E(Q_{NT,l+1}))|\geq x_T/(2l_T)) \nonumber \\
&\leq&\sum_{b=1}^{b_{T,l}} \Pr(|W_{T,l,b} - E(W_{T,l,b}) - (W_{T,l,b}' - E(W_{T,l,b}'))|\geq x_T/(Cl_T)) \nonumber \\
&& + O(1) \left[\frac{T\ell l_T^2}{x_T^2}\right]^{C/4}.
\end{eqnarray}
Note that for any $M>1$, there exists a constant $C_{M,\delta,\theta}$ such that the second term in the right hand of \eqref{EqB.5} is less than $C_{M,\delta,\theta} x_T^{-M}$. For the first term, by using Lemmas \ref{LemmaB6} (2) and (3), we have
\begin{eqnarray}\label{EqB.6}
&& \sum_{b=1}^{b_{T,l}} \Pr(|W_{T,l,b} - E(W_{T,l,b}) - (W_{T,l,b}' - E(W_{T,l,b'}))|\geq x_T/(Cl_T)) \nonumber\\
&\leq& O(1) Tm_{T,l}^{-1}\sqrt{\log T} x_T^{-q/2}(m_{T,l}\ell)^{q/4} m_{T,l+1}^{-(\alpha-1) q/2} \nonumber \\
&\leq& O(1)x_T^{-q/2} \sqrt{\log T} T\ell^{q/4} m_{T,l}^{q/4-1-(\alpha-1) \theta q/2} \nonumber \\
&\leq& O(1)x_T^{-q/2} \sqrt{\log T} \left((T\ell)^{q/4}T^{-(\alpha-1)\theta q/2} + T\ell^{q/2-1-(\alpha-1)\theta q/2}\right).
\end{eqnarray}
Combining \eqref{EqB.5} and \eqref{EqB.6}, we have
\begin{eqnarray*}
&&\Pr(|L_{NT,j,j'} - E(L_{NT,j,j'})|\geq x_T)\nonumber \\
&\leq & \Pr(|Q_{NT,l_T'} - E(Q_{NT,l_T'})|\geq x_T/2) + O(1)x_T^{-M}\\
&&+ O(1)x_T^{-q/2} \sqrt{\log T} \left((T\ell)^{q/4}T^{-(\alpha-1)\theta q/2} + T\ell^{q/2-1-(\alpha-1)\theta q/2}\right).
\end{eqnarray*}

For $\Pr(|Q_{NT,l_T'} - E(Q_{NT,l_T'})|\geq x_T/2)$, split $[T]$ into blocks $B_{1},\ldots,B_{b_{T}}$ with block size $2\ell$ and define
$W_{T,l_T',b} = \sum_{t\in B_b}\sum_{1\leq s \leq t}a_{t-s} \widetilde{\overline{u}}_{j,t,l_T'}\widetilde{\overline{u}}_{j',s,l_T'}$. Similarly, we have
$$
\Pr(|Q_{NT,l_T'} - E(Q_{NT,l_T'})|\geq x_T/2)\leq\sum_{b=1}^{b_{T}} \Pr(|W_{T,l_T',b} - E(W_{T,l_T',b})|\geq x_T/C)+Cx_T^{-M}.
$$
By using part (1) of this lemma, we have
$$
\Pr(|W_{T,l_T',b} - E(W_{T,l_T',b})|\geq x_T/C) \leq O(1)x_T^{-q/2}(\log T) (\ell^{q/2-(\alpha-1)\theta q/2}+\ell),
$$
which follows that
$$
\Pr(|Q_{NT,l_T'} - E(Q_{NT,l_T'})|\geq x_T/2)\leq O(1)x_T^{-q/2}(T\log T) (\ell^{q/2-1-(\alpha-1)\theta q/2}+1).
$$

The proof is now complete.
\end{proof}

\begin{proof}[Proof of Lemma \ref{LemmaB8}]
\item 
\noindent (1). By using Lemma \ref{LEMMA1} (2) and letting $x_T = M\sqrt{T\ell \log d}$ for some sufficient large but fixed $M$ and $\theta = 1/(\alpha-1)$, we have
\begin{eqnarray*}
&&\Pr\left(\max_{1\leq k,l\leq d}\left|\sum_{t,s=1}^{T}a((t-s)/\ell)\left(\overline{u}_{k,t}\overline{u}_{l,s} - E(\overline{u}_{k,t}\overline{u}_{l,s})\right)\right| \geq x_T\right)\\
&\leq&\sum_{k,l=1}^{d}\Pr\left(\left|\sum_{t,s=1}^{T}a((t-s)/\ell)\left(\overline{u}_{k,t}\overline{u}_{l,s} - E(\overline{u}_{k,t}\overline{u}_{l,s})\right)\right| \geq x_T\right)\\
&\leq&O(1)d^2 x_T^{-q/2}\log T(T^{-q/4}\ell^{q/4} + T\ell^{-1} +T ) + O(1)d^{2-M} + O(1)d^{2}T^{-M}\to 0 
\end{eqnarray*}
if $d^2 T \log T /(T\ell \log d)^{q/4} \to 0$. Hence, we have proved
$$
\max_{1\leq k,l\leq d}\left|T^{-1}\sum_{t,s=1}^{T}a((t-s)/\ell)\left(\overline{u}_{k,t}\overline{u}_{l,s} - E(\overline{u}_{k,t}\overline{u}_{l,s})\right)\right| = O_P(\sqrt{\ell \log d/T}).
$$

\noindent (2). By using Cauchy-Schwarz inequality, we have
\begin{eqnarray*}
&&\max_{1\leq k,l,l'\leq d} \left| \frac{1}{T}\sum_{t=1}^{T}\frac{1}{N^2}\sum_{i,j=1}^{N}x_{k,it}x_{k,jt}x_{l,it}x_{l',jt}\right|\\
&\leq&\max_{1\leq k,l\leq d} \left| \frac{1}{T}\sum_{t=1}^{T}\left(\frac{1}{N}\sum_{i=1}^{N}x_{k,it}x_{l,it}\right)^2\right| \leq \max_{1\leq k \leq d} \left| \frac{1}{T}\sum_{t=1}^{T}\left\{\frac{1}{N}\sum_{i=1}^{N}x_{k,it}^2\right\}^2\right| = O_P(1).
\end{eqnarray*}

Hence, by using Cauchy-Schwarz inequality, if $s^2\sqrt{\ell \log d/T}< \infty$, we have
\begin{eqnarray*}
&&\max_{1\leq k,l\leq d}\left|\frac{1}{T}\sum_{t,s=1}^{T}a((t-s)/\ell)\left(\widehat{\overline{u}}_{k,t}-\overline{u}_{k,t}\right)\left(\widehat{\overline{u}}_{l,s} - \overline{u}_{l,s}\right)\right| \\
&\leq& (2\ell+1) \max_{1\leq k\leq d}\left|\frac{1}{T}\sum_{t=1}^{T}\left(\widehat{\overline{u}}_{k,t}-\overline{u}_{k,t}\right)^2\right|\\
&=&N(2\ell+1) \max_{1\leq k\leq d} \left| \sum_{1\leq l,l'\leq d}\frac{1}{T}\sum_{t=1}^{T}\frac{1}{N^2}\sum_{i,j=1}^{N}x_{k,it}x_{k,jt}x_{l,it}x_{l',jt}(\widehat{\beta}_{l}-\beta_{l})(\widehat{\beta}_{l'}-\beta_{l'})\right| \\
&\leq&N(2\ell+1)|\widehat{\pmb{\beta}} - \pmb{\beta}|_1^2\max_{1\leq k\leq d}\max_{l,l'\in J} \left| \frac{1}{T}\sum_{t=1}^{T}\frac{1}{N^2}\sum_{i,j=1}^{N}x_{k,it}x_{k,jt}x_{l,it}x_{l',jt}\right|\\
&=& O_P(\ell s^2 \log d/T) = O_P(\sqrt{\ell\log d/T}).
\end{eqnarray*}
The proof of part (2) is now complete.

\smallskip

\noindent (3). By using similar arguments of Lemma \ref{LEMMA1} (2), if $\frac{d^3T\log T}{(T\ell \log d)^{q/4} }\to 0$ we have
\begin{eqnarray*}
&&\max_{1\leq k,l,l'\leq d}\left|\frac{1}{T}\sum_{t,s=1}^{T}a((t-s)/\ell)\left(N^{-1/2}\sum_{i=1}^{N}x_{k,it}x_{l',it}\overline{u}_{l,s} - E((N^{-1/2}\sum_{i=1}^{N}x_{k,it}x_{l',it}\overline{u}_{l,s})\right)\right|\\
& = &O_P(\sqrt{\ell\log d/T}).
\end{eqnarray*}
Since $E(N^{-1/2}\sum_{i=1}^{N}x_{k,it}x_{l,it}\overline{u}_{l,s}) = 0$ (implied by $E(e_{it}\mid \mathbf{X})=0$), then we have
\begin{eqnarray*}
&&\max_{1\leq k,l\leq d}\left|\frac{1}{T}\sum_{t,s=1}^{T}a((t-s)/\ell)\left(\widehat{\overline{u}}_{k,t}-\overline{u}_{k,t}\right) \overline{u}_{l,s}\right|\\
&\leq& \max_{1\leq k,l,l'\leq d}\left|\frac{1}{T}\sum_{t,s=1}^{T}a((t-s)/\ell)N^{-1/2}\sum_{i=1}^{N}x_{k,it}x_{l',it}\overline{u}_{l,s}\right| \times |\widehat{\pmb{\beta}} - \pmb{\beta}|_1\\
&=&O_P(\sqrt{\ell\log d/T})\times O_P(s\sqrt{\log d/(NT)}) = o_P(\sqrt{\ell\log d/T}).
\end{eqnarray*}
The proof is now complete.
\end{proof}

\section{Results of Section \ref{Sec3}}\label{AP.B2}
\begin{lemma}\label{L7}
Let \eqref{EQUATION4} and Assumptions \ref{ASSUMPTION1}, \ref{ASSUMPTION2} and \ref{ASSUMPTION5} hold, and $s=o ( \max(T^{-1/2},N^{-1/2})NT/\log d )$. Then we have
\begin{itemize}[leftmargin=*, itemsep=0.5pt, parsep=0.5pt, topsep=0.6pt]
\item [1.] $ |\frac{1}{NT}\sum_{t=1}^{T}\mathbf{X}_{J,t}^\top\mathbf{M}_{\widehat{\pmb{\Lambda}}^{(l)}}\mathbf{e}_t |_{\infty} = O_P\left(\max(\sqrt{s}\omega_1, \omega_2)\right)$;

\item [2.] $ |\frac{1}{NT}\sum_{t=1}^{T}\mathbf{X}_{J,t}^\top\mathbf{M}_{\widehat{\pmb{\Lambda}}^{(l)}}\pmb{\Lambda}_0\mathbf{f}_{0t} |_{\infty} = O_P\left(\max(\sqrt{s}\omega_1, \omega_2)\right)$;

\item [3.] $|\widetilde{\pmb{\Sigma}}^{(l)}-\pmb{\Sigma}|_{\max} = O_P(\max(\sqrt{s}\omega_1, \omega_2)$, \\
where $\widetilde{\pmb{\Sigma}}^{(l)} = \frac{1}{NT}\sum_{t=1}^{T}\mathbf{X}_t^\top\mathbf{M}_{\widehat{\pmb{\Lambda}}^{(l)}}\mathbf{X}_t$ and $\pmb{\Sigma} = \plim \frac{1}{NT}\sum_{t=1}^{T}\mathbf{X}_t^\top\mathbf{M}_{\pmb{\Lambda}_0}\mathbf{X}_t$.
\end{itemize}
\end{lemma}

\begin{proposition}\label{Thm5}
Let \eqref{EQUATION4} and Assumptions \ref{ASSUMPTION1}, \ref{ASSUMPTION2} and \ref{ASSUMPTION5} hold with $s=o ( \max(T^{-1/2},N^{-1/2})NT/\log d )$. Then the following results hold.
\begin{enumerate}[leftmargin=*]
\item $|\widetilde{\pmb{\beta}}-\pmb{\beta}_0|_{2} = O_P\left(\max(\sqrt{s}\omega_1,\omega_2)\right)$ and  $\frac{1}{\sqrt{NT}}|\widetilde{\pmb{\Xi}}-\pmb{\Xi}_0|_{F} = O_P\left(\max(\sqrt{s}\omega_1,\omega_2)\right)$.

\item $|\widetilde{\pmb{\Lambda}}-\pmb{\Lambda}_0\widetilde{\mathbf{H}}|_F/\sqrt{N} = O_P\left(\max(\sqrt{s}\omega_1,\omega_2)\right)$ for some rotation matrix $\widetilde{\mathbf{H}}$ depending on $(\mathbf{F}_0,\pmb{\Lambda}_0)$.
\end{enumerate}
\end{proposition}

\begin{proof}[Proof of Lemma \ref{L7}]
\item 
\noindent (1). Write
$$
\left|\frac{1}{NT}\sum_{t=1}^{T}\mathbf{X}_{J,t}^\top\mathbf{M}_{\widehat{\pmb{\Lambda}}^{(l)}}\mathbf{e}_t\right|_{\infty}\leq \left|\frac{1}{NT}\sum_{t=1}^{T}\mathbf{X}_{J,t}^\top\mathbf{e}_t\right|_{\infty} + \left|\frac{1}{N^2T}\sum_{t=1}^{T}\mathbf{X}_{J,t}^\top\widehat{\pmb{\Lambda}}^{(l)}\widehat{\pmb{\Lambda}}^{(l),\top}\mathbf{e}_t\right|_{\infty}.
$$
By using Lemma \ref{LemmaB4} (1), the first term is of order $O_P(\sqrt{\log d/(NT)})$. For the second term, by using Cauchy-Schwarz inequality, we have 
\begin{eqnarray*}
\left|\frac{1}{N^2T}\sum_{t=1}^{T}\mathbf{X}_{J,t}^\top\widehat{\pmb{\Lambda}}^{(l)}\widehat{\pmb{\Lambda}}^{(l),\top}\mathbf{e}_t\right|_{\infty} &\leq& \left\{\max_{j \in J}\frac{1}{T}\sum_{t=1}^{T}\left|\mathbf{X}_{j,t}^\top\widehat{\pmb{\Lambda}}^{(l)}/N \right|_F^2\right\}^{1/2} \left\{\frac{1}{T}\sum_{t=1}^{T}\left|\frac{1}{N}\sum_{i=1}^{N}\widehat{\pmb{\lambda}}_i^{(l)}e_{it} \right|_F^2\right\}^{1/2}.
\end{eqnarray*}
Note that $\left|\mathbf{X}_{j,t}^\top\widehat{\pmb{\Lambda}}^{(l)}/N\right|_F^2 \leq \widehat{r}N^{-1}\sum_{i=1}^{N}x_{j,it}^2$ and $$\max_{j \in J}\left|\frac{1}{NT}\sum_{t=1}^{T}\sum_{i=1}^{N}\left(x_{j,it}^2-E(x_{j,it}^2)\right)\right|=O_P(\sqrt{\log d/(NT)}),$$
we then have
\begin{eqnarray*}
\max_{j \in J} \frac{1}{T}\sum_{t=1}^{T}\left|\mathbf{X}_{j,t}^\top\widehat{\pmb{\Lambda}}^{(l)}/N \right|_F^2 \leq \max_{j \in J}\frac{r}{NT}\sum_{t=1}^{T}\sum_{i=1}^{N}x_{j,it}^2 =O_P(1).
\end{eqnarray*}
In addition, by using Cauchy-Schwarz inequality, we have 
\begin{eqnarray*}
\frac{1}{T}\sum_{t=1}^{T}\left|\frac{1}{N}\sum_{i=1}^{N}\widehat{\pmb{\lambda}}_i^{(l)}e_{it} \right|_F^2
&\leq& \frac{2}{T}|\widetilde{\mathbf{H}}|_F^2\sum_{t=1}^{T}\left|\frac{1}{N}\sum_{i=1}^{N}\pmb{\lambda}_{0i}e_{it} \right|_F^2 + \frac{2}{T}\sum_{t=1}^{T}\left|\frac{1}{N}\sum_{i=1}^{N}(\pmb{\lambda}_{0i}-\widetilde{\mathbf{H}}^\top\widehat{\pmb{\lambda}}_i^{(l)})e_{it} \right|_F^2\\
&\leq&O_P(1/N) + \frac{2}{T}\sum_{t=1}^{T}\frac{1}{N}\sum_{i=1}^{N}\left|\pmb{\lambda}_{0i}-\widetilde{\mathbf{H}}^\top\widehat{\pmb{\lambda}}_i^{(l)}\right|_F^2\frac{1}{N}\sum_{i=1}^{N}e_{it}^2 \\
&=& O_P(1/N + \max(s\omega_1^2, \omega_2^2)).
\end{eqnarray*}
Combining the above analyses, we have proved part (1).

\smallskip

\noindent (2). Note that $|\mathbf{M}_{\widehat{\pmb{\Lambda}}^{(l)}}|_2 = 1$, then
\begin{eqnarray*}
\left|\frac{1}{NT}\sum_{t=1}^{T}\mathbf{X}_{J,t}^\top\mathbf{M}_{\widehat{\pmb{\Lambda}}^{(l)}}\pmb{\Lambda}_{0}\mathbf{f}_{0t}\right|_{\infty}
&=& \max_{j\in J}\left|\frac{1}{NT}\sum_{t=1}^{T} \mathbf{X}_{j,t}^\top\mathbf{M}_{\widehat{\pmb{\Lambda}}^{(l)}}(\pmb{\Lambda}_0\widetilde{\mathbf{H}}-\widehat{\pmb{\Lambda}}^{(l)}) \widetilde{\mathbf{H}}^{-1}\mathbf{f}_t\right|\\
&\leq& \max_{j\in J}\frac{1}{T}\sum_{t=1}^{T} |N^{-1/2}\mathbf{X}_{j,t}|_F|\mathbf{f}_t|_F|\mathbf{M}_{\widehat{\pmb{\Lambda}}^{(l)}}|_2N^{-1/2}|\pmb{\Lambda}_{0}\widetilde{\mathbf{H}}-\widehat{\pmb{\Lambda}}^{(l)}|_F |\widetilde{\mathbf{H}}^{-1}|_F \\
&=& O_P(\max(\sqrt{s}\omega_1, \omega_2)).
\end{eqnarray*}

\smallskip

\noindent (3). Note that $\left|\mathbf{P}_{\widetilde{\pmb{\Lambda}}} - \mathbf{P}_{\pmb{\Lambda}_0}\right|_F = O_P(\max(\sqrt{s}\omega_1,\omega_2))$ by the proof of Proposition \ref{Thm5}, and by using Cauchy-Schwarz inequality, we have
\begin{eqnarray*}
&&\left|\frac{1}{NT}\sum_{t=1}^{T}\mathbf{X}_t^\top\mathbf{M}_{\widehat{\pmb{\Lambda}}^{(l)}}\mathbf{X}_t-\frac{1}{NT}\sum_{t=1}^{T}\mathbf{X}_t^\top\mathbf{M}_{\pmb{\Lambda}_0}\mathbf{X}_t\right|_{\max} \\
&\leq& \left|\mathbf{P}_{\widetilde{\pmb{\Lambda}}} - \mathbf{P}_{\pmb{\Lambda}_0}\right|_F\max_{1\leq j\leq d} \frac{1}{NT}\sum_{i=1}^{N}\sum_{t=1}^{T}\mathbf{x}_{j,it}^2 = O_P(\max(\sqrt{s}\omega_1,\omega_2)).
\end{eqnarray*}
In addition by using Lemma \ref{LemmaB4}, we can show 
$$
\left|\frac{1}{NT}\sum_{t=1}^{T}\mathbf{X}_t^\top\mathbf{M}_{\pmb{\Lambda}_0}\mathbf{X}_t - \pmb{\Sigma} \right|_{\max} = O_P(\sqrt{\log d/(NT)}).
$$
The proof is now complete.
\end{proof}

\begin{proof}[Proof of Proposition \ref{Thm5}]
\item

\noindent (1). We first evaluate the difference of the objective function defined in \eqref{EQUATION10} at $(\widetilde{\pmb{\beta}}, \widetilde{\pmb{\Xi}})$ and $( \pmb{\beta}_0,  \pmb{\Xi}_0)$. Thus, by  the definition of $(\widetilde{\pmb{\beta}}, \widetilde{\pmb{\Xi}})$, we have
\begin{eqnarray*}
&&\frac{1}{2NT}\left(|\mathbf{y}-\mathbf{X}\widetilde{\pmb{\beta}}-\mathrm{vec}(\widetilde{\pmb{\Xi}})|_2^2 - |\mathbf{e}|_2^2\right) + \omega_1(|\widetilde{\pmb{\beta}}|_1-|\pmb{\beta}_0|_1) + \frac{\omega_2}{\sqrt{NT}}\left(|\widetilde{\pmb{\Xi}}|_* -|\pmb{\Xi}_0|_* \right)\\
&\eqqcolon & J_7 + J_8 + J_9 \leq 0,
\end{eqnarray*}
where the definitions of $J_7$, $J_8$ and $J_9$ are self-evident.

Consider $J_7$. Note that the event $\mathcal{A}_{NT} = \left\{\left|\frac{1}{NT}\mathbf{e}^\top \mathbf{X}\right|_\infty \leq \omega_1/2\right\}$ holds with probability larger than $C_{NT}$.
Conditional on $\mathcal{A}_{NT}$, using Assumption \ref{ASSUMPTION5}.1, $(\mathrm{vec}(\mathbf{A}))^\top\mathrm{vec}(\mathbf{B}) = \mathrm{tr}(\mathbf{A}^\top\mathbf{B})$ and $\mathrm{tr}(\mathbf{A}^\top\mathbf{B})\leq|\mathbf{A}|_2|\mathbf{B}|_*$, we have
\begin{eqnarray*}
J_7&\geq&\frac{1}{2NT} |\mathbf{X}(\widetilde{\pmb{\beta}}-\pmb{\beta}_0) + \mathrm{vec}(\widetilde{\pmb{\Xi}}-\pmb{\Xi}_0)|_2^2 -  |\frac{1}{NT}\mathbf{X}^\top\mathbf{e}|_{\infty}|\widetilde{\pmb{\beta}} - \pmb{\beta}_0|_1- \frac{1}{NT}|\mathbf{E}|_{2} |\widetilde{\pmb{\Xi}} - \pmb{\Xi}_0|_{*}\\
&\geq& \frac{1}{2NT} |\mathbf{X}(\widetilde{\pmb{\beta}}-\pmb{\beta}_0) + \mathrm{vec}(\widetilde{\pmb{\Xi}}-\pmb{\Xi}_0)|_2^2 - \frac{\omega_1}{2}|\widetilde{\pmb{\beta}} - \pmb{\beta}_0|_1-\frac{\omega_2}{2\sqrt{NT}}|\widetilde{\pmb{\Xi}} - \pmb{\Xi}_0|_*.
\end{eqnarray*}

For $J_8$, we have $J_8 = \omega_1 (|\widetilde{\pmb{\beta}}_{J}|_1+|\widetilde{\pmb{\beta}}_{J^c}|_1 -|\pmb{\beta}_{0,J}|_1) \geq \omega_1(|\widetilde{\pmb{\beta}}_{J^c}|_1 - |\widetilde{\pmb{\beta}}_J - \pmb{\beta}_{0,J}|_1) $.

Consider $J_9$. Let $\mathbb{P}(\pmb{\Xi}) = \mathbf{M}_{\mathbf{U}_{0,[r]}}\pmb{\Xi}\mathbf{M}_{\mathbf{V}_{0,[r]}}$ and $ \mathbb{M}(\pmb{\Xi}) = \pmb{\Xi} -  \mathbf{M}_{\mathbf{U}_{0,[r]}}\pmb{\Xi}\mathbf{M}_{\mathbf{V}_{0,[r]}}$. Note that $|\mathbf{A}+\mathbf{B}|_* = |\mathbf{A}|_* + |\mathbf{B}|_*$ if $\mathbf{A}^\top\mathbf{B} = 0$ and $\mathbf{A}\mathbf{B}^\top = 0$, and thus we have
\begin{eqnarray*}
|\widetilde{\pmb{\Xi}}|_* &=& |\widetilde{\pmb{\Xi}} - \pmb{\Xi}_0 + \pmb{\Xi}_0|_* = |\pmb{\Xi}_0 + \mathbb{P}(\widetilde{\pmb{\Xi}} - \pmb{\Xi}_0) + \mathbb{M}(\widetilde{\pmb{\Xi}} - \pmb{\Xi}_0)|_* \\
&\geq& |\pmb{\Xi}_0 + \mathbb{P}(\widetilde{\pmb{\Xi}} - \pmb{\Xi}_0)|_* - |\mathbb{M}(\widetilde{\pmb{\Xi}} - \pmb{\Xi}_0)|_* =|\pmb{\Xi}_0|_* + |\mathbb{P}(\widetilde{\pmb{\Xi}} - \pmb{\Xi}_0)|_* - |\mathbb{M}(\widetilde{\pmb{\Xi}} - \pmb{\Xi}_0)|_*,
\end{eqnarray*}
which follows that 
$$J_9 \geq  \frac{\omega_2}{\sqrt{NT}}|\mathbb{P}(\widetilde{\pmb{\Xi}} - \pmb{\Xi}_0)|_* - \frac{\omega_2}{\sqrt{NT}}|\mathbb{M}(\widetilde{\pmb{\Xi}} - \pmb{\Xi}_0)|_*.
$$

Combing the above analyses, we have 
\begin{eqnarray*}
&&\frac{1}{NT} |\mathbf{X}(\widetilde{\pmb{\beta}}-\pmb{\beta}_0) + \mathrm{vec}(\widetilde{\pmb{\Xi}}-\pmb{\Xi}_0)|_2^2 + \omega_1|\widetilde{\pmb{\beta}}_{J^c}|_1 + \frac{\omega_2}{\sqrt{NT}}|\mathbb{P}(\widetilde{\pmb{\Xi}} - \pmb{\Xi}_0)|_* \\
&\leq& 3\omega_1|\widetilde{\pmb{\beta}}_J - \pmb{\beta}_{0,J}|_1 + 3\frac{\omega_2}{\sqrt{NT}}|\mathbb{M}(\widetilde{\pmb{\Xi}} - \pmb{\Xi}_0)|_*.
\end{eqnarray*}
Hence, $(\widetilde{\pmb{\beta}}-\pmb{\beta}_0,\widetilde{\pmb{\Xi}} - \pmb{\Xi}_0)\in \mathbb{C}$ and thus by Assumption \ref{ASSUMPTION5}.2, we have
\begin{eqnarray*}
\kappa_c|\widetilde{\pmb{\beta}}-\pmb{\beta}_0|_2^2 + \kappa_c\frac{1}{NT}|\mathrm{vec}(\widetilde{\pmb{\Xi}}-\pmb{\Xi}_0)|_2^2 \leq \frac{1}{NT} |\mathbf{X}(\widetilde{\pmb{\beta}}-\pmb{\beta}_0) + \mathrm{vec}(\widetilde{\pmb{\Xi}}-\pmb{\Xi}_0)|_2^2.
\end{eqnarray*}
In addition, since $|\pmb{\Xi}|_F^2 = |\mathbb{P}(\pmb{\Xi})|_F^2+|\mathbb{M}(\pmb{\Xi})|_F^2 + 2\mathrm{tr}(\mathbb{P}(\pmb{\Xi})^\top\mathbb{M}(\pmb{\Xi}))=|\mathbb{P}(\pmb{\Xi})|_F^2+|\mathbb{M}(\pmb{\Xi})|_F^2$, $|\pmb{\Xi}|_*^2\leq |\pmb{\Xi}|_F^2\mathrm{rank}(\pmb{\Xi})$ and $\mathrm{rank}(\mathbb{M}(\widetilde{\pmb{\Xi}} - \pmb{\Xi}_0))\leq 2 r$ (which we will prove in the following), we have
\begin{eqnarray*}
&&\kappa_c|\widetilde{\pmb{\beta}}-\pmb{\beta}_0|_2^2 + \kappa_c\frac{1}{NT}|\mathrm{vec}(\widetilde{\pmb{\Xi}}-\pmb{\Xi}_0)|_2^2\\
&\leq& 3\omega_1|\widetilde{\pmb{\beta}}_J - \pmb{\beta}_{0,J}|_1 + 3\frac{\omega_2}{\sqrt{NT}}|\mathbb{M}(\widetilde{\pmb{\Xi}} - \pmb{\Xi}_0)|_*\\
&\leq& 3\sqrt{s}\omega_1 |\widetilde{\pmb{\beta}}-\pmb{\beta}_0|_2 + 3\sqrt{2r}\frac{\omega_2}{\sqrt{NT}}|\widetilde{\pmb{\Xi}}-\pmb{\Xi}_0|_F\\
&\leq&\max\left\{6\sqrt{s}\omega_1, 6\sqrt{2r}\omega_2\right\}\sqrt{|\widetilde{\pmb{\beta}}-\pmb{\beta}_0|_2^2 + \frac{1}{NT}|\widetilde{\pmb{\Xi}}-\pmb{\Xi}_0|_F^2}.
\end{eqnarray*}

We now prove $\mathrm{rank}(\mathbb{M}(\widetilde{\pmb{\Xi}} - \pmb{\Xi}_0))\leq 2 r$. Define the matrix $\pmb{\Delta} = \mathbf{U}_0^\top (\widetilde{\pmb{\Xi}}-\pmb{\Xi}_0) \mathbf{V}_0 = \mathbf{U}_0^\top \mathbb{M}(\widetilde{\pmb{\Xi}}-\pmb{\Xi}_0) \mathbf{V}_0 + \mathbf{U}_0^\top \mathbb{P}(\widetilde{\pmb{\Xi}}-\pmb{\Xi}_0) \mathbf{V}_0$, and write it in block form as
$$
\pmb{\Delta} = \left[\begin{matrix}
\pmb{\Delta}_{11} & \pmb{\Delta}_{12}\\
\pmb{\Delta}_{21} & \pmb{\Delta}_{22}
\end{matrix} \right],
$$
where $\pmb{\Delta}_{11} \in \mathbb{R}^{r\times r}$ and $\pmb{\Delta}_{11} \in \mathbb{R}^{(T-r)\times (N-r)}$. Since $\mathbf{U}_0^\top\mathbf{M}_{\mathbf{U}_{0,[r]}} = [\mathbf{0}_{T\times r}, \mathbf{U}_{0,[T-r]}]^\top$ and $\mathbf{V}_0^\top\mathbf{M}_{\mathbf{V}_{0,[r]}} = [\mathbf{0}_{N\times r}, \mathbf{V}_{0,[N-r]}]^\top$, we have
$$
\mathbf{U}_0^\top \mathbb{M}(\widetilde{\pmb{\Xi}}-\pmb{\Xi}_0) \mathbf{V}_0 = \left[\begin{matrix}
\pmb{\Delta}_{11} & \pmb{\Delta}_{12}\\
\pmb{\Delta}_{21} & \mathbf{0}
\end{matrix} \right],
$$
which implies that
\begin{eqnarray*}
&&\mathrm{rank}(\mathbb{M}(\widetilde{\pmb{\Xi}}-\pmb{\Xi}_0)) = \mathrm{rank}(\mathbf{U}_0^\top \mathbb{M}(\widetilde{\pmb{\Xi}}-\pmb{\Xi}_0) \mathbf{V}_0)\\ &\leq& \mathrm{rank}\left(\left[\begin{matrix}
\mathbf{0} & \pmb{\Delta}_{12}\\
\mathbf{0} & \mathbf{0}
\end{matrix} \right]\right) + \mathrm{rank}\left(\left[\begin{matrix}
\pmb{\Delta}_{11} & \mathbf{0}\\
\pmb{\Delta}_{21} & \mathbf{0}
\end{matrix} \right]\right) \leq 2r.
\end{eqnarray*}

Hence, we have proved $$
|\widetilde{\pmb{\beta}}-\pmb{\beta}_0|_{2} \leq \max\left\{\frac{6\sqrt{s}\omega_1}{\kappa_c}, \frac{6\sqrt{2r}\omega_2}{\kappa_c}\right\}$$
and
$$\frac{1}{\sqrt{NT}}|\widetilde{\pmb{\Xi}}-\pmb{\Xi}_0|_{F} \leq \max\left\{\frac{6\sqrt{s}\omega_1}{\kappa_c}, \frac{6\sqrt{2r}\omega_2}{\kappa_c}\right\}
$$ 
with probability larger than $1 - C_1\left(\frac{dT^{1-q/2}}{(\log d)^{q/2}} +  d^{-C_2}\right)$.

\smallskip

\noindent (2). Let $\widehat{\pmb{\Sigma}}_{\lambda} = \pmb{\Lambda}_0^\top\pmb{\Lambda}_0/N$, $\widehat{\pmb{\Sigma}}_{f} = \mathbf{F}_0^\top\mathbf{F}_0/T$ and $\widehat{n}_1\geq\cdots\geq\widehat{n}_r$ be the $r$ nonzero eigenvalues of $\pmb{\Xi}_0^\top \pmb{\Xi}_0/(NT) = \pmb{\Lambda}_0\widehat{\pmb{\Sigma}}_{f}\pmb{\Lambda}_0^\top/N$. Note that the nonzero eigenvalues of $\pmb{\Lambda}_0\widehat{\pmb{\Sigma}}_{f}\pmb{\Lambda}_0^\top/N$ are the same as those of $\widehat{\pmb{\Sigma}}_{\lambda}^{1/2}\widehat{\pmb{\Sigma}}_{f}\widehat{\pmb{\Sigma}}_{\lambda}^{1/2}$. Then with probability approaching to 1 (w.p.a.1), for some $0<C<\infty$ and $1\leq j \leq r$, by using the Weyl's theorem, we have
$$
|\widehat{n}_j - n_j|\leq |\widehat{\pmb{\Sigma}}_{\lambda}^{1/2}\widehat{\pmb{\Sigma}}_{f}\widehat{\pmb{\Sigma}}_{\lambda}^{1/2} - \pmb{\Sigma}_{\lambda}^{1/2}\pmb{\Sigma}_{f}\pmb{\Sigma}_{\lambda}^{1/2}|_2 \leq C(T^{-1/2}+N^{-1/2}),
$$
which also implies that $|\pmb{\Xi}_0|_2/\sqrt{NT} = \sqrt{n_1 + O_P(T^{-1/2}+N^{-1/2})}$.

Let $\widetilde{n}_1\geq \cdots \geq\widetilde{n}_{N \wedge T}$ be the eigenvalues of $\widetilde{\pmb{\Xi}}^\top\widetilde{\pmb{\Xi}}/(NT)$. Again, by using the Weyl's theorem, part (1) of this lemma and the condition $s=o\left( \max(T^{-1/2},N^{-1/2})NT/\log d \right)$, we have $$|\widetilde{n}_j - n_j|\leq C\cdot \max(\sqrt{s}\omega_1,\omega_2) = o(C\sqrt{\omega_2})$$ w.p.a.1 for all $j\geq 1$ and $|\widetilde{\pmb{\Xi}}|_2/\sqrt{NT} = \sqrt{n_1 + o_P(C\sqrt{\omega_2})}$. In addition, since $\psi_j(\widetilde{\pmb{\Xi}})/\sqrt{NT} = \sqrt{\widetilde{n}_j}$, we have $\psi_j(\widetilde{\pmb{\Xi}})/\sqrt{NT} \geq \sqrt{n_{j}-C \sqrt{\omega_2}}$ w.p.a.1 for $1\leq j \leq r$. The above results imply that $\psi_r(\widetilde{\pmb{\Xi}}) \geq (\omega_2\sqrt{NT}|\widetilde{\pmb{\Xi}}|_2)^{1/2}$ and $\psi_{r+1}(\widetilde{\pmb{\Xi}}) < (\omega_2\sqrt{NT}|\widetilde{\pmb{\Xi}}|_2)^{1/2}$ w.p.a.1.

\smallskip

We next prove $|\widetilde{\pmb{\Lambda}} - \pmb{\Lambda}_0\widetilde{\mathbf{H}}|_F/\sqrt{N} = O_P\left(\max(\sqrt{s}\omega_1,\omega_2)\right)$ for some rotation matrix $\widetilde{\mathbf{H}}$ depending on $(\mathbf{F}_0,\pmb{\Lambda}_0)$. Let $\mathbf{V}$ be the $r\times r$ matrix whose columns are the eigenvectors of $\widehat{\pmb{\Sigma}}_{\lambda}^{1/2}\widehat{\pmb{\Sigma}}_{f}\widehat{\pmb{\Sigma}}_{\lambda}^{1/2}$. Then $\mathbf{D} = \mathbf{V}^\top\widehat{\pmb{\Sigma}}_{\lambda}^{1/2}\widehat{\pmb{\Sigma}}_{f}\widehat{\pmb{\Sigma}}_{\lambda}^{1/2}\mathbf{V}$ is a diagonal matrix of the eigenvalues of $\widehat{\pmb{\Sigma}}_{\lambda}^{1/2}\widehat{\pmb{\Sigma}}_{f}\widehat{\pmb{\Sigma}}_{\lambda}^{1/2}$. Let $\widetilde{\mathbf{H}} = \widehat{\pmb{\Sigma}}_{\lambda}^{-1/2}\mathbf{V}$ and then
\begin{eqnarray*}
\frac{1}{NT}\pmb{\Xi}_0^\top\pmb{\Xi}_0 \pmb{\Lambda}_0\widetilde{\mathbf{H}} &=& \pmb{\Lambda}_0\widehat{\pmb{\Sigma}}_{f} \widehat{\pmb{\Sigma}}_{\lambda}^{1/2}\mathbf{V} = \pmb{\Lambda}_0\widehat{\pmb{\Sigma}}_{\lambda}^{-1/2}\widehat{\pmb{\Sigma}}_{\lambda}^{1/2}\widehat{\pmb{\Sigma}}_{f} \widehat{\pmb{\Sigma}}_{\lambda}^{1/2}\mathbf{V}\\
&=&\pmb{\Lambda}_0\widehat{\pmb{\Sigma}}_{\lambda}^{-1/2} \mathbf{V}\mathbf{D} = \pmb{\Lambda}_0\widetilde{\mathbf{H}}\mathbf{D}.
\end{eqnarray*}
In addition, we have $(\pmb{\Lambda}_0\widetilde{\mathbf{H}})^\top\pmb{\Lambda}_0\widetilde{\mathbf{H}}/N = \mathbf{I}_r$ and thus the columns of $\pmb{\Lambda}_0\widetilde{\mathbf{H}}/\sqrt{N}$ are the eigenvectors of $\pmb{\Xi}^\top\pmb{\Xi}$ with the associated eigenvalues in $\mathbf{D}$. In addition, conditional on the event $\widehat{r} = r$, by using Davis-Kahan sin($\Theta$) theorem, we have
$$
|\widetilde{\pmb{\Lambda}}-\pmb{\Lambda}_0\widetilde{\mathbf{H}}|_F/\sqrt{T} \leq \frac{(NT)^{-1}|\widetilde{\pmb{\Xi}}^\top\widetilde{\pmb{\Xi}}-{\pmb{\Xi}_0}^\top{\pmb{\Xi}_0}|_2}{\min_{j\leq r}\min(|\widehat{n}_{j-1}-\widetilde{n}_j|,|\widetilde{n}_j-\widehat{n}_{j+1}|)} = O_P(\max(\sqrt{s}\omega_1,\omega_2)).
$$
In addition, we have
\begin{eqnarray*}
\left|\mathbf{P}_{\widetilde{\pmb{\Lambda}}} - \mathbf{P}_{\pmb{\Lambda}_0}\right|_F&\leq& \left|\pmb{\Lambda}_0\widetilde{\mathbf{H}}\widetilde{\mathbf{H}}^\top\pmb{\Lambda}_0^\top/N- \pmb{\Lambda}_0(\pmb{\Lambda}_0^\top\pmb{\Lambda}_0)^{-1}\pmb{\Lambda}_0^\top\right|_F+\left|\widetilde{\pmb{\Lambda}}-\pmb{\Lambda}_0\widetilde{\mathbf{H}}\right|_F^2/N\\
&& + 2\left|\widetilde{\pmb{\Lambda}}-\pmb{\Lambda}_0\widetilde{\mathbf{H}}\right|_F\left| \pmb{\Lambda}_0\widetilde{\mathbf{H}}\right|_F/N = O_P(\max(\sqrt{s}\omega_1,\omega_2)).
\end{eqnarray*}

The proof is now complete.
\end{proof}

\begin{proof}[Proof of Proposition \ref{PROPOSITION1}]
\item
\noindent (1). Note that by the proof of Proposition \ref{Thm5} (2), we have proved $\psi_r(\widetilde{\pmb{\Xi}}) \geq (\omega_2\sqrt{NT}|\widetilde{\pmb{\Xi}}|_2)^{1/2}$ and $\psi_{r+1}(\widetilde{\pmb{\Xi}}) < (\omega_2\sqrt{NT}|\widetilde{\pmb{\Xi}}|_2)^{1/2}$ w.p.a.1., which implies that $\Pr(\widehat{r}=r) \to 1$. 

We next prove the sign consistency. Let $\mathbf{G} = \mathrm{diag}(g_1, \ldots,g_{d})$. Note that $$
\max_{j\in J^c}|\widetilde{\beta}_j| = \max_{j\in J^c}|\widetilde{\beta}_j - \beta_{0,j}|\leq |\widehat{\pmb{\beta}} - \pmb{\beta}_0|_2 = O_P(\max(\sqrt{s}\omega_1,\omega_2)) =o_P(\omega_3)$$ by Assumption \ref{ASSUMPTION6}, which follows that $\mathbf{G}_{J^c} = \mathbf{I}_{d-s}$ w.p.a.1. In addition, since $\max_{j\in J}|\widetilde{\beta}_j| \geq \beta_{\min} - |\widetilde{\pmb{\beta}} - \pmb{\beta}_0|_2 \geq \omega_3$ w.p.a.1 by using Assumptions \ref{ASSUMPTION6}, which implies that $\mathbf{G}_{J} = \mathbf{0}_{s}$ w.p.a.1.

By concentrating out $\mathbf{F}$, the estimator $\widehat{\pmb{\beta}}^{(l)}$ can be re-written as
$$
\widehat{\pmb{\beta}}^{(l)} = \argmin_{\pmb{\beta}\in\mathbb{R}^{d}} \frac{1}{2NT}(\mathbf{y} - \mathbf{X}\pmb{\beta})^\top(\mathbf{M}_{\widehat{\pmb{\Lambda}}^{(l-1)}}\otimes \mathbf{I}_T)(\mathbf{y} - \mathbf{X}\pmb{\beta}) + \omega_3\sum_{j=1}^{d}g_j|\beta_j|.
$$
By the properties of convex optimization, we have
$$
\frac{1}{NT}\mathbf{X}^\top(\mathbf{M}_{\widehat{\pmb{\Lambda}}^{(l-1)}}\otimes \mathbf{I}_T)\mathbf{X}\left(\widehat{\pmb{\beta}}^{(l)} -\pmb{\beta}_0 \right) - \frac{1}{NT}\mathbf{X}^\top(\mathbf{M}_{\widehat{\pmb{\Lambda}}^{(l-1)}}\otimes \mathbf{I}_T)\left(\mathrm{vec}(\mathbf{F}_0\pmb{\Lambda}_0^\top)+\mathbf{e}\right)+ \omega_3\mathbf{G}\vec{\mathbf{g}} = \mathbf{0},
$$
where $|\vec{\mathbf{g}}|_{\infty} \leq 1$ and $\vec{g}_j = \sgn(\widehat{\beta}_j^{(l)})$ if $\widehat{\beta}_j^{(l)}\neq0$ for $j=1,\ldots,d$. 

Hence, $\sgn(\widehat{\pmb{\beta}}^{(l)}) = \sgn(\pmb{\beta}_0)$ if and only if
\begin{eqnarray*}
&&\frac{1}{NT}\mathbf{X}_{J^c}^\top(\mathbf{M}_{\widehat{\pmb{\Lambda}}^{(l-1)}}\otimes \mathbf{I}_T)\mathbf{X}_J\left(\widehat{\pmb{\beta}}_{J}^{(l)} -\pmb{\beta}_{0,J} \right) - \frac{1}{NT}\mathbf{X}_{J^c}^\top(\mathbf{M}_{\widehat{\pmb{\Lambda}}^{(l-1)}}\otimes \mathbf{I}_T)\left(\mathrm{vec}(\mathbf{F}_0\pmb{\Lambda}_0^\top)+\mathbf{e}\right)\nonumber \\
& =& -  \omega_3\mathbf{G}_{J^c}\vec{\mathbf{g}}_{J^c},
\end{eqnarray*}
\begin{eqnarray*}
&&\frac{1}{NT}\mathbf{X}_{J}^\top(\mathbf{M}_{\widehat{\pmb{\Lambda}}^{(l-1)}}\otimes \mathbf{I}_T)\mathbf{X}_J\left(\widehat{\pmb{\beta}}_{J}^{(l)} -\pmb{\beta}_{0,J} \right) - \frac{1}{NT}\mathbf{X}_{J}^\top(\mathbf{M}_{\widehat{\pmb{\Lambda}}^{(l-1)}}\otimes \mathbf{I}_T)\left(\mathrm{vec}(\mathbf{F}_0\pmb{\Lambda}_0^\top)+\mathbf{e}\right) \nonumber \\
&=& - \omega_3 \mathbf{G}_{J}\vec{\mathbf{g}}_{J},
\end{eqnarray*}
\begin{eqnarray*}
\sgn(\widehat{\pmb{\beta}}_{J}^{(l)}) = \sgn(\pmb{\beta}_{0,J})\quad \text{and}\quad \widehat{\pmb{\beta}}_{J^c}^{(l)} = \pmb{\beta}_{0,J^c} = \mathbf{0}.
\end{eqnarray*}

Let $\widetilde{\pmb{\Sigma}}^{(l)} = \frac{1}{NT}\mathbf{X}^\top(\mathbf{M}_{\widehat{\pmb{\Lambda}}^{(l)}}\otimes \mathbf{I}_T)\mathbf{X}$. By standard results in matrix perturbation theory, Lemma \ref{L7} (3) and Assumption \ref{ASSUMPTION6}, we have
\begin{eqnarray*}
&&\left|\psi_{\mathrm{min}}((NT)^{-1}\mathbf{X}_{J}^\top(\mathbf{M}_{\widehat{\pmb{\Lambda}}^{(l-1)}}\otimes \mathbf{I}_T)\mathbf{X}_J)  - \psi_{\mathrm{min}}(\pmb{\Sigma}_{J,J}) \right|\\
&\leq& \left|(NT)^{-1}\mathbf{X}_{J}^\top(\mathbf{M}_{\widehat{\pmb{\Lambda}}^{(l-1)}}\otimes \mathbf{I}_T)\mathbf{X}_J - \pmb{\Sigma}_{J,J} \right|_2 \\
&\leq& s \left|(NT)^{-1}\mathbf{X}_{J}^\top(\mathbf{M}_{\widehat{\pmb{\Lambda}}^{(l-1)}}\otimes \mathbf{I}_T)\mathbf{X}_J - \pmb{\Sigma}_{J,J} \right|_{\mathrm{max}}\\
&=& s\times O_P(\max(\sqrt{s}\omega_1,\omega_2)) = o_P(1).
\end{eqnarray*}
Hence, given the invertibility of $\widetilde{\pmb{\Sigma}}_{J,J}^{(l)}$, to prove $\sgn(\widehat{\pmb{\beta}}^{(l)}) = \sgn(\pmb{\beta}_0)$, it suffices to show w.p.a.1.
\begin{equation}\label{Eq.A3}
\left| \widetilde{\pmb{\Sigma}}_{J,J}^{(l-1),-1}\left((NT)^{-1}\mathbf{X}_{J}^\top(\mathbf{M}_{\widehat{\pmb{\Lambda}}^{(l-1)}}\otimes \mathbf{I}_T)\left(\mathrm{vec}(\mathbf{F}_0\pmb{\Lambda}_0^\top)+\mathbf{e}\right) - \omega_3 \mathbf{G}_{J}\vec{\mathbf{g}}_{J}\right)\right|_{\infty} < \beta_{\min}
\end{equation}
and for any $j\in J^c$
\begin{eqnarray}\label{Eq.A4}
&&\left|(NT)^{-1}\mathbf{X}_{j}^\top(\mathbf{M}_{\widehat{\pmb{\Lambda}}^{(l-1)}}\otimes \mathbf{I}_T)\mathbf{X}_{J}\left(\widehat{\pmb{\beta}}_{J}^{(l)} -\pmb{\beta}_{0,J} \right) -  \frac{1}{NT}\mathbf{X}_{j}^\top(\mathbf{M}_{\widehat{\pmb{\Lambda}}^{(l-1)}}\otimes \mathbf{I}_T)\left(\mathrm{vec}(\mathbf{F}_0\pmb{\Lambda}_0^\top)+\mathbf{e}\right)\right|_{\infty} \nonumber \\
& \leq &\omega_3 g_j.
\end{eqnarray}

Consider \eqref{Eq.A3} first, by using Lemmas \ref{L7} (1)--(2) and Assumption \ref{ASSUMPTION6}.1, we have
\begin{eqnarray*}
&&\left| \widetilde{\pmb{\Sigma}}_{J,J}^{(l-1),-1}\left((NT)^{-1}\mathbf{X}_{J}^\top(\mathbf{M}_{\widehat{\pmb{\Lambda}}^{(l-1)}}\otimes \mathbf{I}_T)\left(\mathrm{vec}(\mathbf{F}_0\pmb{\Lambda}_0^\top)+\mathbf{e}\right) - \omega_3 \mathbf{G}_{J}\vec{\mathbf{g}}_{J}\right)\right|_{\infty}\\
&\leq& |\widetilde{\pmb{\Sigma}}_{J,J}^{(l-1),-1}|_{\infty}\Big( |(NT)^{-1}\mathbf{X}_{J}^\top(\mathbf{M}_{\widehat{\pmb{\Lambda}}^{(l-1)}}\otimes \mathbf{I}_T)\mathrm{vec}(\mathbf{F}_0\pmb{\Lambda}_0^\top)|_{\infty} \nonumber \\
&&+|(NT)^{-1}\mathbf{X}_{J}^\top(\mathbf{M}_{\widehat{\pmb{\Lambda}}^{(l-1)}}\otimes \mathbf{I}_T)\mathbf{e}|_{\infty}+\omega_3\Big)\\
&\leq&\sqrt{s}|\widetilde{\pmb{\Sigma}}_{J,J}^{(l-1),-1}|_2\times \left(O_P(\max(\sqrt{s}\omega_1,\omega_2))+\omega_3\right) = o_P(\beta_{\min}),
\end{eqnarray*}
which implies that $\eqref{Eq.A3}$ holds w.p.a.1. Similarly, using Lemma \ref{L7}, $\eqref{Eq.A4}$ holds w.p.a.1. We then have proved $
\Pr\left(\sgn(\widehat{\pmb{\beta}}^{(l)}) = \sgn(\pmb{\beta}_0)\right) \to 1$ as $(N,T) \to (\infty,\infty)$.

\smallskip

\noindent (2). By the proof of part (1), we have
$$
\widehat{\pmb{\beta}}_{J} -\pmb{\beta}_{0,J} =\widehat{\pmb{\Sigma}}_{J,J}^{-1} \frac{1}{NT}\mathbf{X}_{J}^\top(\mathbf{M}_{\widehat{\pmb{\Lambda}}}\otimes \mathbf{I}_T)\left(\mathrm{vec}(\mathbf{F}_0\pmb{\Lambda}_0^\top)+\mathbf{e}\right) - \omega_3 \mathbf{G}_{J}\vec{\mathbf{g}}_{J},
$$
and
$$
\frac{1}{NT}\sum_{t=1}^{T}(\mathbf{y}_t-\mathbf{X}_{J,t}\widehat{\pmb{\beta}}_J)(\mathbf{y}_t-\mathbf{X}_{J,t}\widehat{\pmb{\beta}}_J)^\top \widehat{\pmb{\Lambda}} = \widehat{\pmb{\Lambda}} \mathbf{V}_{NT},
$$
where $\widehat{\pmb{\Sigma}}_{J,J} = \frac{1}{NT}\mathbf{X}_{J}^\top(\mathbf{M}_{\widehat{\pmb{\Lambda}}}\otimes \mathbf{I}_T)\mathbf{X}_{J}$ and $\mathbf{V}_{NT}$ is a diagonal matrix that consists of the $\widehat{r}$ eigenvalues of $\frac{1}{NT}\sum_{t=1}^{T}(\mathbf{y}_t-\mathbf{X}_{J,t}\widehat{\pmb{\beta}}_J)(\mathbf{y}_t-\mathbf{X}_{J,t}\widehat{\pmb{\beta}}_J)^\top$. By part (1), $\max_{j\in J} g_j = 0$ w.p.a.1 and thus $\pmb{\rho}^\top \omega_3 \mathbf{G}_{J}\vec{\mathbf{g}}_{J} = o_P\left( (NT)^{-1/2} \right)$. Then we can follow the analysis of oracle least squares estimator to establish the asymptotic distribution of $\widehat{\pmb{\beta}}_J$. Specifically, by using Proposition \ref{PropB1}, if $s^{3/2}\times \max(1/\sqrt{N},1/\sqrt{T})\to 0$, we have 
$$
\sqrt{NT}\pmb{\rho}^\top(\widehat{\pmb{\beta}}_J - \pmb{\beta}_{0,J})
= \sqrt{N/T}\pmb{\rho}^\top\pmb{\xi} + \sqrt{T/N}\pmb{\rho}^\top\pmb{\zeta} +  \pmb{\rho}^\top\mathbf{D}^{-1}(\pmb{\Lambda}_0) \frac{1}{\sqrt{NT}}\sum_{t=1}^{T}\widetilde{\mathbf{X}}_{J,t}^\top\mathbf{M}_{\pmb{\Lambda}_0} \mathbf{e}_t 
+ o_P(1).
$$
In addition, by using a similar arguments as the proof of Theorem \ref{THEOREM1}, we have 
$$
\pmb{\rho}^\top\mathbf{D}^{-1}(\pmb{\Lambda}_0) \frac{1}{\sqrt{NT}}\sum_{t=1}^{T}\widetilde{\mathbf{X}}_{J,t}^\top\mathbf{M}_{\pmb{\Lambda}_0} \mathbf{e}_t  \to_D N\left(0,\pmb{\rho}^\top\pmb{\Sigma}_{J}^{-1}\pmb{\Theta}_{J}\pmb{\Sigma}_{J}^{-1}\pmb{\rho}\right).
$$

The proof is now complete.
\end{proof}

\begin{proof}[Proof of Proposition \ref{PROPOSITION2}]
\item
\noindent (1). Similar to the proof of Theorem \ref{THEOREM1}, to show $|T_u\left(\widehat{\pmb{\Omega}}_e\right) - \pmb{\Omega}_e|_2 = O_P\left((\log N /T)^{(1-p_e)/2}C_e(N)\right)$ it suffices to prove 
$$
\max_{1\leq i,j\leq N}\left|\frac{1}{T}\sum_{t=1}^{T}(\widehat{e}_{it}\widehat{e}_{jt} - E(e_{it}e_{jt}))\right| = O_P\left(\sqrt{\log N /T}\right).
$$
Write
\begin{eqnarray*}
&& \max_{1\leq i,j\leq N}\left|\frac{1}{T}\sum_{t=1}^{T}(\widehat{e}_{it}\widehat{e}_{jt} - E(e_{it}e_{jt}))\right|\\
&\leq& \max_{1\leq i,j\leq N}\left|\frac{1}{T}\sum_{t=1}^{T}(\widehat{e}_{it}\widehat{e}_{jt} - e_{it}e_{jt})\right| + \max_{1\leq i,j\leq N}\left|\frac{1}{T}\sum_{t=1}^{T}(e_{it}e_{jt} - E(e_{it}e_{jt}))\right|\\
& \eqqcolon & J_{10} + J_{11}.
\end{eqnarray*}

By using concentration inequality in Lemma \ref{LemmaB4}, we have $J_{11} = O_P\left(\sqrt{\log N /T}\right)$.

Now consider $J_{10}$, by using Cauchy-Schwarz inequality, we have
\begin{eqnarray*}
&&\max_{1\leq i,j\leq N}\left|\frac{1}{T}\sum_{t=1}^{T}(\widehat{e}_{it}\widehat{e}_{jt} - e_{it}e_{jt})\right| \\
&\leq& \max_{1\leq i\leq N}\left|\frac{1}{T}\sum_{t=1}^{T}(\widehat{e}_{it} - e_{it})^2\right| + 2\max_{1\leq i\leq N}\left\{\frac{1}{T}\sum_{t=1}^{T}(\widehat{e}_{it} - e_{it})^2\right\}^{1/2}\max_{1\leq i\leq N}\left\{\frac{1}{T}\sum_{t=1}^{T} e_{it}^2\right\}^{1/2}\\
&=&O_P\left(\max_{1\leq i\leq N}\left\{\frac{1}{T}\sum_{t=1}^{T}(\widehat{e}_{it} - e_{it})^2\right\}^{1/2}\right).
\end{eqnarray*}
For $\max_{1\leq i\leq N}\frac{1}{T}\sum_{t=1}^{T}(\widehat{e}_{it} - e_{it})^2$, by using Lemma \ref{L.B4} (1), we have
\begin{eqnarray*}
&&\max_{1\leq i\leq N}\frac{1}{T}\sum_{t=1}^{T}(\widehat{e}_{it} - e_{it})^2\\ &\leq& 3|\widehat{\pmb{\beta}}_{J} - \pmb{\beta}_{0,J}|_2^2\max_{1\leq i\leq N}\frac{1}{T}\sum_{t=1}^{T}|\mathbf{x}_{J,it}|_F^2 + 3\max_{1\leq i\leq N}|\pmb{\lambda}_{0i}^\top\mathbf{H}|_F^2\frac{1}{T}\sum_{t=1}^{T}|\widehat{\mathbf{f}}_t - \mathbf{H}^{-1}\mathbf{f}_{0t}|_F^2 \\
&& + 3\max_{1\leq i\leq N}|\widehat{\pmb{\lambda}}_i-\mathbf{H}^\top\pmb{\lambda}_{0i}|_F^2\frac{1}{T}\sum_{t=1}^{T}|\widehat{\mathbf{f}}_t|_F^2\\
&=& O_P\left(s^2/(NT) + \log N/ \delta_{NT}^2\right).
\end{eqnarray*}
provided that $\max_{1\leq i\leq N}|\widehat{\pmb{\lambda}}_i-\mathbf{H}^\top\pmb{\lambda}_{0i}|_F^2 = O_P(\log N/ \delta_{NT}^2)$. By the proof of Lemma \ref{L.B1} (1), we have
\begin{eqnarray*}
\widehat{\pmb{\lambda}}_i-\mathbf{H}^\top\pmb{\lambda}_{0i}&=& \frac{1}{NT}\sum_{t=1}^{T}\sum_{j=1}^{N}\mathbf{x}_{J,it}^\top(\widehat{\pmb{\beta}}_J - \pmb{\beta}_{0,J})(\widehat{\pmb{\beta}}_J - \pmb{\beta}_{0,J})^\top\mathbf{x}_{J,jt}\widehat{\pmb{\lambda}}_j^\top \mathbf{V}_{NT}^{-1}\\
&& + \frac{1}{NT}\sum_{t=1}^{T}\sum_{j=1}^{N}\mathbf{x}_{J,it}^\top(\widehat{\pmb{\beta}}_J - \pmb{\beta}_{0,J})\mathbf{f}_{0t}^\top\pmb{\lambda}_{0j}\widehat{\pmb{\lambda}}_j^\top \mathbf{V}_{NT}^{-1}\\
&& + \frac{1}{NT}\sum_{t=1}^{T}\sum_{j=1}^{N}\mathbf{x}_{J,it}^\top(\widehat{\pmb{\beta}}_J - \pmb{\beta}_{0,J})e_{jt}\widehat{\pmb{\lambda}}_j^\top \mathbf{V}_{NT}^{-1}\\
&& + \frac{1}{NT}\sum_{t=1}^{T}\sum_{j=1}^{N}\mathbf{f}_{0t}^\top\pmb{\lambda}_{0i}(\widehat{\pmb{\beta}}_J - \pmb{\beta}_{0,J})^\top\mathbf{x}_{J,jt}\widehat{\pmb{\lambda}}_j^\top \mathbf{V}_{NT}^{-1}\\
&& + \frac{1}{NT}\sum_{t=1}^{T}\sum_{j=1}^{N}e_{it}(\widehat{\pmb{\beta}}_J - \pmb{\beta}_{0,J})^\top\mathbf{x}_{J,jt}\widehat{\pmb{\lambda}}_j^\top \mathbf{V}_{NT}^{-1} + \frac{1}{NT}\sum_{t=1}^{T}\sum_{j=1}^{N}\mathbf{f}_{0t}^\top\pmb{\lambda}_{0i}e_{jt}\widehat{\pmb{\lambda}}_j^\top \mathbf{V}_{NT}^{-1}\\
&& + \frac{1}{NT}\sum_{t=1}^{T}\sum_{j=1}^{N}e_{it}\mathbf{f}_{0t}^\top\pmb{\lambda}_{0j}\widehat{\pmb{\lambda}}_j^\top \mathbf{V}_{NT}^{-1} + \frac{1}{NT}\sum_{t=1}^{T}\sum_{j=1}^{N}e_{it}e_{jt}\widehat{\pmb{\lambda}}_j^\top \mathbf{V}_{NT}^{-1}\\
&\eqqcolon &J_{10,1} + \cdots + J_{10,8}.
\end{eqnarray*}
For $J_{10,1}$, by using Cauchy-Schwarz inequality, we have
\begin{eqnarray*}
\max_{1\leq i\leq N}|J_{10,1}|_F^2 &\leq& |\widehat{\pmb{\beta}}_{J} - \pmb{\beta}_{0,J}|_2^4\max_{1\leq i\leq N}\frac{1}{T}\sum_{t=1}^{T}|\mathbf{x}_{J,it}|_F^2 \times \frac{1}{NT}\sum_{t=1}^{T}\sum_{j=1}^{N}|\mathbf{x}_{J,jt}|_F^2 \frac{1}{N}\sum_{j=1}^{N}|\widehat{\pmb{\lambda}}_j^\top \mathbf{V}_{NT}^{-1}|_F^2\\
&=&O_P(s^2/(NT)^2)=o_P(\log N/ \delta_{NT}^2).
\end{eqnarray*}
Similarly, using Cauchy-Schwarz inequality, we have $$\max_{1\leq i\leq N}|J_{10,l}|_F^2=O_P(s/(NT))=o_P(\log N/ \delta_{NT}^2)$$ for $2\leq l\leq 5$.

For $J_{10,6}$, by using Cauchy-Schwarz inequality, we have
$$
\max_{1\leq i\leq N}|J_{10,6}|_F^2 \leq \max_{1\leq i\leq N}|\pmb{\lambda}_{0i}|_F^2 \frac{1}{T}\sum_{t=1}^{T}|\mathbf{f}_{0t}|_F^2\times \frac{1}{T}\sum_{t=1}^{T}|\frac{1}{N}\sum_{j=1}^{N}\widehat{\pmb{\lambda}}_{j}e_{jt}|_F^2 \times |\mathbf{V}_{NT}^{-1}|^2=O_P(\log N/ \delta_{NT}^2)
$$
since by using Lemma \ref{LemmaB3} (1)
\begin{eqnarray*}
&&\frac{1}{T}\sum_{t=1}^{T}|\frac{1}{N}\sum_{j=1}^{N}\widehat{\pmb{\lambda}}_{j}e_{jt}|_F^2\\
&\leq& 2|\mathbf{H}|_F^2 \frac{1}{T}\sum_{t=1}^{T}|\frac{1}{N}\sum_{j=1}^{N}\pmb{\lambda}_{j}e_{jt}|_F^2 + 2 \frac{1}{NT}\sum_{t=1}^{T}\sum_{j=1}^{N}|e_{jt}|^2\times \frac{1}{N}\sum_{j=1}^{N}|\widehat{\pmb{\lambda}}_{j} - \mathbf{H}^\top\pmb{\lambda}_{j}|_F^2\\
&=&O_P(1/N + s/(NT) +\delta_{NT}^{-2}).
\end{eqnarray*}
Similarly, we can show that $\max_{1\leq i\leq N}|J_{10,7}|_F^2$ and $\max_{1\leq i\leq N}|J_{10,8}|_F^2$ are both $O_P(\log N/\delta_{NT}^{2})$. 

Combining the above results, we have proved part (1).

\smallskip

\noindent (2). We first show the consistency of $\widehat{\pmb{\mu}}_{\zeta}$. By using Lemma \ref{L.B3} (1), part (1) of this proposition and the condition $\sqrt{s}(\log N / T)^{(1-p_e)/2}C_e(N) \to 0$, it suffices to consider the following term:
$$
\frac{1}{NT}\sum_{t=1}^T \mathbf{X}_{J,t}^\top \mathbf{M}_{\widehat{\pmb{\Lambda}}} \pmb{\Omega}_e  \widehat{\pmb{\Lambda}} \left(\frac{\widehat{\mathbf{F}}^\top \widehat{\mathbf{F}}}{T}\right)^{-1} \widehat{\mathbf{f}}_t.
$$
since
\begin{eqnarray*}
&&\left|\frac{1}{NT}\sum_{t=1}^T \mathbf{X}_{J,t}^\top \mathbf{M}_{\widehat{\pmb{\Lambda}}} \left(T_u(\widehat{\pmb{\Omega}}_e) - \pmb{\Omega}_e\right)  \widehat{\pmb{\Lambda}} \left(\frac{\widehat{\mathbf{F}}^\top \widehat{\mathbf{F}}}{T}\right)^{-1} \widehat{\mathbf{f}}_t\right|_F\\
&\leq& \frac{1}{T}\sum_{t=1}^T \left|N^{-1/2}\mathbf{X}_{J,t}\right|_F\left|\widehat{\mathbf{f}}_t\right|_F \left|\mathbf{M}_{\widehat{\pmb{\Lambda}}}\right|_2 \left|T_u(\widehat{\pmb{\Omega}}_e) - \pmb{\Omega}_e\right|_2  \left|N^{-1/2}\widehat{\pmb{\Lambda}}\right|_F \left|\left(\frac{\widehat{\mathbf{F}}^\top \widehat{\mathbf{F}}}{T}\right)^{-1}\right|_F  \\
& = & O_P(\sqrt{s}(\log N / T)^{(1-p_e)/2}C_e(N)).
\end{eqnarray*}

To proceed, we first note that
$$
\frac{\widehat{\mathbf{F}}^\top\widehat{\mathbf{F}}}{T} = \frac{1}{T}\sum_{t=1}^T\widehat{\mathbf{f}}_t\widehat{\mathbf{f}}_t^\top = \frac{1}{N^2}\widehat{\pmb{\Lambda}}^\top\frac{1}{T}\sum_{t=1}^T(\mathbf{y}_t- \mathbf{X}_{J,t}\widehat{\pmb{\beta}}_J)(\mathbf{y}_t- \mathbf{X}_{J,t}\widehat{\pmb{\beta}}_J)^\top \widehat{\pmb{\Lambda}}=\mathbf{V}_{NT},
$$
where $\mathbf{V}_{NT}$ includes the largest $r$ eigenvalues of $\frac{1}{T}\sum_{t=1}^T(\mathbf{y}_t- \mathbf{X}_{J,t}\widehat{\pmb{\beta}}_J)(\mathbf{y}_t- \mathbf{X}_{J,t}\widehat{\pmb{\beta}}_J)^\top$ in descending order, the second equality follows from the definition of $\widehat{\mathbf{f}}_t$, and the third equality follows from the definition of PCA.  Second, by Lemmas \ref{L.B1} (1)-(2), we have
$$
\frac{1}{N}\left|\widehat{\pmb{\Lambda}}^\top(\widehat{\pmb{\Lambda}} - \pmb{\Lambda}_0 \pmb{\mathbf{H}}) \right|_F = O_P\left(\sqrt{s/(NT)}+\delta_{NT}^{-2}\right) \quad \text{and}\quad\frac{1}{\sqrt{N}}\left|\widehat{\pmb{\Lambda}}\mathbf{H}^{-1} -\pmb{\Lambda}_0\right|_F = O_P(\sqrt{s/(NT)}+\delta_{NT}^{-1}),
$$
where $\mathbf{H} = (\mathbf{F}_0^\top\mathbf{F}_0/T)(\pmb{\Lambda}_0^\top\widehat{\pmb{\Lambda}}/N)\mathbf{V}_{NT}^{-1}$.

In addition, since $|\mathbf{M}_{\widehat{\pmb{\Lambda}}} - \mathbf{M}_{\pmb{\Lambda}_0}|_F = O_P(\delta_{NT}^{-1})$ by Lemma \ref{L.B1} (4), we have
\begin{eqnarray*}
&&\frac{1}{NT}\sum_{t=1}^T \mathbf{X}_{J,t}^\top \mathbf{M}_{\widehat{\pmb{\Lambda}}}\pmb{\Omega}_e  \widehat{\pmb{\Lambda}} \left(\frac{\widehat{\mathbf{F}}^\top \widehat{\mathbf{F}}}{T}\right)^{-1} \widehat{\mathbf{f}}_t\\
&=&\frac{1}{NT}\sum_{t=1}^T \mathbf{X}_{J,t}^\top \mathbf{M}_{\pmb{\Lambda}_0}\pmb{\Omega}_e  \widehat{\pmb{\Lambda}} \left(\frac{\widehat{\mathbf{F}}^\top \widehat{\mathbf{F}}}{T}\right)^{-1} \widehat{\mathbf{f}}_t+O_P(\sqrt{s}/\delta_{NT})\\
&=&\frac{1}{NT}\sum_{t=1}^T \mathbf{X}_{J,t}^\top \mathbf{M}_{\pmb{\Lambda}_0} \pmb{\Omega}_e  \widehat{\pmb{\Lambda}} \mathbf{V}_{NT}^{-1} \frac{1}{N}\widehat{\pmb{\Lambda}}^\top\pmb{\Lambda}_0\mathbf{f}_t+\frac{1}{NT}\sum_{t=1}^T \mathbf{X}_{J,t}^\top \mathbf{M}_{\pmb{\Lambda}_0} \pmb{\Omega}_e  \widehat{\pmb{\Lambda}} \mathbf{V}_{NT}^{-1} \frac{1}{N}\widehat{\pmb{\Lambda}}^\top\mathbf{e}_t\\
&&+\frac{1}{NT}\sum_{t=1}^T \mathbf{X}_{J,t}^\top \mathbf{M}_{\pmb{\Lambda}_0} \pmb{\Omega}_e  \widehat{\pmb{\Lambda}} \mathbf{V}_{NT}^{-1} \frac{1}{N}\widehat{\pmb{\Lambda}}^\top\mathbf{X}_{J,t}(\pmb{\beta}_{J}-\widehat{\pmb{\beta}}_{J}) + O_P(\sqrt{s}/\delta_{NT})\\
&\eqqcolon &J_{12} + J_{13} + J_{14} + O_P(\sqrt{s}/\delta_{NT}).
\end{eqnarray*}
It is straightforward to show that $J_{13}$ and $J_{14}$ are both $O_P(\sqrt{s}/\delta_{NT})$, we next investigate $J_{12}$. By the proofs of Lemmas \ref{L.B1} (1) and \ref{L.B3} (3), we have
$$
|\mathbf{V}_{NT} - (\widehat{\pmb{\Lambda}}^\top\pmb{\Lambda}_{0}/N)(\mathbf{F}_{0}^\top\mathbf{F}_{0}/T)(\pmb{\Lambda}_{0}^\top\widehat{\pmb{\Lambda}}/N)|_F = O_P(\delta_{NT}^{-1}),
$$
and
$$
N^{-1/2}\left( \widehat{\pmb{\Lambda}}\left(\frac{\pmb{\Lambda}_{0}^\top\widehat{\pmb{\Lambda}}}{N}\right)^{-1} - \pmb{\Lambda}_{0}\left(\frac{\pmb{\Lambda}_{0}^\top\pmb{\Lambda}_{0}}{N}\right)^{-1}\right) = O_P(\delta_{NT}^{-1}),
$$
which follows that
\begin{eqnarray*}
J_{12} &=& \frac{1}{NT}\sum_{t=1}^T \mathbf{X}_{J,t}^\top \mathbf{M}_{\pmb{\Lambda}_0} \pmb{\Omega}_e  \widehat{\pmb{\Lambda}} (\pmb{\Lambda}_{0}^\top\widehat{\pmb{\Lambda}}/N)^{-1} (\mathbf{F}_{0}^\top\mathbf{F}_{0}/T)^{-1}\mathbf{f}_t + O_P(\sqrt{s}\delta_{NT}^{-1})\\
&=&\frac{1}{NT}\sum_{t=1}^T \mathbf{X}_{J,t}^\top \mathbf{M}_{\pmb{\Lambda}_0} \pmb{\Omega}_e\pmb{\Lambda}_0(\pmb{\Lambda}_{0}^\top\pmb{\Lambda}_0/N)^{-1} (\mathbf{F}_{0}^\top\mathbf{F}_{0}/T)^{-1}\mathbf{f}_t + O_P(\sqrt{s}\delta_{NT}^{-1}).
\end{eqnarray*}

We next prove consider the consistency of $\widehat{\pmb{\Sigma}}_J$. Write
$$
\widehat{\pmb{\Sigma}}_J = \frac{1}{NT}\sum_{t=1}^{T}\mathbf{X}_{J,t}^\top\mathbf{M}_{\widehat{\pmb{\Lambda}}}\mathbf{X}_{J,t} - \frac{1}{NT^2}\sum_{t,s=1}^{T}\mathbf{X}_{J,t}^\top\mathbf{M}_{\widehat{\pmb{\Lambda}}}\mathbf{X}_{J,s}\widehat{a}_{st}
$$
where $\widehat{a}_{st} = \widehat{\mathbf{f}}_{t}^\top(\widehat{\mathbf{F}}^\top\widehat{\mathbf{F}}/T)^{-1}\widehat{\mathbf{f}}_{s}$. For the first term, by using Lemma \ref{L.B1} (4), we have
$$
\left|\frac{1}{NT}\sum_{t=1}^{T}\mathbf{X}_{J,t}^\top(\mathbf{M}_{\widehat{\pmb{\Lambda}}}-\mathbf{M}_{\pmb{\Lambda}_0})\mathbf{X}_{J,t}\right|_F \leq \frac{1}{NT}\sum_{t=1}^{T}\left|\mathbf{X}_{J,t}\right|_F^2\left|\mathbf{M}_{\widehat{\pmb{\Lambda}}}-\mathbf{M}_{\pmb{\Lambda}_0}\right|_F =O_P(s/\delta_{NT}).
$$
For the second term, it suffices to show 
$$
\frac{1}{NT^2}\sum_{t,s=1}^{T}\mathbf{X}_{J,t}^\top\mathbf{M}_{\widehat{\pmb{\Lambda}}}\mathbf{X}_{J,s}(\widehat{a}_{st} - a_{st}) = O_P(s/\delta_{NT}).
$$
Adding and subtracting terms yields
\begin{eqnarray*}
\widehat{a}_{st} - a_{st} &=& (\widehat{\mathbf{f}}_{t} - \mathbf{H}^{-1}\mathbf{f}_{0t})^\top(\widehat{\mathbf{F}}^\top\widehat{\mathbf{F}}/T)^{-1}\widehat{\mathbf{f}}_{s} \\
&& + \mathbf{f}_{0t}^\top\mathbf{H}^{-1,\top}\left[(\widehat{\mathbf{F}}^\top\widehat{\mathbf{F}}/T)^{-1} - \mathbf{H}^\top(\mathbf{F}_0^\top\mathbf{F}_0/T)^{-1}\mathbf{H}\right]\widehat{\mathbf{f}}_{s} \\
&& +  \mathbf{f}_{0t}^\top(\mathbf{F}_0^\top\mathbf{F}_0/T)^{-1}\mathbf{H}(\widehat{\mathbf{f}}_{s} - \mathbf{H}^{-1}\mathbf{f}_{0s})\\
&\eqqcolon & b_{st} + c_{st} + d_{st}.
\end{eqnarray*}
For the term $b_{st}$, by using Cauchy-Schwarz inequality, we have
\begin{eqnarray*}
&& \left|\frac{1}{NT^2}\sum_{t,s=1}^{T}\mathbf{X}_{J,t}^\top\mathbf{M}_{\widehat{\pmb{\Lambda}}}\mathbf{X}_{J,s}b_{st}\right|_F\\
&\leq& \left|(\widehat{\mathbf{F}}^\top\widehat{\mathbf{F}}/T)^{-1}\right|_F \frac{1}{T}\sum_{t=1}^{T}\left|N^{-1/2}\mathbf{X}_{J,t}\right|_F\left|\widehat{\mathbf{f}}_{t} - \mathbf{H}^{-1}\mathbf{f}_{0t}\right|_F\frac{1}{T}\sum_{s=1}^{T}\left|N^{-1/2}\mathbf{X}_{J,s}\right|_F|\widehat{\mathbf{f}}_{t}|_F \\
&\leq&O_P(s) \left(\frac{1}{T}\sum_{t=1}^{T}|\widehat{\mathbf{f}}_t - \mathbf{H}^{-1}\mathbf{f}_{0t}|_F^2\right)^{1/2}  = O_P(s/\delta_{NT}).
\end{eqnarray*}
Similarly, we can show $\left|\frac{1}{NT^2}\sum_{t,s=1}^{T}\mathbf{X}_{J,t}^\top\mathbf{M}_{\widehat{\pmb{\Lambda}}}\mathbf{X}_{J,s}d_{st}\right|_F = O_P(s/\delta_{NT})$. For the term $c_{st}$, by using Lemma \ref{L.B4} (4), we have $$\left|\frac{1}{NT^2}\sum_{t,s=1}^{T}\mathbf{X}_{J,t}^\top\mathbf{M}_{\widehat{\pmb{\Lambda}}}\mathbf{X}_{J,s}c_{st}\right|_F = O_P(s^2/\sqrt{NT}+s\delta_{NT}^{-2})=o_P(s/\delta_{NT}).$$

In addition, by using similar arguments as the proof of Theorem \ref{THEOREM1} and using moments inequality in Lemma \ref{LemmaB6} (3), we have $\left|\widehat{\pmb{\Theta}}_{J} - \pmb{\Theta}_{J}\right|_2 = O(\ell^{-q_a}) + O_P\left[s(\ell / T)^{1/2} \right]$.

Combining the above results, we have now completed the proof.
\end{proof}

\section*{Appendix C} 

\renewcommand{\theequation}{C.\arabic{equation}}
\renewcommand{\thesection}{C.\arabic{section}}
\renewcommand{\thefigure}{C.\arabic{figure}}
\renewcommand{\thetable}{C.\arabic{table}}
\renewcommand{\thelemma}{C.\arabic{lemma}}
\renewcommand{\theassumption}{C.\arabic{assumption}}
\renewcommand{\thetheorem}{C.\arabic{theorem}}
\renewcommand{\theproposition}{C.\arabic{proposition}}

\setcounter{equation}{0}
\setcounter{lemma}{0}
\setcounter{section}{0}
\setcounter{table}{0}
\setcounter{figure}{0}
\setcounter{assumption}{0}
\setcounter{proposition}{0}

Given the information of set $J$, the oracle least squares estimator can be written as
$$
(\widehat{\pmb{\beta}}_J,\widehat{\pmb{\Lambda}} ) = \argmin_{\pmb{\beta}\in\mathbb{R}^{s}, \pmb{\Lambda}\in \mathbb{R}^{N\times r}}\sum_{t=1}^{T}(\mathbf{y}_t-\mathbf{X}_{J,t}\pmb{\beta})^\top \mathbf{M}_{\pmb{\Lambda}}(\mathbf{y}_t-\mathbf{X}_{J,t}\pmb{\beta}).
$$
Similar to \cite{bai2009panel}, interchanging the $i$ and $t$ dimensions, we have the following relationship:
\begin{eqnarray*}
\widehat{\pmb{\beta}}_J &=& \left(\sum_{t=1}^{T}\mathbf{X}_{J,t}^\top\mathbf{M}_{\widehat{\pmb{\Lambda}}}\mathbf{X}_{J,t} \right)^{-1}\sum_{t=1}^{T}\mathbf{X}_{J,t}^\top\mathbf{M}_{\widehat{\pmb{\Lambda}}}\mathbf{y}_{t},\\
\widehat{\pmb{\Lambda}} \mathbf{V}_{NT} &=& \left(\frac{1}{NT}\sum_{t=1}^{T}(\mathbf{y}_t - \mathbf{X}_{J,t}\widehat{\pmb{\beta}}_J)(\mathbf{y}_t - \mathbf{X}_{J,t}\widehat{\pmb{\beta}}_J)^\top \right) \widehat{\pmb{\Lambda}},
\end{eqnarray*}
where $\mathbf{V}_{NT}$ is a diagonal matrix that consists of the $r$ largest eigenvalues of the matrix $\frac{1}{NT}\sum_{t=1}^{T}(\mathbf{y}_t - \mathbf{X}_{J,t}\widehat{\pmb{\beta}}_J)(\mathbf{y}_t - \mathbf{X}_{J,t}\widehat{\pmb{\beta}}_J)^\top$ arranged in a descending order.

\begin{proposition}\label{PropB1}
Let $\{\mathbf{e}_t\}$ satisfy the conditions of Assumption C in \cite{bai2009panel}. Suppose further that (1) $E(x_{j,it}^4)<\infty$ and $\inf_{\pmb{\Lambda}} \mathbf{D}(\pmb{\Lambda}) >0$ with $\pmb{\Lambda}^\top \pmb{\Lambda} /N =\mathbf{I}_r$; (2) $\frac{1}{T}\mathbf{F}_0^\top \mathbf{F}_0 \to_P \pmb{\Sigma}_f > 0$ with $E|\mathbf{f}_{0t}|_F^4 < \infty$, and $\frac{1}{N}\pmb{\Lambda}_0^\top \pmb{\Lambda}_0 \to_P\pmb{\Lambda}_{\lambda}>0 $ with $E|\pmb{\lambda}_{0i}|_F^4<\infty$; (3) $\{\mathbf{e}_t\}$ is independent of $\mathbf{X}_J$, $\pmb{\Lambda}_0$ and $\mathbf{F}_0$. Then
\begin{eqnarray*}
\sqrt{NT}\pmb{\rho}^\top(\widehat{\pmb{\beta}}_J - \pmb{\beta}_{0,J})
&=& \sqrt{N/T}\pmb{\rho}^\top\pmb{\xi} + \sqrt{T/N}\pmb{\rho}^\top\pmb{\zeta} +  \pmb{\rho}^\top\mathbf{D}^{-1}(\pmb{\Lambda}_0) \frac{1}{\sqrt{NT}}\sum_{t=1}^{T}\widetilde{\mathbf{X}}_{J,t}^\top\mathbf{M}_{\pmb{\Lambda}_0} \mathbf{e}_t \\
&&+ O_P(s^2/\sqrt{NT} + s^{3/2}\delta_{NT}^{-1}),
\end{eqnarray*}
where $\delta_{NT} = \min(\sqrt{N},\sqrt{T})$, $\mathbf{D}(\pmb{\Lambda}) = \frac{1}{NT}\sum_{t=1}^{T}\widetilde{\mathbf{X}}_{J,t}^\top\mathbf{M}_{\pmb{\Lambda}}\widetilde{\mathbf{X}}_{J,t}$, $\widetilde{\mathbf{X}}_{J,t} = \mathbf{X}_{J,t} - \frac{1}{T}\sum_{s=1}^{T}a_{st}\mathbf{X}_{J,s}$, $a_{st} = \mathbf{f}_{0t}^\top(\mathbf{F}_0^\top\mathbf{F}_0/T)\mathbf{f}_{0s}$,
$$
\pmb{\xi} = - \mathbf{D}^{-1}(\pmb{\Lambda}_0) \frac{1}{NT}\sum_{t,s=1}^{T}\frac{\widetilde{\mathbf{X}}_{J,t}^\top\pmb{\Lambda}_0}{N}\left(\frac{\pmb{\Lambda}_0^\top\pmb{\Lambda}_0}{N}\right)^{-1} \left(\frac{\mathbf{F}_0^\top\mathbf{F}_0}{T}\right)^{-1}\mathbf{f}_{0s}\left(\sum_{i=1}^{N}E(e_{it}e_{is})\right),
$$
and
$$
\pmb{\zeta} = - \mathbf{D}^{-1}(\pmb{\Lambda}_0) \frac{1}{NT}\sum_{t=1}^{T}\mathbf{X}_{J,t}^\top\mathbf{M}_{\pmb{\Lambda}_0}\pmb{\Omega}_e\pmb{\Lambda}_0(\pmb{\Lambda}_0^\top\pmb{\Lambda}_0/N)^{-1}(\mathbf{F}_0^\top\mathbf{F}_0/T)^{-1}\mathbf{f}_{0t}.
$$
\end{proposition}

\begin{lemma}\label{L.B0}
Under the conditions of Proposition \ref{PropB1}, then
\begin{itemize}[leftmargin=*, itemsep=0.5pt, parsep=0.5pt, topsep=0.6pt]
\item [1.] $E\left|\frac{1}{T}\sum_{t=1}^{T}\mathbf{f}_{0t}e_{it}\right|_F^2 = O(1/T)$ and $E\left|\frac{1}{T}\sum_{t=1}^{T}e_{it}\right|^2 = O(1/T)$;

\item [2.] $E\left|\frac{1}{T}\sum_{t=1}^{T}[e_{it}e_{jt} - E(e_{it}e_{jt})]\right|^2 = O(1/T)$ and \\ $E\left|\frac{1}{NT}\sum_{i=1}^{N}\sum_{t=1}^{T}[e_{it}e_{is} - E(e_{it}e_{is})]\right|^2 = O(1/(NT));$

\item [3.] $E\left|\frac{1}{NT}\sum_{i=1}^{N}\sum_{t=1}^{T}\mathbf{f}_{0t}\pmb{\lambda}_{0i}^\top e_{it}\right|_F^2 = O(1/(NT))$, $E\left|\frac{1}{NT}\sum_{i=1}^{N}\sum_{t=1}^{T}\pmb{\lambda}_{0i}e_{it}\right|_F^2 = O(1/(NT))$, and\\ $E\left|\frac{1}{NT}\sum_{i=1}^{N}\sum_{t=1}^{T}x_{j,is}e_{it}\right|_F^2 = O(1/(NT));$

\item [4.] $E\left|\frac{1}{NT}\sum_{i=1}^{N}\sum_{s=1}^{T}x_{j,it}e_{is}\mathbf{f}_{0s}^\top\right|_F^2 = O(1/(NT))$;

\item [5.] $E|\frac{1}{N}\sum_{i=1}^{N}x_{j,it}e_{is}|^2 = O(1/N)$ and $E|\frac{1}{N}\sum_{i=1}^{N}\pmb{\lambda}_{0i}e_{it}|_F^2 = O(1/N)$;

\item [6.] $E|\frac{1}{T}\sum_{s=1}^{T}e_{ks}\frac{1}{N}\sum_{i=1}^{N}x_{j,is}\pmb{\lambda}_{0i}^\top|_F^2 = O(1/T)$;

\item [7.] $E\left|\frac{1}{T}\sum_{s=1}^{T}\frac{1}{N}\sum_{i=1}^{N}\sum_{j=1}^{N}x_{k,it}({e}_{is}{e}_{js}-E({e}_{is}{e}_{js}))\pmb{\lambda}_{0j}^\top \right|_F^2 = O(1/T)$;

\item [8.] $E\left|\frac{1}{NT}\sum_{s=1}^{T}\sum_{i=1}^{N}x_{k,it}({e}_{is}{e}_{js}-E({e}_{is}{e}_{js}))\right|_F^2 = O(1/(NT))$;

\item [9.] $E\left|\frac{1}{T}\sum_{s=1}^{T}\frac{1}{N}\sum_{i=1}^{N}\sum_{j=1}^{N}\pmb{\lambda}_{0i}\pmb{\lambda}_{0j}^\top({e}_{is}{e}_{js}-E({e}_{is}{e}_{js})) \right|_F^2 = O(1/T)$;
\end{itemize}
\end{lemma}

\begin{lemma}\label{L.B1}
Under the conditions of Proposition \ref{PropB1}, then
\begin{itemize}[leftmargin=*, itemsep=0.5pt, parsep=0.5pt, topsep=0.6pt]
\item[1.] $N^{-1/2}|\widehat{\pmb{\Lambda}} - \pmb{\Lambda}_0\mathbf{H}|_F = O_P\left(|\widehat{\pmb{\beta}}_J-\pmb{\beta}_{0,J}|_F+\delta_{NT}^{-1}\right)$, where $\mathbf{H} = (\mathbf{F}_0^\top\mathbf{F}_0/T)(\pmb{\Lambda}_0^\top\widehat{\pmb{\Lambda}}/N)\mathbf{V}_{NT}^{-1}$ and $\mathbf{V}_{NT}$ is a diagonal matrix that consists of the $r$ largest eigenvalues of the matrix $\frac{1}{NT}\sum_{t=1}^{T}(\mathbf{y}_t - \mathbf{X}_{J,t}\widehat{\pmb{\beta}}_J)(\mathbf{y}_t - \mathbf{X}_{J,t}\widehat{\pmb{\beta}}_J)^\top$;

\item[2.] $N^{-1}\pmb{\Lambda}_0^\top(\widehat{\pmb{\Lambda}} - \pmb{\Lambda}_0\mathbf{H}) = O_P\left(|\widehat{\pmb{\beta}}_J-\pmb{\beta}_{0,J}|_F+\delta_{NT}^{-2}\right)$ and\\ $N^{-1}\widehat{\pmb{\Lambda}}^\top(\widehat{\pmb{\Lambda}} - \pmb{\Lambda}_0\mathbf{H}) = O_P\left(|\widehat{\pmb{\beta}}_J-\pmb{\beta}_{0,J}|_F+\delta_{NT}^{-2}\right);$

\item[3.] $\mathbf{H}\mathbf{H}^\top = \left(\pmb{\Lambda}_0^\top\pmb{\Lambda}_0/N\right)^{-1} + O_P(|\widehat{\pmb{\beta}}_J - \pmb{\beta}_{0,J}|_F+\delta_{NT}^{-2})$;

\item[4.] $|\mathbf{P}_{\widehat{\pmb{\Lambda}}} - \mathbf{P}_{\pmb{\Lambda}_0}|_F = O_P(|\widehat{\pmb{\beta}}_J - \pmb{\beta}_{0,J}|_F + \delta_{NT}^{-1})$;

\item[5.] $\left|N^{-1}\mathbf{X}_{J,t}^\top(\widehat{\pmb{\Lambda}} - \pmb{\Lambda}_0\mathbf{H})\right|_F = O_P\left(\sqrt{s}|\widehat{\pmb{\beta}}_J-\pmb{\beta}_{0,J}|_F+\sqrt{s}\delta_{NT}^{-2}\right)$ for $1\leq t \leq T$;

\item[6.] $\left|\frac{1}{NT}\sum_{t=1}^{T}\mathbf{X}_{J,t}^\top\mathbf{M}_{\widehat{\pmb{\Lambda}}}(\widehat{\pmb{\Lambda}} - \pmb{\Lambda}_0\mathbf{H})\right|_F = O_P\left(\sqrt{s}|\widehat{\pmb{\beta}}_J-\pmb{\beta}_{0,J}|_F+\sqrt{s}\delta_{NT}^{-2}\right)$;

\item[7.] $N^{-1}\mathbf{e}_{t}^\top(\widehat{\pmb{\Lambda}} - \pmb{\Lambda}_0\mathbf{H}) = O_P\left(\sqrt{s}\delta_{NT}^{-1}|\widehat{\pmb{\beta}}_J-\pmb{\beta}_{0,J}|_F+\delta_{NT}^{-2}\right)$ for $1\leq t \leq T$;

\item[8.] $N^{-1}T^{-1/2}\sum_{t=1}^{T}\mathbf{e}_{t}^\top(\widehat{\pmb{\Lambda}} - \pmb{\Lambda}_0\mathbf{H}) = O_P\left(\sqrt{s}\delta_{NT}^{-1}|\widehat{\pmb{\beta}}_J-\pmb{\beta}_{0,J}|_F+T^{-1/2}+\delta_{NT}^{-2}\right)$;

\item[9.] $(TN)^{-1}\sum_{t=1}^{T}\mathbf{f}_{0t}\mathbf{e}_{t}^\top(\widehat{\pmb{\Lambda}} - \pmb{\Lambda}_0\mathbf{H}) = O_P\left(\sqrt{s/(NT)}|\widehat{\pmb{\beta}}_J-\pmb{\beta}_{0,J}|_F+T^{-1}+T^{-1/2}\delta_{NT}^{-2}\right)$;

\item[10.]		 $(TN)^{-1}\sum_{t=1}^{T}(\mathbf{X}_{J,t}^\top\pmb{\Lambda}_0/N)(\pmb{\Lambda}_0^\top\pmb{\Lambda}_0/N)(\widehat{\pmb{\Lambda}}\mathbf{H}^{-1} - \pmb{\Lambda}_0)^\top\mathbf{e}_{t} = \frac{1}{T^2}\sum_{t=1}^{T}\sum_{s=1}^{T}(\mathbf{X}_{J,t}^\top\pmb{\Lambda}_0/N)(\pmb{\Lambda}_0^\top\pmb{\Lambda}_0/N)\\\times(\mathbf{F}_0^\top\mathbf{F}_0/T)^{-1}\mathbf{f}_{0s}(\frac{1}{N}\sum_{i=1}^{N}e_{is}e_{it}) + O_P\left(\sqrt{s/(NT)}|\widehat{\pmb{\beta}}_J-\pmb{\beta}_{0,J}|_F+\sqrt{s/T}\delta_{NT}^{-2}\right)$;
\end{itemize}
\end{lemma}

\begin{lemma}\label{L.B2}
Under the conditions of Proposition \ref{PropB1}, then
\begin{itemize}[leftmargin=*, itemsep=0.5pt, parsep=0.5pt, topsep=0.6pt]
\item [1.]  
\begin{eqnarray*}
&&\frac{1}{N^2T^2}\sum_{s=1}^{T}\sum_{t=1}^{T}\mathbf{X}_{J,t}^\top\mathbf{M}_{\widehat{\pmb{\Lambda}}}(\mathbf{e}_s\mathbf{e}_s^\top-\pmb{\Omega}_{e})\widehat{\pmb{\Lambda}}\mathbf{G}\mathbf{f}_{0t}\\
&=& O_P\left( \sqrt{s/(NT)}|\widehat{\pmb{\beta}}_J-\pmb{\beta}_{0,J}|_F +\sqrt{s/T}|\widehat{\pmb{\beta}}_J-\pmb{\beta}_{0,J}|_F^2 + \sqrt{s/(NT)}\delta_{NT}^{-1} \right),
\end{eqnarray*}
where $\mathbf{G} = (\pmb{\Lambda}_0^\top\widehat{\pmb{\Lambda}}/N)^{-1}(\mathbf{F}_0^\top\mathbf{F}_0/T)^{-1}$ and $\pmb{\Omega}_{e}=E(\mathbf{e}_t\mathbf{e}_t^\top)$;

\item [2.]  
\begin{eqnarray*}
\sqrt{NT}(\widehat{\pmb{\beta}}_J - \pmb{\beta}_{0,J}) &=& \mathbf{D}^{-1}(\widehat{\pmb{\Lambda}}) \frac{1}{\sqrt{NT}}\sum_{t=1}^{T}\left[\mathbf{X}_{J,t}^\top\mathbf{M}_{\widehat{\pmb{\Lambda}}} - \frac{1}{T}\sum_{s=1}^{T}a_{st}\mathbf{X}_{J,s}^\top\mathbf{M}_{\widehat{\pmb{\Lambda}}}\right] \mathbf{e}_t + \sqrt{\frac{T}{N}}\pmb{\zeta}_{NT}\\ 
&&+ O_P(\sqrt{sNT}|\widehat{\pmb{\beta}}_J - \pmb{\beta}_{0,J}|_F^2 + \sqrt{sNT}\delta_{NT}^{-1}|\widehat{\pmb{\beta}}_J - \pmb{\beta}_{0,J}|_F + \sqrt{sNT}\delta_{NT}^{-3}),
\end{eqnarray*}
where $a_{ts} = \mathbf{f}_{0t}^\top(\mathbf{F}_0^\top\mathbf{F}_0/T)\mathbf{f}_{0s}$, $\mathbf{D}(\widehat{\pmb{\Lambda}}) = \frac{1}{NT}\sum_{t=1}^{T}\widetilde{\mathbf{X}}_{J,t}^\top\mathbf{M}_{\widehat{\pmb{\Lambda}}}\widetilde{\mathbf{X}}_{J,t}$, $\widetilde{\mathbf{X}}_{J,t} = \mathbf{X}_{J,t} - \frac{1}{s}\sum_{s=1}^{T}a_{st}\mathbf{X}_{J,s}$ and $\pmb{\zeta}_{NT} = - \mathbf{D}^{-1}(\widehat{\pmb{\Lambda}}) \frac{1}{NT}\sum_{t=1}^{T}\mathbf{X}_{J,t}^\top\mathbf{M}_{\widehat{\pmb{\Lambda}}}\pmb{\Omega}_e\widehat{\pmb{\Lambda}}\mathbf{G}\mathbf{f}_{0t}$;

\item [3.] 
\begin{eqnarray*}
&& \frac{1}{\sqrt{NT}}\sum_{t=1}^{T}\left[\mathbf{X}_{J,t}^\top\mathbf{M}_{\widehat{\pmb{\Lambda}}} - \frac{1}{N}\sum_{s=1}^{T}a_{ts}\mathbf{X}_{J,s}^\top\mathbf{M}_{\widehat{\pmb{\Lambda}}}\right] \mathbf{e}_t \\
&=& \frac{1}{\sqrt{NT}}\sum_{t=1}^{T}\left[\mathbf{X}_{J,t}^\top\mathbf{M}_{\pmb{\Lambda}_0} - \frac{1}{N}\sum_{s=1}^{T}a_{ts}\mathbf{X}_{J,s}^\top\mathbf{M}_{\pmb{\Lambda}_0}\right] \mathbf{e}_t \\
&& - \sqrt{\frac{N}{T}}\frac{1}{NT}\sum_{t=1}^{T}\sum_{s=1}^{T}\frac{\widetilde{\mathbf{X}}_{J,t}^\top\pmb{\Lambda}_0}{N}\left(\frac{\pmb{\Lambda}_0^\top\pmb{\Lambda}_0}{N}\right)^{-1} \left(\frac{\mathbf{F}_0^\top\mathbf{F}_0}{T}\right)^{-1}\mathbf{f}_{0s}\left(\sum_{i=1}^{N}e_{it}e_{is}\right)\\
&& + O_P(\sqrt{Ns}|\widehat{\pmb{\beta}}_J - \pmb{\beta}_{0,J}|_F^2 + \sqrt{s}|\widehat{\pmb{\beta}}_J - \pmb{\beta}_{0,J}|_F + \sqrt{Ns}\delta_{NT}^{-2});
\end{eqnarray*}

\item [4.] $|\widehat{\pmb{\beta}}_J - \pmb{\beta}_{0,J}|_F = O_P(\sqrt{s/(NT)})$.
\end{itemize}
\end{lemma}

\begin{lemma}\label{L.B3}
Under the conditions of Proposition \ref{PropB1}, then
\begin{itemize}[leftmargin=*, itemsep=0.5pt, parsep=0.5pt, topsep=0.6pt]
\item [1.] $\left|\mathbf{D}(\widehat{\pmb{\Lambda}}) - \mathbf{D}(\pmb{\Lambda}_0) \right|_F= O_P(s^{3/2}/\sqrt{NT} + s\delta_{NT}^{-1})$;

\item [2.] $|\pmb{\xi}_{NT} - \pmb{\xi}|_F = O_P(\sqrt{s/N} + s^2/\sqrt{NT} + s^{3/2}\delta_{NT}^{-1})$, where 
$$
\pmb{\xi}_{NT} = -\mathbf{D}^{-1}(\widehat{\pmb{\Lambda}}) \frac{1}{NT}\sum_{t=1}^{T}\sum_{s=1}^{T}\frac{\widetilde{\mathbf{X}}_{J,t}^\top\pmb{\Lambda}_0}{N}\left(\frac{\pmb{\Lambda}_0^\top\pmb{\Lambda}_0}{N}\right)^{-1} \left(\frac{\mathbf{F}_0^\top\mathbf{F}_0}{T}\right)^{-1}\mathbf{f}_{0s}\left(\sum_{i=1}^{N}e_{it}e_{is}\right),
$$
and 
$$
\pmb{\xi}=-\mathbf{D}^{-1}(\pmb{\Lambda}_0)\frac{1}{NT}\sum_{t=1}^{T}\sum_{s=1}^{T}\frac{\widetilde{\mathbf{X}}_{J,t}^\top\pmb{\Lambda}_0}{N}\left(\frac{\pmb{\Lambda}_0^\top\pmb{\Lambda}_0}{N}\right)^{-1} \left(\frac{\mathbf{F}_0^\top\mathbf{F}_0}{T}\right)^{-1}\mathbf{f}_{0s}\left(\sum_{i=1}^{N}E(e_{it}e_{is})\right);
$$

\item [3.] $|\pmb{\zeta}_{NT} - \pmb{\zeta}|_F = O_P(s^2/\sqrt{NT} + s^{3/2}\delta_{NT}^{-1})$, where
$$
\pmb{\zeta}=-\mathbf{D}^{-1}(\pmb{\Lambda}_0) \frac{1}{NT}\sum_{t=1}^{T}\mathbf{X}_{J,t}^\top\mathbf{M}_{\pmb{\Lambda}_0}\pmb{\Omega}_e\pmb{\Lambda}_0(\pmb{\Lambda}_0^\top\pmb{\Lambda}_0/N)^{-1}(\mathbf{F}_0^\top\mathbf{F}_0/T)^{-1}\mathbf{f}_{0t}.
$$
\end{itemize}
\end{lemma}

\begin{lemma}\label{L.B4}
Under the conditions of Proposition \ref{PropB1}, then
\begin{itemize}[leftmargin=*, itemsep=0.5pt, parsep=0.5pt, topsep=0.6pt]
\item [1.] $\frac{1}{T}\sum_{t=1}^{T}|\widehat{\mathbf{f}}_t - \mathbf{H}^{-1}\mathbf{f}_{0t}|_F^2 = O_P\left(s^2/(NT) + \delta_{NT}^{-2}\right)$;

\item [2.] $\frac{1}{T}\sum_{t=1}^{T}(\widehat{\mathbf{f}}_t - \mathbf{H}^{-1}\mathbf{f}_{0t})\mathbf{f}_{0t}^\top = O_P\left(s/\sqrt{NT} + \delta_{NT}^{-2}\right)$;

\item [3.] $\widehat{\mathbf{F}}^\top\widehat{\mathbf{F}}/T - \mathbf{H}^{-1} (\mathbf{F}_0^\top\mathbf{F}_0/T)\mathbf{H}^{-1,\top}= O_P\left(s/\sqrt{NT} + \delta_{NT}^{-2}\right)$;

\item [4.] $ (\widehat{\mathbf{F}}^\top\widehat{\mathbf{F}}/T )^{-1} - \mathbf{H}^{\top} (\mathbf{F}_0^\top\mathbf{F}_0/T)^{-1}\mathbf{H} = O_P\left(s/\sqrt{NT} + \delta_{NT}^{-2}\right)$;
\end{itemize}
\end{lemma}

\begin{proof}[Proof of Proposition \ref{PropB1}]
\item
Proposition \ref{PropB1} follows directly from Lemmas \ref{L.B2} (2)--(4) and \ref{L.B3}.
\end{proof}

\begin{proof}[Proof of Lemma \ref{L.B0}]
\item
By using the weak dependence conditions of $\{\mathbf{e}_t\}$ in Assumption C of \cite{bai2009panel}, this lemma can be proved by using similar arguments as in the proof of \cite{bai2009panel}, e.g., the proof of Lemma A.2 in \cite{bai2009panel}. 
\end{proof}

\begin{proof}[Proof of Lemma \ref{L.B1}]
\item
\noindent (1). Note that $\mathbf{y}_t = \mathbf{X}_{J,t}\pmb{\beta}_{0,J} + \pmb{\Lambda}_0\mathbf{f}_{0t} + \mathbf{e}_t$, write
\begin{eqnarray*}
\widehat{\pmb{\Lambda}}\mathbf{V}_{NT} &=& \frac{1}{NT}\sum_{t=1}^{T}\mathbf{X}_{J,t}(\widehat{\pmb{\beta}}_J-\pmb{\beta}_{0,J})(\widehat{\pmb{\beta}}_J-\pmb{\beta}_{0,J})^\top\mathbf{X}_{J,t}^\top\widehat{\pmb{\Lambda}}+\frac{1}{NT}\sum_{t=1}^{T}\mathbf{X}_{J,t}(\pmb{\beta}_{0,J}-\widehat{\pmb{\beta}}_J)\mathbf{f}_{0t}^\top\pmb{\Lambda}_0^\top\widehat{\pmb{\Lambda}}\\
&&+\frac{1}{NT}\sum_{t=1}^{T}\mathbf{X}_{J,t}(\pmb{\beta}_{0,J}-\widehat{\pmb{\beta}}_J)\mathbf{e}_{t}^\top\widehat{\pmb{\Lambda}}+\frac{1}{NT}\sum_{t=1}^{T}\pmb{\Lambda}_0\mathbf{f}_{0t}(\pmb{\beta}_{0,J}-\widehat{\pmb{\beta}}_J)^\top\mathbf{X}_{J,t}^\top\widehat{\pmb{\Lambda}}\\
&& + \frac{1}{NT}\sum_{t=1}^{T}\mathbf{e}_{t}(\pmb{\beta}_{0,J}-\widehat{\pmb{\beta}}_J)^\top\mathbf{X}_{J,t}^\top\widehat{\pmb{\Lambda}}+\frac{1}{NT}\sum_{t=1}^{T}\pmb{\Lambda}_{0}\mathbf{f}_{0t}\mathbf{e}_t^\top\widehat{\pmb{\Lambda}}\\
&&+\frac{1}{NT}\sum_{t=1}^{T}\mathbf{e}_{t}\mathbf{f}_{0t}^\top\pmb{\Lambda}_{0}^\top\widehat{\pmb{\Lambda}}+\frac{1}{NT}\sum_{t=1}^{T}\mathbf{e}_t\mathbf{e}_t^\top\widehat{\pmb{\Lambda}} + \frac{1}{NT}\sum_{t=1}^{T}\pmb{\Lambda}_{0}\mathbf{f}_{0t}\mathbf{f}_{0t}^\top\pmb{\Lambda}_{0}^\top \widehat{\pmb{\Lambda}}\\
&\eqqcolon & \mathbf{K}_1+\cdots + \mathbf{K}_9,
\end{eqnarray*}
which implies that
$$
\widehat{\pmb{\Lambda}}\mathbf{V}_{NT}(\pmb{\Lambda}_{0}^\top\widehat{\pmb{\Lambda}}/N)^{-1}(\mathbf{F}_0^\top\mathbf{F}_{0}/T)^{-1}  - \pmb{\Lambda}_{0} = ( \mathbf{K}_1+\cdots + \mathbf{K}_8)(\pmb{\Lambda}_{0}^\top\widehat{\pmb{\Lambda}}/N)^{-1}(\mathbf{F}_{0}^\top\mathbf{F}_{0}/T)^{-1}.
$$
Since $|N^{-1/2}\widehat{\pmb{\Lambda}}|_F = \sqrt{r}$ and we can show that $N^{-1/2}\mathbf{K}_1,\ldots,N^{-1/2}\mathbf{K}_8$ are all $o_P(1)$ in the below, we have $
|\mathbf{V}_{NT} - (\widehat{\pmb{\Lambda}}^\top\pmb{\Lambda}_{0}/N)(\mathbf{F}_{0}^\top\mathbf{F}_{0}/T)(\pmb{\Lambda}_{0}^\top\widehat{\pmb{\Lambda}}/N)|_F = o_P(1)$.

Given the invertibility of $\mathbf{V}_{NT}$, we have
$$
N^{-1/2}|\widehat{\pmb{\Lambda}} - \pmb{\Lambda}_{0}\mathbf{H}|_F =  N^{-1/2}(\mathbf{K}_1+\cdots + \mathbf{K}_8) \times O_P(1).
$$
We next prove the convergence rate of $\mathbf{K}_1,\ldots,\mathbf{K}_8$. For $\mathbf{K}_1$,
\begin{eqnarray*}
|N^{-1/2}\mathbf{K}_1|_F &\leq& \frac{\sqrt{r}}{NT}\sum_{t=1}^{T}|\mathbf{X}_{J,t}(\widehat{\pmb{\beta}}_J-\pmb{\beta}_{0,J})|_F^2=(\widehat{\pmb{\beta}}_J-\pmb{\beta}_{0,J})^\top \frac{\sqrt{r}}{NT}\sum_{t=1}^{T} \mathbf{X}_{J,t}^\top\mathbf{X}_{J,t}(\widehat{\pmb{\beta}}_J-\pmb{\beta}_{0,J}) \\
&\leq&|\widehat{\pmb{\beta}}_J-\pmb{\beta}_{0,J}|_F^2 \psi_{\max}\left(\frac{\sqrt{r}}{NT}\sum_{i=1}^{N}\sum_{t=1}^{T}\mathbf{x}_{J,it}\mathbf{x}_{J,it}^\top\right) = O_P(|\widehat{\pmb{\beta}}_J-\pmb{\beta}_{0,J}|_F^2).
\end{eqnarray*}
For $\mathbf{K}_2$, by using Cauchy-Schwarz inequality, we have
\begin{eqnarray*}
|T^{-1/2}\mathbf{K}_2|_F &\leq&\sqrt{r}\left\{ \frac{1}{NT}\sum_{t=1}^{T}|\mathbf{X}_{J,t}(\widehat{\pmb{\beta}}_J-\pmb{\beta}_{0,J})|_F^2\right\}^{1/2}\left\{ \frac{1}{NT}\sum_{t=1}^{T}|\pmb{\Lambda}_{0}\mathbf{f}_{0t}|_F^2\right\}^{1/2}\\
&=& O_P(|\widehat{\pmb{\beta}}_J-\pmb{\beta}_{0,J}|_F).
\end{eqnarray*}
Similar to the term $\mathbf{K}_2$, we can show that $\mathbf{K}_3,\ldots,\mathbf{K}_5$ are all $O_P(|\widehat{\pmb{\beta}}_J-\pmb{\beta}_{0,J}|_F)$. For $\mathbf{K}_6$ and $\mathbf{K}_7$, by using Lemma \ref{L.B0} (1), we have
\begin{eqnarray*}
\left|\frac{1}{NT}\sum_{t=1}^{T}\pmb{\Lambda}_{0}\mathbf{f}_{0t}\mathbf{e}_t^\top\right|_F^2&\leq& \frac{1}{N}\sum_{i=1}^{N}|\pmb{\lambda}_{0i}|_F^2\times \frac{1}{N}\sum_{j=1}^{N}\left|\frac{1}{T}\sum_{t=1}^{T}\mathbf{f}_{0t}e_{jt}\right|_F^2 =O_P(1/T).
\end{eqnarray*}
For $\mathbf{K}_8$, by using Lemma \ref{L.B0} (2), we have
\begin{eqnarray*}
\left|\frac{1}{NT}\sum_{t=1}^{T}\mathbf{e}_t\mathbf{e}_t^\top\right|_F&\leq&  \left|\frac{1}{NT}\sum_{t=1}^{T}(\mathbf{e}_t\mathbf{e}_t^\top-E(\mathbf{e}_t\mathbf{e}_t^\top))\right|_F + \frac{1}{N}|E(\mathbf{e}_t\mathbf{e}_t^\top)|_F\\
&\leq& \sqrt{\frac{1}{N^2}\sum_{i,j=1}^{N}\left[\frac{1}{T}\sum_{t=1}^{T}\left(e_{it}e_{jt}-E(e_{it}e_{jt})\right)\right]^2} + \frac{1}{\sqrt{N}}|E(\mathbf{e}_t\mathbf{e}_t^\top)|_2\\
&=& O_P(1/\sqrt{T}) + O_P(1/\sqrt{N})
\end{eqnarray*}
as $N^{-2}\sum_{i,j=1}^{N}E\left[\frac{1}{T}\sum_{t=1}^{T}\left(e_{it}e_{jt}-E(e_{it}e_{jt})\right)\right]^2 = O(1/T)$.

Combing the above analyses, we have proved part (1).

\smallskip

\noindent (2). We first consider the term $N^{-1}\pmb{\Lambda}_{0}^\top(\widehat{\pmb{\Lambda}} - \pmb{\Lambda}_{0}\mathbf{H})$, which can be decomposed into eight terms. Similar to the proof of part (1), it is easy to show the first five terms are each $O_P(|\widehat{\pmb{\beta}}_J-\pmb{\beta}_{0,J}|_F)$. Next, consider the last three terms. For the sixth term, by using Lemma \ref{L.B0} (3), we have
\begin{eqnarray*}
&&\left|\frac{1}{N}\pmb{\Lambda}_{0}^\top\pmb{\Lambda}_{0}\frac{1}{NT}\sum_{t=1}^{T}\mathbf{f}_{0t}\mathbf{e}_t^\top\widehat{\pmb{\Lambda}}\right|_F\\
&\leq&O_P(1)\left|\frac{1}{T\sqrt{N}}\sum_{t=1}^{T}\mathbf{f}_{0t}\mathbf{e}_t^\top\right|_FN^{-1/2}\left|\widehat{\pmb{\Lambda}}-\pmb{\Lambda}_{0}\mathbf{H}\right|_F + O_P(1)\left|\frac{1}{NT}\sum_{i=1}^{N}\sum_{t=1}^{T}\pmb{\lambda}_{0i}\mathbf{f}_{0t}^\top e_{it} \right|_F\\
&=&O_P(1/\sqrt{T}) O_P(|\widehat{\pmb{\beta}}_J-\pmb{\beta}_{0,J}|_F+\delta_{NT}^{-1}) +O_P(1/\sqrt{NT})\\
&=&O_P(\delta_{NT}^{-1}|\widehat{\pmb{\beta}}_J-\pmb{\beta}_{0,J}|_F+\delta_{NT}^{-2}).
\end{eqnarray*}
Similarly, for the seventh term, we have
$$
\left|\frac{1}{NT}\sum_{t=1}^{T}\pmb{\Lambda}_{0}^\top\mathbf{f}_{0t}\mathbf{e}_t^\top\frac{1}{N}\pmb{\Lambda}_{0}^\top\widehat{\pmb{\Lambda}}\right|_F
\leq O_P(1)\left|\frac{1}{NT}\sum_{i=1}^{N}\sum_{t=1}^{T}\pmb{\lambda}_{0i}\mathbf{f}_{0t}^\top e_{it} \right|_F=O_P(1/\sqrt{NT}).
$$
For the last term, by using Cauchy-Schwarz inequality and using Lemma \ref{L.B0} (5), we have
\begin{eqnarray*}
&&\left|\frac{1}{N^2T}\sum_{t=1}^{T}\pmb{\Lambda}_{0}^\top\mathbf{e}_t\mathbf{e}_t^\top \widehat{\pmb{\Lambda}}\right|_F\\
&\leq&|\mathbf{H}|_F\frac{1}{T}\sum_{t=1}^{T}\left|N^{-1}\pmb{\Lambda}_{0}^\top\mathbf{e}_t\right|_F^2 + \frac{1}{T}\sum_{t=1}^{T}\left|N^{-1}\pmb{\Lambda}_{0}^\top\mathbf{e}_t\right|_F|N^{-1/2}\mathbf{e}_t|_F N^{-1/2}|\widehat{\pmb{\Lambda}} - \pmb{\Lambda}_{0}\mathbf{H}|_F\\
&\leq&O_P(1/N) + \left(\frac{1}{T}\sum_{t=1}^{T}\left|N^{-1}\pmb{\Lambda}_{0}^\top\mathbf{e}_t\right|_F^2\right)^{1/2}\left(\frac{1}{T}\sum_{t=1}^{T}|N^{-1/2}\mathbf{e}_t|_F^2\right)^{1/2} N^{-1/2}|\widehat{\pmb{\Lambda}} - \pmb{\Lambda}_{0}\mathbf{H}|_F\\
&=&O_P(|\widehat{\pmb{\beta}}_J-\pmb{\beta}_{0,J}|_F\delta_{NT}^{-1}+\delta_{NT}^{-2}).
\end{eqnarray*}

For the term $N^{-1}\widehat{\pmb{\Lambda}}^\top(\widehat{\pmb{\Lambda}} - \pmb{\Lambda}_{0}\mathbf{H}) $, which is bounded by
$$
\left|N^{-1/2}(\widehat{\pmb{\Lambda}} - \pmb{\Lambda}_{0}\mathbf{H})\right|_F^2 + |\mathbf{H}|_F\left|N^{-1}\pmb{\Lambda}_{0}^\top(\widehat{\pmb{\Lambda}} - \pmb{\Lambda}_{0}\mathbf{H}) \right|_F = O_P(|\widehat{\pmb{\beta}}_J-\pmb{\beta}_{0,J}|_F+\delta_{NT}^{-2}).
$$
The proof is now complete.

\smallskip

\noindent (3). By using $\widehat{\pmb{\Lambda}}^\top\widehat{\pmb{\Lambda}}/N = \mathbf{I}_r$, we have
$$
\mathbf{I}_r=\widehat{\pmb{\Lambda}}^\top\widehat{\pmb{\Lambda}}/N = \frac{1}{N}\widehat{\pmb{\Lambda}}^\top(\widehat{\pmb{\Lambda}} - \pmb{\Lambda}_{0}\mathbf{H}) + \frac{1}{N}(\widehat{\pmb{\Lambda}} - \pmb{\Lambda}_{0}\mathbf{H})^\top\widehat{\pmb{\Lambda}} + \frac{1}{N}(\widehat{\pmb{\Lambda}} - \pmb{\Lambda}_{0}\mathbf{H})^\top(\widehat{\pmb{\Lambda}} - \pmb{\Lambda}_{0}\mathbf{H})+ \frac{1}{N}\mathbf{H}^\top\pmb{\Lambda}_{0}^\top \pmb{\Lambda}_{0}\mathbf{H},
$$
which implies that $\frac{1}{N}\pmb{\Lambda}_{0}^\top \pmb{\Lambda}_{0}/N = \mathbf{H}^{-1,\top}\mathbf{H}^{-1} + O_P(|\widehat{\pmb{\beta}}_J-\pmb{\beta}_{0,J}|_F +\delta_{NT}^{-1})$.

The proof is now complete.

\smallskip

\noindent (4). By using parts (1)-(3), we have
\begin{eqnarray*}
\mathbf{P}_{\widehat{\pmb{\Lambda}}} - \mathbf{P}_{\pmb{\Lambda}_{0}} &=& \frac{1}{N}(\widehat{\pmb{\Lambda}} - \pmb{\Lambda}_{0}\mathbf{H})(\widehat{\pmb{\Lambda}} - \pmb{\Lambda}_{0}\mathbf{H})^\top + \frac{1}{N} \pmb{\Lambda}_{0}\mathbf{H}(\widehat{\pmb{\Lambda}} - \pmb{\Lambda}_{0}\mathbf{H})^\top + \frac{1}{N}(\widehat{\pmb{\Lambda}} - \pmb{\Lambda}_{0}\mathbf{H})\mathbf{H}^\top\pmb{\Lambda}_{0}^\top\\
&& + \frac{\pmb{\Lambda}_{0}}{\sqrt{N}}\left(\mathbf{H}\mathbf{H}^\top - (\pmb{\Lambda}_{0}^\top \pmb{\Lambda}_{0}/N)^{-1}  \right) \frac{\pmb{\Lambda}_{0}^\top}{\sqrt{N}} = O_P(|\widehat{\pmb{\beta}}_J-\pmb{\beta}_{0,J}|_F +\delta_{NT}^{-1}).
\end{eqnarray*}

\smallskip

\noindent (5). Note that $N^{-1}\mathbf{X}_{J,t}^\top(\widehat{\pmb{\Lambda}} - \pmb{\Lambda}_{0}\mathbf{H}) = N^{-1}\mathbf{X}_{J,t}^\top(\mathbf{K}_1+\cdots+ \mathbf{K}_8)\mathbf{V}_{NT}^{-1}$. For $N^{-1}\mathbf{X}_{J,t}^\top\mathbf{K}_1$, we have
\begin{eqnarray*}
|N^{-1}\mathbf{X}_{J,t}^\top\mathbf{K}_1|_F&\leq& \sqrt{r}|N^{-1/2}\mathbf{X}_{J,t}|_F|\widehat{\pmb{\beta}}_J-\pmb{\beta}_{0,J}|_F^2 \psi_{\max}\left(\frac{1}{NT}\sum_{i=1}^{N}\sum_{t=1}^{T}\mathbf{x}_{J,it}\mathbf{x}_{J,it}^\top\right)\\
& =& O_P(\sqrt{s}|\widehat{\pmb{\beta}}_J-\pmb{\beta}_{0,J}|_F^2).
\end{eqnarray*}
For $N^{-1}\mathbf{X}_{J,t}^\top\mathbf{K}_2$, we have
\begin{eqnarray*}
&& |N^{-1}\mathbf{X}_{J,t}^\top\mathbf{K}_2|_F \nonumber \\
&\leq& |N^{-1/2}\mathbf{X}_{J,t}|_F |\pmb{\Lambda}_{0}^\top\widehat{\pmb{\Lambda}}/N|_F\left(\frac{1}{T\sqrt{N}}\sum_{t=1}^{T}|\mathbf{X}_{J,t}(\widehat{\pmb{\beta}}_J-\pmb{\beta}_{0,J})|_F^2\right)^{1/2}\left(\frac{1}{T\sqrt{N}}\sum_{t=1}^{T}|\mathbf{f}_{0t}|_F^2\right)^{1/2}\\
&=& O_P(\sqrt{s}|\widehat{\pmb{\beta}}_J-\pmb{\beta}_{0,J}|_F).
\end{eqnarray*}
Similarly, we can show that $N^{-1}\mathbf{X}_{J,t}^\top\mathbf{K}_3,\ldots,N^{-1}\mathbf{X}_{J,t}^\top\mathbf{K}_5$ are all $O_P(\sqrt{s}|\widehat{\pmb{\beta}}_J-\pmb{\beta}_{0,J}|_F)$. 

Consider the last three terms. For $N^{-1}\mathbf{X}_{J,t}^\top\mathbf{K}_6$, we have
\begin{eqnarray*}
|N^{-1}\mathbf{X}_{J,t}^\top\mathbf{K}_6|_F &\leq& |\mathbf{X}_{J,t}^\top\pmb{\Lambda}_{0}/N|_F \left|\frac{1}{NT}\sum_{i=1}^{N}\sum_{t=1}^{T}e_{it}\pmb{\lambda}_{0i}\mathbf{f}_{0t}^\top\right|_F|\mathbf{H}|_F \\
&& + |\mathbf{X}_{J,t}^\top\pmb{\Lambda}_{0}/N|_F \left|\frac{1}{T\sqrt{N}}\sum_{t=1}^{T}\mathbf{f}_{0t}\mathbf{e}_{t}^\top\right|_FN^{-1/2}|\widehat{\pmb{\Lambda}}-\pmb{\Lambda}_{0}\mathbf{H}|_F\\
&=& O_P(\sqrt{s/(NT)}) + O_P(\sqrt{s/T})O_P(|\widehat{\pmb{\beta}}_J-\pmb{\beta}_{0,J}|_F+\delta_{NT}^{-1}).
\end{eqnarray*}
For $N^{-1}\mathbf{X}_{J,t}^\top\mathbf{K}_7$, we have
\begin{eqnarray*}
N^{-1}\mathbf{X}_{J,t}^\top\mathbf{K}_7&\leq&  \left|\frac{1}{NT}\sum_{i=1}^{N}\sum_{s=1}^{T}\mathbf{x}_{J,it}e_{is}\mathbf{f}_{0s}^\top\right|_F  |\pmb{\Lambda}_{0}^\top\widehat{\pmb{\Lambda}}/N|_F\\
&=& O_P(\sqrt{s/(NT)}).
\end{eqnarray*}
For $N^{-1}\mathbf{X}_{J,t}^\top\mathbf{K}_8$, we have
\begin{eqnarray*}
N^{-1}\mathbf{X}_{J,t}^\top\mathbf{K}_8&\leq&  \frac{1}{T}\sum_{s=1}^{T}\left|\frac{1}{N}\sum_{i=1}^{N}\mathbf{x}_{J,it}e_{is}\right|_F \left|\frac{1}{N}\sum_{i=1}^{N}e_{is}\pmb{\lambda}_{0i} \right|_F |\mathbf{H}|_F\\
&& +|N^{-1/2}\mathbf{X}_{J,t}|_F \left|\frac{1}{NT}\sum_{t=1}^{T}\mathbf{e}_{t}\mathbf{e}_t^\top\right|_F N^{-1/2}|\widehat{\pmb{\Lambda}}-\pmb{\Lambda}_{0}\mathbf{H}|_F\\
&=& O_P(\sqrt{s}/N) + O_P(\sqrt{s}\delta_{NT}^{-1}|\widehat{\pmb{\beta}}_J-\pmb{\beta}_{0,J}|_F+\sqrt{s}\delta_{NT}^{-2}).
\end{eqnarray*}
The proof is now complete.

\smallskip

\noindent (6). Write
$$
\frac{1}{NT}\sum_{t=1}^{T}\mathbf{X}_{J,t}^\top\mathbf{M}_{\widehat{\pmb{\Lambda}}}(\widehat{\pmb{\Lambda}} - \pmb{\Lambda}_{0}\mathbf{H}) = \frac{1}{T}\sum_{t=1}^{T}\frac{1}{N}\mathbf{X}_{J,t}^\top(\widehat{\pmb{\Lambda}} - \pmb{\Lambda}_{0}\mathbf{H}) + \frac{1}{T}\sum_{t=1}^{T}\frac{1}{N}\mathbf{X}_{J,t}^\top\frac{1}{N}\widehat{\pmb{\Lambda}}\widehat{\pmb{\Lambda}}^\top(\widehat{\pmb{\Lambda}} - \pmb{\Lambda}_{0}\mathbf{H}).
$$
The first term on the right is an average of (3) over $t$ and thus it can be show that its order is of same magnitude. For the second term,
\begin{eqnarray*}
&&\left|\frac{1}{T}\sum_{t=1}^{T}\frac{1}{N}\mathbf{X}_{J,t}^\top\frac{1}{N}\widehat{\pmb{\Lambda}}\widehat{\pmb{\Lambda}}^\top(\widehat{\pmb{\Lambda}} - \pmb{\Lambda}_{0}\mathbf{H})\right|_F\\
&\leq&\frac{1}{T}\sum_{t=1}^{T}|\mathbf{X}_{J,t}/\sqrt{N}|_F|\widehat{\pmb{\Lambda}}/\sqrt{T}|_F|\frac{1}{N}\widehat{\pmb{\Lambda}}^\top(\widehat{\pmb{\Lambda}} - \pmb{\Lambda}_{0}\mathbf{H})|_F = O_P(\sqrt{s}|\widehat{\pmb{\beta}}_J-\pmb{\beta}_{0,J}|_F + \sqrt{s}\delta_{NT}^2).
\end{eqnarray*}
The proof is now complete.

\smallskip

\noindent (7). Note that $N^{-1}\mathbf{e}_{t}^\top(\widehat{\pmb{\Lambda}} - \pmb{\Lambda}_{0}\mathbf{H}) = N^{-1}\mathbf{e}_{t}^\top(\mathbf{K}_1+\cdots+ \mathbf{K}_8)\mathbf{V}_{NT}^{-1}$, we next consider these eight terms respectively. For $N^{-1}\mathbf{e}_{t}^\top\mathbf{K}_1$, by Cauchy-Schwarz inequality,
\begin{eqnarray*}
&&\left|N^{-1}\mathbf{e}_{t}^\top\mathbf{K}_1\right|_F \nonumber \\
&\leq&|N^{-1/2}\widehat{\pmb{\Lambda}}|_F|\widehat{\pmb{\beta}}_J - \pmb{\beta}_{0,J}|_F \left(\frac{1}{T}\sum_{t=1}^{T}|N^{-1}\mathbf{e}_{t}^\top\mathbf{X}_{J,t}|^2\right)^{1/2}\left(\frac{1}{NT}\sum_{t=1}^{T}|\mathbf{X}_{J,t}(\widehat{\pmb{\beta}}_J - \pmb{\beta}_{0,J})|^2\right)^{1/2}\\
&=&\sqrt{r}|\widehat{\pmb{\beta}}_J - \pmb{\beta}_{0,J}|_FO_P(\sqrt{s/N})O_P(|\widehat{\pmb{\beta}}_J - \pmb{\beta}_{0,J}|_F) = O_P(\sqrt{s/N}|\widehat{\pmb{\beta}}_J - \pmb{\beta}_{0,J}|_F^2).
\end{eqnarray*}
For $N^{-1}\mathbf{e}_{t}^\top\mathbf{K}_2$, by using Cauchy-Schwarz inequality,
\begin{eqnarray*}
\left|N^{-1}\mathbf{e}_{t}^\top\mathbf{K}_2\right|_F &\leq& |\pmb{\Lambda}_{0}^\top\widehat{\pmb{\Lambda}}/N|_F|\widehat{\pmb{\beta}}_J - \pmb{\beta}_{0,J}|_F\frac{1}{T}\sum_{s=1}^{T}\left|\frac{1}{N}\sum_{i=1}^{N}\mathbf{x}_{J,is}e_{it}\right|_F|\mathbf{f}_{0s}|_F\\
&\leq& O_P\left(|\widehat{\pmb{\beta}}_J - \pmb{\beta}_{0,J}|_F\right)\left\{\frac{1}{T}\sum_{s=1}^{T}\left|\frac{1}{N}\sum_{i=1}^{N}\mathbf{x}_{J,is}e_{it}\right|_F^2\right\}^{1/2}\left\{\frac{1}{T}\sum_{s=1}^{T}|\mathbf{f}_{0s}|_F^2\right\}^{1/2}\\
&=&O_P\left(\sqrt{s/N}|\widehat{\pmb{\beta}}_J - \pmb{\beta}_{0,J}|_F\right).
\end{eqnarray*}
Similarly, we can show that $N^{-1}\mathbf{e}_{t}^\top\mathbf{K}_3,\ldots,N^{-1}\mathbf{e}_{t}^\top\mathbf{K}_5$ are both $O_P(\sqrt{s/N} |\widehat{\pmb{\beta}}_J - \pmb{\beta}_{0,J}|_F)$. For $\left|N^{-1}\mathbf{e}_{t}^\top\mathbf{K}_6\right|_F$, we have
\begin{eqnarray*}
\left|N^{-1}\mathbf{e}_{t}^\top\mathbf{K}_6\right|_F&\leq& |\mathbf{e}_t^\top\pmb{\Lambda}_{0}/N|_F \left(\left|\frac{1}{NT}\sum_{t=1}^{T}\mathbf{f}_{0t}\mathbf{e}_t^\top \pmb{\Lambda}_{0}\right|_F|\mathbf{H}|_F + \left|\frac{1}{NT}\sum_{t=1}^{T}\mathbf{f}_{0t}\mathbf{e}_t^\top \right|_F|\widehat{\pmb{\Lambda}}-\pmb{\Lambda}_{0}\mathbf{H}|_F\right)\\
&=&O_P(1/\sqrt{N})\left(O_P(1/\sqrt{T}) + O_P(1/\sqrt{T})O_P(|\widehat{\pmb{\beta}}_J-\pmb{\beta}_{0,J}|_F+\delta_{NT}^{-1}) \right)\\
&=&O_P(\widehat{\pmb{\beta}}_J-\pmb{\beta}_{0,J}|_F\delta_{NT}^{-1}+\delta_{NT}^{-2}).
\end{eqnarray*}
Similarly, we can show the next two terms are both $O_P(|\widehat{\pmb{\beta}}_J-\pmb{\beta}_{0,J}|_F\delta_{NT}^{-1}+\delta_{NT}^{-2})$.

The proof is now complete.

\smallskip

\noindent (8). Note that $N^{-1}T^{-1/2}\sum_{t=1}^{T}\mathbf{e}_{t}^\top(\widehat{\pmb{\Lambda}} - \pmb{\Lambda}_{0}\mathbf{H}) = N^{-1}T^{-1/2}\sum_{t=1}^{T}\mathbf{e}_{t}^\top(\mathbf{K}_1+\cdots+ \mathbf{K}_8)\widehat{\mathbf{V}}_{NT}^{-1}$.

For the first term, by using Cauchy-Schwarz inequality, 
\begin{eqnarray*}
&& |N^{-1}T^{-1/2}\sum_{t=1}^{T}\mathbf{e}_{t}^\top\mathbf{K}_1|_F\\
&\leq&\frac{1}{NT}\sum_{s=1}^{T}\left|(NT)^{-1/2}\sum_{t=1}^{T}\mathbf{e}_{t}^\top \mathbf{X}_{J,s}\right|_F |\widehat{\pmb{\beta}}_J - \pmb{\beta}_{0,J}|_F|\mathbf{X}_{J,s}(\widehat{\pmb{\beta}}_J - \pmb{\beta}_{0,J})|_F|N^{-1/2}\widehat{\pmb{\Lambda}}|_F\\
&\leq& \sqrt{r}|\widehat{\pmb{\beta}}_J - \pmb{\beta}_{0,J}|_F\left\{\frac{1}{NT}\sum_{s=1}^{T}\left|(NT)^{-1/2}\sum_{i=1}^{N}\sum_{t=1}^{T} \mathbf{x}_{J,is}e_{it}\right|_F^2\right\}^{1/2}\left\{\frac{1}{NT}\sum_{s=1}^{T}\left|\mathbf{X}_{J,s}(\widehat{\pmb{\beta}}_J - \pmb{\beta}_{0,J})\right|_F^2\right\}^{1/2}\\
&=&\sqrt{r}|\widehat{\pmb{\beta}}_J - \pmb{\beta}_{0,J}|_FO_P(\sqrt{s/N})O_P(|\widehat{\pmb{\beta}}_J - \pmb{\beta}_{0,J}|_F)=O_P(\sqrt{s/N}|\widehat{\pmb{\beta}}_J - \pmb{\beta}_{0,J}|_F^2).
\end{eqnarray*}
For the second term, we have
\begin{eqnarray*}
|N^{-1}T^{-1/2}\sum_{t=1}^{T}\mathbf{e}_{t}^\top\mathbf{K}_2|_F
&\leq& N^{-1/2}|\widehat{\pmb{\beta}}_J - \pmb{\beta}_{0,J}|_F\frac{1}{T}\sum_{t=1}^{T}\left|(NT)^{-1/2}\sum_{i=1}^{N}\sum_{s=1}^{T} \mathbf{x}_{J,it}e_{is}\right|_F |\mathbf{f}_{0t}|_F |\pmb{\Lambda}_{0}^\top\widehat{\pmb{\Lambda}}/N|_F\\
&=& O_P(\sqrt{s/N}|\widehat{\pmb{\beta}}_J - \pmb{\beta}_{0,J}|_F).
\end{eqnarray*}
Similarly, we have $N^{-1}T^{-1/2}\sum_{t=1}^{T}\mathbf{e}_{t}^\top\mathbf{K}_3=O_P(\sqrt{s/N}|\widehat{\pmb{\beta}}_J - \pmb{\beta}_{0,J}|_F)$ and $N^{-1}T^{-1/2}\sum_{t=1}^{T}\mathbf{e}_{t}^\top\mathbf{K}_4=O_P(\sqrt{s/N}|\widehat{\pmb{\beta}}_J - \pmb{\beta}_{0,J}|_F)$. 

For the fifth term, by using Cauchy-Schwarz inequality, 
\begin{eqnarray*}
&& |N^{-1}T^{-1/2}\sum_{t=1}^{T}\mathbf{e}_{t}^\top\mathbf{K}_5|_F\\
&\leq&\frac{1}{\sqrt{T}} \frac{1}{N}\sum_{i=1}^{N} \left|\frac{1}{\sqrt{T}}\sum_{t=1}^{T}e_{it}\right|_F\left|\frac{1}{\sqrt{T}}\sum_{s=1}^{T}e_{is}\frac{1}{N}\sum_{k=1}^{N}\mathbf{x}_{J,ks}\pmb{\lambda}_{0k}^\top\right|_F|\mathbf{H}|_F |\widehat{\pmb{\beta}}_J-\pmb{\beta}|_F\\
&&+\frac{1}{\sqrt{T}}\frac{1}{N}\sum_{i=1}^{N} \left|\frac{1}{\sqrt{T}}\sum_{t=1}^{T}e_{it}\right|_F\left|\frac{1}{N}\sum_{k=1}^{N}\mathbf{x}_{J,ks}\left(\frac{1}{\sqrt{T}}\sum_{s=1}^{T}e_{is}\right)(\widehat{\pmb{\lambda}}_k-\mathbf{H}^\top\pmb{\lambda}_{0k})^\top\right|_F |\widehat{\pmb{\beta}}_J-\pmb{\beta}|_F\\
&\leq&O_P(|\widehat{\pmb{\beta}}_J-\pmb{\beta}|_F)\frac{1}{\sqrt{T}} \left\{\frac{1}{N}\sum_{i=1}^{N} \left|\frac{1}{\sqrt{T}}\sum_{t=1}^{T}e_{it}\right|_F^2\right\}^{1/2}\left\{\frac{1}{N}\sum_{i=1}^{N}\left|\frac{1}{\sqrt{T}}\sum_{s=1}^{T}e_{is}\frac{1}{N}\sum_{k=1}^{N}\mathbf{x}_{J,ks}\pmb{\lambda}_{0k}^\top\right|_F^2\right\}^{1/2}\\
&&+O_P(|\widehat{\pmb{\beta}}_J-\pmb{\beta}|_F)\frac{1}{\sqrt{T}}\left\{\frac{1}{N}\sum_{i=1}^{N} \left|\frac{1}{\sqrt{T}}\sum_{t=1}^{T}e_{it}\right|_F^2\right\}^{1/2}\left\{\frac{1}{N}\sum_{k=1}^{N}\frac{1}{N}\sum_{i=1}^{N}\left|\frac{1}{\sqrt{T}}\sum_{s=1}^{T}\mathbf{x}_{J,ks}e_{is}\right|_F^2\right\}^{1/2}\\
&&\times\left\{\frac{1}{N}\sum_{k=1}^{N}|\widehat{\pmb{\lambda}}_k-\mathbf{H}^\top\pmb{\lambda}_{0k}|_F^2\right\}^{1/2}\\
&=& O_P(\sqrt{s/T}|\widehat{\pmb{\beta}}_J-\pmb{\beta}|_F).
\end{eqnarray*}
For the sixth term,
\begin{eqnarray*}
&& |N^{-1}T^{-1/2}\sum_{t=1}^{T}\mathbf{e}_{t}^\top\mathbf{K}_6|_F\\
&\leq& \left|\frac{1}{N\sqrt{T}}\left(\frac{1}{\sqrt{NT}}\sum_{i=1}^{N}\sum_{t=1}^{T}\pmb{\lambda}_{0i}^\top e_{it}\right)\left(\frac{1}{\sqrt{NT}}\sum_{i=1}^{N}\sum_{s=1}^{T}\pmb{\lambda}_{0i}\mathbf{f}_{0s}^\top e_{is} \right) \right|_F|\mathbf{H}|_F\\
&&+\frac{1}{\sqrt{N}}\left|\frac{1}{\sqrt{NT}}\sum_{i=1}^{N}\sum_{t=1}^{T}\pmb{\lambda}_{0i}^\top e_{it}\right|_F \left| \frac{1}{T\sqrt{N}}\sum_{s=1}^{T}\mathbf{f}_{0s} \mathbf{e}_{s}^\top\right|_F N^{-1/2}|\widehat{\pmb{\Lambda}}-\pmb{\Lambda}_{0}\mathbf{H}|_F\\
&=&O_P(N^{-1}T^{-1/2}+N^{-1/2}|\widehat{\pmb{\beta}}_J-\pmb{\beta}_{0,J}|_F+\delta_{NT}^{-2}).
\end{eqnarray*}
For the seventh term,
\begin{eqnarray*}
&&|N^{-1}T^{-1/2}\sum_{t=1}^{T}\mathbf{e}_{t}^\top\mathbf{K}_7|_F\\
&\leq& T^{-1/2} \frac{1}{N}\sum_{i=1}^{N}\left|\frac{1}{\sqrt{T}}\sum_{t=1}^{T}e_{it}\right|_F\left|\frac{1}{\sqrt{T}}\sum_{s=1}^{T}e_{is}\mathbf{f}_{0s}\right|_F |\pmb{\Lambda}_{0}^\top\widehat{\pmb{\Lambda}}/N|_F=O_P(1/\sqrt{T}).
\end{eqnarray*}
For the last term, by using Cauchy-Schwarz inequality, we have
\begin{eqnarray*}
&&\left|N^{-1}T^{-1/2}\sum_{t=1}^{T}\mathbf{e}_{t}^\top\frac{1}{NT}\sum_{s=1}^{T}\mathbf{e}_s\mathbf{e}_s^\top\widehat{\pmb{\Lambda}}\right|_F\\
&\leq&\frac{1}{NT}\sum_{s=1}^{T}\left|\frac{1}{\sqrt{NT}}\sum_{i=1}^{N}\sum_{t=1}^{T}(e_{it}e_{is}-E(e_{it}e_{is}))\right|_F\left|\frac{1}{\sqrt{N}}\sum_{i=1}^{N}e_{is}\pmb{\lambda}_{0i} \right|_F|\mathbf{H}|_F\\
&&+\frac{1}{\sqrt{NT}}\frac{1}{T}\sum_{s=1}^{T}\left|\frac{1}{N}\sum_{i=1}^{N}\sum_{t=1}^{T}E(e_{it}e_{is})\right|_F\left|\frac{1}{\sqrt{N}}\sum_{i=1}^{N}e_{is}\pmb{\lambda}_{0i} \right|_F|\mathbf{H}|_F\\
&&+N^{-1/2}\left(\frac{1}{T}\sum_{s=1}^{T}\left|\frac{1}{\sqrt{NT}}\sum_{i=1}^{N}\sum_{t=1}^{T}(e_{it}e_{is}-E(e_{it}e_{is}))\right|_F^2\right)^{1/2}\left(\frac{1}{T}\sum_{t=1}^{T}|\mathbf{e}_{t}|_F^2/N \right)^{1/2}N^{-1/2}|\widehat{\pmb{\Lambda}}-\pmb{\Lambda}_{0}\mathbf{H}|_F\\
&&+\frac{1}{\sqrt{T}}\frac{1}{T}\sum_{s=1}^{T}\left|\frac{1}{N}\sum_{i=1}^{N}\sum_{t=1}^{T}E(e_{it}e_{is})\right|_F\left|\mathbf{e}_{s}/\sqrt{N} \right|_FN^{-1/2}|\widehat{\pmb{\Lambda}}-\pmb{\Lambda}_{0}\mathbf{H}|_F\\
&=&O_P(1/N) + O_P(1/\sqrt{NT})+O_P(N^{-1/2}|\widehat{\pmb{\beta}}_J-\pmb{\beta}_{0,J}|_F+\delta_{NT}^{-2})+O_P(T^{-1/2}|\widehat{\pmb{\beta}}_J-\pmb{\beta}_{0,J}|_F+\delta_{NT}^{-2}).
\end{eqnarray*}

Combining the above analysis, the proof is now complete.

\smallskip

\noindent (9). The proof of part (9) is similar to that of part (8). The details are omitted.

\smallskip

\noindent (10). The proof of part (10) is similar to that of part (8). The details are omitted.
\end{proof}

\begin{proof}[Proof of Lemma \ref{L.B2}]
\item 
\noindent (1). Write
\begin{eqnarray*}
&&\frac{1}{N^2T^2}\sum_{t=1}^{T}\sum_{s=1}^{T}\mathbf{X}_{J,t}^\top\mathbf{M}_{\widehat{\pmb{\Lambda}}}(\mathbf{e}_s\mathbf{e}_s^\top-\pmb{\Omega}_e)\widehat{\pmb{\Lambda}}\mathbf{G}\mathbf{f}_{0t}\\
&=&\frac{1}{N^2T^2}\sum_{t=1}^{T}\sum_{s=1}^{T}\mathbf{X}_{J,t}^\top(\mathbf{e}_s\mathbf{e}_s^\top-\pmb{\Omega}_e)\pmb{\Lambda}_{0}\mathbf{H}\mathbf{G}\mathbf{f}_{0t} + \frac{1}{N^2T^2}\sum_{t=1}^{T}\sum_{s=1}^{T}\mathbf{X}_{J,t}^\top(\mathbf{e}_s\mathbf{e}_s^\top-\pmb{\Omega}_e)(\widehat{\pmb{\Lambda}}-\pmb{\Lambda}_{0}\mathbf{H})\mathbf{G}\mathbf{f}_{0t}\\
&& - \frac{1}{T}\sum_{t=1}^{T}\mathbf{X}_{J,t}^\top\widehat{\pmb{\Lambda}}/N\frac{1}{N^2T}\sum_{s=1}^{T}\widehat{\pmb{\Lambda}}^\top(\mathbf{e}_s\mathbf{e}_s^\top-\pmb{\Omega}_e)\widehat{\pmb{\Lambda}}\mathbf{G}\mathbf{f}_{0t}\\
&\eqqcolon &\mathbf{K}_{10} + \mathbf{K}_{11} + \mathbf{K}_{12}.
\end{eqnarray*}

For $\mathbf{K}_{10}$, we have
\begin{eqnarray*}
\mathbf{K}_{10}&=&\frac{1}{N^2T^2}\sum_{i=1}^{N}\sum_{j=1}^{N}\sum_{t=1}^{T}\sum_{s=1}^{T}\mathbf{x}_{J,it}({e}_{is}{e}_{js}-E({e}_{is}{e}_{js}))\pmb{\lambda}_{0j}^\top\mathbf{H}\mathbf{G}\mathbf{f}_{0t}\\
&=& \frac{1}{N\sqrt{T}}\frac{1}{T}\sum_{t=1}^{T}\left\{\frac{1}{\sqrt{T}}\sum_{s=1}^{T}\frac{1}{N}\sum_{i=1}^{N}\sum_{j=1}^{N}\mathbf{x}_{J,it}({e}_{is}{e}_{js}-E({e}_{is}{e}_{js}))\pmb{\lambda}_{0j}^\top\right\}\mathbf{H}\mathbf{G}\mathbf{f}_{0t}\\
&=& O_P(\sqrt{s}/(N\sqrt{T})).
\end{eqnarray*}

For $\mathbf{K}_{11}$, by using Cauchy-Schwarz inequality, we have
\begin{eqnarray*}
|\mathbf{K}_{11}|_F &=& \left|\frac{1}{\sqrt{NT}}\frac{1}{T}\sum_{t=1}^{T}\frac{1}{N}\sum_{j=1}^{N}\left\{\frac{1}{\sqrt{NT}}\sum_{s=1}^{T}\sum_{i=1}^{N}\mathbf{x}_{J,it}({e}_{is}{e}_{js}-E({e}_{is}{e}_{js}))\right\}(\widehat{\pmb{\lambda}}_j-\mathbf{H}^\top\pmb{\lambda}_{0j})^\top\mathbf{H}\mathbf{G}\mathbf{f}_{0t}\right|_F\\
&\leq&O_P(1/\sqrt{NT})\left\{\frac{1}{T}\sum_{t=1}^{T}\frac{1}{N}\sum_{j=1}^{N}\left|\frac{1}{\sqrt{NT}}\sum_{s=1}^{T}\sum_{i=1}^{N}\mathbf{x}_{J,it}({e}_{is}{e}_{js}-E({e}_{is}{e}_{js}))\right|_F^2\right\}^{1/2}\\
&&\times\left\{\frac{1}{N}\sum_{j=1}^{N}|\widehat{\pmb{\lambda}}_j-\mathbf{H}^\top\pmb{\lambda}_{0j}|_F^2\right\}^{1/2} \\
&=&O_P(\sqrt{s/NT}|\widehat{\pmb{\beta}}_J - \pmb{\beta}_{0,J}|_F + \sqrt{s/NT}\delta_{NT}^{-1}). 
\end{eqnarray*}

For $\mathbf{K}_{12}$, as $|\mathbf{X}_{J,t}/\sqrt{N}|_F=O_P(\sqrt{s})$ and $|\widehat{\pmb{\Lambda}}/\sqrt{N}|_F=\sqrt{r}$, we have
\begin{eqnarray*}
|\mathbf{K}_{12}|_F&\leq& \frac{1}{T}\sum_{t=1}^{T}\left|\mathbf{X}_{J,t}^\top\widehat{\pmb{\Lambda}}/N\right|_F\left|\mathbf{G}\mathbf{f}_{0t}\right|_F\left|\frac{1}{N^2T}\sum_{s=1}^{T}\widehat{\pmb{\Lambda}}^\top(\mathbf{e}_s\mathbf{e}_s^\top-\pmb{\Omega}_e)\widehat{\pmb{\Lambda}}\right|_F \\
&\leq&O_P(\sqrt{s})\left|\frac{1}{N^2T}\sum_{s=1}^{T}\widehat{\pmb{\Lambda}}^\top(\mathbf{e}_s\mathbf{e}_s^\top-\pmb{\Omega}_e)\widehat{\pmb{\Lambda}}\right|_F\\
&\leq&O_P(\sqrt{s})\left|\frac{1}{N^2T}\sum_{s=1}^{T}\pmb{\Lambda}_{0}^\top(\mathbf{e}_s\mathbf{e}_s^\top-\pmb{\Omega}_e)\pmb{\Lambda}_{0}\right|_F+O_P(\sqrt{s})\left|\frac{1}{N^2T}\sum_{s=1}^{T}\pmb{\Lambda}_{0}^\top(\mathbf{e}_s\mathbf{e}_s^\top-\pmb{\Omega}_e)(\widehat{\pmb{\Lambda}}-\pmb{\Lambda}_{0}\mathbf{H})\right|_F\\
&&+O_P(\sqrt{s})\left|\frac{1}{N^2T}\sum_{s=1}^{T}(\widehat{\pmb{\Lambda}}-\pmb{\Lambda}_{0}\mathbf{H})^\top(\mathbf{e}_s\mathbf{e}_s^\top-\pmb{\Omega}_e)(\widehat{\pmb{\Lambda}}-\pmb{\Lambda}_{0}\mathbf{H})\right|_F\\
&\eqqcolon &O_P(\sqrt{s})(\mathbf{K}_{12,1}+\mathbf{K}_{12,2}+\mathbf{K}_{12,3}).
\end{eqnarray*}
For $\mathbf{K}_{12,1}$, we have
\begin{eqnarray*}
&&\frac{1}{N^2T}\sum_{s=1}^{T}\pmb{\Lambda}_{0}^\top(\mathbf{e}_s\mathbf{e}_s^\top-\pmb{\Omega}_e)\pmb{\Lambda}_{0}\\
&=& \frac{1}{N\sqrt{T}}\frac{1}{N\sqrt{T}}\sum_{s=1}^{T}\sum_{k=1}^{N}\sum_{i=1}^{N}\pmb{\lambda}_{0k}\pmb{\lambda}_{0i}^\top({e}_{is}{e}_{ks}-E({e}_{is}{e}_{ks})) =O_P\left(\frac{1}{N\sqrt{T}}\right).
\end{eqnarray*}
For $\mathbf{K}_{12,2}$, by using Cauchy-Schwarz inequality, we have
\begin{eqnarray*}
|\mathbf{K}_{12,2}|_F&=& \left|\frac{1}{\sqrt{NT}}\frac{1}{N}\sum_{j=1}^{N}\left\{\frac{1}{\sqrt{NT}}\sum_{i=1}^{N}\sum_{s=1}^{T}\pmb{\lambda}_{0i}({e}_{is}{e}_{js}-E({e}_{is}{e}_{js}))\right\}  (\widehat{\pmb{\lambda}}_j-\mathbf{H}^\top\pmb{\lambda}_{0j})^\top \right|_F \\
&\leq&\frac{1}{\sqrt{NT}}\left\{ \frac{1}{N}\sum_{j=1}^{N}\left|\frac{1}{\sqrt{NT}}\sum_{i=1}^{N}\sum_{s=1}^{T}\pmb{\lambda}_{0i}({e}_{is}{e}_{js}-E({e}_{is}{e}_{js}))\right|_F^2 \right\}^{1/2} \left\{ \frac{1}{N}\sum_{j=1}^{N}\left|\widehat{\pmb{\lambda}}_j-\mathbf{H}^\top\pmb{\lambda}_{0j}\right|_F^2\right\}^{1/2} \\
&=& \frac{1}{\sqrt{NT}}O_P(|\widehat{\pmb{\beta}}_J - \pmb{\beta}_{0,J}|_F+\delta_{NT}^{-1}).
\end{eqnarray*}
For $\mathbf{K}_{12,3}$, by using Cauchy-Schwarz inequality, we have
\begin{eqnarray*}
|\mathbf{K}_{12,3}|_F &=& \left|\frac{1}{\sqrt{T}}\frac{1}{N^2}\sum_{i=1}^{N}\sum_{j=1}^{N}(\widehat{\pmb{\lambda}}_j-\mathbf{H}^\top\pmb{\lambda}_{0j})(\widehat{\pmb{\lambda}}_i-\mathbf{H}^\top\pmb{\lambda}_{0i})^\top\left\{\frac{1}{\sqrt{T}}\sum_{t=1}^{T}({e}_{it}{e}_{jt}-E({e}_{it}{e}_{jt}))\right\}\right|_F\\
&\leq& \frac{1}{\sqrt{T}}\left(\frac{1}{N}\sum_{j=1}^{N}|\widehat{\pmb{\lambda}}_j-\mathbf{H}^\top\pmb{\lambda}_{0j}|_F^2\right)\left\{ \frac{1}{N^2}\sum_{i=1}^{N}\sum_{j=1}^{N}\left|\frac{1}{\sqrt{T}}\sum_{t=1}^{T}({e}_{it}{e}_{jt}-E({e}_{it}{e}_{jt}))\right|_F^2\right\}^{1/2}\\
&=&\frac{1}{\sqrt{T}}O_P(|\widehat{\pmb{\beta}}_J - \pmb{\beta}_{0,J}|_F^2+\delta_{NT}^{-2}).
\end{eqnarray*}
Combing the above analyses, we have proved part (1).

\smallskip

\noindent (2). Note that
$$
\left(\frac{1}{NT}\sum_{t=1}^{T}\mathbf{X}_{J,t}^\top \mathbf{M}_{\widehat{\pmb{\Lambda}}}\mathbf{X}_{J,t}\right)(\widehat{\pmb{\beta}}_J-\pmb{\beta}_{0,J}) = \frac{1}{NT}\sum_{t=1}^{T}\mathbf{X}_{J,t}^\top \mathbf{M}_{\widehat{\pmb{\Lambda}}}\pmb{\Lambda}_{0}\mathbf{f}_{0t} + \frac{1}{NT}\sum_{t=1}^{T}\mathbf{X}_{J,t}^\top \mathbf{M}_{\widehat{\pmb{\Lambda}}}\mathbf{e}_t.
$$
Consider $\frac{1}{NT}\sum_{t=1}^{T}\mathbf{X}_{J,t}^\top \mathbf{M}_{\widehat{\pmb{\Lambda}}}\pmb{\Lambda}_{0}\mathbf{f}_{0t}$, since $\mathbf{M}_{\widehat{\pmb{\Lambda}}}\widehat{\pmb{\Lambda}}\mathbf{H}^{-1} = \mathbf{0}$, we have
\begin{eqnarray*}
\frac{1}{NT}\sum_{t=1}^{T}\mathbf{X}_{J,t}^\top \mathbf{M}_{\widehat{\pmb{\Lambda}}}\pmb{\Lambda}_{0}\mathbf{f}_{0t}
&=& - \frac{1}{NT}\sum_{t=1}^{T}\mathbf{X}_{J,t}^\top \mathbf{M}_{\widehat{\pmb{\Lambda}}} \left(\mathbf{K}_1+\cdots+\mathbf{K}_8\right) \mathbf{G}\mathbf{f}_{0t}\\
&\eqqcolon & \mathbf{K}_{13} + \cdots \mathbf{K}_{20}.
\end{eqnarray*}
For $\mathbf{K}_{13}$, since $|\mathbf{X}_{J,t}^\top \mathbf{M}_{\widehat{\pmb{\Lambda}}}|_F\leq |\mathbf{M}_{\widehat{\pmb{\Lambda}}}|_2 |\mathbf{X}_{J,t}|_F=O_P\left(\sqrt{Ns}\right)$, we have $\mathbf{K}_{13} = O_P\left(\sqrt{s}|\widehat{\pmb{\beta}}_J - \pmb{\beta}_{0,J}|_F^2\right)$. For $\mathbf{K}_{14}$, we have
\begin{eqnarray*}
\mathbf{K}_{14}&=& \frac{1}{N} \left[ \frac{1}{T^2}\sum_{t=1}^{T}\sum_{s=1}^{T}\mathbf{X}_{J,t}^\top\mathbf{M}_{\widehat{\pmb{\Lambda}}}\mathbf{X}_{J,s}a_{st}\right] (\widehat{\pmb{\beta}}_J - \pmb{\beta}_{0,J}).
\end{eqnarray*}
For $\mathbf{K}_{15}$, by using Cauchy-Schwarz inequality and $|\mathbf{M}_{\widehat{\pmb{\Lambda}}}|_2 = 1$, we have
\begin{eqnarray*}
|\mathbf{K}_{15}|_F &=& \left|\frac{1}{NT^2}\sum_{t=1}^{T}\sum_{s=1}^{T}\mathbf{X}_{J,s}^\top\mathbf{M}_{\widehat{\pmb{\Lambda}}}\mathbf{X}_{J,t}(\widehat{\pmb{\beta}}_J - \pmb{\beta}_{0,J})\left(\frac{\mathbf{e}_t^\top\widehat{\pmb{\Lambda}}}{N} \right) \mathbf{G}\mathbf{f}_{0s}\right|_F\\
&\leq&|\mathbf{G}|_F\left(\frac{1}{T\sqrt{N}}\sum_{s=1}^{T}\left|\mathbf{X}_{J,s}\right|_F|\mathbf{f}_{0s}|_{F} \right) \left\{\frac{1}{NT}\sum_{t=1}^{T}\left|\mathbf{X}_{J,t}(\widehat{\pmb{\beta}}_J - \pmb{\beta}_{0,J})\right|_F^2\right\}^{1/2}\left\{\frac{1}{T}\sum_{t=1}^{T}\left|\frac{\mathbf{e}_t^\top\widehat{\pmb{\Lambda}}}{N} \right|_F^2\right\}^{1/2}\\
&\leq&O_P(\sqrt{s}|\widehat{\pmb{\beta}}_J - \pmb{\beta}_{0,J}|_F)\left(\frac{1}{T}\sum_{t=1}^{T}\left|\frac{\mathbf{e}_t^\top\pmb{\Lambda}_{0}}{N}\right|_F+\frac{1}{T}\sum_{t=t}^{T}\left|\mathbf{e}_t/\sqrt{N}\right|_F\left|(\widehat{\pmb{\Lambda}}-\pmb{\Lambda}_{0}\mathbf{H})/\sqrt{N}\right|_F\right)\\
&=&O_P(\sqrt{s}\delta_{NT}^{-1}|\widehat{\pmb{\beta}}_J - \pmb{\beta}_{0,J}|_F+\sqrt{s}|\widehat{\pmb{\beta}}_J - \pmb{\beta}_{0,J}|_F^2).
\end{eqnarray*}
For $\mathbf{K}_{16}$, since $\mathbf{M}_{\widehat{\pmb{\Lambda}}}\widehat{\pmb{\Lambda}} = \mathbf{0}$, we have
\begin{eqnarray*}
|\mathbf{K}_{16}|_F &=& \left|\frac{1}{T^2N^2}\sum_{t=1}^{T}\sum_{s=1}^{T}\mathbf{X}_{J,t}^\top\mathbf{M}_{\widehat{\pmb{\Lambda}}}(\pmb{\Lambda}_{0}-\widehat{\pmb{\Lambda}}\mathbf{H}^{-1})\mathbf{f}_{0s}(\widehat{\pmb{\beta}}_J - \pmb{\beta}_{0,J})^\top\mathbf{X}_{J,s}^\top\widehat{\pmb{\Lambda}} \mathbf{G}\mathbf{f}_{0t}\right|_F\\
&\leq&\frac{\sqrt{r}}{\sqrt{N}}|\pmb{\Lambda}_{0}-\widehat{\pmb{\Lambda}}\mathbf{H}^{-1}|_F|\mathbf{G}|_F\left(\frac{1}{T\sqrt{N}}\sum_{t=1}^{T}\left|\mathbf{X}_{J,t}\right|_F|\mathbf{f}_{0t}|_{F} \right) \left(\frac{1}{T\sqrt{N}}\sum_{s=1}^{T}|\mathbf{X}_{J,s}(\widehat{\pmb{\beta}}_J - \pmb{\beta}_{0,J})|_F\left|\mathbf{f}_{0s}\right|_F\right)\\
&=&O_P(\sqrt{s}|\widehat{\pmb{\beta}}_J - \pmb{\beta}_{0,J}|_F^2 + \sqrt{s}|\widehat{\pmb{\beta}}_J - \pmb{\beta}_{0,J}|_F\delta_{NT}^{-1}).
\end{eqnarray*}
For $\mathbf{K}_{17}$, by using Cauchy-Schwarz inequality
\begin{eqnarray*}
|\mathbf{K}_{17}|_F&=& \left|\frac{1}{N^2T^2}\sum_{s=1}^{T}\left(\sum_{t=1}^{T}\mathbf{X}_{J,t}^\top\mathbf{M}_{\widehat{\pmb{\Lambda}}}\mathbf{e}_{s}\mathbf{f}_{0t}^\top\right)(\mathbf{F}_{0}^\top\mathbf{F}_{0}/T)^{-1}(\widehat{\pmb{\Lambda}}^\top\pmb{\Lambda}_{0}/N)^{-1}\widehat{\pmb{\Lambda}}^\top\mathbf{X}_{J,s}(\widehat{\pmb{\beta}}_J - \pmb{\beta}_{0,J}) \right|_F\\
&\leq&O_P(1)\frac{1}{T}\sum_{s=1}^{T}\left|\frac{1}{NT}\sum_{t=1}^{T}\mathbf{X}_{J,t}^\top\mathbf{M}_{\widehat{\pmb{\Lambda}}}\mathbf{e}_{s}\mathbf{f}_{0t}^\top\right|_F\left|\frac{1}{\sqrt{N}}\mathbf{X}_{J,s}(\widehat{\pmb{\beta}}_J - \pmb{\beta}_{0,J}) \right|_F \\
&\leq&O_P(1)\left\{\frac{1}{T}\sum_{s=1}^{T}\left|\frac{1}{NT}\sum_{t=1}^{T}\mathbf{X}_{J,t}^\top\mathbf{M}_{\widehat{\pmb{\Lambda}}}\mathbf{e}_{s}\mathbf{f}_{0t}^\top\right|_F^2\right\}^{1/2}\left\{\frac{1}{NT}\sum_{s=1}^{N}\left|\mathbf{X}_{J,s}(\widehat{\pmb{\beta}}_J - \pmb{\beta}_{0,J}) \right|_F^2\right\}^{1/2}\\
&=& O_P(|\widehat{\pmb{\beta}}_J - \pmb{\beta}_{0,J}|_F)\left\{\frac{1}{T}\sum_{s=1}^{T}\left|\frac{1}{NT}\sum_{t=1}^{T}\mathbf{X}_{J,t}^\top\mathbf{M}_{\widehat{\pmb{\Lambda}}}\mathbf{e}_{s}\mathbf{f}_{0t}^\top\right|_F^2\right\}^{1/2}.
\end{eqnarray*} 
Next, consider $\frac{1}{NT}\sum_{t=1}^{T}\mathbf{X}_{J,t}^\top\mathbf{M}_{\widehat{\pmb{\Lambda}}}\mathbf{e}_{s}\mathbf{f}_{0t}^\top$, by using $|\mathbf{P}_{\pmb{\Lambda}_{0}}-\mathbf{P}_{\widehat{\pmb{\Lambda}}}|_F = O_P(|\widehat{\pmb{\beta}}_J - \pmb{\beta}_{0,J}|_F+\delta_{NT}^{-1})$, we have
\begin{eqnarray*}
&&\left|\frac{1}{NT}\sum_{t=1}^{T}\mathbf{X}_{J,t}^\top\mathbf{M}_{\widehat{\pmb{\Lambda}}}\mathbf{e}_{s}\mathbf{f}_{0t}^\top\right|_F\\
&\leq&\left|\frac{1}{T\sqrt{N}}\sum_{t=1}^{T}\frac{1}{\sqrt{N}}\sum_{i=1}^{N}\mathbf{x}_{J,it}{e}_{is}\mathbf{f}_{0t}^\top\right|_F +  \left|\frac{1}{T\sqrt{N}}\sum_{t=1}^{T}\frac{\mathbf{X}_{J,t}^\top\pmb{\Lambda}_{0}}{N}(\pmb{\Lambda}_{0}^\top\pmb{\Lambda}_{0}/N)^{-1}\frac{\pmb{\Lambda}_{0}^\top\mathbf{e}_{s}}{\sqrt{N}}\mathbf{f}_{0t}^\top\right|_F \\
&& + \left|\frac{1}{NT}\sum_{t=1}^{T}\mathbf{X}_{J,t}^\top(\mathbf{P}_{\widehat{\pmb{\Lambda}}}-\mathbf{P}_{\pmb{\Lambda}_{0}})\mathbf{e}_{s}\mathbf{f}_{0t}^\top\right|_F\\
&\leq&O_P(\sqrt{s}/\sqrt{N}) + O_P(\sqrt{s/N}) + \sqrt{s}O_P(|\widehat{\pmb{\beta}}_J - \pmb{\beta}_{0,J}|_F+\delta_{NT}^{-1}).
\end{eqnarray*}
Hence, we have $|\mathbf{K}_{17}|_F=O_P(\sqrt{s}|\widehat{\pmb{\beta}}_J - \pmb{\beta}_{0,J}|_F^2 + \sqrt{s}|\widehat{\pmb{\beta}}_J - \pmb{\beta}_{0,J}|_F\delta_{NT}^{-1})$.

For $\mathbf{K}_{18}$, using $\mathbf{M}_{\widehat{\pmb{\Lambda}}}\pmb{\Lambda}_{0}=\mathbf{M}_{\widehat{\pmb{\Lambda}}}(\pmb{\Lambda}_{0} - \widehat{\pmb{\Lambda}}\mathbf{H}^{-1})$, we have
\begin{eqnarray*}
\left|\mathbf{K}_{18}\right|&\leq&O_P(1)\frac{1}{T}\sum_{t=1}^{T}\left|\mathbf{X}_{J,t}/\sqrt{N}\right|_F|\mathbf{f}_{0t}|_FN^{-1/2}|\pmb{\Lambda}_{0} - \widehat{\pmb{\Lambda}}\mathbf{H}^{-1}|_F\left|\frac{1}{NT}\sum_{s=1}^{T} \mathbf{f}_{0s}\mathbf{e}_s^\top\widehat{\pmb{\Lambda}}\right|_F\\
&\leq&\sqrt{s}O_P(|\widehat{\pmb{\beta}}_J - \pmb{\beta}_{0,J}|_F+\delta_{NT}^{-1})\left|\frac{1}{NT}\sum_{s=1}^{T} \mathbf{f}_{0s}\mathbf{e}_s^\top\widehat{\pmb{\Lambda}}\right|_F.
\end{eqnarray*}
In addition, by using Lemma \ref{L.B1} (9), 
\begin{eqnarray*}
\frac{1}{NT}\sum_{s=1}^{T} \mathbf{f}_{0s}\mathbf{e}_s^\top\widehat{\pmb{\Lambda}}&=&\frac{1}{NT}\sum_{s=1}^{T} \mathbf{f}_{0s}\mathbf{e}_s^\top\pmb{\Lambda}_{0}\mathbf{H} + \frac{1}{NT}\sum_{s=1}^{T} \mathbf{f}_{0s}\mathbf{e}_s^\top(\widehat{\pmb{\Lambda}}-\pmb{\Lambda}_{0}\mathbf{H}) \\
&=& O_P(1/\sqrt{NT}) + O_P(\sqrt{s/(NT)}|\widehat{\pmb{\beta}}_J - \pmb{\beta}_{0,J}|_F+T^{-1}+T^{-1/2}\delta_{NT}^{-2}).
\end{eqnarray*}
Hence, $\left|\mathbf{K}_{18}\right|_F = O_P(\sqrt{s}\delta_{NT}^{-2}|\widehat{\pmb{\beta}}_J - \pmb{\beta}_{0,J}|_F + s/\sqrt{NT}|\widehat{\pmb{\beta}}_J - \pmb{\beta}_{0,J}|_F^2 + \sqrt{s}\delta_{NT}^{-3})$. The term $\mathbf{K}_{19}$ is simply $\mathbf{K}_{19} = - \frac{1}{NT^2}\sum_{t=1}^T\sum_{s=1}^Ta_{st}  \mathbf{X}_{J,t}^\top\mathbf{M}_{\widehat{\pmb{\Lambda}}}\mathbf{e}_{s}$. 

For $\mathbf{K}_{20}$, by part (1), we have
\begin{eqnarray*}
\mathbf{K}_{20}&=& - \frac{1}{N^2T}\sum_{t=1}^T \mathbf{X}_{J,t}^\top\mathbf{M}_{\widehat{\pmb{\Lambda}}}\pmb{\Omega}_e \widehat{\pmb{\Lambda}}\mathbf{G}\mathbf{f}_{0t} \\
&& -  \frac{1}{N^2T^2}\sum_{t=1}^T\sum_{s=1}^T \mathbf{X}_{J,t}^\top\mathbf{M}_{\widehat{\pmb{\Lambda}}}(\mathbf{e}_s\mathbf{e}_s^\top -\pmb{\Omega}_e) \widehat{\pmb{\Lambda}}\mathbf{G}\mathbf{f}_{0t} \\
&=& - \frac{1}{N^2T}\sum_{t=1}^T \mathbf{X}_{J,t}^\top\mathbf{M}_{\widehat{\pmb{\Lambda}}}\pmb{\Omega}_e \widehat{\pmb{\Lambda}}\mathbf{G}\mathbf{f}_{0t}\\
&& +O_P\left( \sqrt{s/(NT)}|\widehat{\pmb{\beta}}_J-\pmb{\beta}_{0,J}|_F +\sqrt{s/T}|\widehat{\pmb{\beta}}_J-\pmb{\beta}_{0,J}|_F^2 + \sqrt{s/(NT)}\delta_{NT}^{-1} \right). 
\end{eqnarray*}

Combining the above analyses, we have 
\begin{eqnarray*}
\frac{1}{NT}\sum_{t=1}^{T}\mathbf{X}_{J,t}^\top \mathbf{M}_{\widehat{\pmb{\Lambda}}}\pmb{\Lambda}_{0}\mathbf{f}_{0t}
&=& \frac{1}{N} \left[ \frac{1}{T^2}\sum_{t=1}^{T}\sum_{s=1}^{T}\mathbf{X}_{J,t}^\top\mathbf{M}_{\widehat{\pmb{\Lambda}}}\mathbf{X}_{J,s}a_{st}\right] (\widehat{\pmb{\beta}}_J - \pmb{\beta}_{0,J}) \\
&&- \frac{1}{NT^2}\sum_{t=1}^T\sum_{s=1}^Ta_{st}  \mathbf{X}_{J,t}^\top\mathbf{M}_{\widehat{\pmb{\Lambda}}}\mathbf{e}_{s} - \frac{1}{N^2T}\sum_{t=1}^T \mathbf{X}_{J,t}^\top\mathbf{M}_{\widehat{\pmb{\Lambda}}}\pmb{\Omega}_e \widehat{\pmb{\Lambda}}\mathbf{G}\mathbf{f}_{0t} \\
&&+ O_P(\sqrt{s}|\widehat{\pmb{\beta}}_J - \pmb{\beta}_{0,J}|_F^2 + \sqrt{s}\delta_{NT}^{-1}|\widehat{\pmb{\beta}}_J - \pmb{\beta}_{0,J}|_F + \sqrt{s}\delta_{NT}^{-3}).
\end{eqnarray*}
In addition, by using the equality $\frac{1}{N}\sum_{k=1}^{T}a_{tk}a_{sk} = a_{ts}$, we have
$$
\frac{1}{NT}\sum_{t=1}^{T}\mathbf{X}_{J,t}^\top\mathbf{M}_{\widehat{\pmb{\Lambda}}}\mathbf{X}_{J,t} - \frac{1}{N} \left[ \frac{1}{T^2}\sum_{t=1}^{T}\sum_{s=1}^{T}\mathbf{X}_{J,t}^\top\mathbf{M}_{\widehat{\pmb{\Lambda}}}\mathbf{X}_{J,s}a_{st}\right] = \frac{1}{NT}\sum_{t=1}^{T}\widetilde{\mathbf{X}}_{J,t}^\top\mathbf{M}_{\widehat{\pmb{\Lambda}}}\widetilde{\mathbf{X}}_{J,t}= \mathbf{D}(\widehat{\pmb{\Lambda}}),
$$
which implies that
\begin{eqnarray*}
\sqrt{NT}(\widehat{\pmb{\beta}}_J - \pmb{\beta}_{0,J}) &=& \mathbf{D}^{-1}(\widehat{\pmb{\Lambda}}) \frac{1}{\sqrt{NT}}\sum_{t=1}^{T}\left[\mathbf{X}_{J,t}^\top\mathbf{M}_{\widehat{\pmb{\Lambda}}} - \frac{1}{T}\sum_{s=1}^{T}a_{st}\mathbf{X}_{J,s}^\top\mathbf{M}_{\widehat{\pmb{\Lambda}}}\right] \mathbf{e}_t + \sqrt{\frac{T}{N}}\pmb{\zeta}_{NT}\\ 
&&+ O_P(\sqrt{sNT}|\widehat{\pmb{\beta}}_J - \pmb{\beta}_{0,J}|_F^2 + \sqrt{sNT}\delta_{NT}^{-1}|\widehat{\pmb{\beta}}_J - \pmb{\beta}_{0,J}|_F + \sqrt{sNT}\delta_{NT}^{-3}).
\end{eqnarray*}

\smallskip

\noindent (3). Consider $\frac{1}{\sqrt{NT}}\sum_{t=1}^{T}\mathbf{X}_{J,t}^\top(\mathbf{M}_{\pmb{\Lambda}_{0}}-\mathbf{M}_{\widehat{\pmb{\Lambda}}})\mathbf{e}_{t}$, write
\begin{eqnarray*}
&&\frac{1}{\sqrt{NT}}\sum_{t=1}^{T}\mathbf{X}_{J,t}^\top(\mathbf{M}_{\pmb{\Lambda}_{0}}-\mathbf{M}_{\widehat{\pmb{\Lambda}}})\mathbf{e}_{t}\\
&=&\frac{1}{\sqrt{NT}}\sum_{t=1}^{T}\frac{\mathbf{X}_{J,t}^\top(\widehat{\pmb{\Lambda}}-\pmb{\Lambda}_{0}\mathbf{H})}{N}\mathbf{H}^\top\pmb{\Lambda}_{0}^\top\mathbf{e}_{t} + \frac{1}{\sqrt{NT}}\sum_{t=1}^{T}\frac{\mathbf{X}_{J,i}^\top(\widehat{\pmb{\Lambda}}-\pmb{\Lambda}_{0}\mathbf{H})}{N}(\widehat{\pmb{\Lambda}}-\pmb{\Lambda}_{0}\mathbf{H})^\top\mathbf{e}_{t} \\
&& + \frac{1}{\sqrt{NT}}\sum_{t=1}^{T}\frac{\mathbf{X}_{J,t}^\top \pmb{\Lambda}_{0}\mathbf{H}}{N}(\widehat{\pmb{\Lambda}}-\pmb{\Lambda}_{0}\mathbf{H})^\top\mathbf{e}_{t} + \frac{1}{\sqrt{NT}}\sum_{t=1}^{T}\frac{\mathbf{X}_{J,t}^\top \pmb{\Lambda}_{0}}{N}(\mathbf{H}\mathbf{H}^\top-(\pmb{\Lambda}_{0}^\top\pmb{\Lambda}_{0}/N)^{-1})\pmb{\Lambda}_{0}^\top\mathbf{e}_{t} \\
&\eqqcolon & \mathbf{K}_{21} + \mathbf{K}_{22} + \mathbf{K}_{23} + \mathbf{K}_{24}.
\end{eqnarray*}

For $\mathbf{K}_{21}$, by using Cauchy-Schwarz inequality, we have
\begin{eqnarray*}
|\mathbf{K}_{21}|_F &=& \left|\frac{1}{N}\sum_{i=1}^{N} \left(\frac{1}{\sqrt{NT}}\sum_{j=1}^{N}\sum_{t=1}^{T}\mathbf{x}_{J,it}e_{jt}\pmb{\lambda}_{0j}^\top \right)\mathbf{H}(\widehat{\pmb{\lambda}}_i - \mathbf{H}^\top\pmb{\lambda}_{0i})\right|_F \\
&\leq& \left\{\frac{1}{N}\sum_{i=1}^{N}\left|\frac{1}{\sqrt{NT}}\sum_{j=1}^{N}\sum_{t=1}^{T}\mathbf{x}_{J,it}e_{jt}\pmb{\lambda}_{0j}^\top\right|_F^2\right\}^{1/2}\left\{\frac{1}{N}\sum_{i=1}^{N}\left|\widehat{\pmb{\lambda}}_i - \mathbf{H}^\top\pmb{\lambda}_{0i}\right|_F^2\right\}^{1/2} |\mathbf{H}|_F \\
&=&\sqrt{s} O_P(|\widehat{\pmb{\beta}}_J - \pmb{\beta}_{0,J}|_F + \delta_{NT}^{-1}).
\end{eqnarray*}
Similarly, for $\mathbf{K}_{22}$, by using Cauchy-Schwarz inequality, we have
\begin{eqnarray*}
|\mathbf{K}_{22}|_F &=& \left|\sqrt{N}\frac{1}{N^2}\sum_{i=1}^{N}\sum_{j=1}^{N} (\widehat{\pmb{\lambda}}_i - \mathbf{H}^\top\pmb{\lambda}_{0i})^\top(\widehat{\pmb{\lambda}}_j - \mathbf{H}^\top\pmb{\lambda}_{0j})\left(\frac{1}{\sqrt{T}}\sum_{t=1}^{T}\mathbf{x}_{J,it}e_{jt}\right)\right|_F \\
&\leq&\sqrt{N}\left(\frac{1}{N}\sum_{i=1}^{N}\left|\widehat{\pmb{\lambda}}_i - \mathbf{H}^\top\pmb{\lambda}_{0i}\right|_F^2\right) \left\{\frac{1}{N^2}\sum_{i=1}^{N}\sum_{j=1}^{N}\left|\frac{1}{\sqrt{T}}\sum_{t=1}^{T}\mathbf{x}_{J,it}e_{jt}\right|_F^2\right\}^{1/2} \\
&=& \sqrt{Ns}O_P(|\widehat{\pmb{\beta}}_J - \pmb{\beta}_{0,J}|_F^2 + \delta_{NT}^{-2}).
\end{eqnarray*}
For $\mathbf{K}_{23}$, write
\begin{eqnarray*}
\mathbf{K}_{23} &=& \frac{1}{\sqrt{NT}}\sum_{t=1}^{T}\frac{\mathbf{X}_{J,t}^\top\pmb{\Lambda}_{0}}{N}\mathbf{H}\mathbf{H}^\top (\widehat{\pmb{\Lambda}}\mathbf{H}^{-1} - \pmb{\Lambda}_{0})^\top\mathbf{e}_t  \\
&=&\frac{1}{\sqrt{NT}}\sum_{t=1}^{T}\frac{\mathbf{X}_{J,t}^\top\pmb{\Lambda}_{0}}{N}\left(\pmb{\Lambda}_{0}^\top\pmb{\Lambda}_{0}/N\right)^{-1} (\widehat{\pmb{\Lambda}}\mathbf{H}^{-1} - \pmb{\Lambda}_{0})^\top\mathbf{e}_t  \\
&& + \frac{1}{\sqrt{NT}}\sum_{t=1}^{T}\frac{\mathbf{X}_{J,t}^\top\pmb{\Lambda}_{0}}{N}\left(\mathbf{H}\mathbf{H}^\top - \left(\pmb{\Lambda}_{0}^\top\pmb{\Lambda}_{0}/N\right)^{-1}\right) (\widehat{\pmb{\Lambda}}\mathbf{H}^{-1} - \pmb{\Lambda}_{0})^\top\mathbf{e}_t \\
&\eqqcolon &\mathbf{K}_{23,1} + \mathbf{K}_{23,2}.
\end{eqnarray*}
For $\mathbf{K}_{23,1}$, by using Lemma \ref{L.B1} (10), we have
\begin{eqnarray*}
\mathbf{K}_{23,1}&=& \sqrt{\frac{N}{T}}\frac{1}{NT}\sum_{t=1}^{T}\sum_{s=1}^{T}(\mathbf{X}_{J,t}^\top\pmb{\Lambda}_{0}/N)(\pmb{\Lambda}_{0}^\top\pmb{\Lambda}_{0}/N)\times(\mathbf{F}_{0}^\top\mathbf{F}_{0}/T)^{-1}\mathbf{f}_{0s}\left(\sum_{i=1}^{N}e_{is}e_{it}\right) \\
&& + O_P(\sqrt{s}|\widehat{\pmb{\beta}}_J - \pmb{\beta}_{0,J}|_F + \sqrt{Ns}\delta_{NT}^{-2}).
\end{eqnarray*}
For $\mathbf{K}_{23,2}$, by using Lemma \ref{L.B1} (3) and the arguments of Lemma \ref{L.B1}  (9), we have
\begin{eqnarray*}
\mathbf{K}_{23,2}&=& \sqrt{NT}\left(\frac{1}{NT}\sum_{t=1}^{T}\mathbf{e}_t^\top(\widehat{\pmb{\Lambda}}\mathbf{H}^{-1} - \pmb{\Lambda}_{0})\otimes\left(\frac{\mathbf{X}_{J,t}^\top\pmb{\Lambda}_{0}}{N}\right) \right) \mathrm{vec}\left(\mathbf{H}\mathbf{H}^\top -(\pmb{\Lambda}_{0}^\top\pmb{\Lambda}_{0}/N)^{-1}\right)\\
&=& \sqrt{NT}O_P\left(\sqrt{s/(NT)}|\widehat{\pmb{\beta}}_J - \pmb{\beta}_{0,J}|_F + \sqrt{s}/T+ \sqrt{s/T}\delta_{NT}^{-2} \right)O_P\left(|\widehat{\pmb{\beta}}_J - \pmb{\beta}_{0,J}|_F + \delta_{NT}^{-2} \right).
\end{eqnarray*}
Similarly, by using Lemma \ref{L.B1} (3), we have $\mathbf{K}_{24} = O_P\left(|\widehat{\pmb{\beta}}_J - \pmb{\beta}_{0,J}|_F + \delta_{NT}^{-2} \right)$. In summary, we have
\begin{eqnarray*}
&&\frac{1}{\sqrt{NT}}\sum_{t=1}^{T}\mathbf{X}_{J,t}^\top(\mathbf{M}_{\pmb{\Lambda}_{0}}-\mathbf{M}_{\widehat{\pmb{\Lambda}}})\mathbf{e}_{t}\\
&=& \sqrt{\frac{N}{T}}\frac{1}{NT}\sum_{t=1}^{T}\sum_{s=1}^{T}(\mathbf{X}_{J,t}^\top\pmb{\Lambda}_{0}/N)(\pmb{\Lambda}_{0}^\top\pmb{\Lambda}_{0}/N)\times(\mathbf{F}_{0}^\top\mathbf{F}_{0}/T)^{-1}\mathbf{f}_{0s}\left(\sum_{i=1}^{N}e_{is}e_{it}\right)\\
&& +  O_P(\sqrt{s}|\widehat{\pmb{\beta}}_J - \pmb{\beta}_{0,J}|_F +  \sqrt{Ns}|\widehat{\pmb{\beta}}_J - \pmb{\beta}_{0,J}|_F^2 + \sqrt{Ns}\delta_{NT}^{-2}).
\end{eqnarray*}

Let $\mathbf{V}_{J,t} = \frac{1}{T}\sum_{s=1}^{T}a_{ts}\mathbf{X}_{J,s}$. Then replacing $\mathbf{X}_{J,s}$ with $\mathbf{V}_{J,s}$ and using same arguments as above, we have
\begin{eqnarray*}
&&\frac{1}{\sqrt{NT}}\sum_{t=1}^{T}\mathbf{V}_{J,t}^\top(\mathbf{M}_{\pmb{\Lambda}_{0}}-\mathbf{M}_{\widehat{\pmb{\Lambda}}})\mathbf{e}_{t}\\
&=&\sqrt{\frac{N}{T}}\frac{1}{NT}\sum_{t=1}^{T}\sum_{s=1}^{T}(\mathbf{X}_{J,t}^\top\pmb{\Lambda}_{0}/N)(\pmb{\Lambda}_{0}^\top\pmb{\Lambda}_{0}/N)\times(\mathbf{F}_{0}^\top\mathbf{F}_{0}/T)^{-1}\mathbf{f}_{0s}\left(\sum_{i=1}^{N}e_{is}e_{it}\right)\\
&& +  O_P(\sqrt{s}|\widehat{\pmb{\beta}}_J - \pmb{\beta}_{0,J}|_F +  \sqrt{Ns}|\widehat{\pmb{\beta}}_J - \pmb{\beta}_{0,J}|_F^2 + \sqrt{Ns}\delta_{NT}^{-2}).
\end{eqnarray*}
Combing the above results, we have completed the proof.

\smallskip

\noindent (4). Part (4) follows directly from parts (1)-(3). 
\end{proof}

\begin{proof}[Proof of Lemma \ref{L.B3}]
\item
\noindent (1). Write
\begin{eqnarray*}
\left|\mathbf{D}(\widehat{\pmb{\Lambda}}) - \mathbf{D}(\pmb{\Lambda}_{0})\right|_F&\leq& \frac{1}{NT}\sum_{t=1}^{T}|\widetilde{\mathbf{X}}_{J,t}|_F^2 \times |\mathbf{P}_{\widehat{\pmb{\Lambda}}} - \mathbf{P}_{\pmb{\Lambda}_{0}}|_F \\
&=& O_P(s|\widehat{\pmb{\beta}}_J - \pmb{\beta}_{0,J}|_F + s\delta_{NT}^{-1}).
\end{eqnarray*}

\smallskip

\noindent (2). Note that $|\mathbf{D}^{-1}(\widehat{\pmb{\Lambda}}) - \mathbf{D}^{-1}(\pmb{\Lambda}_{0})|_F \leq |\mathbf{D}(\widehat{\pmb{\Lambda}}) - \mathbf{D}(\pmb{\Lambda}_{0})|_F\frac{1}{\psi_{\mathrm{min}}(\mathbf{D}(\widehat{\pmb{\Lambda}}))}\frac{1}{\psi_{\mathrm{min}}(\mathbf{D}(\pmb{\Lambda}_{0}))} = O_P(s|\widehat{\pmb{\beta}}_J - \pmb{\beta}_{0,J}|_F + s\delta_{NT}^{-1})$. It is sufficient to show 
$$
\frac{1}{NT}\sum_{t=1}^{T}\sum_{s=1}^{T}\frac{\widetilde{\mathbf{X}}_{J,t}^\top\pmb{\Lambda}_{0}}{N}\left(\frac{\pmb{\Lambda}_{0}^\top\pmb{\Lambda}_{0}}{N}\right)^{-1} \left(\frac{\mathbf{F}_{0}^\top\mathbf{F}_{0}}{T}\right)^{-1}\mathbf{f}_{0s}\left(\sum_{i=1}^{N}[e_{it}e_{is}-E(e_{it}e_{is})]\right) = O_P(\sqrt{s/N}),
$$
which can be proved by using Assumption C (iv) of \cite{bai2009panel}.

\smallskip

\noindent (3). It suffices to show that 
\begin{eqnarray*}
&&\frac{1}{NT}\sum_{t=1}^{T}\mathbf{X}_{J,t}^\top\mathbf{M}_{\widehat{\pmb{\Lambda}}}\pmb{\Omega}_e\widehat{\pmb{\Lambda}}\mathbf{G}\mathbf{f}_{0t}\\
&=&\frac{1}{NT}\sum_{t=1}^{T}\mathbf{X}_{J,t}^\top\mathbf{M}_{\pmb{\Lambda}_{0}}\pmb{\Omega}_e\pmb{\Lambda}_{0}(\pmb{\Lambda}_{0}^\top\pmb{\Lambda}_{0}/N)^{-1}(\mathbf{F}_{0}^\top\mathbf{F}_{0}/T)^{-1}\mathbf{f}_{0t} + O_P(\sqrt{s}|\widehat{\pmb{\beta}}_J - \pmb{\beta}_{0,J}|_F + \sqrt{s}\delta_{NT}^{-1}).
\end{eqnarray*}
By adding and subtracting terms, we first consider
\begin{eqnarray*}
&&\left|\frac{1}{NT}\sum_{t=1}^{T}\mathbf{X}_{J,t}^\top(\mathbf{M}_{\widehat{\pmb{\Lambda}}} - \mathbf{M}_{\pmb{\Lambda}_{0}})\pmb{\Omega}_e\widehat{\pmb{\Lambda}}(\pmb{\Lambda}_{0}^\top\widehat{\pmb{\Lambda}}/N)^{-1}(\mathbf{F}_{0}^\top\mathbf{F}_{0}/T)^{-1}\mathbf{f}_{0t} \right|_F\\
&\leq& O_P(1)\frac{1}{T}\sum_{t=1}^{T}|\mathbf{X}_{J,t}/\sqrt{N}|_F|\mathbf{f}_{0t}|_F \times |\mathbf{P}_{\widehat{\pmb{\Lambda}}} - \mathbf{P}_{\pmb{\Lambda}_{0}}|_F \\
&=& O_P(\sqrt{s}|\widehat{\pmb{\beta}}_J - \pmb{\beta}_{0,J}|_F + \sqrt{s}\delta_{NT}^{-1}).
\end{eqnarray*}
Next, note that $|\mathbf{M}_{\pmb{\Lambda}_{0}}\pmb{\Omega}_e|_2 = O(1)$, we have
\begin{eqnarray*}
&&\left|\frac{1}{NT}\sum_{t=1}^{T}\mathbf{X}_{J,t}^\top\mathbf{M}_{\pmb{\Lambda}_{0}}\pmb{\Omega}_e\left[\widehat{\pmb{\Lambda}}\left(\frac{\pmb{\Lambda}_{0}^\top\widehat{\pmb{\Lambda}}}{N}\right)^{-1} - \pmb{\Lambda}_{0}\left(\frac{\pmb{\Lambda}_{0}^\top\pmb{\Lambda}_{0}}{N}\right)^{-1} \right] \left(\frac{\mathbf{F}_{0}^\top\mathbf{F}_{0}}{T}\right)^{-1}\mathbf{f}_{0t}\right|_F\\
&\leq& O_P(1\frac{1}{T}\sum_{t=1}^{T}|\mathbf{X}_{J,t}/\sqrt{N}|_F|\mathbf{f}_{0t}|_F \times N^{-1/2}\left| \widehat{\pmb{\Lambda}}\left(\frac{\pmb{\Lambda}_{0}^\top\widehat{\pmb{\Lambda}}}{N}\right)^{-1} - \pmb{\Lambda}_{0}\left(\frac{\pmb{\Lambda}_{0}^\top\pmb{\Lambda}_{0}}{N}\right)^{-1}\right|_F\\
&=&O_P(\sqrt{s}|\widehat{\pmb{\beta}}_J - \pmb{\beta}_{0,J}|_F + \sqrt{s}\delta_{NT}^{-1})
\end{eqnarray*}
since
\begin{eqnarray*}
&&N^{-1/2}\left( \widehat{\pmb{\Lambda}}\left(\frac{\pmb{\Lambda}_{0}^\top\widehat{\pmb{\Lambda}}}{N}\right)^{-1} - \pmb{\Lambda}_{0}\left(\frac{\pmb{\Lambda}_{0}^\top\pmb{\Lambda}_{0}}{N}\right)^{-1}\right)\\
&=&N^{-1/2}\left(\widehat{\pmb{\Lambda}} - \mathbf{P}_{\pmb{\Lambda}_{0}}\widehat{\pmb{\Lambda}}\right)\left(\frac{\pmb{\Lambda}_{0}^\top\widehat{\pmb{\Lambda}}}{N}\right)^{-1} = (\mathbf{P}_{\widehat{\pmb{\Lambda}}}-\mathbf{P}_{\pmb{\Lambda}_{0}}) N^{-1/2}\widehat{\pmb{\Lambda}} \left(\frac{\pmb{\Lambda}_{0}^\top\widehat{\pmb{\Lambda}}}{N}\right)^{-1}\\
&=&O_P(|\widehat{\pmb{\beta}}_J - \pmb{\beta}_{0,J}|_F+\delta_{NT}^{-1}).
\end{eqnarray*}
\end{proof}

\begin{proof}[Proof of Lemma \ref{L.B4}]
\item 
\noindent (1). By using $\frac{1}{N}\widehat{\pmb{\Lambda}}^\top \widehat{\pmb{\Lambda}} = \mathbf{I}_r$, we have
$$
\widehat{\mathbf{f}}_t - \mathbf{H}^{-1}\mathbf{f}_{0t} = \frac{1}{N}\widehat{\pmb{\Lambda}}^\top \left(\pmb{\Lambda}_0 - \widehat{\pmb{\Lambda}}^\top\mathbf{H}^{-1}\right)\mathbf{f}_{0t} + \frac{1}{N}\widehat{\pmb{\Lambda}}^\top\mathbf{e}_t + \frac{1}{N}\widehat{\pmb{\Lambda}}^\top\mathbf{X}_{J,t}^\top(\pmb{\beta}_{0,J} - \widehat{\pmb{\beta}}_J),
$$
which follows that
\begin{eqnarray*}
\frac{1}{T}\sum_{t=1}^{T}|\widehat{\mathbf{f}}_t - \mathbf{H}^{-1}\mathbf{f}_{0t}|_F^2&\leq&\frac{3}{T}\sum_{t=1}^{T}|\mathbf{f}_{0t}|_F^2 \times \left|\frac{1}{N}\widehat{\pmb{\Lambda}}^\top \left(\pmb{\Lambda}_0 - \widehat{\pmb{\Lambda}}^\top\mathbf{H}^{-1}\right)\right|_F^2 + \frac{3}{T}\sum_{t=1}^{T}\left|\frac{1}{N}\widehat{\pmb{\Lambda}}^\top\mathbf{e}_t\right|_F^2\\
&& + \frac{3}{T}\sum_{t=1}^{T}|N^{-1/2}\mathbf{X}_{J,t}|_F^2 \left|N^{-1/2}\widehat{\pmb{\Lambda}}\right|_F^2 |\pmb{\beta}_{0,J} - \widehat{\pmb{\beta}}_J|_F^2 \\
&=& O_P(s/(NT) + \delta_{NT}^{-4}) + O_P(1/N + s/(NT) + \delta_{NT}^{-2}) + O_P(s^2/(NT))
\end{eqnarray*}
provided that
\begin{eqnarray*}
K_{26}&\leq& \frac{2}{T}\sum_{t=1}^{T}\left|\frac{1}{N}\pmb{\Lambda}_0^\top\mathbf{e}_t\right|_F^2 \times |\mathbf{H}|_F^2 + \frac{2}{T}\sum_{t=1}^{T}|N^{-1/2}\mathbf{e}_t|_F^2\left|N^{-1/2}(\widehat{\pmb{\Lambda}}-\mathbf{H}\pmb{\Lambda}_0)\right|_F^2\\
&=&O_P(1/N + s/(NT) + \delta_{NT}^{-2}).
\end{eqnarray*}

\smallskip

\noindent (2). Write
\begin{eqnarray*}
&&\left|\frac{1}{T}\sum_{t=1}^{T}(\widehat{\mathbf{f}}_t - \mathbf{H}^{-1}\mathbf{f}_{0t})\mathbf{f}_{0t}^\top\right|_F\\
&\leq&\frac{1}{T}\sum_{t=1}^{T}|\mathbf{f}_{0t}|_F^2 \times \left|\frac{1}{N}\widehat{\pmb{\Lambda}}^\top \left(\pmb{\Lambda}_0 - \widehat{\pmb{\Lambda}}^\top\mathbf{H}^{-1}\right)\right|_F + \left|\frac{1}{NT}\sum_{t=1}^{T}\widehat{\pmb{\Lambda}}^\top\mathbf{e}_t\mathbf{f}_{0t}^\top\right|_F\\
&& + \frac{1}{T}\sum_{t=1}^{T}|N^{-1/2}\mathbf{X}_{J,t}|_F|\mathbf{f}_{0t}|_F \left|N^{-1/2}\widehat{\pmb{\Lambda}}\right|_F |\pmb{\beta}_{0,J} - \widehat{\pmb{\beta}}_J|_F \\
&=& O_P(\delta_{NT}^{-2}) + O_P(1/N + \delta_{NT}^{-2}) + O_P(s/\sqrt{NT})
\end{eqnarray*}
provided that
\begin{eqnarray*}
\left|\frac{1}{NT}\sum_{t=1}^{T}\widehat{\pmb{\Lambda}}^\top\mathbf{e}_t\mathbf{f}_{0t}^\top\right|_F &\leq& \left|\frac{1}{NT}\sum_{t=1}^{T}\pmb{\Lambda}_0^\top\mathbf{e}_t\mathbf{f}_{0t}^\top\right|_F|\mathbf{H}|_F + \left|\frac{1}{TN^{1/2}}\sum_{t=1}^{T}\mathbf{e}_t\mathbf{f}_{0t}^\top\right|_FN^{-1/2}|\widehat{\pmb{\Lambda}}-\pmb{\Lambda}_0\mathbf{H}|_F\\
&=& O_P\left(1/\sqrt{NT} + T^{-1/2}(\delta_{NT}^{-1} +\sqrt{s/(NT)})\right).
\end{eqnarray*}

\smallskip

\noindent (3)-(4). By using the identity $\widehat{\mathbf{F}} = \widehat{\mathbf{F}} - \mathbf{F}_0 \mathbf{H}^{-1,\top} + \mathbf{F}_0 \mathbf{H}^{-1,\top}$, parts (3)-(4) follow directly from parts (1) and (3).
\end{proof}

}

\end{document}